\documentclass[smallextended,envcountsect]{svjour3}
\usepackage[paperwidth=170mm,paperheight=240mm,textwidth=119mm,textheight=194mm,centering,headheight=12pt,headsep=5mm,footskip=8mm]{geometry}
\usepackage{amsmath,amssymb}
\usepackage{etoolbox}
\usepackage{setspace}
\usepackage{booktabs,graphicx}
\usepackage{subcaption}
\usepackage{placeins}
\usepackage{algorithm}
\usepackage{algorithmic}
\usepackage{tikz}
\usetikzlibrary{positioning,fit,backgrounds,calc,shapes.misc,arrows.meta}
\usepackage[hidelinks]{hyperref}
\hypersetup{pdfauthor={Chang Liu and Kun Xu}}


\newcommand{\bx}{\boldsymbol x}
\newcommand{\bxi}{\boldsymbol\xi}
\newcommand{\bu}{\boldsymbol u}
\newcommand{\bc}{\boldsymbol c}
\newcommand{\bq}{\boldsymbol q}
\newcommand{\bn}{\boldsymbol n}
\newcommand{\bj}{\boldsymbol j}
\newcommand{\bpsi}{\boldsymbol\psi}
\newcommand{\bU}{\mathbf U}
\newcommand{\bF}{\mathbf F}
\newcommand{\bI}{\mathbf I}
\newcommand{\bA}{\mathbf A}
\newcommand{\bM}{\mathbf M}
\newcommand{\avg}[1]{\langle #1\rangle}
\newcommand{\avgh}[1]{\langle #1\rangle_h}
\newcommand{\dd}{\,\mathrm d}
\newcommand{\cD}{\mathcal D}
\newcommand{\cS}{\mathcal S}
\newcommand{\cX}{\mathbb X}
\newcommand{\sNS}{\boldsymbol\sigma_{\rm NS}}
\newcommand{\qNS}{\bq_{\rm NS}}
\newcommand{\Ma}{\mathrm{Ma}}
\newcommand{\Rey}{\mathrm{Re}}
\newcommand{\vk}{\varkappa}
\newcommand{\Fl}[1]{\bF[#1]}
\newlength{\figw}
\extrafloats{150}
\makeatletter
\newcommand{\SMref}[1]{S\@nameuse{smfig@#1}}
\@namedef{smfig@fig:shock_dist}{1}
\@namedef{smfig@fig:couette_residual}{2}
\@namedef{smfig@fig:couette_scan}{3}
\@namedef{smfig@fig:cavity_kn1_fields}{4}
\@namedef{smfig@fig:cavity_kn0075_fields}{5}
\@namedef{smfig@fig:cavity_rarefied_profiles}{6}
\@namedef{smfig@fig:cavity_re100}{7}
\@namedef{smfig@fig:cavity_residual}{8}
\@namedef{smfig@fig:cavity_scan}{9}
\@namedef{smfig@fig:cyl_kn1_profiles}{10}
\@namedef{smfig@fig:cyl_kn1_dist_log}{11}
\@namedef{smfig@fig:cyl_kn1em2_fields}{12}
\@namedef{smfig@fig:cyl_kn1em2_profiles}{13}
\@namedef{smfig@fig:cyl_kn1em4_fields}{14}
\@namedef{smfig@fig:cyl_kn1em4_profiles}{15}
\@namedef{smfig@fig:cyl_residual}{16}
\@namedef{smfig@fig:x38_kn1_fields}{17}
\@namedef{smfig@fig:x38_kn1_profiles}{18}
\@namedef{smfig@fig:x38_kn1_dist_lin}{19}
\@namedef{smfig@fig:x38_kn1_dist_log}{20}
\@namedef{smfig@fig:x38_kn1em2_fields}{21}
\@namedef{smfig@fig:x38_kn1em4_profiles}{22}
\@namedef{smfig@fig:apollo_kn1_fields}{23}
\@namedef{smfig@fig:apollo_kn1_profiles}{24}
\@namedef{smfig@fig:apollo_kn1_dist}{25}
\@namedef{smfig@fig:apollo_kn1em2_fields}{26}
\@namedef{smfig@fig:apollo_kn1em2_profiles}{27}
\@namedef{smfig@fig:apollo_kn1em4_profiles}{28}
\@namedef{smfig@fig:apollo_kn1em4_fields}{29}
\@namedef{smfig@fig:apollo_residual}{30}
\makeatother

\journalname{Journal of Scientific Computing}

\renewcommand{\makeheadbox}{}
\makeatletter
\patchcmd{\@maketitle}{\hrule\@height}{\vskip}{}{\PackageError{title-layout}{Title rule patch failed}{Check the document class title definition.}}
\makeatother
\begin{document}
\raggedbottom
\title{An implicit multiscale solver for steady kinetic equations based on wave--particle decomposition}
\titlerunning{An implicit multiscale solver based on WPD}
\author{Chang Liu \and Kun Xu}
\authorrunning{C. Liu and K. Xu}
\institute{Chang Liu (corresponding author) \at
Institute of Applied Physics and Computational Mathematics, Beijing, P.R. China\\
\email{liuchang@iapcm.ac.cn}
\and Kun Xu \at
Department of Mathematics and Department of Mechanical and Aerospace Engineering, Hong Kong University of Science and Technology, Clear Water Bay, Kowloon, Hong Kong\\
Shenzhen Research Institute, Hong Kong University of Science and Technology, Shenzhen, China\\
\email{makxu@ust.hk}}
\date{}
\maketitle

\begin{abstract}
Steady kinetic computations require accurate transport on meshes that do not resolve the mean free path and efficient removal of slowly decaying macroscopic errors. We develop a closed wave--particle decomposition (WPD) and a residual-based implicit solver for this problem. The wave accumulates only the Maxwellian along a local characteristic horizon and is determined by the conservative variables; the particle carries the kinetic complement and collision remainder. Under moment compatibility, the closed system is equivalent to the kinetic equation. For the Bhatnagar--Gross--Krook and Shakhov models, a hydrodynamic expansion gives the wave and particle contributions to the Euler and Navier--Stokes fluxes. The implicit iteration couples macroscopic prediction to microscopic correction through a discrete endpoint term and synchronizes the conservative moments after each accepted correction. Its stationary states satisfy the target discrete kinetic equation independently of the iteration controls. Formal analysis gives the smooth fixed-mesh fluid limit. Linear model analysis relates convergence to the horizon-to-relaxation ratio, incomplete macroscopic solves and explicit collision remainders. One- and two-dimensional computations of a normal shock, internal flows and hypersonic external flows test these properties across rarefied and continuum regimes. The results show kinetic fidelity and substantial reductions in iteration counts and wall time relative to the conventional iterative scheme near the continuum limit. Collision-model comparisons and parameter studies support the predicted convergence mechanisms.

\end{abstract}

\section{Introduction}\label{sec:intro}
Steady kinetic computations must resolve non-equilibrium distributions in rarefied regions while recovering fluid transport where the mean free path is much smaller than the mesh size. Collision-dominated regions also produce slowly decaying macroscopic errors in conventional kinetic iterations. We address these two requirements through a closed wave--particle decomposition and an implicit macroscopic--microscopic iteration.

Direct simulation Monte Carlo (DSMC) is widely used for rarefied gas flows, but collision sampling and statistical noise become costly near the continuum limit~\cite{bird1994}. Hybrid particle--continuum methods couple DSMC and fluid solvers across selected interfaces~\cite{schwartzentruber2006}. Deterministic discrete-velocity methods~\cite{yang1995,mieussens2000} and discrete-ordinates methods~\cite{adams2002} avoid sampling noise, but their conventional source iterations can converge slowly when collisions dominate transport.

Asymptotic-preserving (AP) schemes recover a consistent macroscopic discretization at fixed mesh parameters as the kinetic scale vanishes~\cite{jin1999}. Implicit--explicit (IMEX) integration and collision penalization provide AP schemes without a fully implicit nonlinear kinetic solve~\cite{pareschi2005,dimarco2013imex,filbet2010}; Navier--Stokes consistency additionally requires control of the first-order nonequilibrium contribution~\cite{bennoune2008}. For steady problems, diffusion synthetic acceleration~\cite{adams2002}, high-order/low-order coupling~\cite{chacon2017holo} and the general synthetic iterative scheme (GSIS)~\cite{su2020,su2020sisc,zhu2021gsis,su2021cmame} use macroscopic corrections to accelerate slowly decaying kinetic errors. GSIS retains explicit Navier--Stokes terms and distribution-dependent higher-order corrections in its macroscopic equations. Newton-based iterations provide another route to steady Boltzmann solutions~\cite{zhang2026newton}.

The unified gas-kinetic scheme (UGKS) couples transport and collisions in the numerical flux through a local integral solution~\cite{xu2010,xu2014direct}; DUGKS uses a characteristic reconstruction at cell interfaces~\cite{guo2013dugks}. Implicit UGKS and related discrete-velocity solvers combine macroscopic prediction with distribution updates~\cite{zhu2016implicit,zhu2019implicit,xu2022iugks,yuan2021multi}. The unified gas-kinetic wave--particle method (UGKWP) combines an analytic wave with sampled particles and reduces particle work near equilibrium~\cite{liu2020ugkwp,zhu2019ugkwp,li2020ugkwp,guo2026representation}. These methods motivate deriving the macroscopic closure and implicit coupling from a common wave--particle system at the equation level.

The WPD formulation of Liu and Xu~\cite{liu_wpd_I} separates the characteristic integral solution into an accumulated relaxation target and its kinetic complement, establishing the kernel identities, wave equation and an explicit scheme. Accumulating a non-BGK target makes the wave depend on non-conserved moments, such as the Shakhov heat flux. Here the wave accumulates only the Maxwellian and the particle retains the collision remainder. For BGK the two accumulation choices coincide.

The first contribution is a closed WPD system for the conservative variables and particle distribution. We establish its equivalence with the kinetic equation and derive the hydrodynamic expansion of both components. Their contributions recover the Euler and Navier--Stokes fluxes, with the ratio $\eta=\ell/\tau$ controlling the decomposition. The formulation admits conservative collision operators; the solver and computations here use the BGK and Shakhov models.

The second contribution is a residual-based implicit iteration derived from the closed WPD system. Macroscopic prediction enters the microscopic correction through a discrete endpoint term; each accepted correction is followed by moment synchronization and reconstruction of the decomposition. The implicit WPD solver is denoted IWPD, with WP retained as the compact label in figures and subscripts. We prove stationary-state consistency under discrete moment compatibility and examine the formal fluid limit. Linear model analysis connects error decay to the particle weight $(1+\eta)e^{-\eta}$, incomplete macroscopic solves and explicit collision remainders.

Shock and internal-flow profiles test kinetic fidelity and hydrodynamic behaviour. Residual histories, collision-model comparisons and parameter scans examine the convergence mechanisms. Cylinder and X38 flows extend the tests to body-fitted meshes, and cavity results on the grids of a published GSIS study provide an additional method comparison.

Section~\ref{sec:theory} presents the closed WPD system and its hydrodynamic expansion. Section~\ref{sec:algorithm} constructs the implicit iteration, and Section~\ref{sec:analysis} examines stationary-state consistency and convergence. Section~\ref{sec:results} reports the numerical verification, followed by the conclusions.

\section{Closed wave--particle decomposition}\label{sec:theory}
We define the wave and particle components, assemble the closed WPD system and establish kinetic equivalence. Its hydrodynamic expansion then identifies how the two components contribute to the macroscopic fluxes. The kernel identities and wave equation follow~\cite{liu_wpd_I}.

Throughout, $\cD=\partial_t+\bxi\cdot\nabla_{\bx}$ denotes the transport derivative, $\avg{\cdot}$ the velocity integral, $\bpsi$ the collision invariants, $\Fl{Z}=\avg{\bpsi\bxi^TZ}$ the flux tensor of a distribution $Z$, and a superscript $h$ a discrete operator.

\subsection{Kinetic equation and wave--particle components}\label{sec:model}
We consider a monatomic gas described by the distribution function $f(\bx,\bxi,t)$ of position $\bx$, molecular velocity $\bxi\in\mathbb R^3$ and time $t$, nondimensionalized so that the pressure of a gas with density $\rho$ and temperature $T$ is $p=\rho T$. The distribution satisfies a kinetic equation
\begin{equation}\label{eq:kinetic}
 \cD f=Q[f],\qquad \avg{\bpsi Q[f]}=\boldsymbol 0,
\end{equation}
whose collision operator $Q$ conserves mass, momentum and energy. Conservation is sufficient for the exact decomposition below; the fluid expansion additionally requires a collision model with the stated Chapman--Enskog limit. For a conservative state $\bU=(\rho,\rho\bu,\rho E)^T$ let $M[\bU]=\rho(2\pi T)^{-3/2}e^{-|\bc|^2/(2T)}$ be its Maxwellian, so that $\avg{\bpsi M[\bU]}=\bU$, where $\bpsi=(1,\bxi^T,|\bxi|^2/2)^T$, $\bu$ is the macroscopic velocity and $\bc=\bxi-\bu$ the peculiar velocity. With any state $\bU$ and any positive relaxation time $\tau[\bU]$ the collision term splits into a relaxation part and a remainder,
\begin{equation}\label{eq:general_collision}
 Q[f]=\frac{M-f}{\tau}+R,\qquad R=Q[f]-\frac{M-f}{\tau},\qquad \avg{\bpsi R}=\frac{\avg{\bpsi f}-\bU}{\tau},
\end{equation}
with $M=M[\bU]$, and the remainder has zero conservative moments when $\bU=\avg{\bpsi f}$. The collision model enters the wave--particle equations only through $R$. The Boltzmann collision integral is the prime example, and the kinetic models of relaxation type, for which $R$ is known in closed form, are the models treated below. They relax the distribution toward a target $S$ that depends on $f$ only through its moments,
\begin{equation}\label{eq:relaxation}
 Q[f]=\frac{S-f}{\tau},\qquad \tau=\frac{\mu(T)}{p},\qquad R=\frac{S-M}{\tau},
\end{equation}
where $\mu$ is the viscosity, and the target has the density, momentum and energy of $f$ and equals the local Maxwellian $M=M[\avg{\bpsi f}]$ when $f$ is a Maxwellian. The BGK model~\cite{bgk1954} takes $S=M$, so that $R=0$, and the Shakhov model~\cite{shakhov1968}
\begin{equation}\label{eq:shakhov}
 S=M\Big[1+(1-\Pr)\frac{\bc\cdot\bq}{5pT}\Big(\frac{|\bc|^2}{T}-5\Big)\Big],\qquad \bq=\frac12\avg{|\bc|^2\bc f},
\end{equation}
corrects the heat flux $\bq$ so that the Navier--Stokes limit has the Prandtl number $\Pr$; it reduces to the BGK model for $\Pr=1$. When the flow depends on only $d<3$ spatial coordinates, the molecular velocity splits into $d$ resolved components $\bxi_r$ and $K=3-d$ transverse components $\bxi_t$, and the transverse components can be integrated out~\cite{chu1965}. Writing $f^{\rm 3D}$ for the original distribution, this yields the reduced pair
\begin{equation}\label{eq:reduced}
 f=(f_1,f_2),\qquad f_1=\int_{\mathbb R^K}f^{\rm 3D}\dd\bxi_t,\qquad f_2=\int_{\mathbb R^K}|\bxi_t|^2f^{\rm 3D}\dd\bxi_t,
\end{equation}
of functions of $\bxi_r\in\mathbb R^d$, which we denote again by $\bxi$. The conservative variables and the heat flux of a pair are
\begin{equation}\label{eq:moments}
 \begin{aligned}
 &\bU=\avg{\bpsi f}=\int_{\mathbb R^d}\Big[\big(1,\bxi,\tfrac12|\bxi|^2\big)^Tf_1+\big(0,\boldsymbol 0,\tfrac12\big)^Tf_2\Big]\dd\bxi=(\rho,\rho\bu,\rho E)^T,\\
 &\bq=\frac12\int_{\mathbb R^d}\bc\big(|\bc|^2f_1+f_2\big)\dd\bxi,
 \end{aligned}
\end{equation}
with $E=|\bu|^2/2+3T/2$. For reduced pairs, $\bpsi$ denotes the linear moment map displayed in \eqref{eq:moments}; products with $\bpsi$ below use this map, rather than applying the full three-dimensional invariant vector to each component separately. The pair satisfies two relaxation equations \eqref{eq:relaxation}, $\cD f_r=(S_r-f_r)/\tau$ for $r=1,2$, with the reduced Maxwellian $M=(M_1,KTM_1)$, where $M_1=\rho(2\pi T)^{-d/2}e^{-|\bc|^2/(2T)}$, and with the reduced target $S=M$ of the BGK model or the reduced Shakhov target
\begin{equation}\label{eq:reduced_target}
 \begin{aligned}
 &S_1=M_1\Big[1+\frac{(1-\Pr)\bc\cdot\bq}{5pT}\Big(\frac{|\bc|^2}{T}+K-5\Big)\Big],\\
 &S_2=KTM_1\Big[1+\frac{(1-\Pr)\bc\cdot\bq}{5pT}\Big(\frac{|\bc|^2}{T}+K-3\Big)\Big].
 \end{aligned}
\end{equation}
For every $\bq$ the pair $S-M$ has zero conservative moments and the heat flux of $S$ is $(1-\Pr)\bq$. The remainder $R=(S-M)/\tau$ is thus the Prandtl-number correction of the Shakhov model, and the model-specific formulas below reduce to BGK when $\Pr=1$. We write $S=S[f]$ when the dependence of the target on the moments of $f$ matters.

Let $\varepsilon$ denote the Knudsen number and write $\tau=\varepsilon\bar\tau$ with $\bar\tau=\mathcal O(1)$. In a smooth region the Chapman--Enskog~\cite{chapman1970} expansion of the reduced model is
\begin{equation}\label{eq:ce}
 f=M+\varepsilon f^{(1)}+\mathcal O(\varepsilon^2),\qquad f^{(1)}=S^{(1)}-\bar\tau\,\cD^{(0)}M,\qquad \cD^{(0)}=\partial_t^{(0)}+\bxi\cdot\nabla_{\bx},
\end{equation}
where $\partial_t^{(0)}$ is the time derivative determined by the Euler equations and $S-M=\varepsilon S^{(1)}+\mathcal O(\varepsilon^2)$, so that $\varepsilon S^{(1)}$ is the first-order part of $\tau R$. The moments of \eqref{eq:ce} are the Navier--Stokes equations $\partial_t\bU+\nabla_{\bx}\cdot(\bF_{\rm E}+\bF_{\rm R}+\bF_{\rm S})=\boldsymbol 0$, in which the Euler flux $\bF_{\rm E}$, the relaxation flux $\bF_{\rm R}$ and the correction flux $\bF_{\rm S}$ are the fluxes of $M$, of $-\varepsilon\bar\tau\cD^{(0)}M$ and of $\varepsilon S^{(1)}$, respectively,
\begin{equation}\label{eq:ns_split}
 \begin{aligned}
 &\bF_{\rm E}=\begin{pmatrix}(\rho\bu)^T\\ \rho\bu\bu^T+p\bI\\ ((\rho E+p)\bu)^T\end{pmatrix},\\
 &\bF_{\rm R}=\begin{pmatrix}\boldsymbol0^T\\ \sNS\\ (\sNS\bu+\Pr\,\qNS)^T\end{pmatrix},\qquad
 \bF_{\rm S}=\begin{pmatrix}\boldsymbol0^T\\ \mathbf0\\ ((1-\Pr)\qNS)^T\end{pmatrix},
 \end{aligned}
\end{equation}
with the viscous stress and the Fourier heat flux
\begin{equation}\label{eq:ns_closure}
 \sNS=-\mu\Big[\nabla_{\bx}\bu+(\nabla_{\bx}\bu)^T-\tfrac23(\nabla_{\bx}\cdot\bu)\bI\Big],\qquad
 \qNS=-\kappa\nabla_{\bx}T,\qquad \kappa=\frac{5\mu}{2\Pr}.
\end{equation}
The relaxation flux thus carries the full viscous stress and the fraction $\Pr$ of the heat flux, and the correction flux carries the remaining fraction. For the BGK model the correction flux vanishes and the relaxation flux carries the complete Navier--Stokes fluxes with unit Prandtl number.

\paragraph{Wave and particle components}\label{sec:wave_particle}
Let $\bx_s=\bx-\bxi s$, $t_s=t-s$ be the backward characteristic through $(\bx,t)$, and let a subscript $s$ denote evaluation at $(\bx_s,t_s)$. Along the characteristic we introduce the cumulative collision frequency and the collision kernel
\begin{equation}\label{eq:kernel}
 \nu(s)=\int_0^s\frac{\dd r}{\tau_r},\qquad \mathcal K(s)=\frac{e^{-\nu(s)}}{\tau_s}.
\end{equation}
For a relaxation model, $e^{-\nu(s)}$ is the collisionless survival factor and $\mathcal K(s)$ is its collision kernel. For a general collision operator, these are the integrating factors of the chosen relaxation splitting. Let $\ell(\bx,t)\ge0$ be the local kinetic horizon, measured in time. Written as $\cD f=(M-f)/\tau+R$, the kinetic equation has along the characteristic the integral form $f=\int_0^\ell(M+\tau R)_s\mathcal K(s)\dd s+e^{-\nu(\ell)}f_\ell$, and splitting its source into the Maxwellian and the remainder defines the wave and the particle,
\begin{equation}\label{eq:wp_def}
 W=\int_0^{\ell}M_s\mathcal K(s)\dd s,\qquad P=f-W=e^{-\nu(\ell)}f_\ell+\int_0^{\ell}(\tau R)_s\mathcal K(s)\dd s.
\end{equation}
Here $\tau R=S-M$ for a relaxation model. The wave is a functional $W[\bU]$ on the backward characteristic segment because $M$ and $\tau$ depend on the conservative state. The particle carries the attenuated incoming distribution and the accumulated collision remainder (Figure~\ref{fig:schematic}(a)). The ratio $\eta=\ell/\tau$ controls the decomposition.

\begin{figure}[!htb]
\centering
\begin{subfigure}[t]{\linewidth}\centering
\begin{tikzpicture}[font=\footnotesize, >={Stealth[length=2.0mm,width=1.5mm]}, x=1.42cm, y=3.3cm]
  \useasboundingbox (-0.45,-0.50) rectangle (4.75,1.30);
  \def\eh{1.3}
  \fill[blue!13] (0,0) -- plot[domain=0:\eh, samples=40, smooth] (\x,{exp(-\x)}) -- (\eh,0) -- cycle;
  \fill[orange!20] (\eh,0) -- plot[domain=\eh:4.2, samples=50, smooth] (\x,{exp(-\x)}) -- (4.2,0) -- cycle;
  \draw[thick] plot[domain=0:4.2, samples=80, smooth] (\x,{exp(-\x)});
  \draw[dashed, black!70] (\eh,0) -- (\eh,0.90);
  \draw[->] (0,0) -- (4.5,0) node[below, xshift=-1mm] {$s/\tau$};
  \draw[->] (0,0) -- (0,1.14) node[above, xshift=7mm] {$\tau \mathcal K(s)=e^{-s/\tau}$};
  \draw (0,0.02) -- (0,-0.02) node[below] {$0$};
  \draw (\eh,0.02) -- (\eh,-0.02) node[below] {$\eta=\ell/\tau$};
  \draw (0.03,1) -- (-0.03,1) node[left] {$1$};
  \node[text=blue!60!black, align=center] at (0.60,0.21) {wave\\ $1-e^{-\eta}$};
  \node[text=orange!80!black, align=center] (pl) at (2.75,0.50) {particle\\ $e^{-\eta}$};
  \draw[->, orange!80!black] (pl.south west) -- (1.95,0.075);
  \node[black!75] at (\eh,0.98) {horizon};
  \draw[<-, black!65] (0,-0.29) -- (4.2,-0.29);
  \fill[black!65] (0,-0.29) circle (1.1pt);
  \node[black!75, anchor=north west, inner sep=1pt] at (0,-0.33) {$(\bx,t)$};
  \node[black!75, anchor=north east, inner sep=1pt] at (4.2,-0.33) {$\bx_s=\bx-\bxi s$};
\end{tikzpicture}
\caption{Split of the integral solution at the horizon.}
\end{subfigure}\par\medskip
\begin{subfigure}[t]{\linewidth}\centering
\begin{tikzpicture}[font=\footnotesize, >={Stealth[length=2.0mm,width=1.5mm]}, x=0.74cm, y=3.3cm]
  \useasboundingbox (-0.75,-0.50) rectangle (9.2,1.30);
  \def\ec{1.678}
  \draw[black!25] (0,0.5) -- (8,0.5);
  \draw[dashed, black!70] (\ec,0) -- (\ec,0.5);
  \draw[thick, blue!60!black, densely dashed] plot[domain=0:8, samples=80, smooth] (\x,{1-exp(-\x)});
  \draw[thick, blue!60!black] plot[domain=0:8, samples=80, smooth] (\x,{1-(1+\x)*exp(-\x)});
  \draw[thick, orange!80!black] plot[domain=0:8, samples=80, smooth] (\x,{(1+\x)*exp(-\x)});
  \draw[->] (0,0) -- (8.6,0) node[below, xshift=-1mm] {$\eta$};
  \draw[->] (0,0) -- (0,1.14);
  \foreach \t in {0,2,4,6,8} \draw (\t,0.02) -- (\t,-0.02) node[below] {$\t$};
  \draw (0.06,1) -- (-0.06,1) node[left] {$1$};
  \draw (0.06,0.5) -- (-0.06,0.5) node[left] {$\tfrac12$};
  \node[black!75, anchor=west, inner sep=1.5pt] at (\ec,0.11) {$\eta_c$};
  \draw[thick, blue!60!black, densely dashed] (4.25,0.74) -- (5.05,0.74) node[right, text=blue!60!black] {$A_0=1-e^{-\eta}$};
  \draw[thick, blue!60!black] (4.25,0.57) -- (5.05,0.57) node[right, text=blue!60!black] {$A_1=1-\beta$};
  \draw[thick, orange!80!black] (4.25,0.40) -- (5.05,0.40) node[right, text=orange!80!black] {$\beta=(1+\eta)e^{-\eta}$};
  \node[black!75, anchor=north west, inner sep=1pt] at (0,-0.33) {kinetic};
  \node[black!75, anchor=north east, inner sep=1pt] at (8.4,-0.33) {Navier--Stokes};
  \draw[<->, black!65] (0,-0.29) -- (8.4,-0.29);
\end{tikzpicture}
\caption{Transport shares of wave and particle.}
\end{subfigure}
\caption{The wave--particle decomposition. (a) For a frozen relaxation time the kernel of the integral solution is the probability density of the time $s$ elapsed since the last collision. The horizon splits it into the fraction $1-e^{-\eta}$ that defines the wave and the fraction $e^{-\eta}$ that the particle carries together with the remainder of the collision term. (b) Near equilibrium (the transport shares \eqref{eq:shares}) the wave carries the fractions $A_0$ of the Euler flux and $A_1$ of the relaxation flux, and the particle share $\beta$ of the relaxation transport falls from one in the kinetic limit to zero in the Navier--Stokes limit, with $\beta=1/2$ at $\eta_c\approx1.68$.}
\label{fig:schematic}
\end{figure}
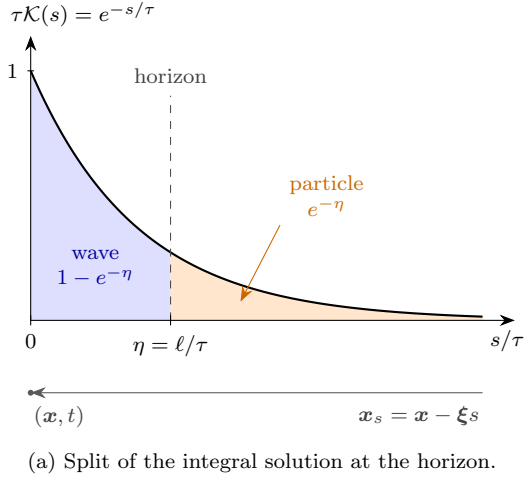
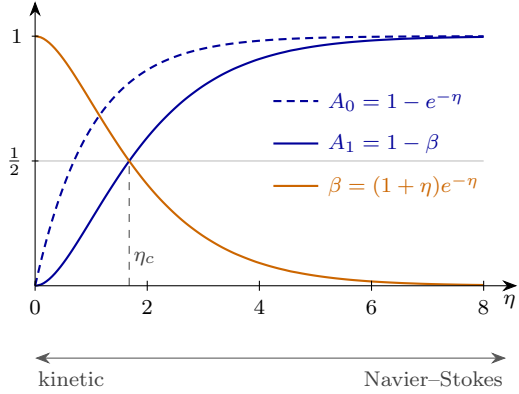

\subsection{Closed equations and kinetic equivalence}\label{sec:closed_system}
For a component $Z\in\{W,P\}$ let $\bU_Z=\avg{\bpsi Z}$ be its conservative state, and let
\begin{equation}\label{eq:component_moments}
 \begin{aligned}
 &\rho_Z=\int_{\mathbb R^d}Z_1\dd\bxi,\qquad
 \bj_Z=\int_{\mathbb R^d}\bc\,Z_1\dd\bxi,\qquad
 p_Z=\frac13\int_{\mathbb R^d}\big(|\bc|^2Z_1+Z_2\big)\dd\bxi,\\
 &\boldsymbol\sigma_Z=\int_{\mathbb R^d}\bc\bc^TZ_1\dd\bxi-p_Z\bI,\qquad
 \bq_Z=\frac12\int_{\mathbb R^d}\bc\big(|\bc|^2Z_1+Z_2\big)\dd\bxi
 \end{aligned}
\end{equation}
be its density, diffusion flux relative to $\bu$, pressure, stress and heat flux. Since $f=W+P$, the moments of $f$ are the sums of those of the components, so that $\bj_W+\bj_P=\boldsymbol 0$, while $\boldsymbol\sigma=\boldsymbol\sigma_W+\boldsymbol\sigma_P$ and $\bq=\bq_W+\bq_P$ are the stress and the heat flux of $f$. 

Since $W=W[\bU]$, the total conservation law and the particle equation close in $(\bU,P)$. Written out, the WPD multiscale equations are
\begin{equation}\label{eq:closed}
\boxed{\begin{aligned}
 &\partial_t\rho+\nabla_{\bx}\cdot(\rho\bu)=0,\\
 &\partial_t(\rho\bu)+\nabla_{\bx}\cdot\big(\rho\bu\bu^T+p\bI
       +\boldsymbol\sigma_W+\boldsymbol\sigma_P\big)=\boldsymbol0,\\
 &\partial_t(\rho E)+\nabla_{\bx}\cdot\big[(\rho E+p)\bu
       +(\boldsymbol\sigma_W+\boldsymbol\sigma_P)\bu+\bq_W+\bq_P\big]=0,\\
 &\partial_tP+\bxi\cdot\nabla_{\bx}P
       =R+\frac{P_e-P}{\tau}.
\end{aligned}}
\end{equation}
The closure is explicit:
\begin{equation}\label{eq:wp_closure}
\begin{aligned}
 &W[\bU]=\int_0^\ell\frac{M[\bU_s]}{\tau_s}e^{-\nu(s)}\dd s,
 \qquad f=W[\bU]+P,\qquad \bU=\avg{\bpsi f},\\
 &P_e=\frac{\tau}{\tau_\ell}M[\bU_\ell]e^{-\nu(\ell)}
       \big(1-\partial_t\ell-\bxi\cdot\nabla_{\bx}\ell\big),\\
 &R=Q[W+P]-\frac{M[\bU]-W-P}{\tau}.
\end{aligned}
\end{equation}
For a relaxation model, $R=(S[W+P]-M[\bU])/\tau$. Here $\tau=\tau[\bU]$, and the wave and particle stresses and heat fluxes are the moments in \eqref{eq:component_moments}. Thus the first three equations contain a wave flux determined by the macroscopic state and a correction supplied by $P$; the last equation retains the velocity-resolved kinetic information. The conservative moment constraint in \eqref{eq:wp_closure} is essential for steady solution and is enforced by the acceptance step of the implicit algorithm. The collision remainder $R$ is distinguished throughout from the iteration residuals $R_U$ and $R_P$ introduced in Section~\ref{sec:residuals}.

\paragraph{Regularity conditions}
\begin{enumerate}\renewcommand{\theenumi}{(A\arabic{enumi})}\renewcommand{\labelenumi}{\theenumi}
\item\label{A1} The reduced distribution $f$ is continuously differentiable along characteristics, has $\rho>0$, $T>0$ and finite moments $\avg{(1+|\bxi|^3)|f_r|}<\infty$, $r=1,2$, and $\tau$ satisfies $0<\tau_{\min}\le\tau\le\tau_{\max}$ and is continuously differentiable along characteristics.
\item\label{A2} The horizon $\ell\ge0$ is Lipschitz continuous, so that $\cD\ell$ exists almost everywhere, and the backward characteristic segment of duration $\ell$ lies in the domain where $f$ is defined; where it would leave the domain through the initial time or an inflow boundary at the flight time $s_b<\ell$, the horizon is replaced by $s_b$, which depends on $\bxi$, and $f_\ell$ by the prescribed data.
\item\label{A3} In a near-equilibrium region the expansion \eqref{eq:ce} holds, $\partial_tM=\partial_t^{(0)}M+\mathcal O(\varepsilon)$, and the second derivative of $M$ along the characteristic within the horizon is bounded by a function $C(\bxi)$ with $\avg{(1+|\bxi|^3)C}<\infty$.
\end{enumerate}

\begin{lemma}[Kernel identities]\label{lem:kernel}
Under \ref{A1}, for fixed $s$ the kernel satisfies
\begin{equation}\label{eq:kernel_identity}
 \cD \mathcal K(s)+\partial_s\mathcal K(s)=-\frac{\mathcal K(s)}{\tau},\qquad \int_0^{\ell}\mathcal K(s)\dd s=1-e^{-\nu(\ell)}.
\end{equation}
Here $\cD$ acts on $(\bx,t)$ at fixed $s$. The kernel also gives the componentwise bound $0\le W\le(1-e^{-\nu(\ell)})\sup_{0\le s\le\ell}M_s$.
\end{lemma}
\begin{proof}
The identity follows from $\cD a_s=-\partial_s a_s$ along the backward characteristic, which gives $\cD\nu(s)=\tau^{-1}-\tau_s^{-1}$. The integral and positivity bounds follow directly from the nonnegative kernel.
\end{proof}

\begin{lemma}[Wave and particle equations]\label{lem:wave_particle}
Let \ref{A1}--\ref{A2} hold and let $\bU$ be a continuously differentiable field. Then $W=W[\bU]$ satisfies the wave equation
\begin{equation}\label{eq:wave_equation}
 \cD W=\frac{M-W}{\tau}-\cS_\ell,\qquad \cS_\ell=\frac{M_\ell}{\tau_\ell}e^{-\nu(\ell)}\big(1-\cD\ell\big),
\end{equation}
identically in $\bU$, where $M=M[\bU]$ and $\tau=\tau[\bU]$. If moreover $f$ solves \eqref{eq:kinetic} and $\bU=\avg{\bpsi f}$, the particle $P=f-W$ satisfies
\begin{equation}\label{eq:particle_equation}
 \cD P=R+\frac{P_e-P}{\tau},\qquad P_e=\tau\cS_\ell ,
\end{equation}
with the remainder $R$ of \eqref{eq:general_collision}.
\end{lemma}
\begin{proof}
Differentiating \eqref{eq:wp_def}, using \eqref{eq:kernel_identity} and integrating by parts gives
\begin{equation*}
 \cD W=\frac{M-W}{\tau}-M_\ell\mathcal K(\ell)(1-\cD\ell).
\end{equation*}
Subtracting this identity from $\cD f=(M-f)/\tau+R$ yields \eqref{eq:particle_equation}.
\end{proof}

The wave equation holds identically for $W[\bU]$. The exchange $\cS_\ell$ accounts for the moving endpoint of the characteristic integral, and $P_e=\tau\cS_\ell$ is the endpoint state in the particle relaxation term. Both are nonnegative where $1-\cD\ell\ge0$ and vanish at a boundary-limited flight time $\ell=s_b$, for which $\cD s_b=1$.

\begin{theorem}[Kinetic equivalence of the closed WPD system]\label{thm:equivalence}
Under \ref{A1}--\ref{A2}, every solution $f$ of \eqref{eq:kinetic} defines a solution $(\bU,P)$ of \eqref{eq:closed}--\eqref{eq:wp_closure} through $\bU=\avg{\bpsi f}$ and $P=f-W[\bU]$. Conversely, every solution of this closed system satisfying the moment constraint reconstructs a solution $f=W[\bU]+P$ of the kinetic equation. This equivalence holds for every admissible horizon and for steady as well as time-dependent solutions.
\end{theorem}
\begin{proof}
For a kinetic solution, subtracting the wave identity \eqref{eq:wave_equation} gives the particle equation, and taking conservative moments gives the first three equations of \eqref{eq:closed}. Conversely, adding the wave identity to the last equation of \eqref{eq:closed} gives $\cD f=(M[\bU]-f)/\tau+R=Q[f]$ by \eqref{eq:wp_closure}. The moment constraint identifies the reconstructed macroscopic state with $\bU$, so its flux has the expanded form in \eqref{eq:closed}.
\end{proof}

The role of the moment constraint is especially important for steady computation. With $W=W[\bU]$, the moments of the steady particle equation already imply $\nabla_{\bx}\cdot(\Fl{W}+\Fl{P})=\boldsymbol0$ by \eqref{eq:moment_equations}, even if the provisional state $\bU$ differs from $\avg{\bpsi f}$. A macroscopic flux balance alone therefore cannot enforce their agreement. In the algorithm it propagates the conservative variables across the domain, while the microscopic correction and acceptance restore moment compatibility.

\begin{corollary}[Horizon invariance]\label{cor:horizon_invariance}
Let $0\le\ell_1\le\ell_2$ be two horizons and $(W_i,P_i)$ the decompositions of the same solution $f$. Then
\begin{equation}\label{eq:horizon_invariance}
 P_1-P_2=W_2-W_1=\int_{\ell_1}^{\ell_2}M_s\mathcal K(s)\dd s\ \ge0,
\end{equation}
and if \eqref{eq:kinetic} has a unique solution for the given data, the solutions of \eqref{eq:closed} for $\ell_1$ and $\ell_2$ with compatible data have the same $\bU$ and the same total distribution.
\end{corollary}
\begin{proof}
The first relation is the definition \eqref{eq:wp_def} applied to the same $f$, and the second follows from Theorem~\ref{thm:equivalence} and uniqueness.
\end{proof}

The equivalence applies to conservative collision operators through the split \eqref{eq:general_collision}. The BGK and Shakhov models provide the relaxation scales and fluid expansions used in the transport analysis and numerical realization below.

For $\eta=0$ the wave vanishes and \eqref{eq:closed} is the kinetic equation for $P=f$, whereas as $\eta\to\infty$ the free-transport term disappears from the particle, which retains only the accumulated remainder and vanishes for the BGK model. In between, $\eta$ parametrizes a continuous family of equivalent representations of the same dynamics, and near equilibrium this spectrum can be made explicit (Figure~\ref{fig:schematic}(b)).

The component fluxes and moment balances are expanded below. Their sum yields the total conservation law when $\bU=\avg{\bpsi f}$,
\begin{equation}\label{eq:total_conservation}
 \begin{aligned}
 &\partial_t\bU+\nabla_{\bx}\cdot\big(\Fl{W}+\Fl{P}\big)=\boldsymbol 0,\\
 &\Fl{W}+\Fl{P}=\begin{pmatrix}(\rho\bu)^T\\ \rho\bu\bu^T+p\bI+\boldsymbol\sigma\\ ((\rho E+p)\bu+\boldsymbol\sigma\bu+\bq)^T\end{pmatrix}.
 \end{aligned}
\end{equation}
The component moments contain their exact stresses and heat fluxes; no fluid constitutive relation has been imposed. The exchange cancels in their sum.

\paragraph{Component fluxes and moment balances}\label{app:component_balance}
For the component moments defined in \eqref{eq:component_moments}, the expanded fluxes retain the diffusion relative to the total velocity and make cancellation of the horizon exchange explicit.

Writing $(\rho E)_Z=\tfrac12\int_{\mathbb R^d}(|\bxi|^2Z_1+Z_2)\dd\bxi$, the state and flux of each component are explicitly
\begin{equation}\label{eq:component_flux}
\begin{aligned}
 \bU_Z&=\big(\rho_Z,(\rho_Z\bu+\bj_Z)^T,(\rho E)_Z\big)^T,\\
 \Fl{Z}&=\begin{pmatrix}
       (\rho_Z\bu+\bj_Z)^T\\
       \rho_Z\bu\bu^T+p_Z\bI+\bu\bj_Z^T+\bj_Z\bu^T+\boldsymbol\sigma_Z\\
       \big[((\rho E)_Z+p_Z)\bu+\boldsymbol\sigma_Z\bu+\bq_Z
                  +\tfrac12|\bu|^2\bj_Z\big]^T
      \end{pmatrix}.
\end{aligned}
\end{equation}
These are the extended fluid fluxes of the two components, including their relative diffusion. Taking the moments of \eqref{eq:wave_equation} and \eqref{eq:particle_equation} with $M=M[\bU]$ and the remainder $R$ of \eqref{eq:general_collision} formed with this state, which for a relaxation model is $(S[f]-M[\bU])/\tau$, and using $\avg{\bpsi M[\bU]}=\bU$, $\avg{\bpsi R}=(\avg{\bpsi f}-\bU)/\tau$ and $\avg{\bpsi P_e}=\tau\avg{\bpsi\cS_\ell}$, we obtain the moment equations
\begin{equation}\label{eq:moment_equations}
 \begin{aligned}
 &\partial_t\bU_W+\nabla_{\bx}\cdot\Fl{W}=\frac{\bU-\bU_W}{\tau}-\avg{\bpsi\cS_\ell},\\
 &\partial_t\bU_P+\nabla_{\bx}\cdot\Fl{P}=\frac{\avg{\bpsi f}-\bU}{\tau}+\avg{\bpsi\cS_\ell}-\frac{\bU_P}{\tau}.
 \end{aligned}
\end{equation}

Adding the two balances gives the total conservation law \eqref{eq:total_conservation} under the moment constraint $\bU=\avg{\bpsi f}$.

\subsection{Wave--particle decomposition in the hydrodynamic limit}\label{sec:spectrum}

\begin{lemma}[Relaxation weights]\label{lem:weights}
For a relaxation time frozen along the characteristic within the horizon and $k\ge0$,
\begin{equation}\label{eq:weights}
 \int_0^{\ell}\frac{s^k}{k!}\,\frac{e^{-s/\tau}}{\tau}\dd s=\tau^kA_k(\eta),\qquad
 A_k(\eta)=1-e^{-\eta}\sum_{j=0}^{k}\frac{\eta^j}{j!},
\end{equation}
and $0\le A_{k+1}(\eta)\le A_k(\eta)\le1$. In particular $A_0=1-e^{-\eta}$, $A_1=1-(1+\eta)e^{-\eta}$ and $A_2=1-(1+\eta+\eta^2/2)e^{-\eta}$.
\end{lemma}
\begin{proof}
Integration by parts after $v=s/\tau$ gives $A_k=A_{k-1}-e^{-\eta}\eta^k/k!$. The ordering follows from the Poisson interpretation: $A_k$ is the probability of at least $k+1$ collisions for mean $\eta$.
\end{proof}

\begin{theorem}[Hydrodynamic expansion of the wave and particle components]\label{thm:spectrum}
Let \ref{A1}--\ref{A3} hold, let the relaxation time be frozen along the characteristic within the horizon, $\tau=\varepsilon\bar\tau$, and let $\eta$ be locally constant. Then, for the BGK and Shakhov models the following equalities hold up to remainders of order $\varepsilon^2$, uniformly in $\eta\in[0,\infty)$:
\begin{equation}\label{eq:spectrum}
 \begin{aligned}
 W&=A_0(\eta)M-\varepsilon A_1(\eta)\,\bar\tau\cD^{(0)}M,\\
 P&=e^{-\eta}M+\varepsilon\big[S^{(1)}-(1+\eta)e^{-\eta}\,\bar\tau\cD^{(0)}M\big],\\
 \Fl{W}&=A_0\bF_{\rm E}+A_1\bF_{\rm R},\\
 \Fl{P}&=e^{-\eta}\bF_{\rm E}+(1+\eta)e^{-\eta}\bF_{\rm R}+\bF_{\rm S},
 \end{aligned}
\end{equation}
and $\bU_W=A_0\bU$, $\bU_P=e^{-\eta}\bU$. To this order, the sum of the two fluxes is the Navier--Stokes flux $\bF_{\rm E}+\bF_{\rm R}+\bF_{\rm S}$ for every $\eta$.
\end{theorem}
\begin{proof}
We integrate the Taylor expansion of the Maxwellian along the characteristic, whose remainder is bounded by \ref{A3}, $|M_s-M+s\,\cD M|\le Cs^2$, against the kernel $e^{-s/\tau}/\tau$ over $[0,\ell]$. By the weights \eqref{eq:weights},
\begin{equation*}
 \big|W-A_0M+A_1\tau\,\cD M\big|\le2C\tau^2A_2(\eta)\le2C\tau^2,
\end{equation*}
uniformly in $\eta$, and by \ref{A3}, $\tau\cD M=\varepsilon\bar\tau\cD^{(0)}M+\mathcal O(\varepsilon^2)$, which gives the expansion of $W$. Subtracting it from \eqref{eq:ce} and using $1-A_0=e^{-\eta}$ and $1-A_1=(1+\eta)e^{-\eta}$ gives the expansion of $P$. The conservative moments of $\cD^{(0)}M$ vanish, because they are the Euler equations, whence $\bU_W=A_0\bU$ and $\bU_P=e^{-\eta}\bU$ up to $\mathcal O(\varepsilon^2)$. Finally, the fluxes of $M$, $-\varepsilon\bar\tau\cD^{(0)}M$ and $\varepsilon S^{(1)}$ are $\bF_{\rm E}$, $\bF_{\rm R}$ and $\bF_{\rm S}$ by Section~\ref{sec:model}.
\end{proof}

\begin{corollary}[Wave and particle transport shares]\label{cor:transport_shares}
Under the hypotheses of Theorem~\ref{thm:spectrum}, with $\tau=\mu/p$ and up to $\mathcal O(\varepsilon^2)$,
\begin{equation}\label{eq:shares}
 \begin{aligned}
 &\boldsymbol\sigma_W=A_1\sNS,\qquad &&\bq_W=A_1\Pr\,\qNS,\\
 &\boldsymbol\sigma_P=\beta\,\sNS,\qquad &&\bq_P=\beta\Pr\,\qNS+(1-\Pr)\qNS,\qquad \beta(\eta)=(1+\eta)e^{-\eta}.
 \end{aligned}
\end{equation}
As $\eta\to0$ the particle carries all three fluxes, whereas as $\eta\to\infty$ the wave carries $\bF_{\rm E}+\bF_{\rm R}$ and the particle only $\bF_{\rm S}$, and the conservation part of \eqref{eq:closed} becomes the Navier--Stokes system with the Prandtl number of the model. The particle share $\beta=1-A_1$ of the relaxation transport decreases from $1$ to $0$ and equals $1/2$ at $\eta_c\approx1.68$.
\end{corollary}
\begin{proof}
These shares follow by taking the stress and heat-flux moments of \eqref{eq:spectrum}; monotonicity follows from $\beta'(\eta)=-\eta e^{-\eta}$.
\end{proof}

Substituting the wave shares into \eqref{eq:closed} gives, through first order in $\varepsilon$,
\begin{equation}\label{eq:macroscopic_spectrum}
\left\{\begin{aligned}
 &\partial_t\rho+\nabla_{\bx}\cdot(\rho\bu)=0,\\
 &\partial_t(\rho\bu)+\nabla_{\bx}\cdot\big(\rho\bu\bu^T+p\bI
       +A_1(\eta)\sNS+\boldsymbol\sigma_P\big)=\boldsymbol0,\\
 &\partial_t(\rho E)+\nabla_{\bx}\cdot\big[(\rho E+p)\bu
       +\big(A_1(\eta)\sNS+\boldsymbol\sigma_P\big)\bu
       +A_1(\eta)\Pr\,\qNS+\bq_P\big]\!=0.
\end{aligned}\right.
\end{equation}
The momentum and energy fluxes in \eqref{eq:macroscopic_spectrum} omit terms of order $\varepsilon^2$ under the smoothness assumptions of Theorem~\ref{thm:spectrum}; the exact equations remain \eqref{eq:closed}. These equations are coupled to the same particle equation. At $\eta=0$, $W=0$, $P_e=M$ and $P=f$, so the kinetic equation carries the full transport. As $\eta\to\infty$ in a smooth continuum region, $A_1\to1$, $\boldsymbol\sigma_P\to\boldsymbol0$ and $\bq_P\to(1-\Pr)\qNS$ to this order, giving the complete Navier--Stokes transport. At intermediate scales the particle moments supply the complement continuously. No interface or switch between kinetic and fluid models is introduced.

The hypothesis of a frozen relaxation time simplifies the weights but is not essential. If $\tau$ varies along the characteristic, the wave acquires the additional term $\tfrac12\eta^2e^{-\eta}(\cD\tau)M$, which vanishes in both limits of $\eta$; the particle compensates it, and the sum of the two fluxes remains the Navier--Stokes flux.

The closed system also determines the implicit coupling. With $P$ fixed, the conservation law advances $\bU$ through the wave flux; with $W$ fixed, the particle equation is a relaxation--transport equation with an explicit remainder. At finite $\eta$, the particle contains the equilibrium share $e^{-\eta}M$, which must follow the macroscopic prediction. The endpoint correction in Section~\ref{sec:residuals} transfers this change to the particle update.

\section{Implicit wave--particle iteration}\label{sec:algorithm}
IWPD advances the discrete closed system by $N_U$ macroscopic updates with the particle fixed, followed by $N_P$ microscopic updates with the wave fixed during each update. These are called W and P updates. Both use delta--residual form. The endpoint correction transfers the W prediction into the first P update, and acceptance restores moment compatibility. Figure~\ref{fig:iwpd_flow} shows this coupling; the velocity-resolved realization uses the BGK and Shakhov models of Section~\ref{sec:model}.

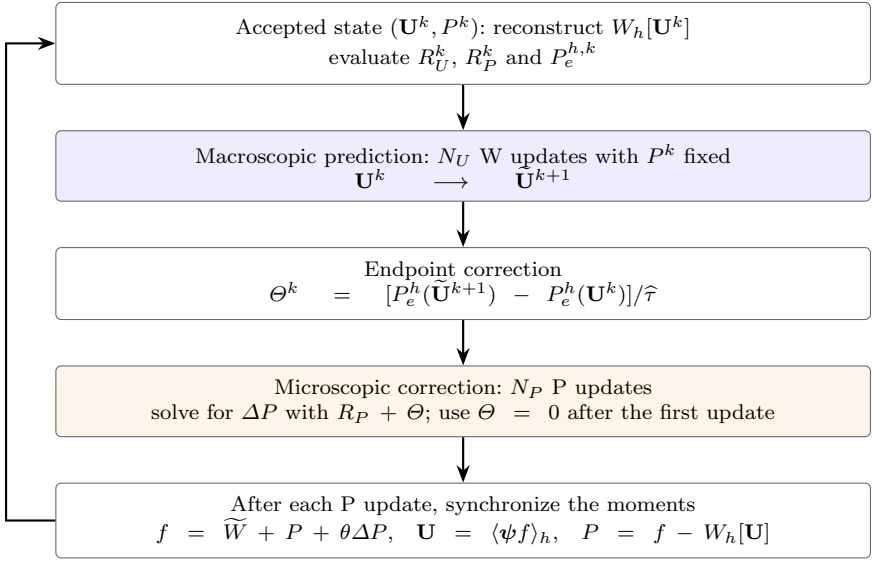
\begin{figure}[!htb]
\centering
\begin{tikzpicture}[font=\small, node distance=6mm,
  box/.style={draw=black!60, rounded corners=2pt, align=center, text width=10.4cm, inner sep=5pt},
  flow/.style={->, thick, >=Stealth}]
\node[box] (state) {Accepted state $(\bU^k,P^k)$: reconstruct $W_h[\bU^k]$\\ evaluate $R_U^k$, $R_P^k$ and $P_e^{h,k}$};
\node[box, fill=blue!7, below=of state] (wave) {Macroscopic prediction: $N_U$ W updates with $P^k$ fixed\\ $\bU^k\longrightarrow\widetilde\bU^{k+1}$};
\node[box, below=of wave] (endpoint) {Endpoint correction\\ $\Theta^k=[P_e^h(\widetilde\bU^{k+1})-P_e^h(\bU^k)]/\widehat\tau$};
\node[box, fill=orange!9, below=of endpoint] (particle) {Microscopic correction: $N_P$ P updates\\ solve for $\Delta P$ with $R_P+\Theta$; use $\Theta=0$ after the first update};
\node[box, below=of particle] (accept) {After each P update, synchronize the moments\\ $f=\widetilde W+P+\theta\Delta P$,\quad $\bU=\avgh{\bpsi f}$,\quad $P=f-W_h[\bU]$};
\draw[flow] (state)--(wave);
\draw[flow] (wave)--(endpoint);
\draw[flow] (endpoint)--(particle);
\draw[flow] (particle)--(accept);
\draw[flow] (accept.west)--++(-0.65,0) |- (state.west);
\end{tikzpicture}
\caption{IWPD coupling. Each outer iteration predicts the macroscopic state, corrects the particle residual through the endpoint, and synchronizes the moments after each microscopic update. Iteration stops when both residual norms satisfy the tolerance.}
\label{fig:iwpd_flow}
\end{figure}

\subsection{Discrete formulation and residuals}\label{sec:setting}
Physical space is covered by finite-volume cells $\Omega_i$, $i\in\mathcal I$, of volume $|\Omega_i|$, centre $\bx_i$ and size $\Delta_i$, the smallest cell width. The faces of $\Omega_i$ form the set $\mathcal F_i$. A face $\Gamma$ has area $|\Gamma|$, midpoint $\bx_\Gamma$ and unit normal $\bn_\Gamma$, outward for $\Omega_i$, and $i_\Gamma$ denotes the cell across $\Gamma$. For a velocity $\bxi$ we write $\xi_n=\bxi\cdot\bn_\Gamma$. The reduced velocity space is represented by nodes $\bxi_j$, $j\in\mathcal J$, with quadrature weights $\varpi_j$, and $\bpsi_j$ is the $(d+2)\times2$ matrix mapping a reduced pair to its conservative moments as in \eqref{eq:moments}. A discrete pair $Z=(Z_{i,j})$ has the discrete moments $\avg{\bpsi Z}_{h,i}=\sum_{j\in\mathcal J}\varpi_j\bpsi_jZ_{i,j}$, and the discrete state space is $\cX=(\mathbb R^{d+2})^{\mathcal I}\times(\mathbb R^2)^{\mathcal I\times\mathcal J}$, whose elements are pairs $X=(\bU,P)$ of cell states and nodal particle values. The discrete Maxwellian $M_i=M_h[\bU_i]$ of a cell is the reduced Maxwellian of the state $\bU_i$ on the velocity nodes; its discrete moments equal $\bU_i$ up to the quadrature error of the velocity grid, and exactly after a correction of the form \eqref{eq:lift}. The discrete target $S^h(f)$ is the target of the model evaluated with the quadrature moments of $f$, that is, $M_h[\avgh{\bpsi f}]$ for the BGK model and \eqref{eq:reduced_target} for the Shakhov model.

Derivatives of Maxwellians are represented by micro-slopes~\cite{xu2001}. For a coefficient vector $\mathbf a=(a_0,\mathbf a_1,a_2)\in\mathbb R^{d+2}$ and a temperature $T$ let
\begin{equation}\label{eq:micro_poly}
 \Phi[\mathbf a]=(\Phi_1[\mathbf a],\Phi_2[\mathbf a]),\qquad
 \Phi_r[\mathbf a]=a_0+\mathbf a_1\cdot\bxi+a_2\Big(\frac{|\bxi|^2}{2}+\frac{(K+2r-2)T}{2}\Big),
\end{equation}
for $r=1,2$, and write $M\Phi[\mathbf a]=(M_1\Phi_1[\mathbf a],M_2\Phi_2[\mathbf a])$. If $\bU$ has the derivative $\partial\bU$ in some direction, the derivative of the reduced Maxwellian in that direction is $M\Phi[\mathbf a]$, where the micro-slope $\mathbf a$ solves the linear system $\avg{\bpsi M\Phi[\mathbf a]}=\partial\bU$ of Gaussian moments. The time slope $\mathbf a_t$ follows from the slopes $\mathbf a_\alpha$ in the mesh directions through the compatibility condition
\begin{equation}\label{eq:gks_compatibility}
 \avg{\bpsi M\big(\Phi[\mathbf a_t]+\textstyle\sum_{\alpha=1}^{d}\xi_\alpha\Phi[\mathbf a_\alpha]\big)}=\boldsymbol 0,
\end{equation}
which takes the time derivative of $\bU$ from the Euler equations.

In cell $i$ the horizon and the horizon-to-relaxation ratio are
\begin{equation}\label{eq:discrete_horizon}
 \ell_i=C_\ell\frac{\Delta_i}{|\bu_i|+a_i},\qquad \eta_i=\frac{\ell_i}{\tau_i},
\end{equation}
with the sound speed $a_i$ and the horizon number $C_\ell\ge0$. The horizon is thus a multiple of the time in which the fastest acoustic signal crosses the cell, and $\eta_i$ is of the order of the cell size measured in mean free paths. Where the mesh resolves the mean free path, $\eta_i$ is small and the particle, then nearly the distribution itself, is advanced as in a discrete-velocity method; where the cells are large compared with the mean free path, the wave carries the transport through a Navier--Stokes flux. The decomposition thus adapts to the local resolution without a switching criterion. The horizon is tied to the cell because the discrete wave represents the Maxwellian along the characteristic by the data of one cell; all computations use $C_\ell=1$, to assess the method with a common horizon convention across the test cases. At a face, $\ell_\Gamma$ is the distance-weighted average of the two adjacent horizons, $\tau_\Gamma$ the relaxation time of the face state and $\eta_\Gamma=\ell_\Gamma/\tau_\Gamma$. The discrete wave is the first-order truncation of \eqref{eq:spectrum} with $\eta$ frozen in the cell,
\begin{equation}\label{eq:discrete_wave}
 W_h[\bU]_{i,j}=M_{i,j}\Big[A_0(\eta_i)-A_1(\eta_i)\,\tau_i\big(\Phi_j[\mathbf a_{t,i}]+\textstyle\sum_{\alpha=1}^{d}\xi_{\alpha,j}\Phi_j[\mathbf a_{\alpha,i}]\big)\Big],
\end{equation}
in which the transport derivative of the Maxwellian is expressed through the micro-slopes $\mathbf a_{\alpha,i}$ of the limited cell gradient of $\bU$ and the time slope $\mathbf a_{t,i}$ of the gas-kinetic scheme~\cite{xu2001}, defined above in \eqref{eq:micro_poly}--\eqref{eq:gks_compatibility}. By the proof of Theorem~\ref{thm:spectrum}, for a relaxation time frozen within the horizon, the truncation error is bounded for smooth solutions by $2C\tau^2A_2(\eta)$, of order $\ell^2\eta$ for small and $\tau^2$ for large $\eta$. Up to the quadrature error, the discrete moments of $M_i$ equal $\bU_i$ and those of the bracketed derivative vanish by the compatibility condition \eqref{eq:gks_compatibility} of the time slope, so that $\avg{\bpsi W_h[\bU]}_{h,i}=A_0(\eta_i)\,\bU_i$, the discrete form of $\bU_W=A_0\bU$ in Theorem~\ref{thm:spectrum}.

\paragraph{Numerical fluxes and transport operators}\label{sec:fluxes}
The wave is transported by the gas-kinetic Navier--Stokes flux~\cite{xu2001}, whose equilibrium and non-equilibrium parts are weighted by $A_{0,\Gamma}=A_0(\eta_\Gamma)$ and $A_{1,\Gamma}=A_1(\eta_\Gamma)$. Its macroscopic form at a face is
\begin{equation}\label{eq:gks_flux}
 \bF^h_{W,\Gamma}=A_{0,\Gamma}\Big[\bF_{\rm E}(\bU_\Gamma)+\frac{\Delta t}{2}\partial_t\bF_{\rm E}(\bU_\Gamma)\Big]\bn_\Gamma+A_{1,\Gamma}\bF_{\rm R}(\bU_\Gamma)\bn_\Gamma,
\end{equation}
where $\bU_\Gamma$ is the face state reconstructed to second order, $\bF_{\rm R}$ is evaluated with $\tau_\Gamma$ and the face gradients, $\partial_t\bF_{\rm E}$ is the Euler time derivative of the flux, and $\Delta t$ is the averaging interval of the gas-kinetic flux, set by an explicit CFL number $C_{\rm exp}$ on the velocity grid. For $A_{0,\Gamma}=A_{1,\Gamma}=1$ it supplies the Euler and relaxation contributions to the second-order gas-kinetic flux; the Shakhov heat-flux correction remains in the particle. The velocity-resolved form $\phi^h_{W,\Gamma,j}$ of the wave flux is corrected by a Maxwellian-weighted combination of the collision invariants, so that its discrete moments equal \eqref{eq:gks_flux} exactly, and the particle is transported node by node by a second-order upwind flux $\phi_{P,\Gamma,j}$. Their velocity-resolved forms are given below. 

Both components are reconstructed to second order with van Leer limited slopes~\cite{vanleer1979}, which form the cell gradients $\nabla_h\bU_i$ and $\nabla_hP_{i,j}$. At a face $\Gamma$ let $\bU_\Gamma$ be the average of the two values extrapolated to $\bx_\Gamma$ from the adjacent cells, $M_\Gamma$ its Maxwellian, $\mathbf a_n$ and $\mathbf a_\parallel$ the micro-slopes of the normal and tangential difference quotients of $\bU$ at the face, $\xi_n$ and $\xi_\parallel$ the normal and tangential velocity components and $\mathbf a_t$ the time slope from \eqref{eq:gks_compatibility}. The velocity-resolved wave flux of the face is
\begin{equation}\label{eq:face_wave}
 \begin{aligned}
 \phi_{W,\Gamma}(\bxi)=\xi_nM_\Gamma\Big[&A_{0,\Gamma}-A_{1,\Gamma}\tau_\Gamma\big(\xi_n\Phi[\mathbf a_n]+\xi_\parallel\Phi[\mathbf a_\parallel]\big)\\
 &+\Big(\frac{A_{0,\Gamma}\Delta t}{2}-A_{1,\Gamma}\tau_\Gamma\Big)\Phi[\mathbf a_t]\Big],
 \end{aligned}
\end{equation}
that is, the average of the gas-kinetic distribution over the interval $\Delta t$. Its moments are Gaussian moments of $M_\Gamma$ and give \eqref{eq:gks_flux} with $\partial_t\bF_{\rm E}=\avg{\bpsi\bxi^T\Phi[\mathbf a_t]M_\Gamma}$. At the node $\bxi_j$ the flux is corrected by a Maxwellian-weighted combination of the collision invariants,
\begin{equation}\label{eq:lift}
 \phi^h_{W,\Gamma,j}=\phi_{W,\Gamma}(\bxi_j)+\operatorname{diag}(M_{1,\Gamma,j},M_{2,\Gamma,j})\,\bpsi_j^T\boldsymbol\lambda_\Gamma,\qquad
 \avgh{\bpsi\phi^h_{W,\Gamma}}=\bF^h_{W,\Gamma}.
\end{equation}
Here $M_{r,\Gamma,j}=M_{r,\Gamma}(\bxi_j)$, $r=1,2$. The $d+2$ coefficients $\boldsymbol\lambda_\Gamma$ are defined by the linear system on the right and remove the quadrature error of the velocity grid. The particle is transported node by node with a second-order upwind flux. With the upwind cell $c=i$ for $\xi_{n,j}>0$ and $c=i_\Gamma$ otherwise, this flux reads
\begin{equation}\label{eq:sn_face}
 \phi_{P,\Gamma,j}=\xi_{n,j}\Big[P_{c,j}+(\bx_\Gamma-\bx_c)\cdot\nabla_hP_{c,j}-\frac{\Delta t}{2}\xi_{n,j}\,\bn_\Gamma\cdot\nabla_hP_{c,j}\Big],
\end{equation}
whose last term matches the time averaging of the wave flux. The CIS transports the full distribution with the same flux. At a boundary face the extrapolation is taken from the interior cell, and at a diffuse wall the entering wave and particle are Maxwellians of the wall state whose densities make the normal mass flux of each component vanish.

The transport operators and the macroscopic particle flux are
\begin{equation}\label{eq:transport_ops}
 \begin{aligned}
 &D^h_W(W)_{i,j}=\frac{1}{|\Omega_i|}\sum_{\Gamma\in\mathcal F_i}|\Gamma|\phi^h_{W,\Gamma,j},\qquad
 D^h_P(P)_{i,j}=\frac{1}{|\Omega_i|}\sum_{\Gamma\in\mathcal F_i}|\Gamma|\phi_{P,\Gamma,j},\\
 &\bF^h_{P,\Gamma}=\avgh{\bpsi\phi_{P,\Gamma}},
 \end{aligned}
\end{equation}
with every face oriented outward for $\Omega_i$. The operator $D^h_W$ depends on $\bU$ through $W_h[\bU]$ and the face states, whereas $D^h_P$ depends on $P$ alone.

\begin{lemma}[Discrete flux compatibility]\label{lem:flux_compatibility}
For every $(\bU,P)\in\cX$ the velocity-resolved transport and the macroscopic fluxes satisfy, at interior and boundary faces alike,
\begin{equation}\label{eq:flux_compatibility}
 \avg{\bpsi\big[D^h_W(W_h[\bU])+D^h_P(P)\big]}_{h,i}=\nabla_h\cdot\big(\bF^h_W+\bF^h_P\big)\big|_i,
\end{equation}
where $\nabla_h\cdot\bF^h|_i=|\Omega_i|^{-1}\sum_{\Gamma\in\mathcal F_i}|\Gamma|\bF^h_\Gamma$ is the finite-volume divergence of the face fluxes.
\end{lemma}
\begin{proof}
The discrete moments of $\phi^h_{W,\Gamma}$ equal $\bF^h_{W,\Gamma}$ by the correction \eqref{eq:lift}, and $\bF^h_{P,\Gamma}$ is defined as the moments of $\phi_{P,\Gamma}$. Summing over the faces of the cell gives \eqref{eq:flux_compatibility}.
\end{proof}

\paragraph{Residuals and endpoint correction}\label{sec:residuals}
For a general conservative discretization $Q^h$ of the collision operator, the microscopic residual has the form
\begin{equation}\label{eq:general_residual}
 R_P=Q^h[f]-D^h_W(W)-D^h_P(P),\qquad
 R_U=-\nabla_h\cdot(\bF_W^h+\bF_P^h)=\avgh{\bpsi R_P}.
\end{equation}
The last identity requires conservative collision moments and compatible transport fluxes. Thus the residual coupling applies at the kinetic-equation level. In the realization studied here, $Q^h[f]=(S^h(f)-f)/\tau$. At a state $X=(\bU,P)$ with $W=W_h[\bU]$, $f=W+P$, relaxation time $\tau$ and discrete Maxwellian $M$, the residuals and the discrete endpoint are
\begin{equation}\label{eq:residuals}
 \begin{aligned}
 &R_U(X)=-\nabla_h\cdot\big(\bF^h_W(\bU)+\bF^h_P(P)\big),\\
 &R_P(X)=\frac{S^h(f)-M}{\tau}-\frac{P}{\tau}+\frac{M-W}{\tau}-D^h_W(W)-D^h_P(P),\\
 &P^h_e(\bU)=M-W-\tau D^h_W(W).
 \end{aligned}
\end{equation}
The particle residual is the residual of the discrete particle equation. Its third and fourth terms are the residual of the discrete steady wave equation, which, multiplied by $\tau$, defines the discrete endpoint, the counterpart of $P_e=\tau\cS_\ell$ in \eqref{eq:particle_equation}. With the discrete remainder $R^h=(S^h(f)-M)/\tau$, and under the conservation of the discrete target (conditions \ref{D1}--\ref{D2} below), the residuals satisfy the identities
\begin{equation}\label{eq:residual_moments}
 \begin{aligned}
 &R_P=R^h+\frac{P^h_e-P}{\tau}-D^h_P(P)=\frac{S^h(f)-f}{\tau}-D^h_W(W)-D^h_P(P),\\
 &\avgh{\bpsi R_P(X)}=R_U(X),
 \end{aligned}
\end{equation}
the second by the compatibility identity \eqref{eq:flux_compatibility}. The macroscopic residual therefore vanishes whenever the particle residual does, which is the discrete counterpart of the moment dependence discussed after Theorem~\ref{thm:equivalence}. The constraint $\bU=\avgh{\bpsi f}$, which the steady system does not contain, is supplied by the acceptance step of Section~\ref{sec:p_update}, and the discrete steady system $R_U=\boldsymbol 0$, $R_P=0$ is thus the discrete kinetic equation
\begin{equation}\label{eq:discrete_kinetic}
 D^h_W\big(W_h[\bU]\big)+D^h_P(P)=\frac{S^h(f)-f}{\tau},\qquad f=W_h[\bU]+P .
\end{equation}
The complement $P=f-W_h[\bU]$ retains the part of the distribution omitted by the wave approximation. The stationary equation therefore discretizes the full kinetic equation, while $C_\ell$ and the averaging interval $\Delta t$ affect the transport discretization and its error. They are distinct from the iteration controls in Theorem~\ref{thm:stationary}.

The iteration is monitored by the discrete $L^2$ norms of the two residuals,
\begin{equation}\label{eq:residual_norms}
 \|R_U\|=\Big(\sum_{i\in\mathcal I}|\Omega_i|\,\big|R_{U,i}\big|^2\Big)^{1/2},\qquad
 \|R_P\|=\Big(\sum_{i\in\mathcal I}|\Omega_i|\sum_{j\in\mathcal J}\varpi_j\,|R_{P,i,j}|^2\Big)^{1/2},
\end{equation}
where $|R_{P,i,j}|^2=R_{P,1,i,j}^2+R_{P,2,i,j}^2$ sums the two reduced components and $|R_{U,i}|$ is the Euclidean norm of the residuals of the conservation laws in cell $i$, in the nondimensional variables of Section~\ref{sec:model} without further scaling. By \eqref{eq:residual_moments} the moments of $R_P$ are $R_U$, and $\|R_P\|$ measures in addition the non-equilibrium part of the residual, which $\|R_U\|$ does not see. For the CIS, $R_P$ is the residual of the discrete kinetic equation for the full distribution and $R_U$ its moments, so that both methods are monitored by the same quantities.

Once the W iteration has produced the predicted variables $\widetilde\bU$, the endpoint is evaluated again, and the endpoint correction
\begin{equation}\label{eq:endpoint_correction}
 \Theta=\frac{P^h_e(\widetilde\bU)-P^h_e(\bU)}{\widehat\tau}
\end{equation}
is added to the particle residual, where $1/\widehat\tau$ is the implicit relaxation rate of the P update, with $\widehat\tau=\tau$ for the BGK and the Shakhov model. By \eqref{eq:residual_moments}, $R_P+\Theta$ is then the particle residual with the endpoint taken at the prediction; without this correction the P update would relax the particle toward the endpoint of the old state and undo part of the prediction. The correction vanishes at a stationary state, and the P updates after the first, which start from a state whose endpoint is current, need no such term.

\subsection{Macroscopic prediction and microscopic correction}\label{sec:w_update}
A W update advances the conservative variables while the particle and its flux $\bF^h_P$ are held fixed. Starting from $\bU^{(0)}=\bU$, the $n$th update forms the residual $R_U^{(n-1)}=-\nabla_h\cdot(\bF^h_W(\bU^{(n-1)})+\bF^h_P)$, in which only the wave flux is recomputed, and advances the conservation law in pseudo-time by the backward Euler method in delta form,
\begin{equation}\label{eq:macro_delta}
 \begin{aligned}
 &\frac{|\Omega_i|}{\Delta t_{\rm num}}\Delta\bU_i+\sum_{\Gamma\in\mathcal F_i}|\Gamma|\,\Delta\bF_\Gamma=|\Omega_i|R^{(n-1)}_{U,i},\\
 &\Delta\bF_\Gamma=\frac12\Big[\big(\Delta\bF_{{\rm E},i}+\Delta\bF_{{\rm E},i_\Gamma}\big)\bn_\Gamma+r_\Gamma\big(\Delta\bU_i-\Delta\bU_{i_\Gamma}\big)\Big],
 \end{aligned}
\end{equation}
where $\Delta t_{\rm num}$ is set by an implicit CFL number $C_{\rm imp}$ and the flux increment through a face is approximated by the Rusanov-type splitting of the Euler-flux increments $\Delta\bF_{{\rm E},m}=\bF_{\rm E}(\bU_m+\Delta\bU_m)-\bF_{\rm E}(\bU_m)$ of the two cells. The increments are computed by $m_U$ forward--backward point-relaxation pairs; \ref{app:implicit_sweeps} gives the diagonal coefficient, spectral radius and sweep equations. The state is then updated as $\bU^{(n)}=\bU^{(n-1)}+\theta_U\,\Delta\bU^{[m_U]}$, with $\theta_U\in\{1,\frac12,\frac14,\dots\}$ selected by backtracking to maintain positive density and temperature. When enabled, the distribution safeguard checks $W_h[\bU^{(n)}]+P$ on resolved velocity ordinates, together with moment admissibility and wall outflow. An inadmissible direction is retried using the local physical residual. The face states, wave fluxes and discrete wave are then rebuilt. We write $\bU^{(n)}=\mathcal T_W(\bU^{(n-1)};P)$ for one update and $\widetilde\bU=\mathcal T_W^{N_U}(\bU;P)$ for the prediction of the W iteration.

\paragraph{Microscopic correction and moment synchronization}\label{sec:p_update}
Each P update holds fixed the wave $\widetilde W=W_h[\widetilde\bU]$ associated with its current conservative state. The first update uses the macroscopic prediction; subsequent updates use the state and decomposition refreshed at the preceding acceptance. Linearized about the accepted state, the implicit particle equation reads
\begin{equation}\label{eq:particle_increment}
 \bA_P\Delta P=R_P+\Theta,\qquad
 \bA_P=\Big(\frac{1}{\Delta t_{\rm num}}+\frac{1}{\widehat\tau}\Big)\bI+D^h_{\rm up},
\end{equation}
where $D^h_{\rm up}$ is the first-order upwind discretization of $\bxi\cdot\nabla_{\bx}$, which preconditions the increment of the second-order residual, and $\Theta$ is the correction \eqref{eq:endpoint_correction} in the first P update of an outer iteration and zero afterwards. The remainder $R$ of the collision term, which for the Shakhov model depends on the heat flux of the current distribution, is kept explicit in $R_P$, and the implicit rate $1/\widehat\tau$ bounds the stiffness of the relaxation. For each velocity node, $m_P$ lower--upper symmetric Gauss--Seidel (LU--SGS) pairs approximate this solve from a zero initial increment, as specified in \ref{app:implicit_sweeps}. The accepted state is then formed from the total distribution,
\begin{equation}\label{eq:acceptance}
 f=\widetilde W+P+\theta\,\Delta P^{[m_P]},\qquad \bU\leftarrow\avgh{\bpsi f},\qquad P\leftarrow f-W_h[\bU].
\end{equation}
The discrete backtracking choice $\theta\in\{1,\frac12,\frac14,\dots\}$ used in macroscopic acceptance is replaced here by a continuous cell-wise factor computed from the admissible decrement ratios. When distribution limiting is enabled, it protects $f$ on the resolved velocity ordinates; signed-distribution configurations bypass this check. The limit $\theta\to0$ retains the base distribution $\widetilde W+P$, whose admissibility is checked before acceptance. Because the particle carries conservative moments, the acceptance synchronizes the conservative variables with the corrected total distribution and refreshes the decomposition without changing $f$. We write $(\bU,P)\leftarrow\mathcal T_P(\widetilde\bU,P;\Theta)$ for one P update including the acceptance.

In a closed domain with diffuse walls the steady equations do not determine the total mass, which is fixed by the initial state and which the approximate implicit updates do not conserve exactly. For the cavity and for Couette flow the conservative variables and the particle are therefore multiplied by a common factor at the end of each outer iteration, so that the total mass keeps its initial value, and the wave is rebuilt; the CIS and NS computations are rescaled in the same way. The factor tends to one as the iteration converges, and the constraint $\bU=\avgh{\bpsi f}$, which the rescaling perturbs in proportion to the deviation of the factor from one, is restored by the next acceptance.

\begin{algorithm}[tbp]
\caption{Implicit WPD iteration (IWPD)}
\label{alg:wp}
\begin{algorithmic}
\STATE Given $\bU^0$, set $P^0=M_h[\bU^0]-W_h[\bU^0]$.
\FOR{$k=0,1,\dots$, until the residual norms $\|R_U\|$ and $\|R_P\|$ of \eqref{eq:residual_norms} are below the tolerance}
\STATE evaluate $R_U$, $R_P$ and $P^h_e$ at $(\bU^k,P^k)$;
\STATE W iteration: $\widetilde\bU^{k+1}=\mathcal T_W^{N_U}(\bU^k;P^k)$;
\STATE form the endpoint correction $\Theta^k$ by \eqref{eq:endpoint_correction};
\STATE P iteration: apply $\mathcal T_P$ $N_P$ times, the first time with $\Theta^k$ and afterwards with $\Theta=0$, which gives $(\bU^{k+1},P^{k+1})$.
\ENDFOR
\end{algorithmic}
\end{algorithm}

\subsection{Implementation}\label{app:implicit_sweeps}
The W and P increments use forward--backward sweeps started from zero. Boundary neighbour increments are set to zero, so the prescribed boundary data enter through the residuals.

\subsubsection{Macroscopic point relaxation}
Since $\sum_\Gamma|\Gamma|\bn_\Gamma=\boldsymbol 0$ on a closed cell, the contribution of $\Delta\bF_{{\rm E},i}$ cancels, and \eqref{eq:macro_delta} becomes
\begin{equation}\label{eq:macro_increment}
 d_i\,\Delta\bU_i+\sum_{\Gamma\in\mathcal F_i}G_{i,\Gamma}\big(\Delta\bU_{i_\Gamma}\big)=|\Omega_i|R^{(n-1)}_{U,i},
\end{equation}
with the diagonal coefficient, the neighbour term and the face spectral radius
\begin{equation}\label{eq:neighbour_term}
 \begin{aligned}
 &d_i=\frac{|\Omega_i|}{\Delta t_{\rm num}}+\frac12\sum_{\Gamma\in\mathcal F_i}r_\Gamma|\Gamma|,\qquad
 r_\Gamma=\max_{m\in\{i,i_\Gamma\}}\big(|\bu_m\cdot\bn_\Gamma|+a_m\big)+\frac{2\chi_\Gamma\mu_\Gamma}{\rho_\Gamma\Delta_\Gamma},\\
 &G_{i,\Gamma}(\boldsymbol\delta)=\frac{|\Gamma|}{2}\Big[\big(\bF_{\rm E}(\bU_{i_\Gamma}+\boldsymbol\delta)-\bF_{\rm E}(\bU_{i_\Gamma})\big)\bn_\Gamma-r_\Gamma\boldsymbol\delta\Big],
 \end{aligned}
\end{equation}
where $\mu_\Gamma$ and $\rho_\Gamma$ are the viscosity and the density of the face state, $\Delta_\Gamma$ is the mesh size at the face, and $\chi_\Gamma\in[0,1]$ weights the viscous part of the spectral radius. System \eqref{eq:macro_increment} is solved approximately by $m_U$ point-relaxation pairs. Let $\partial\Omega_i^\mp$ denote the faces whose neighbours precede and follow $\Omega_i$ in a fixed ordering of the cells. Starting from $\Delta\bU^{[0]}=\boldsymbol 0$, the $q$th pair is a forward sweep
\begin{equation}\label{eq:pr_sweeps}
 d_i\Delta\bU_i^{[q-1/2]}=|\Omega_i|R^{(n-1)}_{U,i}-\sum_{\Gamma\in\partial\Omega_i^-}G_{i,\Gamma}\big(\Delta\bU^{[q-1/2]}_{i_\Gamma}\big)-\sum_{\Gamma\in\partial\Omega_i^+}G_{i,\Gamma}\big(\Delta\bU^{[q-1]}_{i_\Gamma}\big),
\end{equation}
followed by a backward sweep, which produces $\Delta\bU^{[q]}$ from $\Delta\bU^{[q-1/2]}$ with the roles of $\partial\Omega_i^-$ and $\partial\Omega_i^+$ interchanged, that is, one LU--SGS step. Neighbour increments at boundary faces are set to zero, so that boundary data enter through the residual only. 
\subsubsection{Microscopic LU--SGS}
For cell $i$ and node $j$ the row of \eqref{eq:particle_increment} is
\begin{equation}\label{eq:particle_row}
 \begin{aligned}
 &d_{i,j}\Delta P_{i,j}+\sum_{\Gamma\in\mathcal F_i}c_{i,\Gamma,j}\Delta P_{i_\Gamma,j}=R_{P,i,j}+\Theta_{i,j},\\
 &d_{i,j}=\frac{1}{\Delta t_{\rm num}}+\frac{1}{\widehat\tau_i}+\sum_{\Gamma\in\mathcal F_i}\frac{(\xi_{n,j})^+|\Gamma|}{|\Omega_i|},\qquad
 c_{i,\Gamma,j}=\frac{(\xi_{n,j})^-|\Gamma|}{|\Omega_i|},
 \end{aligned}
\end{equation}
with $x^\pm=(x\pm|x|)/2$, so that only inflow faces couple a cell to its upwind neighbours. As in the W update, neighbour increments at boundary faces are set to zero. Let $\bA_P=\mathbf D+\bA_P^-+\bA_P^+$ be the splitting into the diagonal part and the parts that couple a cell to the cells preceding and following it in the cell ordering. One kinetic LU--SGS pair, a forward sweep followed by a backward sweep carried out independently for every node, applies $\bM_P^{-1}=(\mathbf D+\bA_P^+)^{-1}\mathbf D\,(\mathbf D+\bA_P^-)^{-1}$, and the increment is obtained by $m_P$ preconditioned iterations started from $\Delta P^{[0]}=0$,
\begin{equation}\label{eq:particle_lusgs}
 \Delta P^{[q]}=\Delta P^{[q-1]}+\bM_P^{-1}\big(R_P+\Theta-\bA_P\Delta P^{[q-1]}\big),\qquad q=1,\dots,m_P .
\end{equation}

\paragraph{Computational cost}\label{sec:implementation}
Algorithm~\ref{alg:wp} specifies the outer iteration, with configuration W$N_U$/\allowbreak PR$m_U$/\allowbreak P$N_P$ and $m_P$ stated separately. The macroscopic solver receives the fixed particle flux, and the microscopic solver receives the wave transport and endpoint correction. The collision remainder is explicit in the P update. IWPD stores three additional velocity-resolved arrays relative to CIS---the wave, particle and endpoint---and its kinetic sweeps are independent across velocity nodes.

The cost of one outer iteration is approximately
\begin{equation}\label{eq:cost}
 C_{\rm outer}\approx N_U\big(m_UC_{\rm sweep}+C_{\rm wave}\big)+N_P\big(C_{\rm res}+m_PC_{\rm kin}\big),
\end{equation}
where the macroscopic sweeps $C_{\rm sweep}$ are independent of the velocity grid and the rebuild of the wave $C_{\rm wave}$ requires one evaluation of the discrete wave per velocity node, whereas the residual evaluation $C_{\rm res}$ and the kinetic sweeps $C_{\rm kin}$ transport both components on the velocity grid. The macroscopic part is comparable in cost with $N_U$ implicit steps of a Navier--Stokes computation with the same solver, and the kinetic part grows with the number of velocity nodes like an iteration of the CIS. The choice of the controls is examined in Section~\ref{sec:controls}.

\section{Consistency and convergence analysis}\label{sec:analysis}
The analysis establishes the relation between the closed formulation and the implicit iteration. The stationary-state result is exact under discrete moment compatibility. The fluid-limit and prediction-transfer arguments are formal analyses for smooth states. Convergence factors are derived for linear model problems and compared with the numerical histories in Section~\ref{sec:convergence}.

\paragraph{Discrete collision conditions}
\begin{enumerate}\renewcommand{\theenumi}{(D\arabic{enumi})}\renewcommand{\labelenumi}{\theenumi}
\item\label{D1} The discrete target conserves the quadrature moments, $\avgh{\bpsi(S^h(f)-f)}=\boldsymbol 0$ for every discrete pair $f$.
\item\label{D2} The discrete Maxwellian satisfies $\avgh{\bpsi M_h[\bU]}=\bU$ and has zero discrete heat flux. At prescribed conservative moments it is the unique fixed point of $S^h$. For the discrete Shakhov target, the heat-flux moment is $(1-\Pr)$ times that of $f$, with $\Pr>0$.
\item\label{D3} The horizons satisfy $\ell_i\ge\ell_{\min}>0$ independently of $\varepsilon$.
\end{enumerate}

These conditions specify the discrete operators used in the analysis. Condition \ref{D1} is enforced by conservative moment matching. In condition \ref{D2}, the conservative moments can likewise be matched by a lift of the form \eqref{eq:lift}, provided the local lifting system is nonsingular; the equilibrium and Shakhov heat-flux identities are additional requirements, satisfied to velocity-quadrature accuracy for sampled analytic targets. The exact statements below use these identities, with quadrature defects accounted for in their numerical realization. Condition \ref{D3} follows from \eqref{eq:discrete_horizon} with $C_\ell>0$ on a fixed mesh when cell velocities and sound speeds remain bounded.

\subsection{Stationary-state consistency}\label{sec:stationary}
The initial state of Algorithm~\ref{alg:wp} and every accepted state satisfy the constraint $\bU=\avgh{\bpsi f}$ with $f=W_h[\bU]+P$, by the choice of $P^0$ and by the acceptance \eqref{eq:acceptance}; we call such states consistent. A consistent state $X\in\cX$ is stationary if every W update and every P update of one outer iteration started from $X$ produces a zero increment. Since the acceptance leaves a consistent state with zero increments unchanged, a stationary state is a fixed point of the outer iteration, and a superscript $*$ marks quantities evaluated at such a state.

Two properties of the updates establish the stationary-state result. A W update must return a zero increment when the macroscopic residual vanishes, which is evident from \eqref{eq:pr_sweeps}, since the sweeps start from $\Delta\bU^{[0]}=\boldsymbol 0$ and $G_{i,\Gamma}(\boldsymbol 0)=\boldsymbol 0$. A P update must return a zero increment only when its right-hand side vanishes, which follows for the boundary-lagged upwind increment operator from its M-matrix structure.

\begin{lemma}[Invertibility of the particle increment]\label{lem:particle_increment}
With incoming boundary increments held fixed during a sweep, for every node $j$ the matrix $\bA_P$ of \eqref{eq:particle_row} is a nonsingular M-matrix, the LU--SGS splitting $\bA_P=\bM_P-\mathbf E_P$ with $\mathbf E_P=\bA_P^-\mathbf D^{-1}\bA_P^+$ is a regular splitting with $\rho(\bM_P^{-1}\mathbf E_P)<1$, and for every $m_P\ge1$ the increment of \eqref{eq:particle_lusgs} is
\begin{equation}\label{eq:increment_formula}
 \Delta P^{[m_P]}=\big(\bI-\mathbf G_P^{m_P}\big)\bA_P^{-1}(R_P+\Theta),\qquad \mathbf G_P=\bM_P^{-1}\mathbf E_P,
\end{equation}
with $\bI-\mathbf G_P^{m_P}$ nonsingular. Hence $\Delta P^{[m_P]}=0$ if and only if $R_P+\Theta=0$.
\end{lemma}
\begin{proof}
The diagonal entries $d_{i,j}$ are positive and the off-diagonal entries $c_{i,\Gamma,j}$ are nonpositive. Since $\sum_{\Gamma\in\mathcal F_i}|\Gamma|\bn_\Gamma=\boldsymbol 0$,
\begin{equation*}
 \sum_{\Gamma\in\mathcal F_i}\frac{\big[(\xi_{n,j})^++(\xi_{n,j})^-\big]|\Gamma|}{|\Omega_i|}=\sum_{\Gamma\in\mathcal F_i}\frac{\xi_{n,j}|\Gamma|}{|\Omega_i|}=0,
\end{equation*}
so that the row sum of $\bA_P$ over the interior faces is at least $1/\Delta t_{\rm num}+1/\widehat\tau_i>0$, and boundary faces remove nonpositive neighbour entries while retaining outgoing contributions on the diagonal. This argument concerns the boundary-lagged operator; a fully coupled diffuse-reflection increment also couples velocity nodes and requires a separate matrix analysis. Thus $\bA_P$ is a strictly diagonally dominant Z-matrix and hence a nonsingular M-matrix. The triangular factors $\mathbf D+\bA_P^-$ and $\mathbf D+\bA_P^+$ are M-matrices with nonnegative inverses, so that $\bM_P^{-1}\ge0$, and $\mathbf E_P=\bM_P-\bA_P=\bA_P^-\mathbf D^{-1}\bA_P^+\ge0$. The splitting is therefore regular and, since $\bA_P^{-1}\ge0$, convergent, $\rho(\mathbf G_P)<1$, by the theorem on regular splittings~\cite{varga2000}. Iterating \eqref{eq:particle_lusgs} with $\mathbf G_P=\bI-\bM_P^{-1}\bA_P$ gives
\begin{equation*}
 \Delta P^{[m]}=\sum_{q=0}^{m-1}\mathbf G_P^q\bM_P^{-1}(R_P+\Theta)=(\bI-\mathbf G_P^{m})(\bI-\mathbf G_P)^{-1}\bM_P^{-1}(R_P+\Theta),
\end{equation*}
which is \eqref{eq:increment_formula} because $(\bI-\mathbf G_P)^{-1}\bM_P^{-1}=\bA_P^{-1}$, and $\bI-\mathbf G_P^m$ is nonsingular because $\rho(\mathbf G_P)<1$.
\end{proof}

\begin{theorem}[Stationary states]\label{thm:stationary}
Let \ref{D1} hold. Then a consistent state $X^*=(\bU^*,P^*)$ is stationary if and only if $R_P(X^*)=0$, that is, if and only if $f^*=W_h[\bU^*]+P^*$ solves the discrete kinetic equation \eqref{eq:discrete_kinetic} with $\bU^*=\avgh{\bpsi f^*}$. In that case $R_U(X^*)=\boldsymbol 0$. The set of stationary states does not depend on $\Delta t_{\rm num}$, $\widehat\tau$, $r_\Gamma$, $\chi_\Gamma$, $N_U$, $m_U$, $N_P$ or $m_P$.
\end{theorem}
\begin{proof}
Let $X^*$ be stationary. Since the W updates produce no increment, the prediction equals $\bU^*$ and the correction \eqref{eq:endpoint_correction} vanishes. Since the first P update produces no increment either, $R_P(X^*)=0$ by the nonsingular increment map \eqref{eq:increment_formula}, which is \eqref{eq:discrete_kinetic}, and $R_U(X^*)=\avgh{\bpsi R_P(X^*)}=\boldsymbol 0$ follows from \eqref{eq:residual_moments} and \ref{D1}. Conversely, if $R_P(X^*)=0$, then $R_U(X^*)=\boldsymbol 0$ by the same identity. The sweeps \eqref{eq:pr_sweeps} of every W update then start from a zero increment with a zero right-hand side and return a zero increment, the correction vanishes, every P update has a zero increment by the nonsingular increment map \eqref{eq:increment_formula}, and the acceptance \eqref{eq:acceptance} returns $X^*$ unchanged, because $X^*$ is consistent. The listed quantities enter only the increment operators, which multiply residuals that vanish at a stationary state.
\end{proof}

By \eqref{eq:residual_moments}, $R_P$ is the residual of the discrete kinetic equation \eqref{eq:discrete_kinetic} and contains none of the controls, so that a state that passes the stopping test of Algorithm~\ref{alg:wp} is an approximate solution of this equation, whatever the controls.

\subsection{Hydrodynamic consistency and prediction transfer}\label{sec:ap}
For smooth stationary states and positive horizons, the following formal expansion identifies the limit equations as $\varepsilon\to0$ on a fixed mesh. We write $f^{(1),h}$, $W^{(1),h}$ and $S^{(1),h}$ for the first-order terms in $\varepsilon$ of the discrete distribution, wave and target at the stationary state.

\begin{proposition}[Formal fluid limit]\label{prop:formal_fluid_limit}
Let \ref{D1}--\ref{D3} hold, let $\tau=\varepsilon\bar\tau$, and let the stationary solution be smooth on a fixed mesh of cell size $\Delta_i$. Then, formally as $\varepsilon\to0$,
\begin{equation}\label{eq:discrete_ce}
 f^*=M^*+\varepsilon f^{(1),h}+\mathcal O(\varepsilon^2),\qquad P^*=\varepsilon S^{(1),h}+\mathcal O(\varepsilon\Delta_i)+\mathcal O(\varepsilon^2),
\end{equation}
and the quadrature moments of \eqref{eq:discrete_kinetic} reduce to the finite-volume Euler equations at leading order and to the gas-kinetic Navier--Stokes scheme \eqref{eq:gks_flux} with $A_0=A_1=1$ and Prandtl number $\Pr$ at first order, up to terms of order $\varepsilon\Delta_i$.
\end{proposition}
\begin{proof}
Equation \eqref{eq:discrete_kinetic} reads $f^*=S^h(f^*)-\varepsilon\bar\tau[D^h_W(W^*)+D^h_P(P^*)]$. At leading order $f^*=S^h(f^*)$, whose unique solution at fixed conservative moments is the discrete Maxwellian by \ref{D2}; hence $f^*=M^*+\mathcal O(\varepsilon)$. By \ref{D3}, $e^{-\eta_i}$ is smaller than any power of $\varepsilon$, so that \eqref{eq:discrete_wave} gives $W^*=M^*+\varepsilon W^{(1),h}+\mathcal O(\varepsilon^2)$ with $W^{(1),h}=-\bar\tau[\cD^{(0)}M]^h$, the transport derivative formed with the micro-slopes, while $P^*=\mathcal O(\varepsilon)$ enters the transport only at the next order. At first order, $f^{(1),h}=S^{(1),h}-\bar\tau\,D^h_W(M^*)$, where $D^h_W(M^*)$ is the leading-order part of $D^h_W(W^*)$; its solvability condition, the vanishing of the quadrature moments of $D^h_W(M^*)$, is the finite-volume Euler system by the compatibility identity \eqref{eq:flux_compatibility}. The particle is $P^*=\varepsilon(f^{(1),h}-W^{(1),h})=\varepsilon S^{(1),h}$ up to $\varepsilon$ times the consistency error between $D^h_W(M^*)$ and $[\cD^{(0)}M]^h$, which is of the order of the cell size and whose time-derivative part vanishes at leading order in a steady state. The gas-kinetic wave flux then supplies $\bF_{\rm E}+\bF_{\rm R}$ as in Theorem~\ref{thm:spectrum}, and the flux of the $\mathcal O(\varepsilon)$ particle supplies $\bF_{\rm S}$ with an error of order $\varepsilon\Delta_i$.
\end{proof}

This is the asymptotic-preserving property in the sense of consistency: on meshes that need not resolve the mean free path, the wave flux supplies the Euler and relaxation fluxes without quadrature error, and the particle supplies the flux of the remainder, for the Shakhov model the Prandtl-number correction.

\paragraph{Formal transfer of the macroscopic prediction}\label{sec:transfer}
For one outer iteration with $N_P=1$, write $\widetilde M=M_h[\widetilde\bU^{k+1}]$ for the Maxwellian of the macroscopic prediction. At smooth near-equilibrium states, locally fixed $\eta$ and $\widehat\tau=\tau$, the endpoint correction satisfies
\begin{equation}\label{eq:endpoint_transfer}
 W_h[\bU]=(1-e^{-\eta})M+\mathcal O(\varepsilon),\qquad
 \tau\Theta^k=e^{-\eta}(\widetilde M-M^k)+\mathcal O(\varepsilon).
\end{equation}
It therefore supplies the equilibrium change missing from the wave. For an exact particle solve, multiplying \eqref{eq:particle_increment} by $\tau$ gives
\begin{equation}\label{eq:transfer_solve}
 \big[\bI+\tau(\Delta t_{\rm num}^{-1}\bI+D^h_{\rm up})\big]\Delta P
       =\tau R_P^k+\tau\Theta^k.
\end{equation}
If $\tau R_P^k=\mathcal O(\varepsilon)$, $\tau/\Delta t_{\rm num}=\mathcal O(\varepsilon)$, the transport of the smooth predicted increment is bounded, and the inverse in \eqref{eq:transfer_solve} is uniformly bounded, then the exact solve gives $\Delta P=e^{-\eta}(\widetilde M-M^k)+\mathcal O(\varepsilon)$. With an inexact solve the same conclusion requires an $\mathcal O(\varepsilon)$ error in this increment. Under these conditions and with the acceptance factor equal to unity, acceptance gives
\begin{equation}\label{eq:hard_pull}
 f^{k+1}=\widetilde M+\mathcal O(\varepsilon),\qquad
 \bU^{k+1}=\widetilde\bU^{k+1}+\mathcal O(\varepsilon).
\end{equation}
This transfer result links the accuracy of the microscopic increment to that of the macroscopic prediction. Keeping $\eta$ finite under mesh refinement can make $\tau/\Delta_i$ of order one, so the transport matrix cannot then be discarded as an $\mathcal O(\varepsilon)$ perturbation in operator norm. On a fixed mesh with positive horizons, $\eta\to\infty$ and the equilibrium endpoint correction becomes exponentially small.

\subsection{Model convergence analysis}\label{sec:rates}
We quantify the convergence rate in the simplest model that retains transport, relaxation and the horizon split, namely the steady scalar kinetic equation
\begin{equation}\label{eq:model_problem}
 \bxi\cdot\nabla_{\bx}f=\frac{r-f}{\tau},\qquad r=\frac{1}{4\pi}\int_{|\bxi|=1}f\dd\bxi.
\end{equation}
This single-invariant model isolates the transport mechanism at constant relaxation time $\tau$, with molecular velocities of unit speed and isotropic directions. The angular average $r$ of the distribution is its only collision invariant and takes the place of the conservative variables. For a periodic perturbation of wavenumber $\vk$ let $\zeta\in[-1,1]$ be the cosine between $\bxi$ and the wave vector, so that the angular average is $\frac12\int_{-1}^1\cdot\,\dd\zeta$, and put $z=\vk\tau$ and $\eta=\ell/\tau$, with $i$ the imaginary unit. In this model the W iteration is converged with $P$ fixed and the P update is inverted exactly with $\Delta t_{\rm num}\to\infty$, $\widehat\tau=\tau$, $N_P=1$ and $\theta=1$, whereas the CIS obtains $f^{k+1}$ from $\bxi\cdot\nabla_{\bx}f^{k+1}+f^{k+1}/\tau=r^k/\tau$.

\begin{theorem}[Convergence factors]\label{thm:factors}
In the model problem one CIS iteration multiplies the error amplitude of the density mode by $g_{\rm CIS}$ and one outer iteration of IWPD by $g_{\rm WP}$, where
\begin{equation}\label{eq:factors}
 g_{\rm CIS}(z)=\frac{\arctan z}{z},\qquad
 g_{\rm WP}(z,\eta)=\Big|1-g_{\rm CIS}(z)+\frac{\varphi_P(z,\eta)}{\varphi_W(z,\eta)}\Big|,
\end{equation}
with the flux responses of the particle and of the wave
\begin{equation}\label{eq:flux_responses}
 \varphi_P=\frac{e^{-\eta}}{2}\int_{-1}^{1}\frac{\zeta e^{-iz\eta\zeta}}{1+iz\zeta}\dd\zeta,\qquad
 \varphi_W=\frac12\int_{-1}^{1}\frac{\zeta\big(1-e^{-\eta}e^{-iz\eta\zeta}\big)}{1+iz\zeta}\dd\zeta .
\end{equation}
With $\beta(\eta)=(1+\eta)e^{-\eta}$, the factors satisfy, for $0<z\le1$ and $\beta(\eta)+3z^2/5<1$,
\begin{equation}\label{eq:bounds}
 g_{\rm CIS}(z)\ge1-\frac{z^2}{3},\qquad
 g_{\rm WP}(z,\eta)\le\frac{z^2}{3}+\frac{\beta(\eta)}{1-\beta(\eta)-3z^2/5}.
\end{equation}
\end{theorem}
\begin{proof}
Scale time by $\tau$, and let $r^k$ be the amplitude of the accepted state and $\widetilde r^{k+1}$ that of the prediction. For an amplitude $r$ the wave operator and the endpoint are
\begin{equation*}
 \begin{gathered}
 W=w\,r,\qquad w=\frac{1-e^{-\eta}e^{-iz\eta\zeta}}{1+iz\zeta},\\
 P_e=e^{-\eta}e^{-iz\eta\zeta}\,r,\qquad (1+iz\zeta)\,w+\frac{P_e}{r}=1 .
 \end{gathered}
\end{equation*}
The CIS gives $f^{k+1}=r^k/(1+iz\zeta)$, whose angular average is $g_{\rm CIS}r^k$. In the WPD-based iteration the particle residual is $R_P^k=r^k-(1+iz\zeta)f^k$, and the P update with the fixed wave $w\widetilde r^{k+1}$ solves
\begin{equation*}
 (1+iz\zeta)\,\Delta P=R_P^k+\big(\widetilde r^{k+1}-r^k\big)e^{-\eta}e^{-iz\eta\zeta}.
\end{equation*}
By the identity above, $f^{k+1}=w\widetilde r^{k+1}+P^k+\Delta P=\widetilde r^{k+1}/(1+iz\zeta)$, and hence $r^{k+1}=g_{\rm CIS}\widetilde r^{k+1}$. The converged W iteration balances the wave flux of the next prediction against the flux of the fixed particle,
\begin{equation*}
 \varphi_W\,\widetilde r^{k+2}+\frac12\int_{-1}^1\zeta P^{k+1}\dd\zeta=0,\qquad P^{k+1}=f^{k+1}-w\,r^{k+1}.
\end{equation*}
Since $\frac12\int\zeta(1+iz\zeta)^{-1}\dd\zeta=\varphi_W+\varphi_P$, the particle flux equals $(\varphi_W+\varphi_P-g_{\rm CIS}\varphi_W)\widetilde r^{k+1}$, and therefore
\begin{equation*}
 \widetilde r^{k+2}=-\Big(1-g_{\rm CIS}+\frac{\varphi_P}{\varphi_W}\Big)\widetilde r^{k+1}.
\end{equation*}
For the bounds, the alternating series of $\arctan z$ gives $g_{\rm CIS}\ge1-z^2/3$. In $\varphi_P$ the odd parts of the integrand cancel,
\begin{equation*}
 \varphi_P=-\frac{i}{2}e^{-\eta}\int_{-1}^{1}\frac{\zeta\sin(z\eta\zeta)+z\zeta^2\cos(z\eta\zeta)}{1+z^2\zeta^2}\dd\zeta,
\end{equation*}
and $|\sin x|\le|x|$ yields $|\varphi_P|\le\beta z/3$. Finally, $\varphi_W+\varphi_P=-i(z-\arctan z)/z^2$, whose modulus is at least $z/3-z^3/5$ for $0<z\le1$, so that $|\varphi_W|\ge(z/3)(1-\beta-3z^2/5)$, and the triangle inequality gives \eqref{eq:bounds}.
\end{proof}

\begin{corollary}[Scale dependence of the bound]\label{cor:scale_dependence}
Within its stated range of validity, the bound in \eqref{eq:bounds} decreases with $\eta$, and for $z\to0$ it tends to $\beta/(1-\beta)$, which is smaller than one precisely for $\eta>\eta_c\approx1.68$. This limit is attained by $g_{\rm WP}$ itself, so that for $\eta<\eta_c$ the model iteration with a converged W iteration amplifies long waves. On a fixed mesh of cell size $\Delta$ a resolved mode has $\vk\le\pi/\Delta$, so that $z=\mathcal O(\varepsilon)$ while $\eta=\mathcal O(\varepsilon^{-1})$. Consequently the CIS reduces the error by $1-\mathcal O(\varepsilon^2)$ per iteration and needs $\mathcal O(\varepsilon^{-2})$ iterations, whereas the WPD-based iteration contracts by $\mathcal O(\varepsilon^2)+\mathcal O(\beta(\eta))$ per outer iteration, with $\beta$ decaying exponentially in $\eta$.
\end{corollary}
\begin{proof}
The share $\beta$ decreases since $\beta'(\eta)=-\eta e^{-\eta}$, and $\beta/(1-\beta)<1$ if and only if $\beta<1/2$. The limit of $g_{\rm WP}$ follows from $\varphi_P=-i\beta z/3+\mathcal O(z^3)$ and $\varphi_W+\varphi_P=-iz/3+\mathcal O(z^3)$ in the proof of Theorem~\ref{thm:factors}. The remaining statements are \eqref{eq:bounds} with $z=\vk\tau\le\pi\varepsilon\bar\tau/\Delta$ and $\eta=\ell/(\varepsilon\bar\tau)$.
\end{proof}

Two extensions connect the model of Theorem~\ref{thm:factors} with the practical iteration. First, the W iteration consists of a fixed number $N_U$ of updates, a single one below $\eta_c$, where a converged W iteration would amplify long waves. Second, the remainder of the collision term, for the Shakhov model the heat-flux correction, is evaluated at the accepted state. The following estimates account for these features within their respective model settings.

\begin{proposition}[Incomplete macroscopic solution]\label{prop:incomplete_macro}
In the model problem, let the W iteration consist of $N_U$ updates of a linear iteration for the macroscopic balance, started from the accepted amplitude, each of which reduces the error of the prediction by a factor $\varrho_W\in[0,1)$. Then one outer iteration multiplies the amplitude of the prediction by a factor $g$ with
\begin{equation}\label{eq:inexact}
 g\le g_{\rm WP}+\varrho_W^{N_U}\big(g_{\rm CIS}+g_{\rm WP}\big)\le g_{\rm WP}+\varrho_W^{N_U}\big(1+g_{\rm WP}\big).
\end{equation}
\end{proposition}
\begin{proof}
By the proof of Theorem~\ref{thm:factors} the converged W iteration gives $\widetilde r^{k+2}_*=-\Lambda\widetilde r^{k+1}$ with $|\Lambda|=g_{\rm WP}$, and the inexact iteration starts from the accepted amplitude $r^{k+1}=g_{\rm CIS}\widetilde r^{k+1}$. After $N_U$ updates,
\begin{equation*}
 \big|\widetilde r^{k+2}-\widetilde r^{k+2}_*\big|\le\varrho_W^{N_U}\big|r^{k+1}-\widetilde r^{k+2}_*\big|\le\varrho_W^{N_U}\big(g_{\rm CIS}+g_{\rm WP}\big)\big|\widetilde r^{k+1}\big|,
\end{equation*}
whence $|\widetilde r^{k+2}|\le[g_{\rm WP}+\varrho_W^{N_U}(g_{\rm CIS}+g_{\rm WP})]\,|\widetilde r^{k+1}|$, and the second inequality follows from $g_{\rm CIS}\le1$.
\end{proof}

\begin{proposition}[Explicit collision remainder]\label{prop:collision_remainder}
Assume the moment and heat-flux identities in \ref{D1}--\ref{D2}. Consider a spatially uniform state on a periodic mesh with uniform $\eta$, whose distribution $f$ differs from the Maxwellian $M$ of its conservative variables, let $\widehat\tau=\tau$, and let \eqref{eq:particle_increment} be solved exactly. One P update of Algorithm~\ref{alg:wp}, and likewise one CIS iteration, maps $f$ to $(1-\omega)f+\omega S^h(f)$ with $\omega=\Delta t_{\rm num}/(\Delta t_{\rm num}+\tau)$. For the BGK model the deviation $f-M$ is multiplied by $1-\omega$, which vanishes as $\Delta t_{\rm num}/\tau\to\infty$; for the Shakhov model the heat flux $\bq$ is multiplied by $1-\omega\Pr$, which tends to $1-\Pr$ independently of $\eta$.

In particular, the homogeneous relaxation factors are
\begin{equation}\label{eq:remainder_factor}
 g_{\rm BGK}=1-\omega,\qquad g_{q,\rm Shakhov}=|1-\omega\Pr|,
 \qquad \omega=\frac{\Delta t_{\rm num}}{\Delta t_{\rm num}+\tau}.
\end{equation}
\end{proposition}
\begin{proof}
For a uniform state the transport terms vanish, so that $R_U=\boldsymbol 0$, the W iteration leaves $\bU$ unchanged, $\Theta=0$ and, by \eqref{eq:residuals}, $R_P=(S^h(f)-f)/\tau$. Equation \eqref{eq:particle_increment} then gives $\Delta P=\omega(S^h(f)-f)$, that is, $f\leftarrow(1-\omega)f+\omega S^h(f)$. For the BGK model $S^h(f)=M$, because the conservative variables do not change, so that $f-M\leftarrow(1-\omega)(f-M)$. The heat flux of the Shakhov target is $(1-\Pr)$ times the heat flux of the distribution it is built from, whence $\bq\leftarrow(1-\omega\Pr)\bq$. The CIS solves the same relaxation equation.
\end{proof}

The remainder factor \eqref{eq:remainder_factor} identifies a floor of the convergence factor set by the explicitly treated remainder. Once the macroscopic iteration has removed the transport modes, the part of the remainder that depends on the non-equilibrium moments of $f$ leaves a mode that contracts only at its own rate, which is $1-\Pr=1/3$ per P update for the Shakhov model with $\Pr=2/3$ and does not exist for the BGK model. The floor is a property of the explicit treatment of the remainder, which the CIS shares, and not of the decomposition, and Section~\ref{sec:convergence} tests this mechanism by computing the same flows with both models.

\section{Numerical verification}\label{sec:results}
The numerical assessment follows the two contributions of the method. Profiles and distributions test whether the WPD representation retains kinetic behaviour and recovers fluid transport, while convergence factors, collision-model comparisons and parameter scans test the implicit coupling and the mechanisms identified in Section~\ref{sec:analysis}. Wall times quantify the gain over CIS, and cavity calculations on the grids used in a published GSIS study provide a comparison with an established accelerated solver. CIS supplies the rarefied-flow reference; continuum comparisons use the reported NS solutions and the independent cavity benchmark of Ghia et al.~\cite{ghia1982}. The common-grid comparisons assess consistency of the discrete fields and distributions, while the refinement study examines iterative performance as the spatial resolution changes.

\subsection{Setup and comparison criteria}\label{sec:assessment}
The test cases are a normal shock, planar Couette flow, the lid-driven cavity and hypersonic flows around a circular cylinder and the X38 section, with $\varepsilon$ from $1$ to $10^{-4}$ (Table~\ref{tab:settings}). Unless stated otherwise, the model is the reduced Shakhov model \eqref{eq:reduced_target} for a monatomic gas with $\gamma=5/3$ and $\Pr=2/3$, with $d=1$ and $K=2$ for the shock and $d=2$ and $K=1$ otherwise; the viscosity follows the variable-hard-sphere law of argon, $\mu=\mu_{\rm ref}(T/T_{\rm ref})^{0.81}$ with $T_{\rm ref}=0.5$ and $\mu_{\rm ref}=0.73103\,\varepsilon$, the relaxation time is $\tau=\mu/p$ and the horizon number $C_\ell=1$. The external flows have the free-stream Mach number $\Ma$, the free-stream state $(\rho_\infty,T_\infty)=(1,0.5)$ and the wall temperature $T_w=0.5$. IWPD denotes the WPD-based multiscale solver of Algorithm~\ref{alg:wp} with the controls W$N_U$/PR$m_U$/P$N_P$ and $m_P=1$, CIS the conventional iterative scheme for the full reduced distribution with one kinetic LU--SGS pair per iteration, and NS the gas-kinetic Navier--Stokes reference computation on the same mesh.
The cavity at $\Rey=100$ and $1000$ employs a sixth-order central reconstruction with one-sided stencils next to the walls, and all other computations the second-order van Leer reconstruction of Section~\ref{sec:fluxes}, with \eqref{eq:gks_flux} for the wave and a gas-kinetic flux for NS and the upwind flux \eqref{eq:sn_face} for the particle of IWPD and the distribution of the CIS. IWPD and CIS use the same mesh, velocity grid and Courant--Friedrichs--Lewy (CFL) numbers, and all computations, NS included, run with 64 OpenMP threads on one node, to the tolerance $10^{-7}$ for the residual norms \eqref{eq:residual_norms} unless stated otherwise. The normal shock is the exception, since its steady state is fixed only up to a translation.

The agreement with a reference solution is measured by the relative difference $E_{L_2}(q)=\|q_{\rm WP}-q_{\rm ref}\|_2/\|q_{\rm ref}\|_2$, where $q$ is a macroscopic field, a profile or a sampled distribution and the Euclidean norm is taken over the corresponding cells, samples or velocity nodes; the values are printed in the figures and collected in Table~\ref{tab:accuracy} for the external flows. Efficiency is quantified by the ratios $S_{\rm iter}=N_{\rm CIS}/N_{\rm WP}$ of the numbers of outer iterations and $S_t=t_{\rm CIS}/t_{\rm WP}$ of the solver wall times, evaluated at the common residual threshold of each comparison, and in the continuum regime by the ratio $t_{\rm WP}/t_{\rm NS}$ of the wall times to the residual $10^{-7}$. The threshold is the tolerance for the internal flows, $\|R_U\|=10^{-6}$ for the external flows and $\|R_U\|=10^{-5}$ for the cylinder at $\varepsilon=10^{-4}$, and is marked by a grey line in the residual figures. The NS computation shares the mesh, reconstruction, implicit macroscopic solver and CFL numbers with IWPD but carries no velocity-space unknowns; its prescribed update counts and point-relaxation pairs define an implementation-specific cost reference for the macroscopic solver alone, without separate NS parameter optimization. The plotted residual is $\|R_U\|$ unless $\|R_P\|$ is stated. The supplementary material shows, in Figures~\SMref{fig:shock_dist}--\SMref{fig:apollo_residual}, the cases and quantities for which the main text gives numbers only.

\begin{table}[!htb]
\centering\small\setlength{\tabcolsep}{2pt}
\caption{Discretizations and controls. Velocity grids refer to the reduced velocity space; the CFL column gives the explicit and the implicit CFL numbers. The Reynolds numbers $100$ and $1000$ correspond to $\varepsilon=2.05\times10^{-3}$ and $2.05\times10^{-4}$. The cavity meshes in parentheses belong to the refinement study of Section~\ref{sec:efficiency}, which uses W24/PR6/P1.}
\label{tab:settings}
\resizebox{\linewidth}{!}{\begin{tabular}{llllll}
\toprule
Case & Regime & Physical mesh & Velocity grid & IWPD control & CFL\\
\midrule
Shock & \shortstack{$\varepsilon=1$,\\ $\Ma=10$} & \shortstack{176 cells\\ on $[-80,40]$} & \shortstack{100\\ on $[-26.4,26.4]$} & W1/PR4/P1 & 0.25/1000\\
Couette & \shortstack{$\varepsilon=1,0.6,0.2,10^{-2}$,\\ $10^{-4}$} & $4\times100$ & \shortstack{$32^2$\\ on $[-6,6]^2$} & \shortstack{W1/PR6/P1,\\ W128/PR2/P1} & 0.02/1000\\
Cavity & $\varepsilon=1,0.075$ & $64^2$ & \shortstack{$64^2$\\ on $[-5,5]^2$} & W1/PR6/P1 & 0.2/1000\\
Cavity & $\Rey=100$ & \shortstack{$64^2$\\ ($150^2$)} & \shortstack{$16^2$\\ on $[-5,5]^2$} & W24/PR4/P1 & 0.2/150\\
Cavity & $\Rey=1000$ & \shortstack{$100^2$\\ ($200^2$, $300^2$)} & \shortstack{$16^2$\\ on $[-5,5]^2$} & W24/PR3/P1 & 0.2/150\\
Cylinder & \shortstack{$\varepsilon=1$,\\ $\Ma=5$} & $129\times80$ O-grid & \shortstack{$96^2$\\ on $[-9,9]^2$} & W1/PR2/P1 & 0.2/1000\\
Cylinder & \shortstack{$\varepsilon=10^{-2}$,\\ $10^{-4}$} & $129\times40$ O-grid & \shortstack{$96^2$,\\ $64^2$\\ on $[-9,9]^2$} & \shortstack{W1,\\ W120/PR2/P1} & 0.2/1000\\
X38 & \shortstack{$\varepsilon=1$,\\ $\Ma=5$} & $48\times31$ O-grid & \shortstack{$96^2$\\ on $[-9,9]^2$} & W1/PR2/P1 & 0.2/1000\\
X38 & \shortstack{$\varepsilon=10^{-2}$,\\ $10^{-4}$} & $48\times20$ O-grid & \shortstack{$96^2$,\\ $64^2$\\ on $[-9,9]^2$} & \shortstack{W1,\\ W120/PR2/P1} & 0.2/1000\\
\bottomrule
\end{tabular}}
\end{table}

\subsection{Kinetic and hydrodynamic behaviour}\label{sec:test_summary}
The shock tests the non-equilibrium distribution and its transport moments. Couette flow follows the transition from rarefied slip to continuum shear, and the cavity tests fluid behaviour on meshes that do not resolve the mean free path. The external-flow fields and distributions are examined in Section~\ref{sec:external}.

\subsubsection{Normal shock}\label{sec:res_shock}
The Mach-10 normal shock at $\varepsilon=1$ connects the Rankine--Hugoniot states $(\rho,u,T)=(1,9.13,0.5)$ and $(3.88,2.35,16.06)$, where $u$ is the velocity normal to the shock, with relaxation times $1.46$ and $0.19$; the initial discontinuity is centred at $x=-16$ with width $2$. A free-standing shock is determined only up to a translation, and its centre keeps oscillating slightly about the steady position, so that the residual levels off once the profile shape is converged. The two methods are therefore compared on the converged profiles, each plotted against the distance from its own density midpoint (Figure~\ref{fig:shock_macro}). The normalized density, velocity and temperature of IWPD and CIS differ by $3.44\times10^{-4}$ to $4.37\times10^{-4}$ and the heat flux and the normal stress, which reach their extrema upstream of the density midpoint, by $1.95\times10^{-3}$. The reduced distributions at four positions across the shock (Figure~\SMref{fig:shock_dist}) differ by $6.39\times10^{-4}$ to $1.66\times10^{-3}$. Upstream of the midpoint they consist of a narrow peak at the upstream velocity and a broad hot population transported from downstream, and downstream the relaxed distribution still carries a shoulder of the upstream beam. IWPD captures both structures, neither of which a single Maxwellian can represent.

\begin{figure}[!htb]
\centering\singlespacing\setlength{\figw}{\linewidth}
\begin{subfigure}[t]{0.4641\figw}\centering
\raisebox{0.0000\figw}{\includegraphics[width=0.4641\figw]{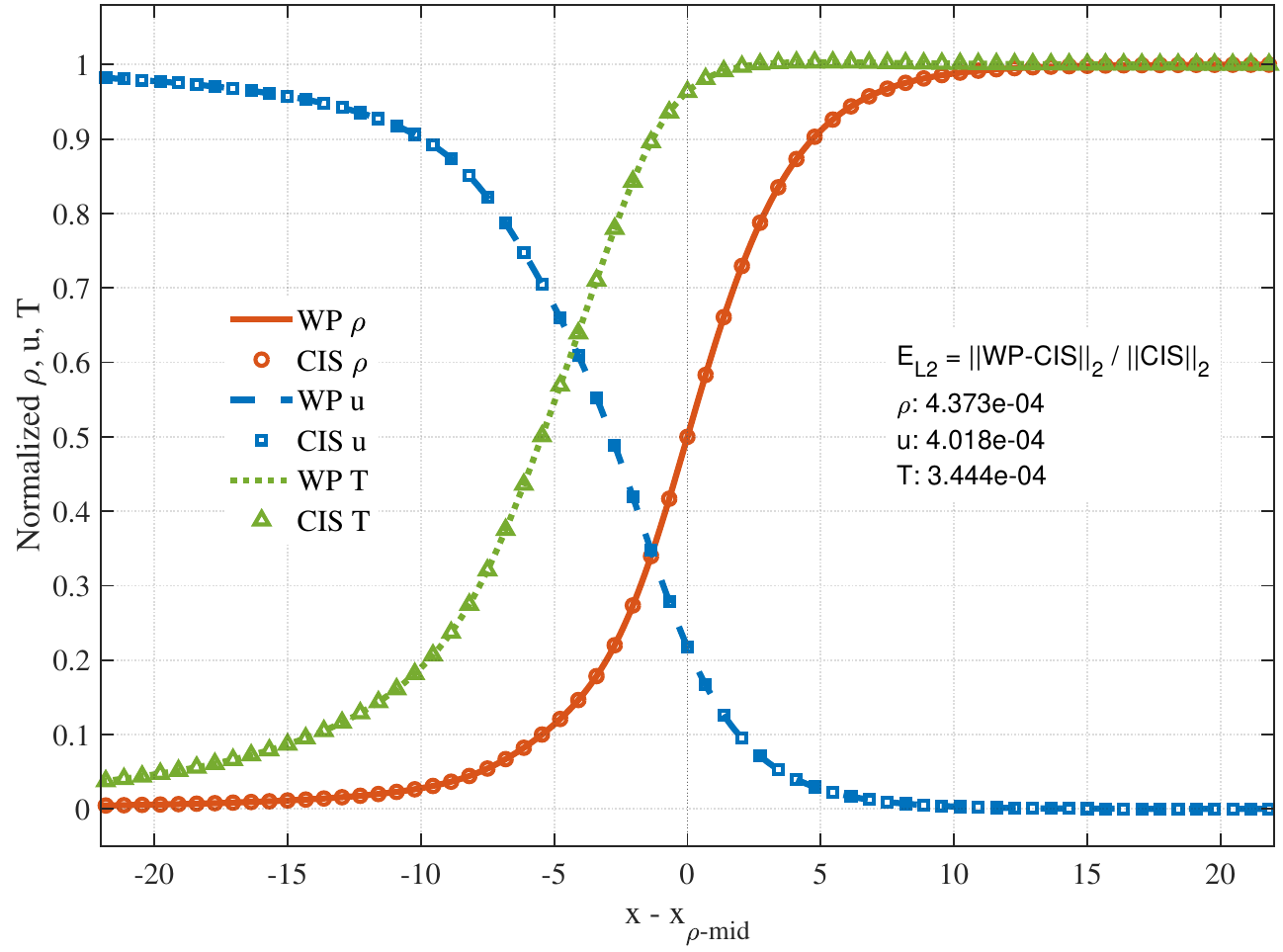}}
\caption{Density, velocity and temperature.}
\end{subfigure}\hspace{0.0200\figw}%
\begin{subfigure}[t]{0.5109\figw}\centering
\raisebox{0.0000\figw}{\includegraphics[width=0.5109\figw]{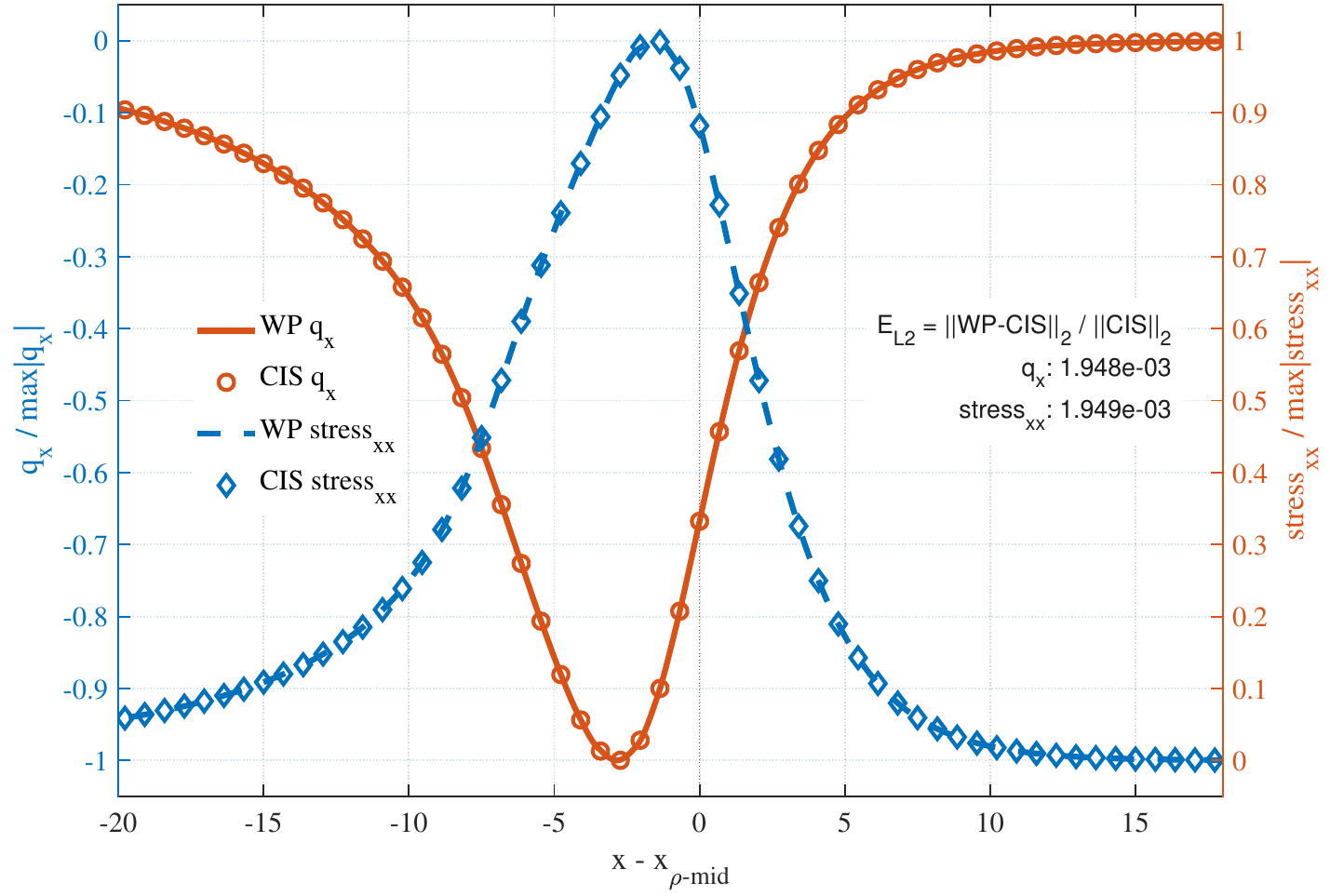}}
\caption{Heat flux and normal stress.}
\end{subfigure}
\caption{Normal shock at $\mathrm{Ma}=10$ and $\varepsilon=1$: normalized macroscopic profiles plotted against the distance from the density midpoint. Lines denote IWPD and symbols CIS; the relative differences $E_{L_2}$ are given in the panels.}\label{fig:shock_macro}
\end{figure}

\subsubsection{Couette flow}\label{sec:res_couette}
Planar Couette flow is computed between diffuse walls at $y=0$ and $y=1$ with temperature $0.5$, the upper wall moving in the $x$ direction with speed $0.15$, on a mesh that is periodic in $x$, with the gas initially at rest at $\rho=1$ and $T=0.5$. The profiles in Figure~\ref{fig:couette_profiles} span the transition from strongly rarefied to near-continuum flow. For $\varepsilon\geq10^{-2}$ IWPD and CIS differ by $2.96\times10^{-8}$ to $2.94\times10^{-6}$ in velocity and by $7.59\times10^{-11}$ to $9.27\times10^{-9}$ in temperature. The velocity slip at the walls, pronounced at $\varepsilon=1$, decreases with the Knudsen number, and at $\varepsilon=10^{-4}$ the velocity profile is the linear continuum solution, and the two methods differ by $8.73\times10^{-4}$ in velocity and $1.56\times10^{-5}$ in temperature. 

\begin{figure}[!htb]
\centering\singlespacing\setlength{\figw}{\linewidth}
\begin{subfigure}[t]{0.4880\figw}\centering
\raisebox{0.0030\figw}{\includegraphics[width=0.4880\figw]{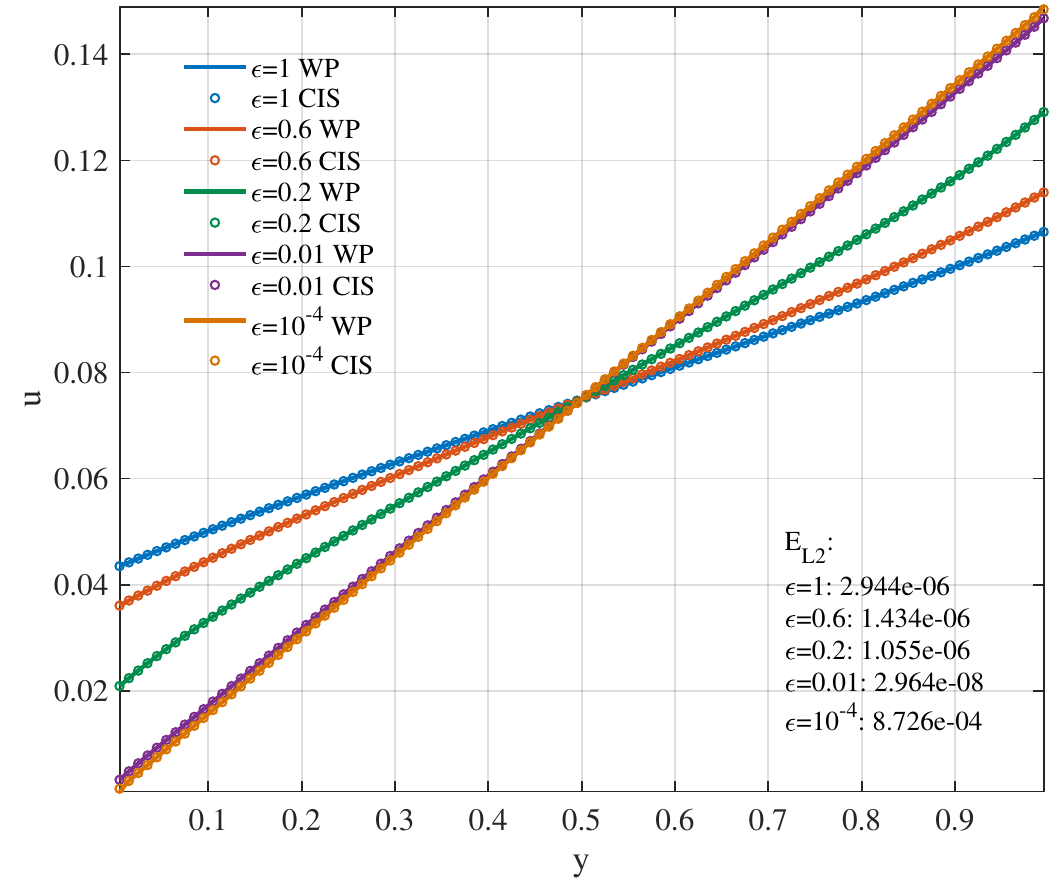}}
\caption{Velocity.}
\end{subfigure}\hspace{0.0200\figw}%
\begin{subfigure}[t]{0.4870\figw}\centering
\raisebox{0.0000\figw}{\includegraphics[width=0.4870\figw]{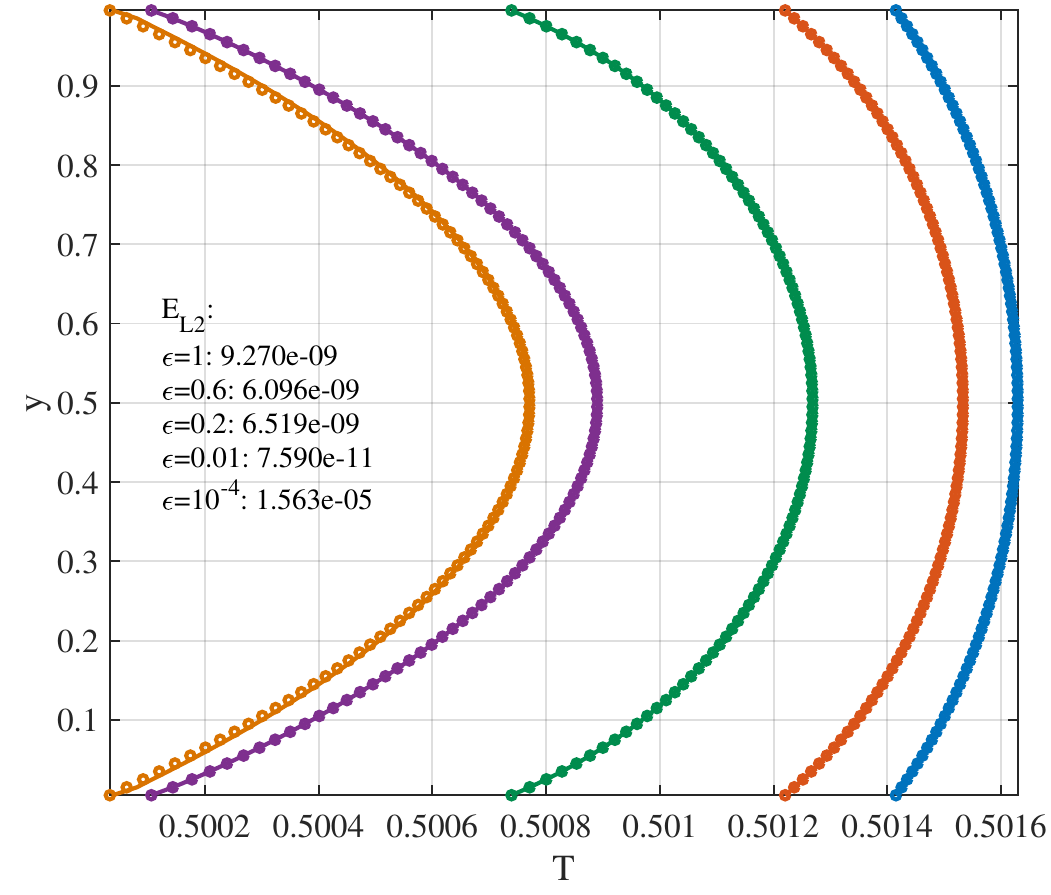}}
\caption{Temperature.}
\end{subfigure}
\caption{Couette flow at $\varepsilon=1$, 0.6, 0.2, $10^{-2}$ and $10^{-4}$: velocity and temperature profiles across the channel. Lines denote IWPD and symbols CIS; the legend of (a) applies to both panels, and the relative differences for each Knudsen number are listed in the panels.}\label{fig:couette_profiles}
\end{figure}

\subsubsection{Lid-driven cavity}\label{sec:res_cavity}
\begin{figure}[!htb]
\centering\singlespacing\setlength{\figw}{\linewidth}
\begin{subfigure}[t]{0.4869\figw}\centering
\raisebox{0.0000\figw}{\includegraphics[width=0.4869\figw]{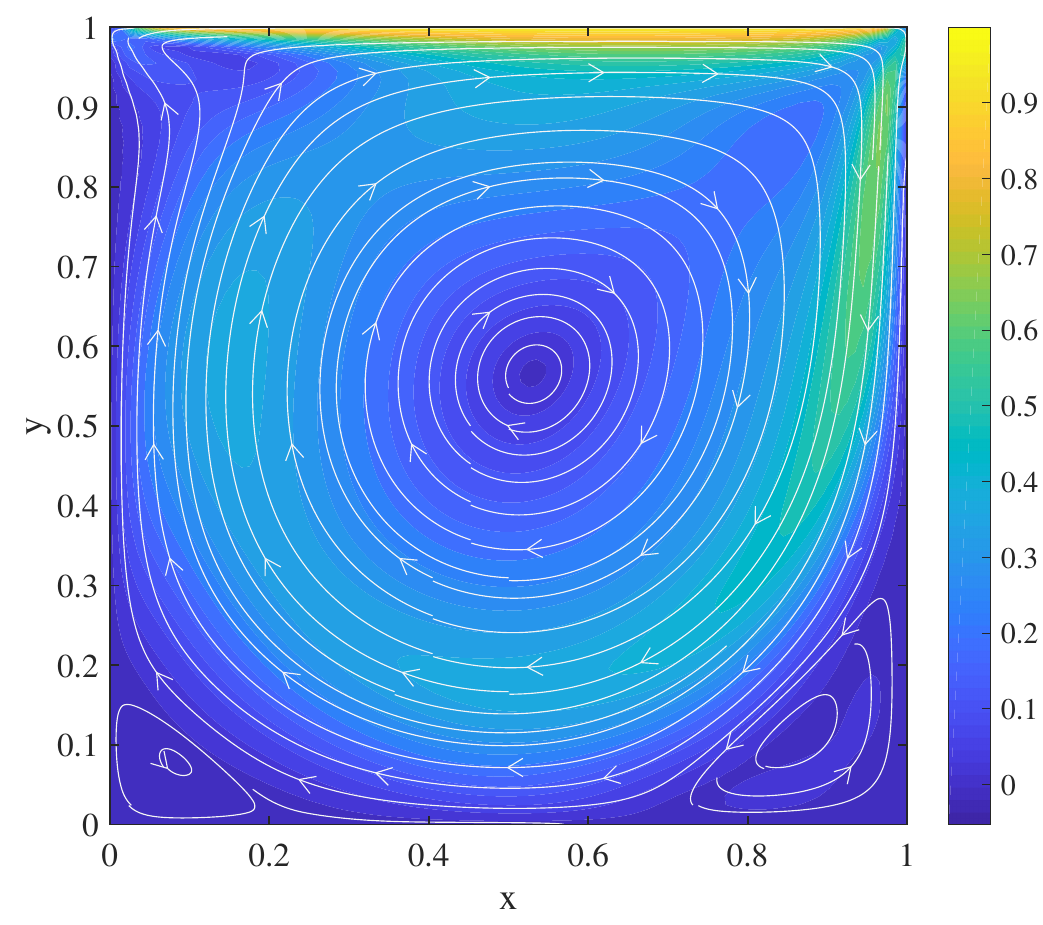}}
\caption{Speed and streamlines.}
\end{subfigure}\hspace{0.0200\figw}%
\begin{subfigure}[t]{0.4881\figw}\centering
\raisebox{0.0000\figw}{\includegraphics[width=0.4881\figw]{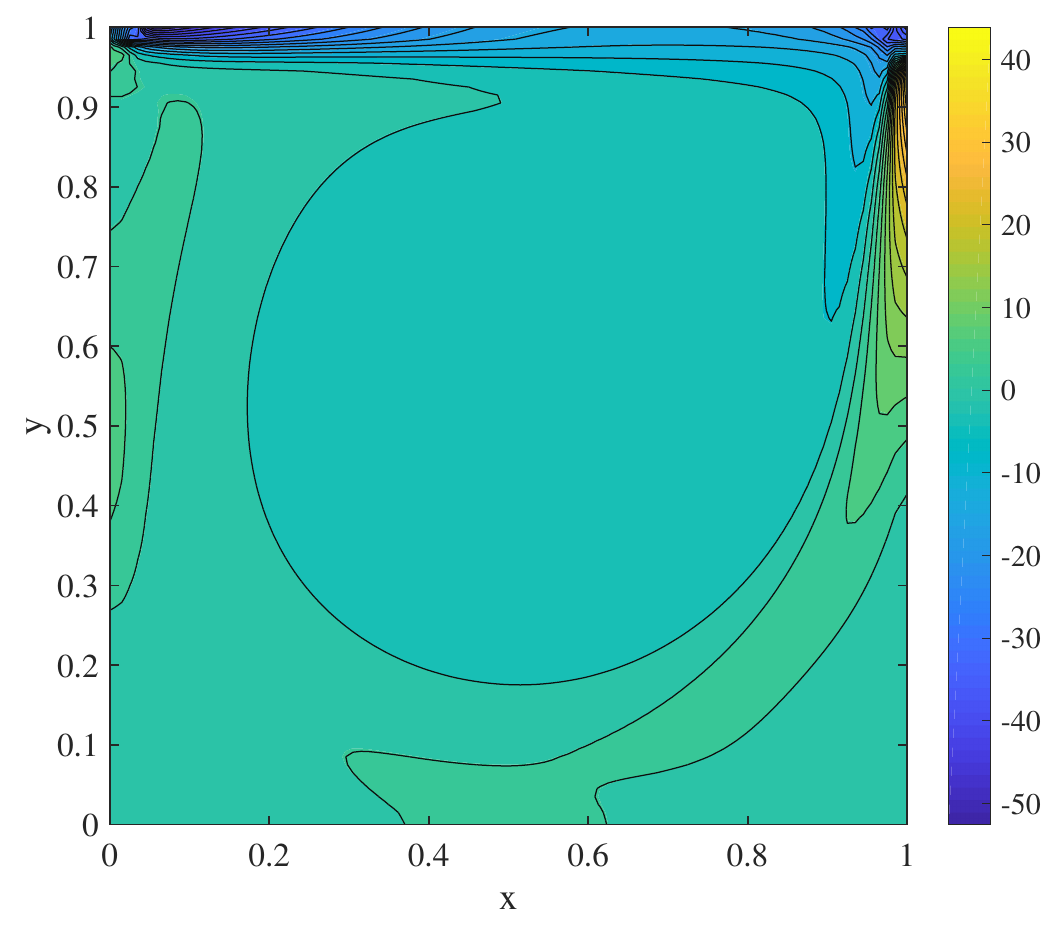}}
\caption{Vorticity.}
\end{subfigure}
\par\smallskip
\begin{subfigure}[t]{0.4812\figw}\centering
\raisebox{0.0000\figw}{\includegraphics[width=0.4812\figw]{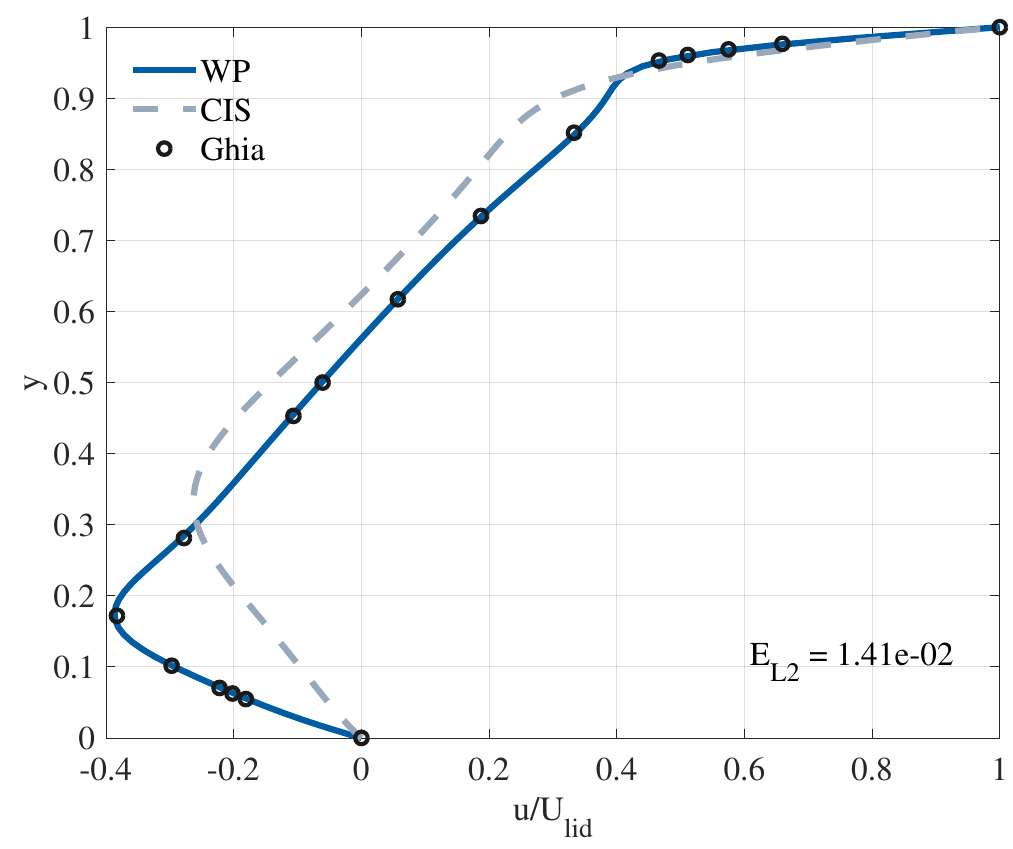}}
\caption{$u/U_{\rm lid}$ on the vertical centreline.}
\end{subfigure}\hspace{0.0200\figw}%
\begin{subfigure}[t]{0.4938\figw}\centering
\raisebox{0.0070\figw}{\includegraphics[width=0.4938\figw]{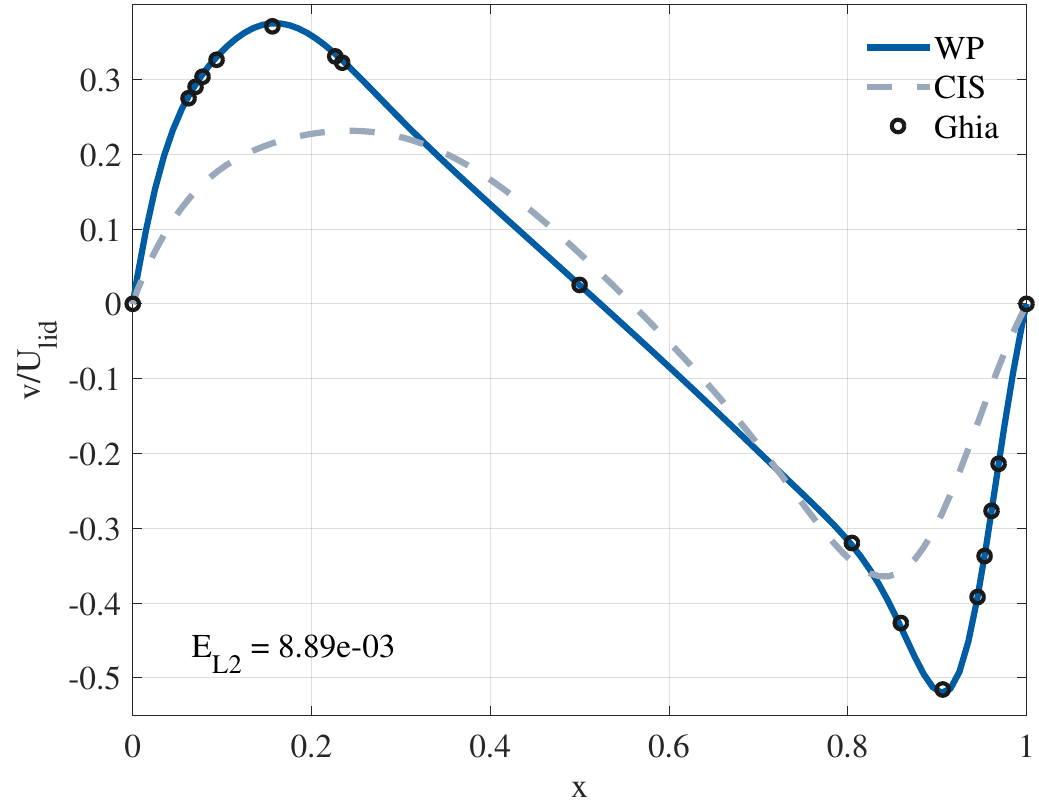}}
\caption{$v/U_{\rm lid}$ on the horizontal centreline.}
\end{subfigure}
\caption{Lid-driven cavity at $\mathrm{Re}=1000$: (a,b) fields of IWPD (filled contours) and NS (lines, streamlines in (a)); (c,d) centreline profiles of IWPD (solid) and CIS (dashed) compared with the benchmark data of Ghia et al.~\cite{ghia1982} (symbols); the panels give the relative difference $E_{L_2}$ between IWPD and the benchmark data.}\label{fig:cavity_re1000}
\end{figure}

The lid-driven cavity is the unit square with stationary diffuse walls at temperature $0.5$ and a lid moving with speed $U_{\rm lid}=0.15$, and the Reynolds number $\Rey=\rho U_{\rm lid}L/\mu_{\rm ref}$ is based on the cavity length $L=1$. At $\varepsilon=1$ and $0.075$ (Figures~\SMref{fig:cavity_kn1_fields}--\SMref{fig:cavity_rarefied_profiles}) the IWPD and CIS fields differ by less than $5\times10^{-8}$ and the centreline velocity profiles by $2.46\times10^{-8}$ to $1.08\times10^{-7}$. The continuum cases $\Rey=100$ and $1000$ are computed with a $16^2$ velocity grid on meshes that do not resolve the mean free path and therefore probe the formal expansion in Section~\ref{sec:ap}. The IWPD speed fields differ from NS by $8.08\times10^{-3}$ and $1.54\times10^{-2}$, and the streamlines reproduce the primary vortex and, at $\Rey=1000$, the two secondary vortices in the lower corners (Figure~\ref{fig:cavity_re1000} and Figure~\SMref{fig:cavity_re100}). The centreline profiles of $u$ and $v$ pass through the benchmark points of Ghia et al.~\cite{ghia1982}, from which they differ by $1.87\times10^{-3}$ to $1.41\times10^{-2}$, whereas the CIS profiles computed on the same mesh and velocity grid depart from the benchmark. On these meshes the two methods thus converge to different discrete solutions. The kinetic transport of the CIS acts on the full distribution, on a mesh and a velocity grid that are too coarse for it, whereas in IWPD it acts only on the particle, which is of order $\varepsilon$ (the formal expansion in Section~\ref{sec:ap}), while the wave supplies the Euler and relaxation fluxes in closed form. Thus the comparison tests the transport discretization as well as iterative convergence.

\subsubsection{Hypersonic external flows}\label{sec:external}
The cylinder tests the non-equilibrium distribution in detail, while X38 extends the verification to a more complex section. Table~\ref{tab:accuracy} reports field differences from CIS at $\varepsilon=1$ and $10^{-2}$ and from the NS reference at $\varepsilon=10^{-4}$. Wall profiles use values in the cells adjacent to the surface. Additional results for the Apollo section, including collision-model comparisons, are collected in Supplementary Section~S6.

\begin{table}[!htb]
\centering\footnotesize
\caption{External flows: relative differences $E_{L_2}$ between the macroscopic fields of IWPD and of the reference, which is CIS at $\varepsilon=1$ and $10^{-2}$ and NS at $\varepsilon=10^{-4}$.}
\label{tab:accuracy}
\begin{tabular}{llrrrr}
\toprule
Case & $\varepsilon$ & Density & Speed & Temperature & Pressure\\
\midrule
Cylinder & $1$       & $1.94\times10^{-5}$ & $2.75\times10^{-6}$ & $8.93\times10^{-6}$ & $7.11\times10^{-6}$\\
         & $10^{-2}$ & $4.44\times10^{-4}$ & $1.89\times10^{-4}$ & $8.67\times10^{-4}$ & $4.86\times10^{-4}$\\
         & $10^{-4}$ & $4.12\times10^{-4}$ & $2.85\times10^{-4}$ & $1.06\times10^{-3}$ & $7.77\times10^{-4}$\\
X38      & $1$       & $2.66\times10^{-6}$ & $8.31\times10^{-7}$ & $2.84\times10^{-6}$ & $2.99\times10^{-6}$\\
         & $10^{-2}$ & $5.71\times10^{-5}$ & $1.38\times10^{-5}$ & $8.40\times10^{-5}$ & $6.26\times10^{-5}$\\
         & $10^{-4}$ & $4.42\times10^{-4}$ & $2.10\times10^{-4}$ & $9.03\times10^{-4}$ & $7.89\times10^{-4}$\\
\bottomrule
\end{tabular}
\end{table}

\paragraph{Circular cylinder}\label{sec:res_cylinder}
\begin{figure}[!htb]
\centering\singlespacing\setlength{\figw}{\linewidth}
\begin{subfigure}[t]{0.4856\figw}\centering
\raisebox{0.0000\figw}{\includegraphics[width=0.4856\figw]{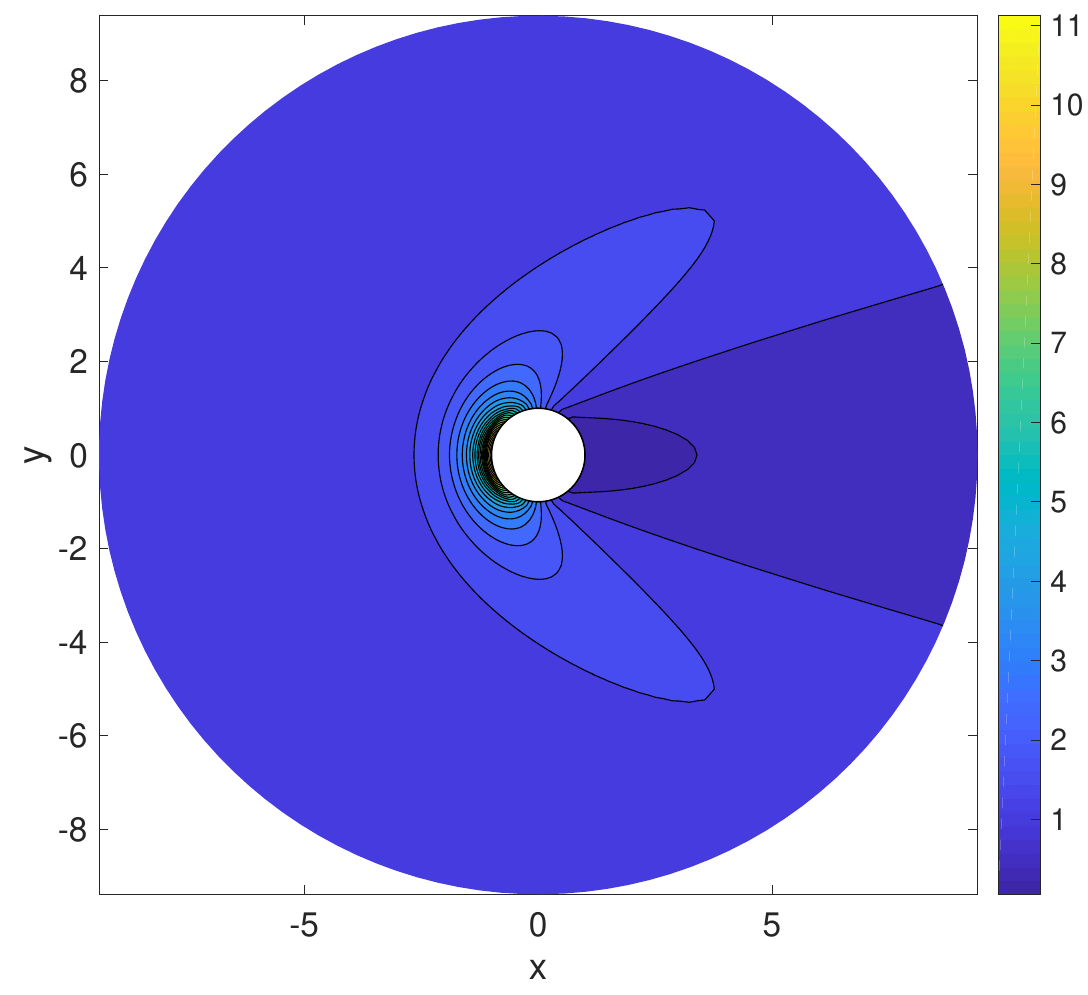}}
\caption{Density.}
\end{subfigure}\hspace{0.0200\figw}%
\begin{subfigure}[t]{0.4894\figw}\centering
\raisebox{0.0000\figw}{\includegraphics[width=0.4894\figw]{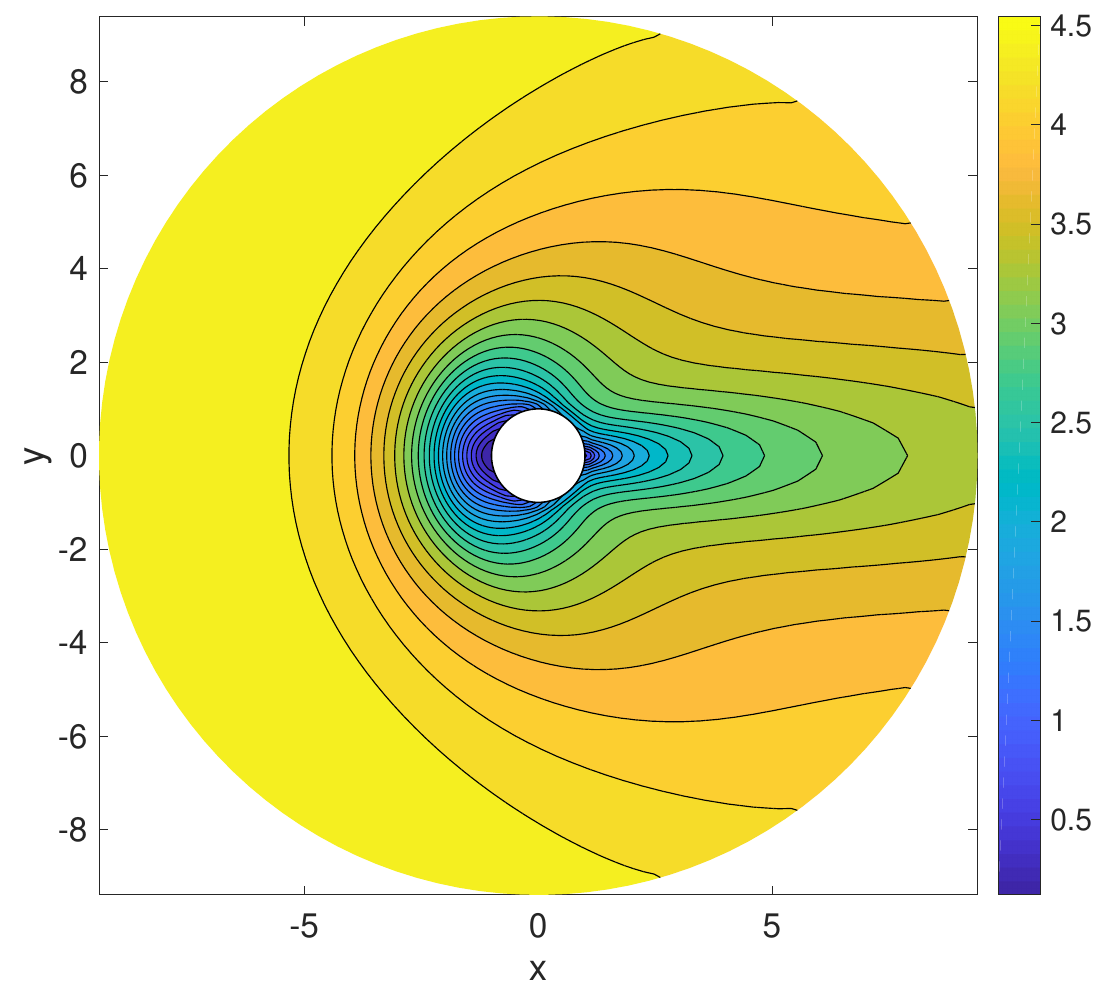}}
\caption{Speed.}
\end{subfigure}
\par\smallskip
\begin{subfigure}[t]{0.4894\figw}\centering
\raisebox{0.0000\figw}{\includegraphics[width=0.4894\figw]{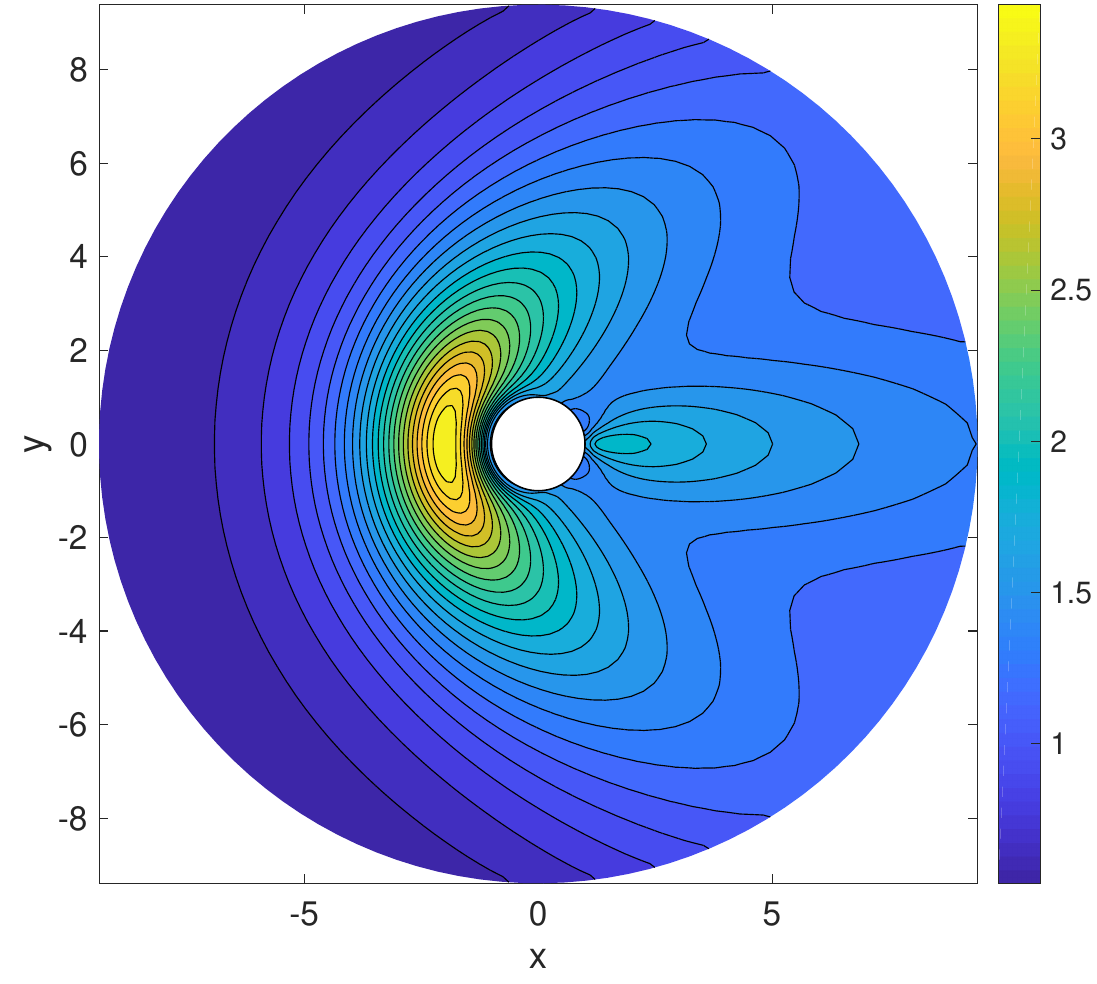}}
\caption{Temperature.}
\end{subfigure}\hspace{0.0200\figw}%
\begin{subfigure}[t]{0.4856\figw}\centering
\raisebox{0.0000\figw}{\includegraphics[width=0.4856\figw]{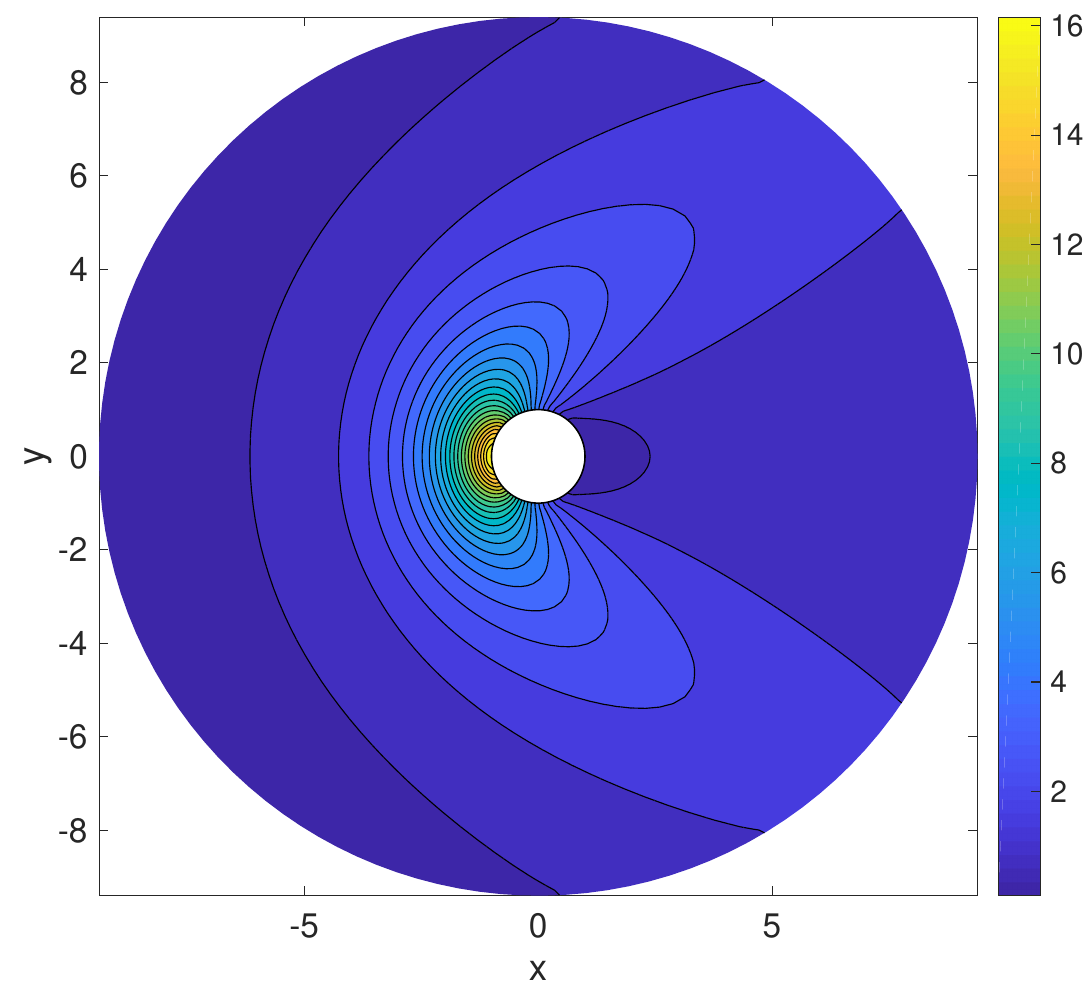}}
\caption{Pressure.}
\end{subfigure}
\caption{Circular cylinder at $\mathrm{Ma}=5$ and $\varepsilon=1$: macroscopic fields of IWPD (filled contours) and CIS (lines) on common levels.}\label{fig:cyl_kn1_fields}
\end{figure}

Mach-5 flow around a cylinder of unit radius with a stationary diffuse wall is computed on body-fitted O-grids with outer radius $9.40$ at $\varepsilon=1$ and $6.26$ at $10^{-2}$ and $10^{-4}$; the field differences from the reference are listed in Table~\ref{tab:accuracy}. At $\varepsilon=1$ the strongly rarefied shock layer extends several radii upstream (Figure~\ref{fig:cyl_kn1_fields}), and the temperature and pressure along the stagnation line and at the wall coincide for the two methods, including the maximum stagnation-line pressure of $16.16$ (Figure~\SMref{fig:cyl_kn1_profiles}). The reduced distributions on the symmetry line (Figure~\ref{fig:cyl_kn1_dist_lin}, and on a logarithmic scale Figure~\SMref{fig:cyl_kn1_dist_log}) consist ahead of the body of the free-stream beam and a slow population reflected from the wall, which dominates next to the stagnation point, and in the wake of two lobes populated from the upper and the lower side of the body. IWPD reproduces these structures on both scales with relative differences of $3.72\times10^{-6}$ to $3.63\times10^{-4}$, the largest next to the stagnation point.

\begin{figure}[!htb]
\centering\singlespacing\setlength{\figw}{\linewidth}
\begin{subfigure}[t]{0.4908\figw}\centering
\raisebox{0.0000\figw}{\includegraphics[width=0.4908\figw]{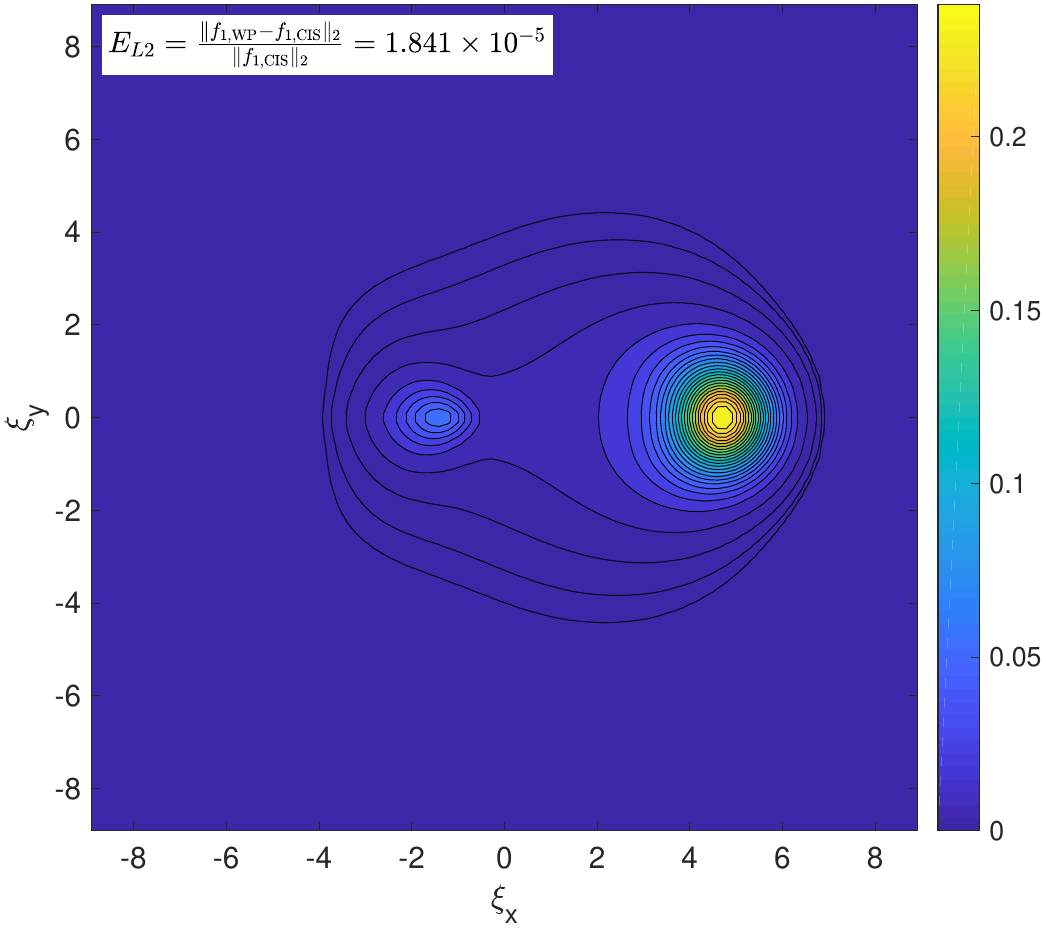}}
\caption{$x=-3.09$.}
\end{subfigure}\hspace{0.0200\figw}%
\begin{subfigure}[t]{0.4842\figw}\centering
\raisebox{0.0000\figw}{\includegraphics[width=0.4842\figw]{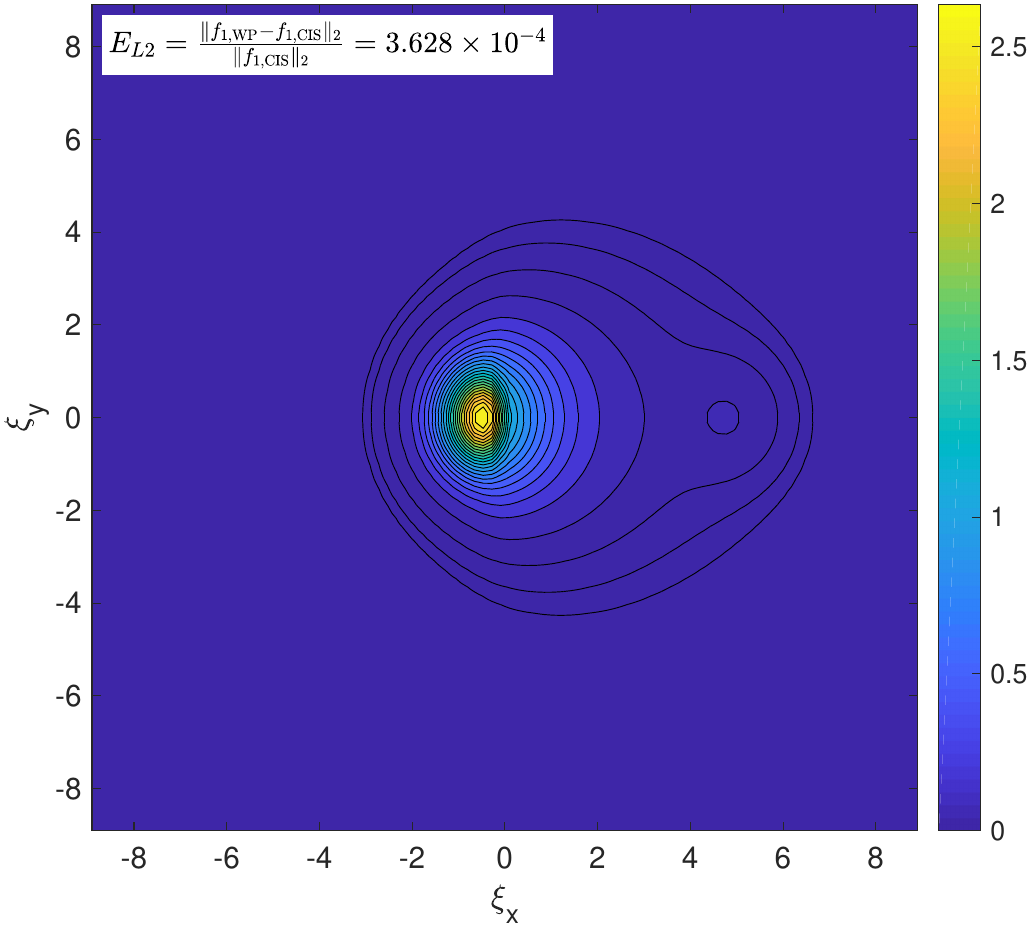}}
\caption{$x=-1.05$.}
\end{subfigure}
\par\smallskip
\begin{subfigure}[t]{0.4873\figw}\centering
\raisebox{0.0000\figw}{\includegraphics[width=0.4873\figw]{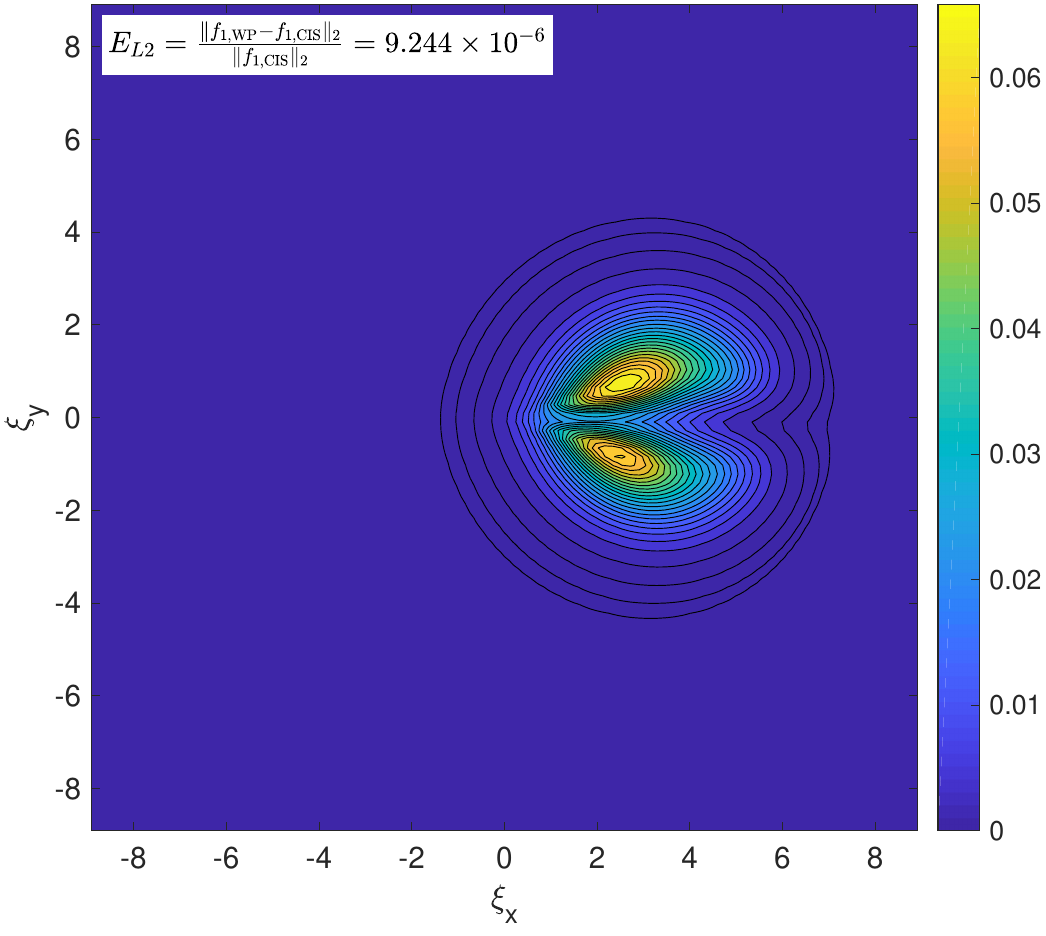}}
\caption{$x=5.16$.}
\end{subfigure}\hspace{0.0200\figw}%
\begin{subfigure}[t]{0.4877\figw}\centering
\raisebox{0.0000\figw}{\includegraphics[width=0.4877\figw]{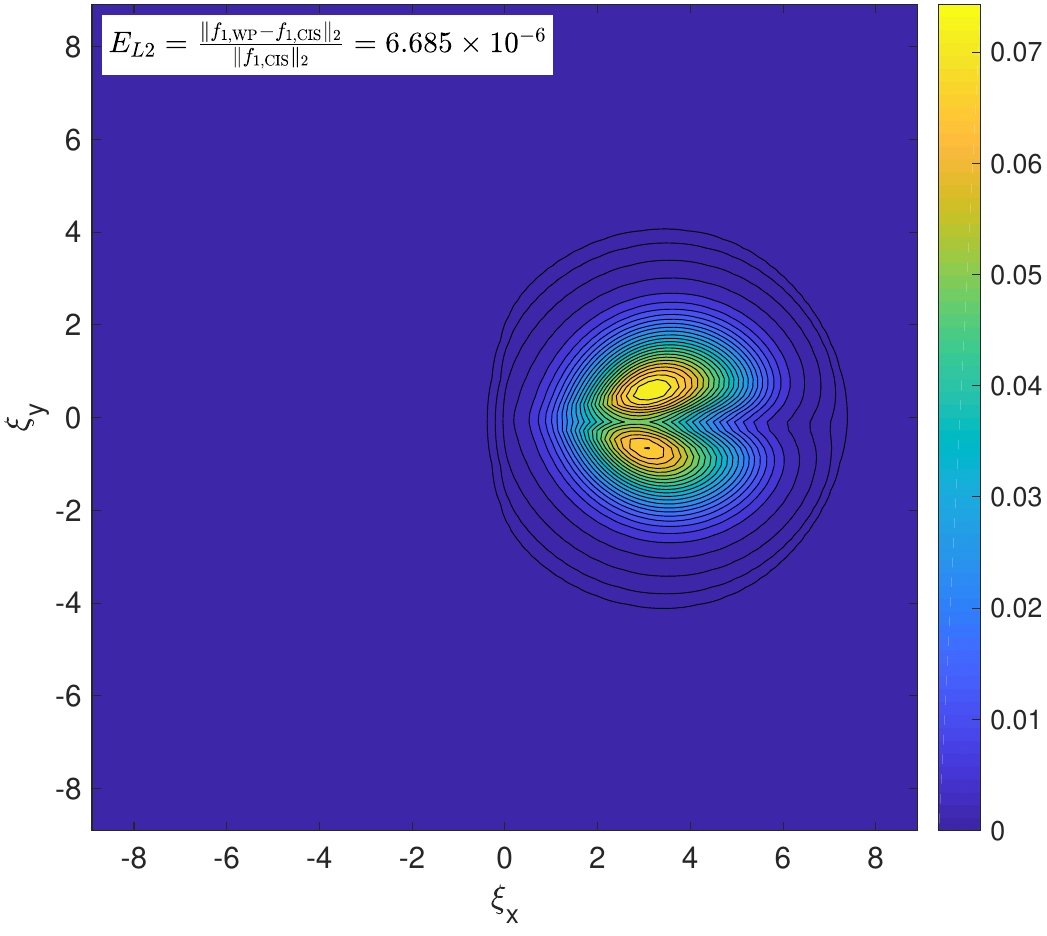}}
\caption{$x=9.34$.}
\end{subfigure}
\caption{Circular cylinder at $\mathrm{Ma}=5$ and $\varepsilon=1$: reduced distributions $f_1(\xi_x,\xi_y)$ on a linear scale at two points ahead of the body (a,b) and two points in the wake (c,d), all close to the symmetry axis. IWPD is shown by filled contours and CIS by lines.}\label{fig:cyl_kn1_dist_lin}
\end{figure}

At $\varepsilon=10^{-2}$ the shock layer is much thinner, and at $\varepsilon=10^{-4}$ the stagnation-line and wall profiles of IWPD and NS coincide across the bow shock and along the body (Figures~\SMref{fig:cyl_kn1em2_fields}--\SMref{fig:cyl_kn1em4_profiles}). At this Knudsen number the meshes do not resolve the boundary layer, and the wall profiles, which show the values in the cells adjacent to the wall, differ accordingly from the wall temperature. The residual histories (Figure~\SMref{fig:cyl_residual}) show the acceleration toward the continuum regime. At $\varepsilon=10^{-4}$ IWPD passes $\|R_U\|=10^{-5}$ after $15$ outer iterations and reaches the tolerance after $20$, whereas the CIS needs $856\,940$ iterations to reach $\|R_U\|=10^{-5}$; NS needs $390$ implicit steps and $46.3$~s to the tolerance against $121.8$~s for IWPD, so that $t_{\rm WP}/t_{\rm NS}=2.63$ although IWPD transports the particle on $64^2$ velocity nodes.

\paragraph{X38 section}\label{sec:res_x38}
\begin{figure}[!htb]
\centering\singlespacing\setlength{\figw}{\linewidth}
\begin{subfigure}[t]{0.4875\figw}\centering
\raisebox{0.0000\figw}{\includegraphics[width=0.4875\figw]{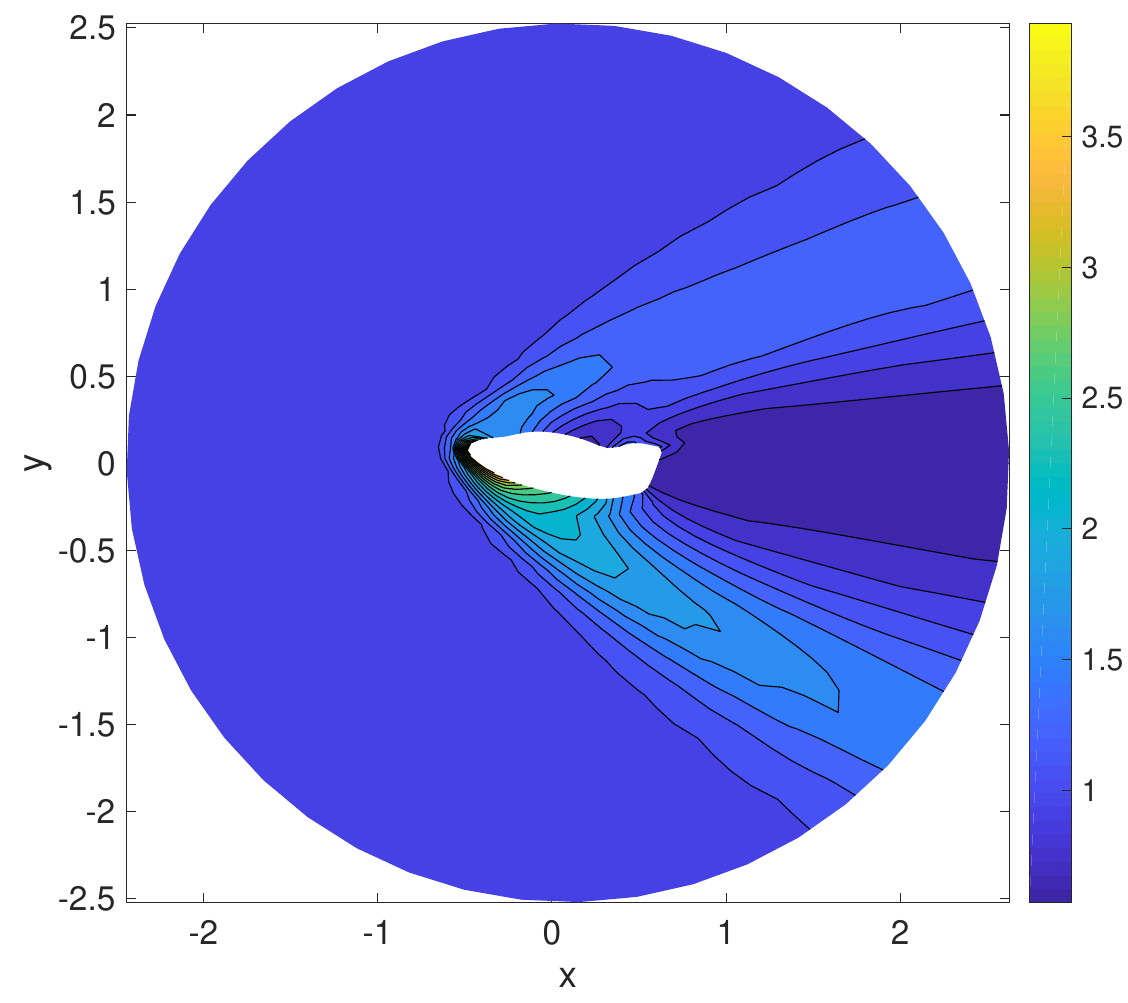}}
\caption{Density.}
\end{subfigure}\hspace{0.0200\figw}%
\begin{subfigure}[t]{0.4875\figw}\centering
\raisebox{0.0000\figw}{\includegraphics[width=0.4875\figw]{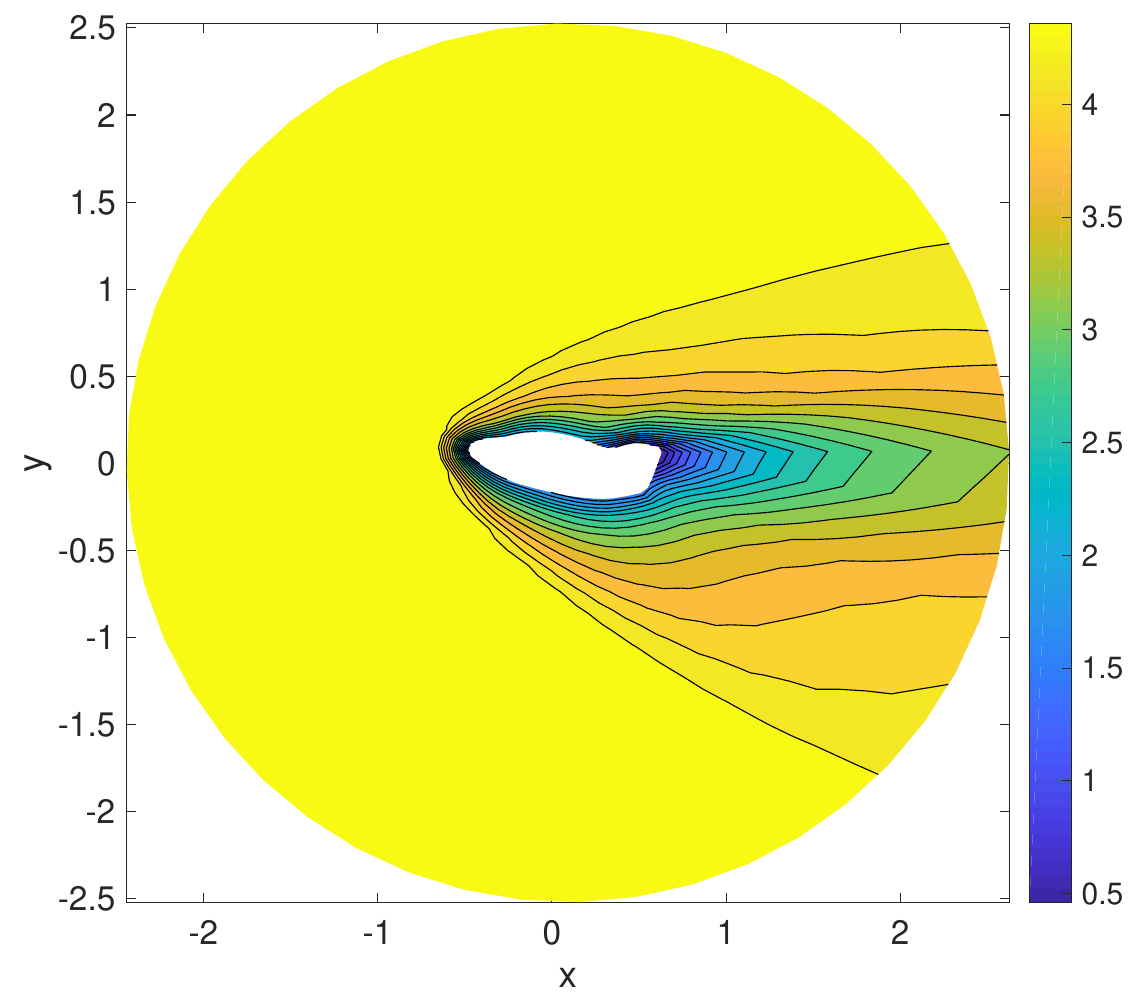}}
\caption{Speed.}
\end{subfigure}
\par\smallskip
\begin{subfigure}[t]{0.4930\figw}\centering
\raisebox{0.0000\figw}{\includegraphics[width=0.4930\figw]{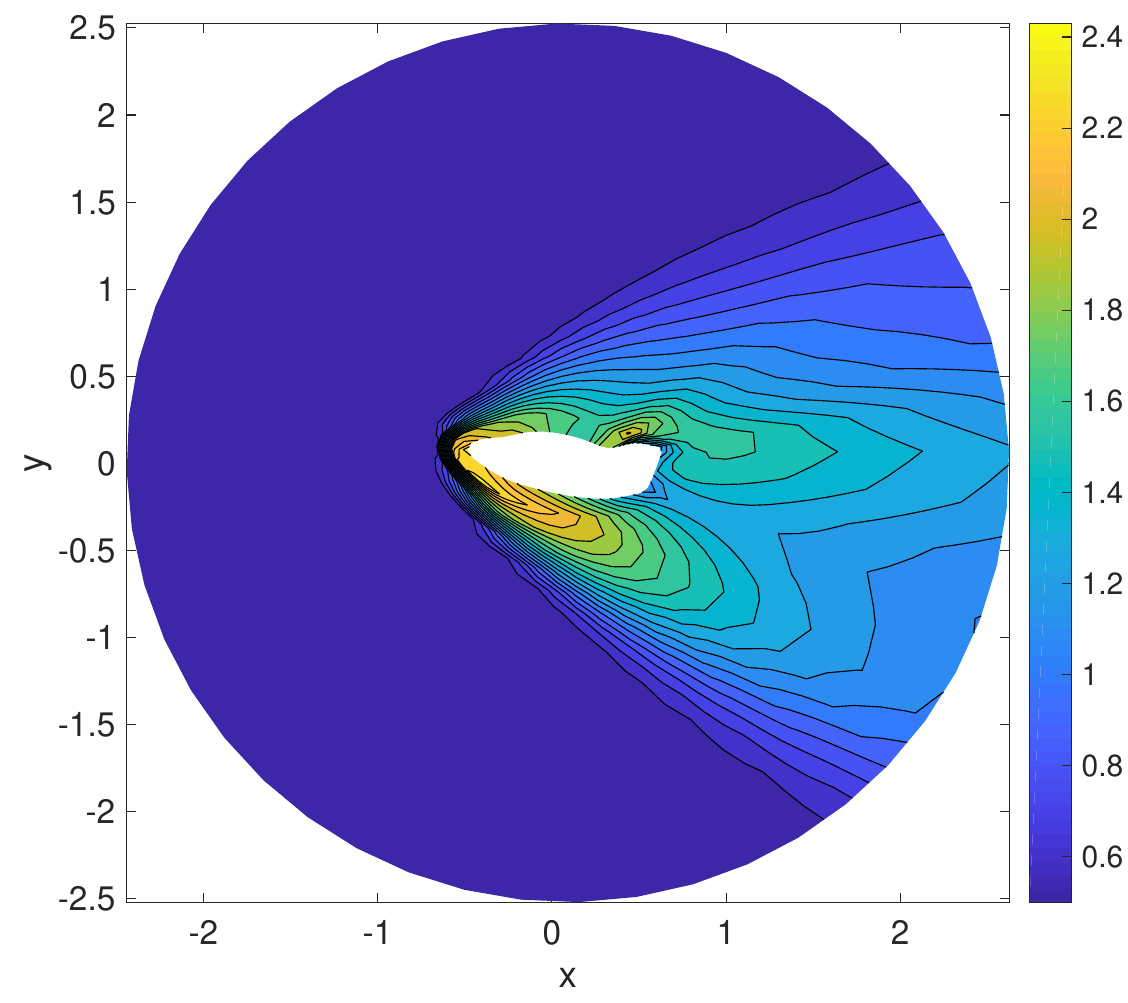}}
\caption{Temperature.}
\end{subfigure}\hspace{0.0200\figw}%
\begin{subfigure}[t]{0.4820\figw}\centering
\raisebox{0.0000\figw}{\includegraphics[width=0.4820\figw]{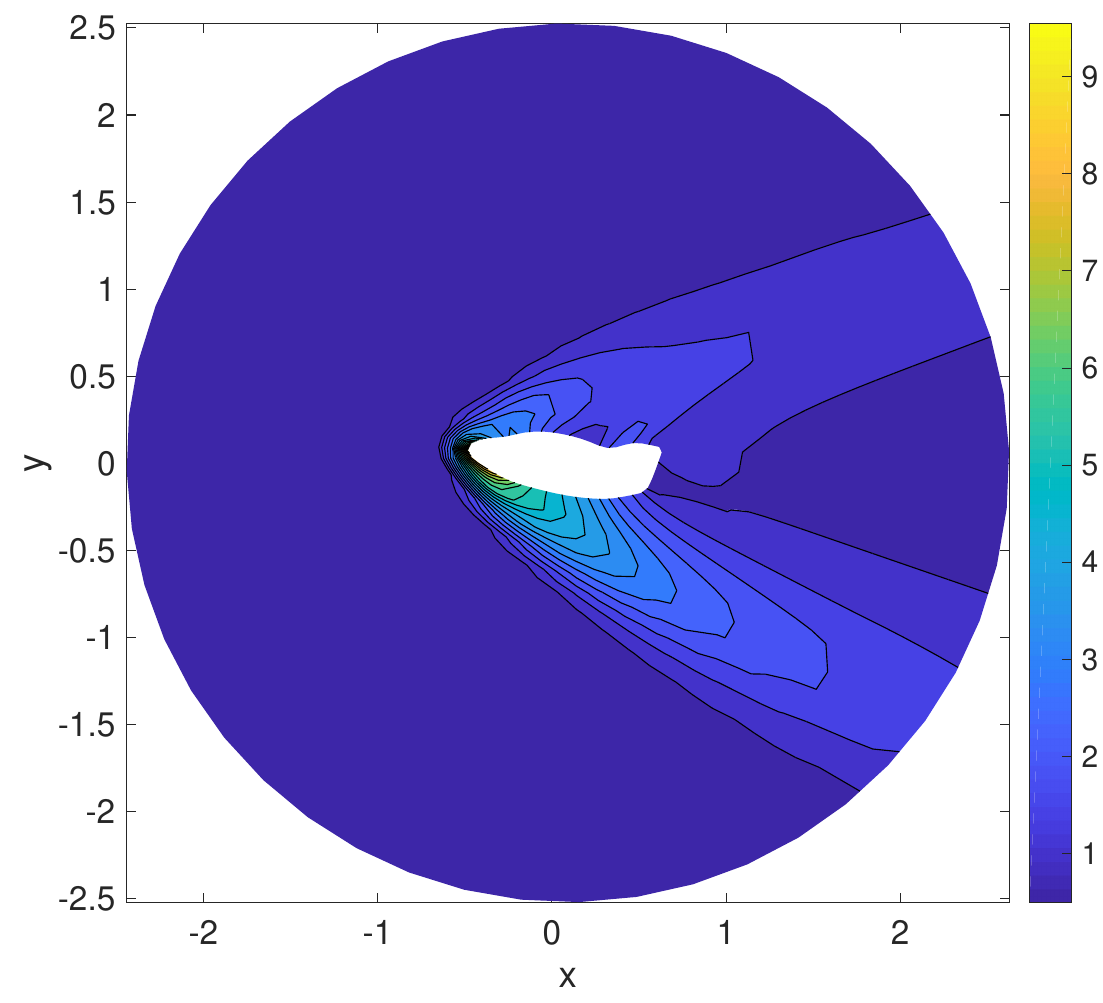}}
\caption{Pressure.}
\end{subfigure}
\caption{X38 section at $\mathrm{Ma}=5$ and $\varepsilon=10^{-4}$: macroscopic fields of IWPD (filled contours) and NS (lines) on common levels.}\label{fig:x38_kn1em4_fields}
\end{figure}

The X38 section uses periodic body-fitted O-grids with the settings in Table~\ref{tab:settings}. The stationary walls are diffuse, and the incoming far-field distribution is the free-stream Maxwellian. Figure~\ref{fig:x38_kn1em4_fields} shows the continuum fields; Supplementary Figures~\SMref{fig:x38_kn1_fields}--\SMref{fig:x38_kn1em4_profiles} provide the other fields, profiles and sampled distributions. The rarefied distributions contain a free-stream beam and a population reflected from the wall. At two points near X38, the IWPD--CIS distribution differences are $1.69\times10^{-4}$ and $1.82\times10^{-4}$ at $\varepsilon=10^{-2}$.

Figure~\ref{fig:x38_residual} shows that the X38 convergence follows the cylinder trend. The same WPD iteration covers the three regimes; the gain over CIS grows as collisions become dominant. Iteration counts and wall-time ratios are collected in Table~\ref{tab:efficiency}.

\begin{figure}[!htb]
\centering\singlespacing\setlength{\figw}{.97\linewidth}
\begin{subfigure}[t]{0.9950\figw}\centering
\makebox[\linewidth]{\raisebox{0.0000\figw}{\includegraphics[width=0.4944\figw]{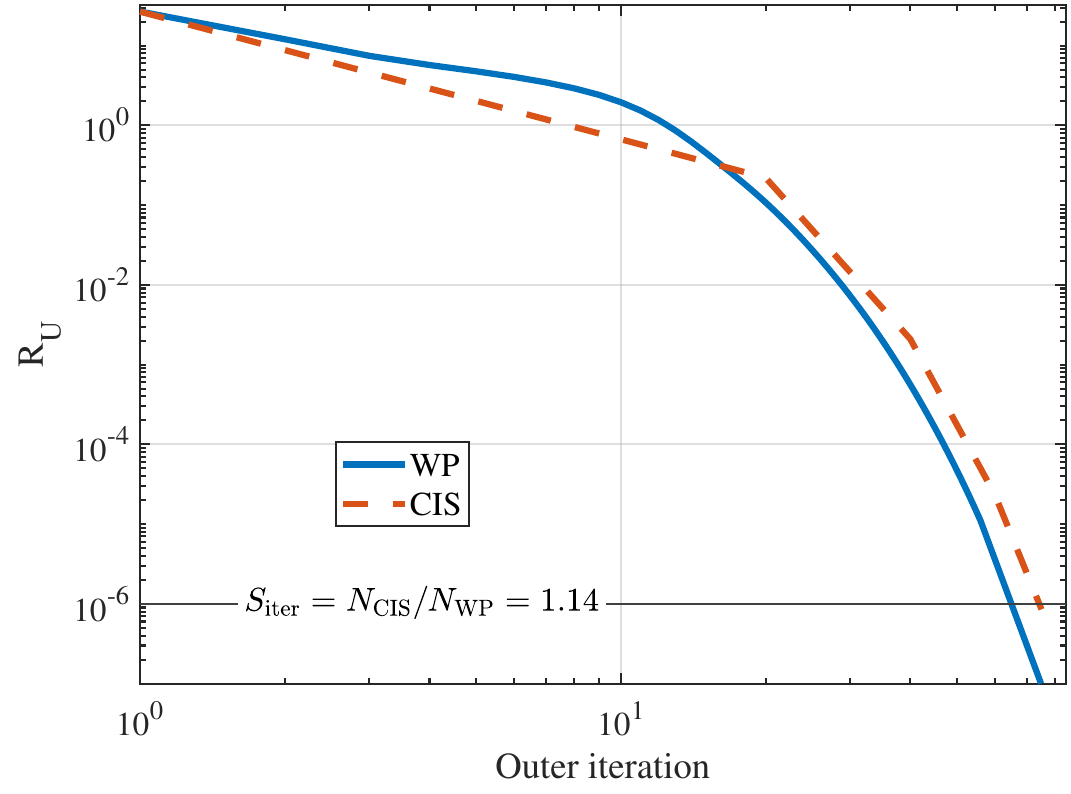}}\hspace{0.0200\figw}\raisebox{0.0002\figw}{\includegraphics[width=0.4806\figw]{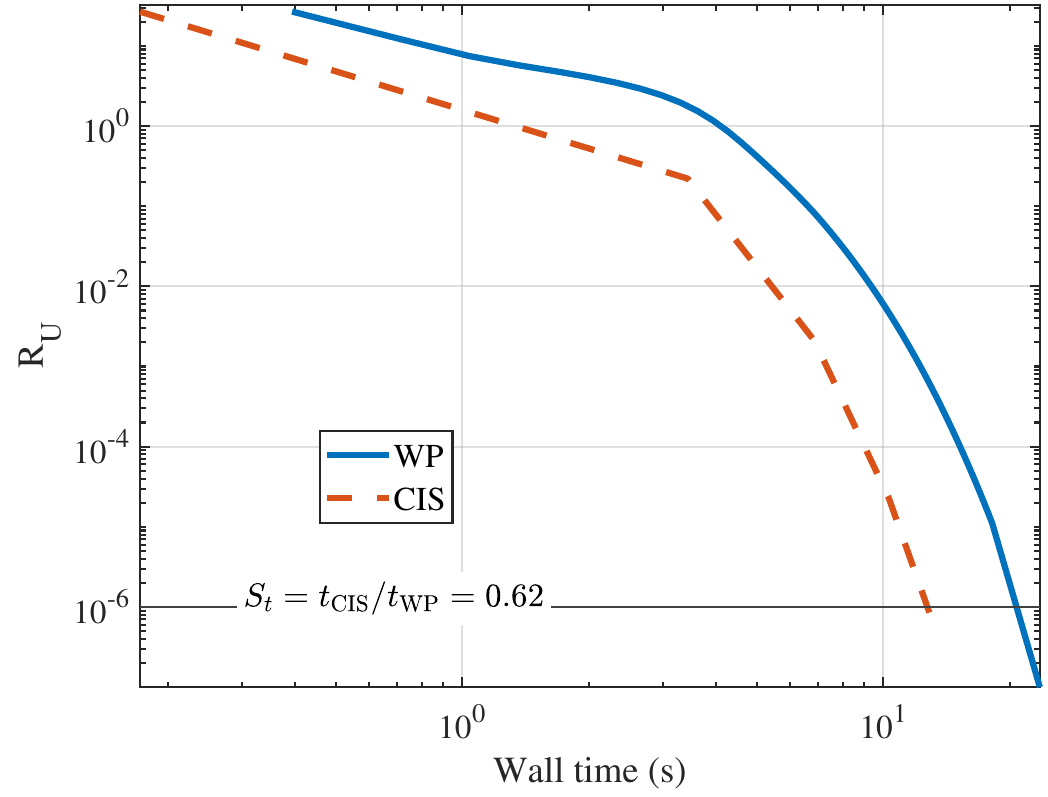}}}
\caption{$\varepsilon=1$.}
\end{subfigure}
\par\smallskip
\begin{subfigure}[t]{0.9950\figw}\centering
\makebox[\linewidth]{\raisebox{0.0002\figw}{\includegraphics[width=0.4842\figw]{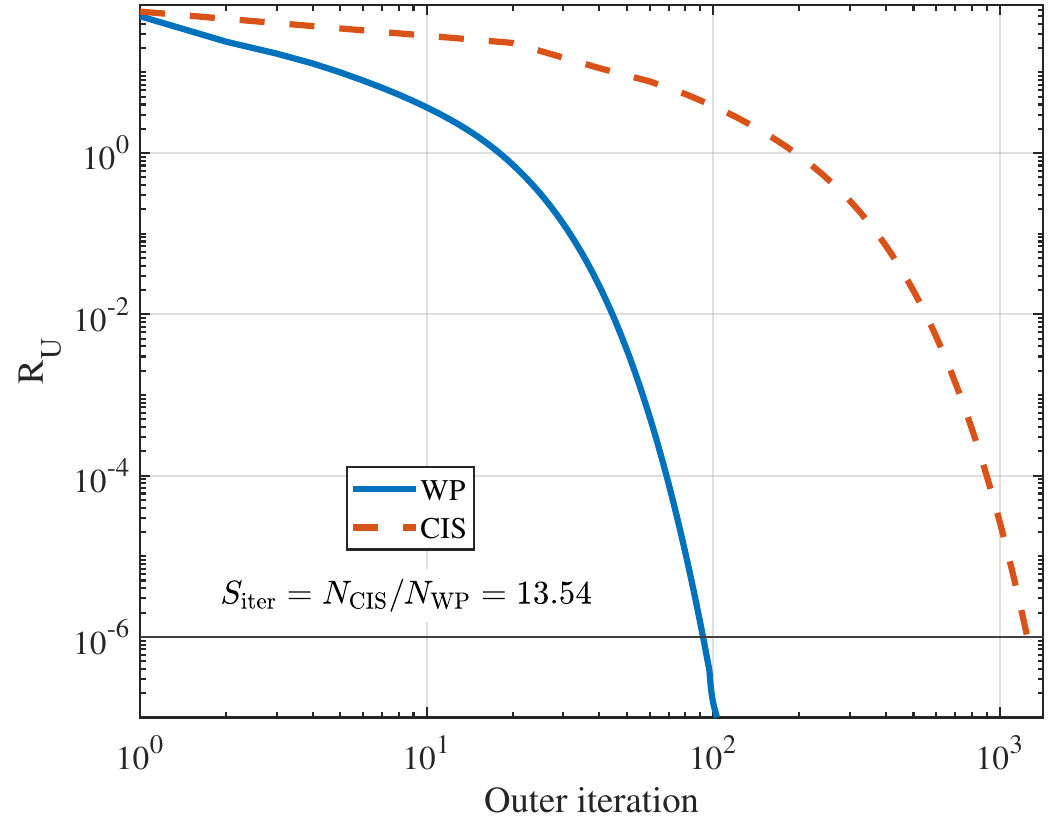}}\hspace{0.0200\figw}\raisebox{0.0000\figw}{\includegraphics[width=0.4908\figw]{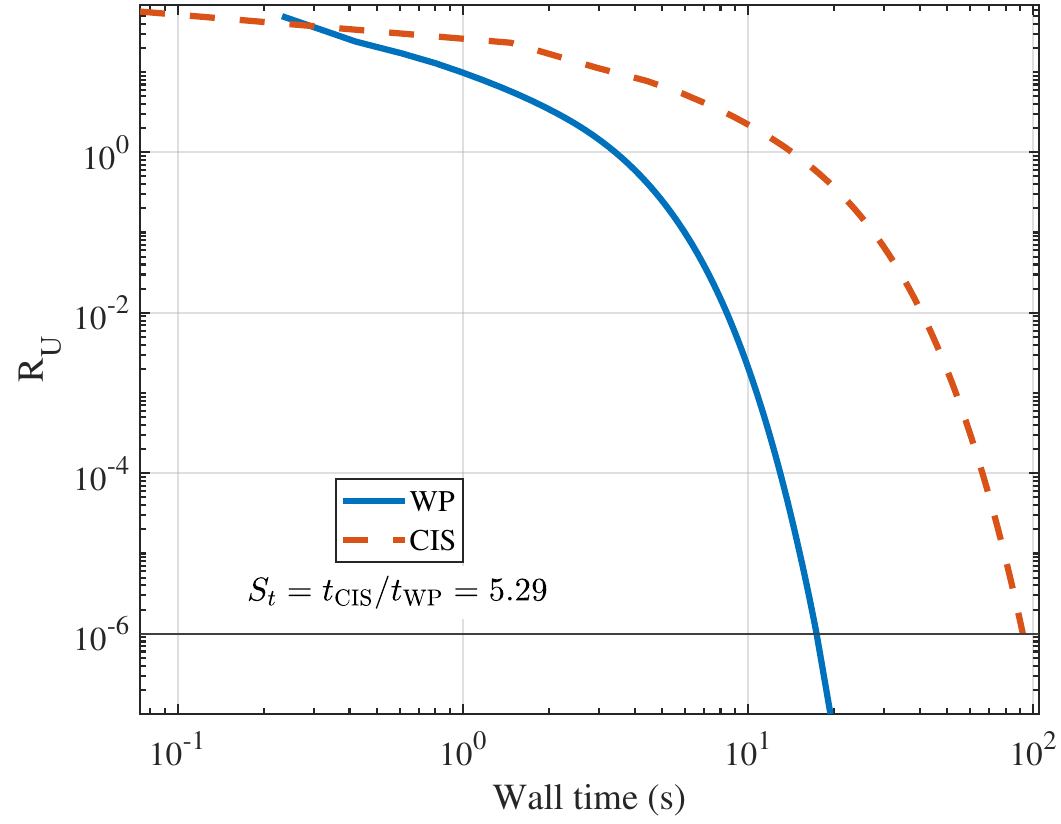}}}
\caption{$\varepsilon=10^{-2}$.}
\end{subfigure}
\par\smallskip
\begin{subfigure}[t]{0.9950\figw}\centering
\makebox[\linewidth]{\raisebox{0.0000\figw}{\includegraphics[width=0.5167\figw]{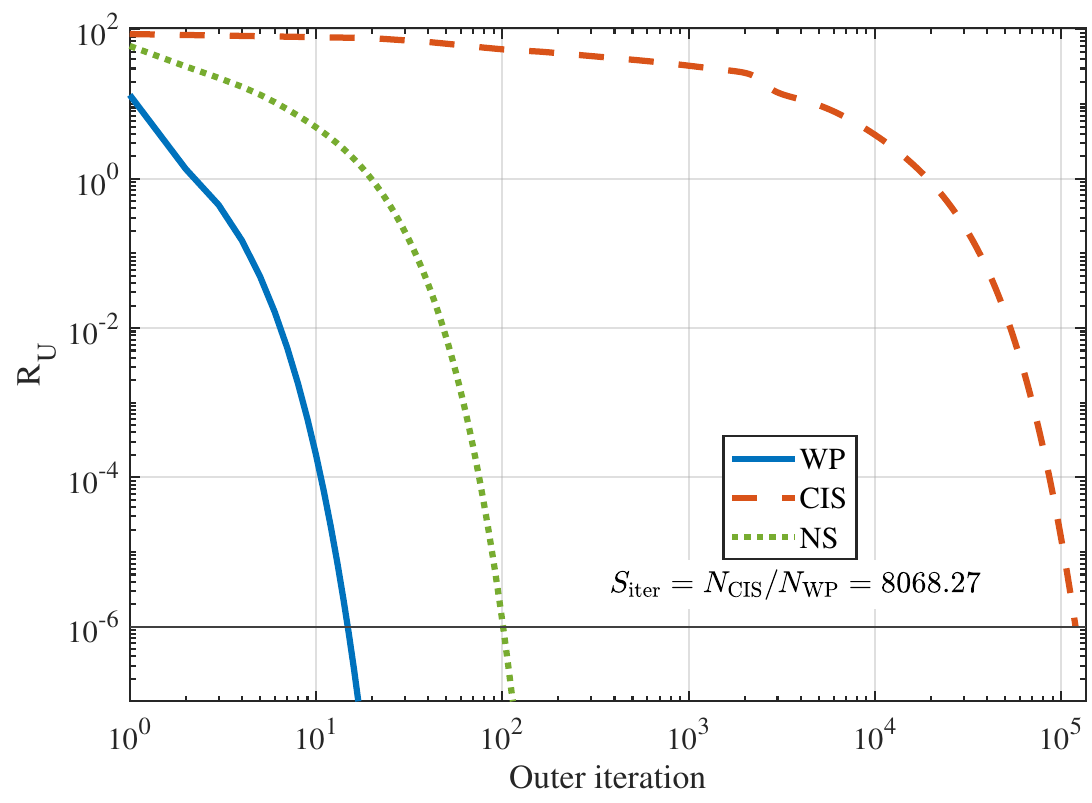}}\hspace{0.0200\figw}\raisebox{0.0037\figw}{\includegraphics[width=0.4583\figw]{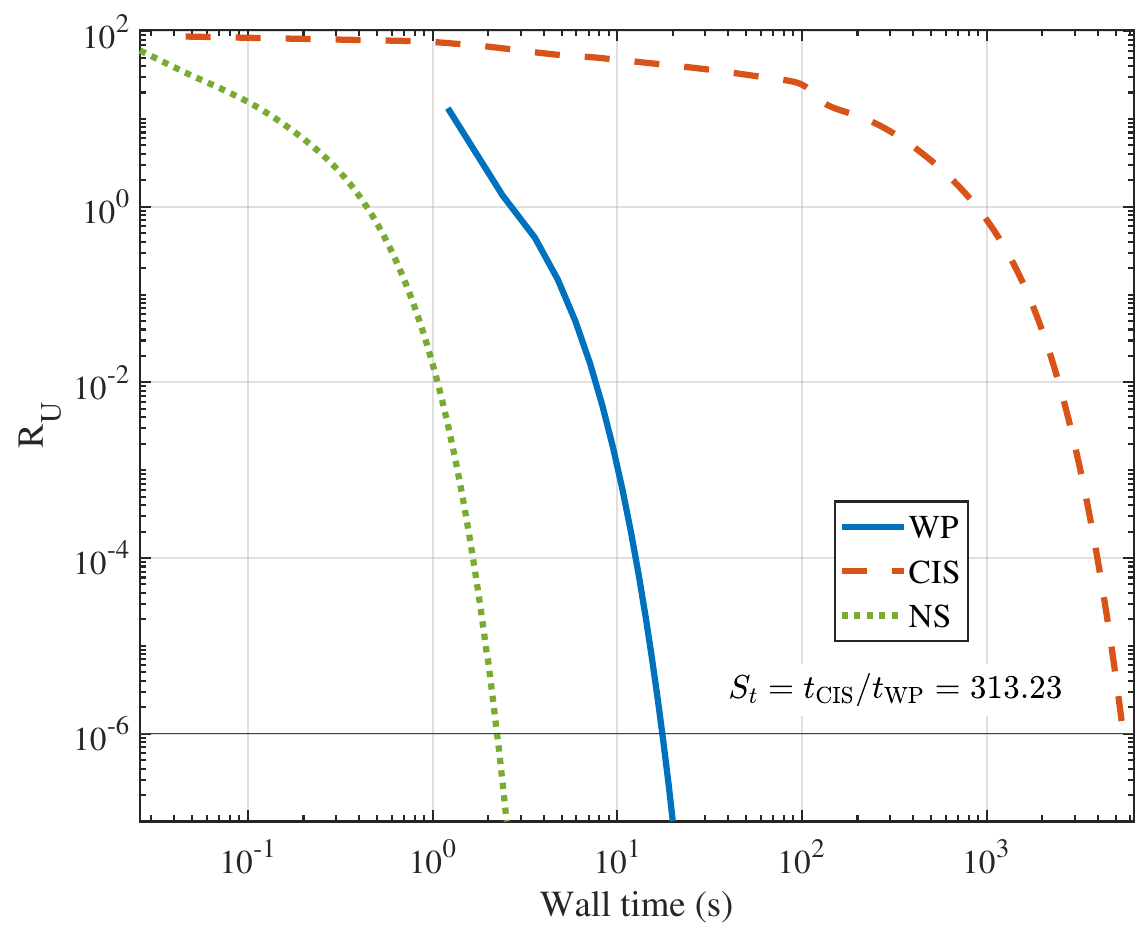}}}
\caption{$\varepsilon=10^{-4}$.}
\end{subfigure}
\caption{X38 section at $\mathrm{Ma}=5$: residual versus iteration (left) and wall time (right) for IWPD (solid), CIS (dashed) and, at $\varepsilon=10^{-4}$, NS (dotted). The grey line marks $\|R_U\|=10^{-6}$, at which $S_{\rm iter}$ and $S_t$ are evaluated.}\label{fig:x38_residual}
\end{figure}

\subsection{Convergence and computational cost}\label{sec:convergence}
The residual histories (Figures~\ref{fig:internal_residual} and~\ref{fig:x38_residual}, and Figures~\SMref{fig:couette_residual}, \SMref{fig:cavity_residual} and~\SMref{fig:cyl_residual}) show that the CIS slows down as $\varepsilon$ decreases, as the scale-dependent bound \eqref{eq:bounds} predicts, whereas IWPD converges in tens to hundreds of outer iterations. Table~\ref{tab:factors} compares the asymptotic contraction factors, fitted by least squares to $\ln\|R_U\|$, with the predictions of Section~\ref{sec:rates}. For the internal flows the nominal ratio $\eta$ and wavenumber parameter $z$ are evaluated with the cell size $\Delta$, the reference state and the slowest mode $\vk=\pi/L$ of the domain of length $L=1$,
\begin{equation}\label{eq:nominal}
 \eta=\frac{\Delta}{(|\bu|+a)\,\tau},\qquad z=\pi\tau.
\end{equation}
For the external flows the horizon varies too strongly over the O-grid for a nominal value to be meaningful. The macroscopic contraction $\varrho_W$ of the inexact-solve estimate \eqref{eq:inexact} is taken from the histories of the NS computations, which use the same macroscopic solver.
\begin{table}[!htb]
\centering\small
\caption{Measured contraction factors per iteration of the CIS and per outer iteration of IWPD, fitted over the residual range $[10^{-6},10^{-3}]$, and the predictions of Section~\ref{sec:rates}: the CIS factor \eqref{eq:factors}, the prediction $\varrho_W^{N_U}$ of the inexact-solve estimate \eqref{eq:inexact} with $\varrho_W$ the measured factor of the NS computation, and the bound \eqref{eq:bounds} of the model problem with a converged W iteration, where it is smaller than one. }
\label{tab:factors}
\resizebox{\linewidth}{!}{\begin{tabular}{llrrrrrrrr}
\toprule
Case & $\varepsilon$ & $N_U$ & $\eta$ & $z$ & $g_{\rm CIS}$ measured & $\arctan z/z$ & $g_{\rm WP}$ measured & $\varrho_W^{N_U}$ & bound \eqref{eq:bounds}\\
\midrule
Couette  & $1$          & 1   & 0.0075 & 4.59   & 0.916  & 0.295  & 0.907 & -- & --\\
Couette  & $10^{-2}$    & 1   & 0.75   & 0.046  & 0.9987 & 0.9993 & 0.976 & -- & --\\
Couette  & $10^{-4}$    & 128 & 75     & 0.0005 & 1.0000 & 1.0000 & 0.990 & 0.974 & $7\times10^{-8}$\\
Cavity   & $1$          & 1   & 0.012  & 4.59   & 0.944  & 0.295  & 0.927 & -- & --\\
Cavity   & $0.075$      & 1   & 0.16   & 0.34   & 0.969  & 0.963  & 0.858 & -- & --\\
Cavity   & $2.05\times10^{-3}$ & 24 & 5.7 & 0.0094 & 0.9997 & 1.0000 & 0.744 & 0.752 & 0.023\\
Cavity   & $2.05\times10^{-4}$ & 24 & 36.5 & 0.0009 & 1.0000 & 1.0000 & 0.975 & 0.933 & $3\times10^{-7}$\\
Cylinder & $1$          & 1   & --     & --     & 0.911  & --     & 0.878 & -- & --\\
Cylinder & $10^{-2}$    & 1   & --     & --     & 0.998  & --     & 0.948 & -- & --\\
Cylinder & $10^{-4}$    & 120 & --     & --     & 1.0000 & --     & 0.332 & 0.002 & --\\
X38      & $10^{-2}$    & 1   & --     & --     & 0.987  & --     & 0.823 & -- & --\\
X38      & $10^{-4}$    & 120 & --     & --     & 0.9999 & --     & 0.332 & $<10^{-6}$ & --\\
\bottomrule
\end{tabular}}
\end{table}

Table~\ref{tab:factors} shows the loss of contraction of CIS toward the continuum limit. For $z\le0.35$, the measured factors differ from $\arctan z/z$ by less than $0.006$, but agreement near one does not imply agreement in iteration counts: the measured continuum runs converge one to two orders of magnitude faster than the scalar model predicts. At $\varepsilon=1$, the model omits diffuse boundaries and approximate transport inversion and is not quantitative. In the intermediate regime IWPD reduces the measured factors even where the bound \eqref{eq:bounds} does not apply. In the continuum regime its measured factors remain well above the ideal bound, which assumes a converged W iteration; the next paragraph shows that they are set instead by the incomplete macroscopic solve and the explicit collision remainder.

For the internal flows, the measured IWPD factors are close to the estimate $\varrho_W^{N_U}$, with differences of about $0.01$ at $\Rey=100$ and $0.02$--$0.04$ for Couette flow and the $\Rey=1000$ cavity. This is consistent with incomplete macroscopic convergence limiting the outer iteration. For hypersonic flows at $\varepsilon=10^{-4}$, the measured factors lie within $0.002$ of the Shakhov relaxation value $1-\Pr=1/3$, consistent with the explicit remainder mechanism in \eqref{eq:remainder_factor} and tested directly by the BGK comparison below.

The BGK--Shakhov comparison in Table~\ref{tab:bgk} tests the role of the explicit collision remainder with the same programs, meshes, velocity grids, controls and tolerance. Removing the Shakhov remainder removes the homogeneous heat-flux mode in \eqref{eq:remainder_factor}. The BGK residual then falls by more than an order of magnitude per outer iteration once it is below $10^{-1}$, and both the iteration count and wall time decrease. At $\varepsilon=1$ and $10^{-2}$ the cylinder computations instead converge similarly for the two collision models: at the residual $10^{-6}$, BGK gives $139$ and $337$ IWPD iterations against $186$ and $9890$ CIS iterations, with $S_{\rm iter}=1.34$ and $29.35$ and $S_t=0.75$ and $15.95$. The Shakhov iteration ratios are $1.35$ and $29.36$. The IWPD--CIS field differences for BGK are $2.70\times10^{-6}$ to $1.81\times10^{-5}$ at $\varepsilon=1$ and $2.10\times10^{-4}$ to $9.02\times10^{-4}$ at $\varepsilon=10^{-2}$. Thus the collision remainder limits the observed continuum convergence but not these rarefied and transition cases.

\begin{table}[!htb]
\centering\small
\caption{Hypersonic flows at $\varepsilon=10^{-4}$ computed with the Shakhov and the BGK model by IWPD with W120/PR2/P1: outer iterations and wall time to the tolerance $10^{-7}$, asymptotic contraction factor per outer iteration, and wall-time ratio to the NS computation on the same mesh. The Shakhov factors are fitted over the residual range $[10^{-6},10^{-3}]$. For the BGK model the residual contracts faster than linearly and crosses this range within two or three outer iterations, and the factor given is the largest ratio of successive residuals once $\|R_U\|<10^{-1}$.}
\label{tab:bgk}
\begin{tabular}{llrrrr}
\toprule
Case & Model & Outer iterations & Factor & $t_{\rm WP}$ (s) & $t_{\rm WP}/t_{\rm NS}$\\
\midrule
Cylinder & Shakhov & 20 & 0.332 & 121.8 & 2.63\\
         & BGK     & 8  & 0.047 & 44.7  & 0.97\\
X38      & Shakhov & 17 & 0.332 & 20.00 & 7.92\\
         & BGK     & 6  & 0.015 & 6.71  & 2.65\\
\bottomrule
\end{tabular}
\end{table}

\subsubsection{Computational cost}\label{sec:cost_assess}
Table~\ref{tab:efficiency} compares wall time as well as iteration counts. At $\varepsilon=1$, the two schemes require similar iteration counts, with the additional macroscopic work of IWPD giving $S_t=0.40$--$0.74$. Near the continuum limit, the wall-time gain reaches $13.56$--$50.07$ for internal flows and $313.23$--$2118.23$ for external flows. These gains are smaller than the iteration-count ratios because each IWPD outer iteration includes multiple W updates. The ratio $S_{\rm iter}/S_t$ measures that added work: it is $1.84$--$3.76$ with W1 and rises to $25.76$--$142.40$ for the continuum configurations. The reported cost relative to the accompanying NS solver is $0.82$--$7.92$. These ratios characterize the shared implementation and prescribed controls described in the computational setup.

\begin{table}[!htb]
\centering\footnotesize
\caption{Efficiency of IWPD relative to CIS at the common residual threshold of each comparison, which is $\|R_U\|=10^{-6}$ for the external flows and $\|R_U\|=10^{-5}$ for the cylinder at $\varepsilon=10^{-4}$, and wall-time ratio of IWPD to the NS computation on the same mesh at the tolerance $10^{-7}$ in the continuum cases.}
\label{tab:efficiency}
\begin{tabular}{llcrrr}
\toprule
Case & Regime & IWPD control & $S_{\rm iter}$ & $S_t$ & $t_{\rm WP}/t_{\rm NS}$\\
\midrule
Couette  & $\varepsilon=1$       & W1/PR6/P1   & 1.17     & 0.40    & --\\
Couette  & $\varepsilon=10^{-2}$ & W1/PR6/P1   & 18.89    & 5.02    & --\\
Couette  & $\varepsilon=10^{-4}$ & W128/PR2/P1 & 2987.73  & 20.98   & 1.71\\
Cavity   & $\varepsilon=1$       & W1/PR6/P1   & 1.36     & 0.74    & --\\
Cavity   & $\varepsilon=0.075$   & W1/PR6/P1   & 4.55     & 1.93    & --\\
Cavity   & $\Rey=100$            & W24/PR4/P1  & 806.22   & 13.56   & 0.86\\
Cavity   & $\Rey=1000$           & W24/PR3/P1  & 2001.63  & 50.07   & 0.82\\
Cylinder & $\varepsilon=1$       & W1/PR2/P1   & 1.35     & 0.72    & --\\
Cylinder & $\varepsilon=10^{-2}$ & W1/PR2/P1   & 29.36    & 11.37   & --\\
Cylinder & $\varepsilon=10^{-4}$ & W120/PR2/P1 & 57391.23 & 2118.23 & 2.63\\
X38      & $\varepsilon=1$       & W1/PR2/P1   & 1.14     & 0.62    & --\\
X38      & $\varepsilon=10^{-2}$ & W1/PR2/P1   & 13.54    & 5.29    & --\\
X38      & $\varepsilon=10^{-4}$ & W120/PR2/P1 & 8068.27  & 313.23  & 7.92\\
\bottomrule
\end{tabular}
\end{table}

\begin{figure}[!htb]
\centering\singlespacing\setlength{\figw}{.97\linewidth}
\begin{subfigure}[t]{0.9950\figw}\centering
\makebox[\linewidth]{\raisebox{0.0006\figw}{\includegraphics[width=0.4875\figw]{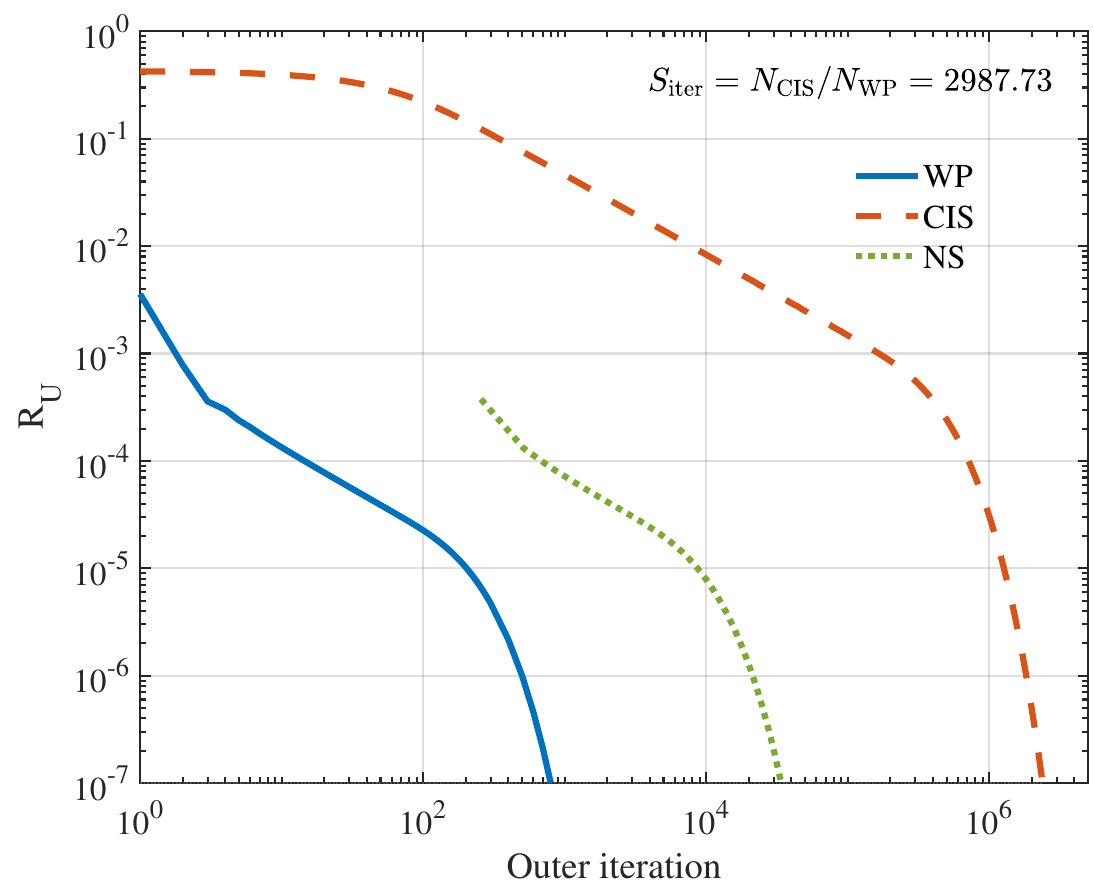}}\hspace{0.0200\figw}\raisebox{0.0000\figw}{\includegraphics[width=0.4875\figw]{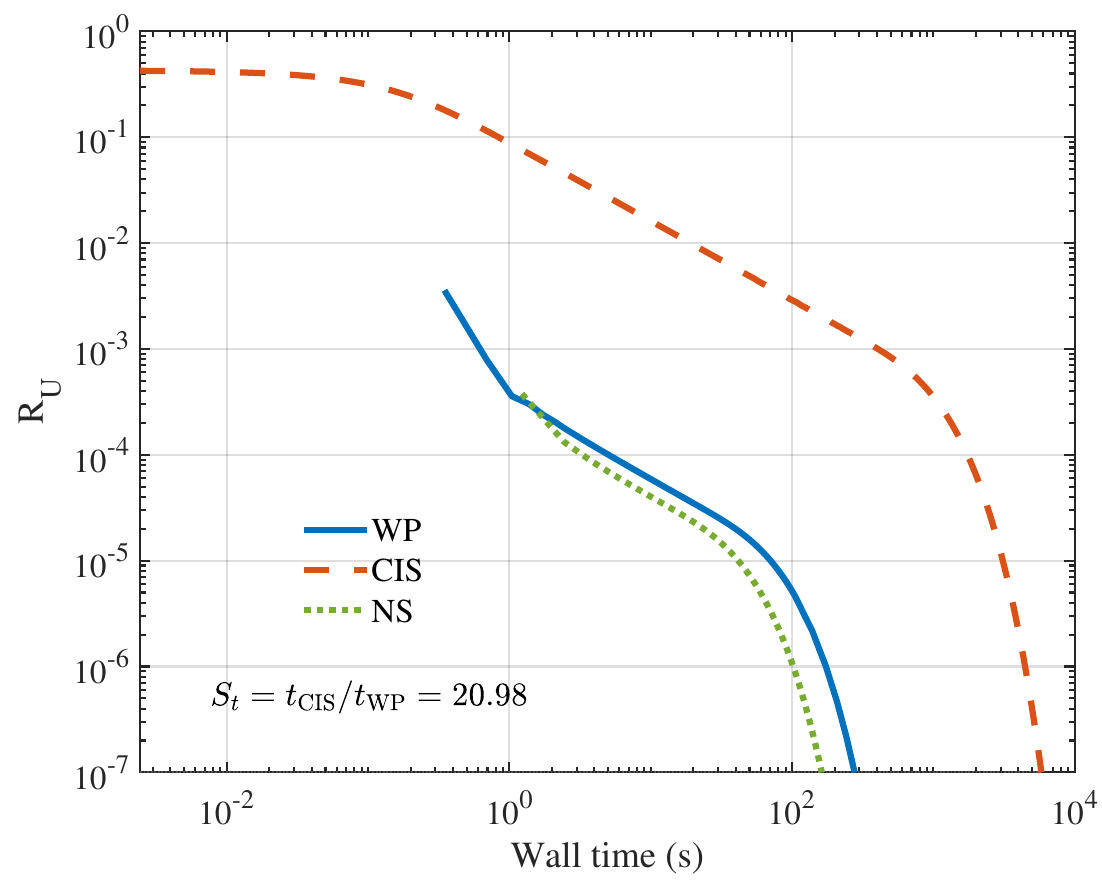}}}
\caption{Couette flow, $\varepsilon=10^{-4}$.}
\end{subfigure}
\par\smallskip
\begin{subfigure}[t]{0.9950\figw}\centering
\makebox[\linewidth]{\raisebox{0.0011\figw}{\includegraphics[width=0.3681\figw]{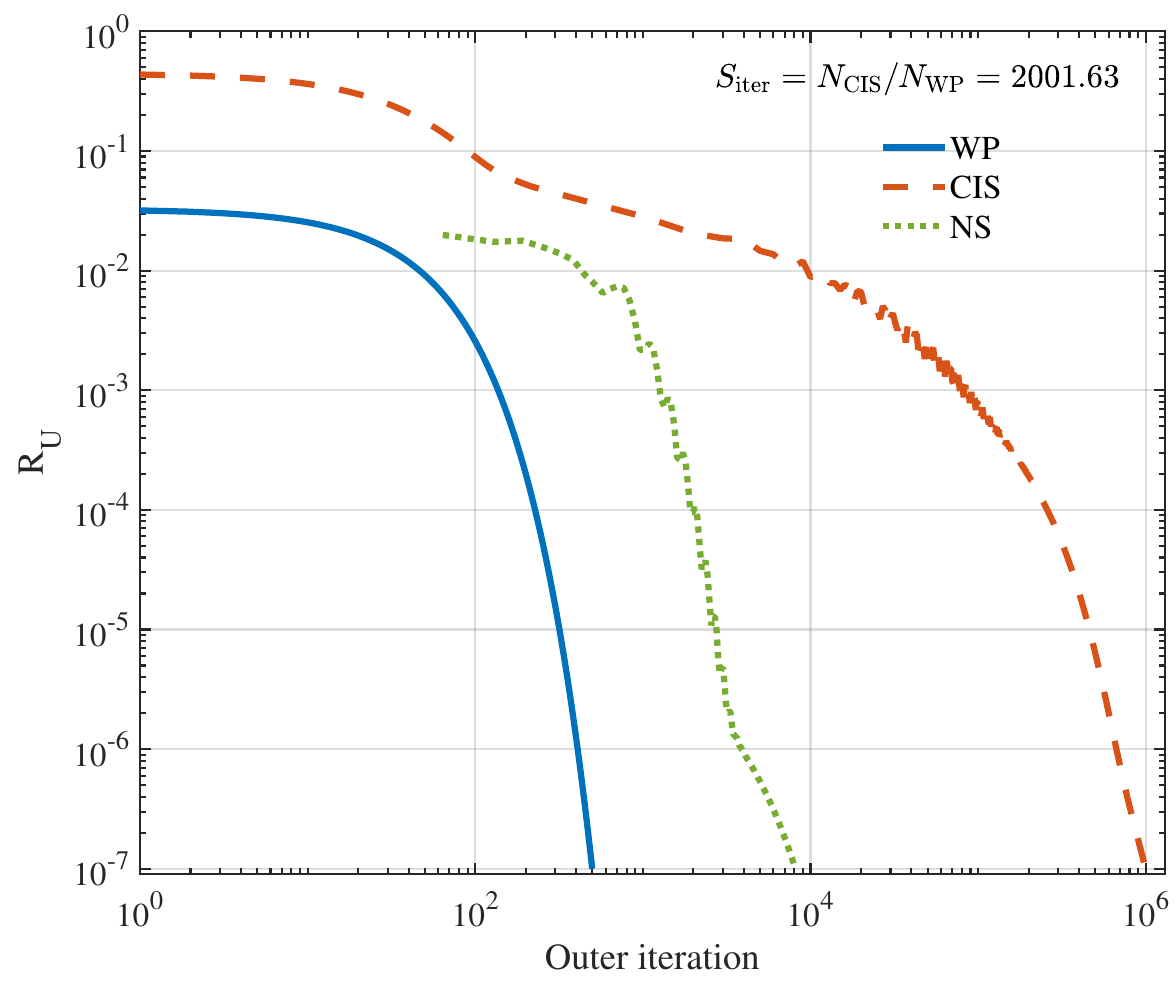}}\hspace{0.0200\figw}\raisebox{0.0000\figw}{\includegraphics[width=0.4875\figw]{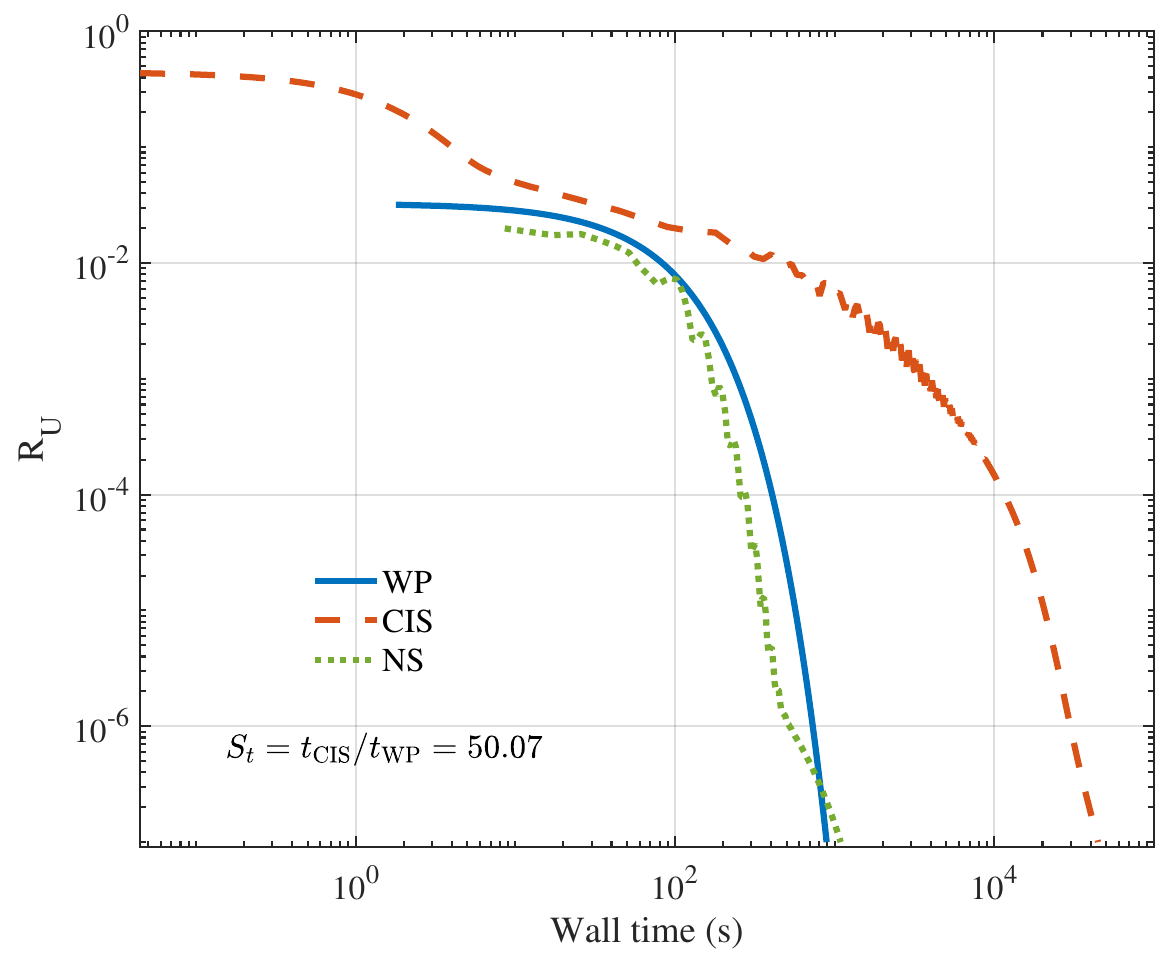}}}
\caption{Lid-driven cavity, $\mathrm{Re}=1000$.}
\end{subfigure}
\caption{Internal flows in the continuum regime: macroscopic residual norm $\|R_U\|$ versus outer iteration (left) and wall time (right) for IWPD (solid), CIS (dashed) and NS (dotted). The ratios $S_{\rm iter}$ and $S_t$ are evaluated at the common residual threshold.}\label{fig:internal_residual}
\end{figure}

For continuum internal flows, Figure~\ref{fig:internal_residual} and Table~\ref{tab:efficiency} show that the reduction in outer iterations outweighs the added macroscopic work. In the cavity, the macroscopic work is distributed differently between IWPD and NS: IWPD uses PR4 or PR3 in its W updates, whereas NS uses PR6. The wall-time ratios below one therefore reflect this allocation of work within the shared macroscopic solver.

\subsection{Parameter and mesh dependence}\label{sec:controls}
The baseline computations use one W update at $\varepsilon\ge10^{-2}$ and $24$--$128$ W updates in the continuum cases, with one P update and one kinetic LU--SGS pair. The small nominal $\eta$ in the former cases motivates limiting the W work: a fully converged W solve amplifies long waves below $\eta_c$ in the scalar model. The following scans assess the practical allocation of W and P work across this range of nominal horizon ratios.

\begin{figure}[!htb]
\centering\singlespacing\setlength{\figw}{\linewidth}
\begin{subfigure}[t]{0.5118\figw}\centering
\raisebox{0.0048\figw}{\includegraphics[width=0.5118\figw]{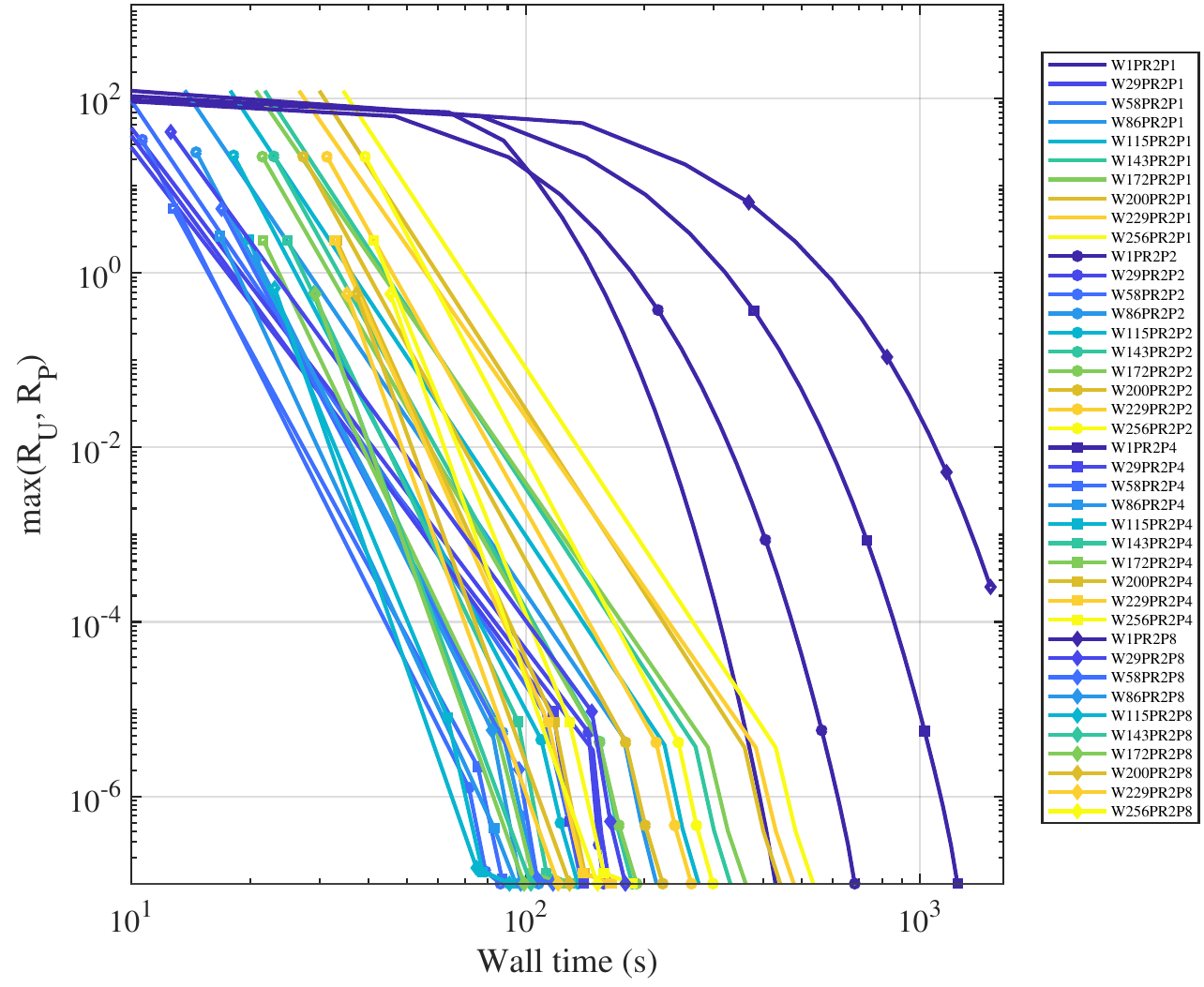}}
\caption{Residual versus wall time.}
\end{subfigure}\hspace{0.0200\figw}%
\begin{subfigure}[t]{0.4632\figw}\centering
\raisebox{0.0000\figw}{\includegraphics[width=0.4632\figw]{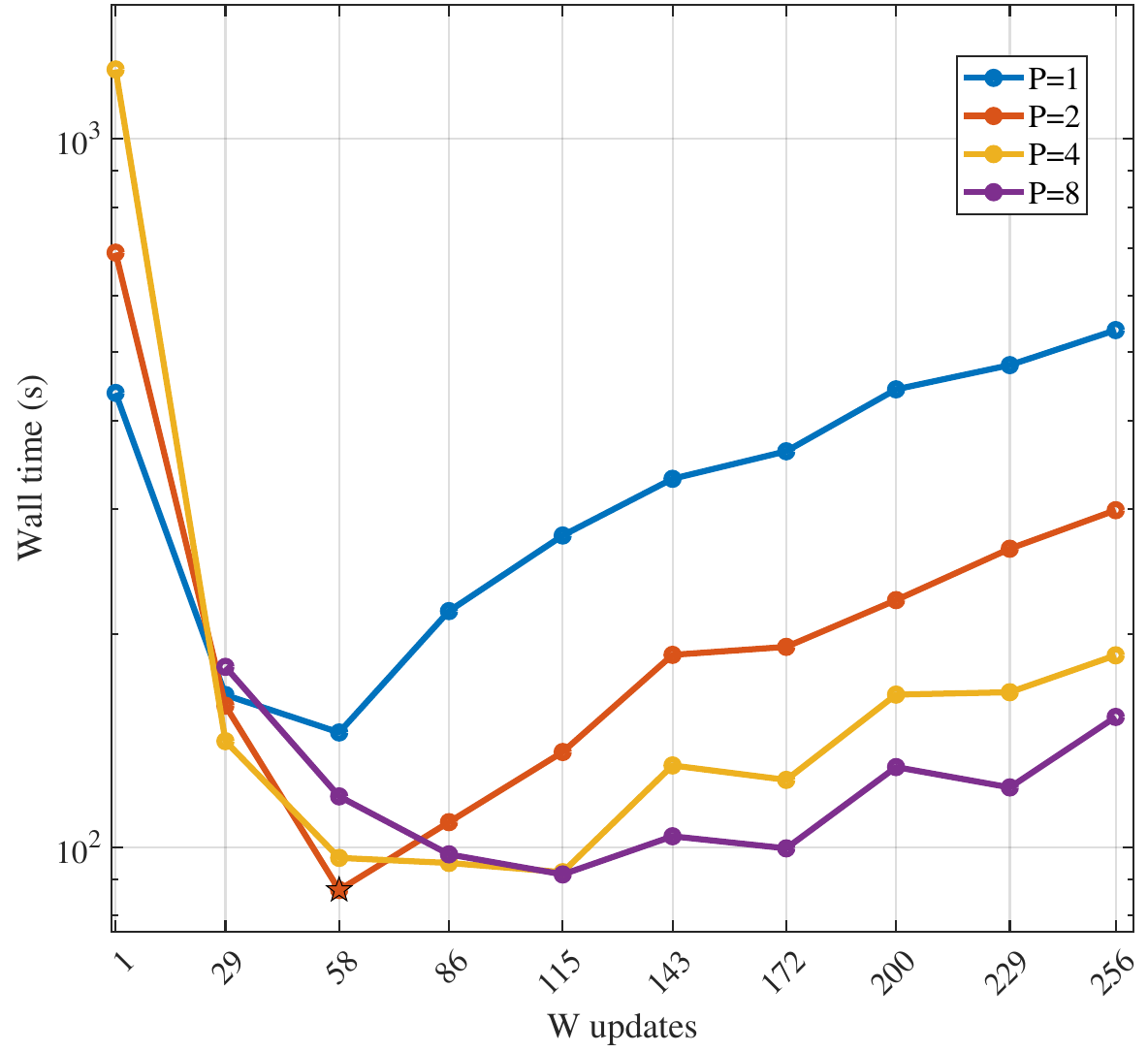}}
\caption{Wall time versus the number of W updates.}
\end{subfigure}
\caption{Circular cylinder at $\mathrm{Ma}=5$ and $\varepsilon=10^{-4}$: scan of the W and P controls with PR2 on a $100\times100$ mesh, covering W1--W256 and P1, P2, P4 and P8; the optimum is marked by a star. The residual in (a) is the larger of the residual norms $\|R_U\|$ and $\|R_P\|$.}\label{fig:cyl_scan}
\end{figure}

For Couette flow at $\varepsilon=10^{-4}$ a scan of W with PR1, PR2 and PR4 and P1, continued to $\|R_U\|\leq10^{-8}$ (Figure~\SMref{fig:couette_scan}), gives with PR2 a wall time that falls from $559.2$~s at W8 to a minimum of $380.6$~s at W128, the setting used in Section~\ref{sec:res_couette}, and rises to $382.0$~s at W256. With PR1 the runs from W16 to W256 need $655.1$~s to $826.6$~s, and with PR4 the wall time grows from $421.2$~s at W1 to $6904.4$~s at W64. For the cavity the scans of W and PR with P1 (Figure~\SMref{fig:cavity_scan}) select the settings of Section~\ref{sec:res_cavity}, W24/PR4 at $\Rey=100$ and W24/PR3 at $\Rey=1000$, as the fastest runs to $\|R_U\|\leq10^{-7}$ (not shown), and to $\|R_U\|\leq10^{-8}$ they show a flat wall time over a broad range of W around the optimum W28/PR3 at both Reynolds numbers. For the cylinder at $\varepsilon=10^{-4}$ the scan in Figure~\ref{fig:cyl_scan}, on a $100\times100$ mesh with the same velocity grid, CFL numbers, PR2 and tolerance, covers $N_U$ from $1$ to $256$ and $N_P=1$, $2$, $4$ and $8$. The lowest wall time, $86.97$~s, is obtained with W58/P2, and for large $N_U$ additional P updates shorten it further, consistently with the remainder factor \eqref{eq:remainder_factor}, which acts per P update.

The cost is relatively flat around the optimal number of W updates, but is more sensitive to the number of point-relaxation pairs: Couette flow with PR4/W64 is about $18$ times slower than its optimum. The internal-flow controls were selected by scans, whose cost is excluded from the reported solver times. The common external-flow setting W120/PR2/P1 provides a baseline across geometries; the cylinder scan identifies further savings, with W58/P2 costing about a third as much as W115/P1. These trends motivate allocating iteration work according to the local $\eta$ relative to the model threshold $\eta_c$.

\subsubsection{Mesh refinement and comparison with other methods}\label{sec:efficiency}
Table~\ref{tab:meshref} follows the continuum cavity flows under mesh refinement with a single control setting, W24/PR6/P1, and all computations converged to $\|R_U\|<10^{-8}$; this setting and this tolerance differ from those of Table~\ref{tab:efficiency}, and each table is interpreted using its stated controls and stopping criterion. The NS computations are organized in cycles of $64$ implicit steps, which the table counts. On the $64^2$ mesh at $\Rey=100$ the two computations perform the same number of macroscopic updates, $48\times24=18\times64=1152$, and need the same wall time to the printed accuracy. At $\Rey=1000$, where $\eta$ stays above $12$, the outer iterations of IWPD and the cycles of NS grow in proportion to the number of cells per direction, because both are limited by the point-relaxation solver, whose convergence slows down under refinement and enters the outer factor through the inexact-solve estimate \eqref{eq:inexact}. The wall time of IWPD therefore remains at two thirds of that of NS on all three meshes, and its advantage over the CIS, whose number of iterations grows by $62\%$ between the $100^2$ and $200^2$ meshes, stays at $S_t=56.50$ and $56.15$. At $\Rey=100$ the refinement lowers the nominal ratio $\eta$ from $5.7$ to $2.4$, toward $\eta_c$, because the finer mesh begins to resolve the mean free path; the number of CIS iterations grows only from $43888$ to $50071$, whereas the outer iterations of IWPD grow from $48$ to $187$ and $S_t$ drops from $10.89$ to $4.12$. These runs show that the gain diminishes as refinement reduces the cell-to-mean-free-path ratio. The table measures iterative cost for the converged configurations on the stated meshes. Additional finer-grid trials required a reduced number of W updates to obtain convergence, demonstrating the importance of matching the macroscopic work to the local resolution.

\begin{table}[!htb]
\centering\small
\caption{Lid-driven cavity under mesh refinement: nominal ratio $\eta$ of \eqref{eq:nominal}, outer iterations of IWPD (W24/PR6/P1), cycles of $64$ implicit steps of NS and iterations of CIS, $N$, with the wall times $t$ (s), all to $\|R_U\|<10^{-8}$, and the ratios $S_t=t_{\rm CIS}/t_{\rm WP}$ and $t_{\rm WP}/t_{\rm NS}$.}
\label{tab:meshref}
\resizebox{\linewidth}{!}{\begin{tabular}{llrrrrrrrrr}
\toprule
Case & Mesh & $\eta$ & $N_{\rm WP}$ & $t_{\rm WP}$ & $N_{\rm NS}$ & $t_{\rm NS}$ & $N_{\rm CIS}$ & $t_{\rm CIS}$ & $S_t$ & $t_{\rm WP}/t_{\rm NS}$\\
\midrule
$\Rey=100$  & $64^2$  & 5.7  & 48  & 61.6   & 18  & 61.6  & 43888   & 670.9    & 10.89 & 1.00\\
            & $150^2$ & 2.4  & 187 & 1354.2 & 48  & 928.5 & 50071   & 5577.0   & 4.12 & 1.46\\
$\Rey=1000$ & $100^2$ & 36.5 & 345 & 1118.7 & 194 & 1659.4 & 1415368 & 63208.8 & 56.50 & 0.67\\
            & $200^2$ & 18.3 & 665  & 8604.9  & 379 & 13006.6 & 2291601 & 483139.4 & 56.15 & 0.66\\
            & $300^2$ & 12.2 & 1007 & 29858.0 & 570 & 44501.3 & --      & --       & --   & 0.67\\
\bottomrule
\end{tabular}}
\end{table}

Table~\ref{tab:iugks} compares IWPD with published IUGKS~\cite{zhu2016implicit,xu2022iugks} and GSIS~\cite{zhu2021gsis} results. In addition to the baseline calculations, IWPD is run on the spatial and velocity grids of the GSIS study. On those grids it requires $150$ and $431$ outer iterations at $\Rey=100$ and $1000$, compared with the reported GSIS counts of $234$ and $1410$. At $\varepsilon=0.075$, the corresponding counts are $116$ for IWPD and $52$ for GSIS. The smaller continuum counts and larger transition count reflect the regime dependence examined above.

The comparison uses matched spatial and velocity grids with the stopping criteria of the respective studies: GSIS uses a volume-weighted relative change below $10^{-8}$, whereas IWPD uses $\|R_U\|<10^{-8}$. With the respective stopping criteria and work per iteration, these counts demonstrate competitive iterative convergence; a wall-time comparison requires a common implementation and stopping test. The published IUGKS results use further grid and stopping-criterion choices, as detailed in the caption. The corresponding CIS baselines illustrate the effect of these computational choices: at $\varepsilon=1$, the GSIS study reports $24$ iterations, compared with $170$ here.

\begin{table}[!htb]
\centering\small
\caption{Lid-driven cavity: iterations of the CIS and of IWPD of this work to $\|R_U\|<10^{-6}$ (settings of Table~\ref{tab:settings}); iterations of the implicit UGKS (IUGKS) and of its two-step versions without and with multigrid from Table~1 of~\cite{xu2022iugks} ($64^2$ cells or a stretched mesh, $80^2$, $60^2$ and $28^2$ velocity nodes, relative change of the conservative variables below $10^{-6}$); iterations of the CIS and of GSIS from Table~1 of~\cite{zhu2021gsis} ($64^2$ cells with $48^2$ and $28^2$ velocity nodes at $\varepsilon=1$ and $0.075$, $128^2$ cells with $28^2$ velocity nodes at $\Rey=100$ and $1000$, volume-weighted relative change below $10^{-8}$); and outer iterations of IWPD on the mesh and velocity grid of~\cite{zhu2021gsis} to $\|R_U\|<10^{-8}$ (W1/PR6/P1 at $\varepsilon=0.075$, W24/PR6/P1 at $\Rey=100$ and $1000$).}
\label{tab:iugks}
\resizebox{\linewidth}{!}{\begin{tabular}{lrrrrrrrr}
\toprule
 & \multicolumn{2}{c}{This work} & \multicolumn{3}{c}{IUGKS~\cite{xu2022iugks}} & \multicolumn{2}{c}{GSIS~\cite{zhu2021gsis}} & This work\\
\cmidrule(lr){2-3}\cmidrule(lr){4-6}\cmidrule(lr){7-8}\cmidrule(lr){9-9}
Case & CIS & IWPD & IUGKS & two-step & two-step, multigrid & CIS & GSIS & IWPD, GSIS grid\\
\midrule
$\varepsilon=1$     & 170    & 127 & 79   & 77  & 33  & 24      & 23   & --\\
$\varepsilon=0.075$ & 323    & 71  & 169  & 51  & 20  & 226     & 52   & 116\\
$\Rey=100$          & 27318  & 36  & 923  & 74  & 36  & 40917   & 234  & 150\\
$\Rey=1000$         & 668472 & 408 & 3480 & 263 & 187 & 1283068 & 1410 & 431\\
\bottomrule
\end{tabular}}
\end{table}

\FloatBarrier
\section{Conclusions}\label{sec:conclusion}
The closed WPD system combines conservation laws for the macroscopic state with a kinetic equation for the particle component. Accumulating only the Maxwellian makes the wave a functional of the conservative variables. Under moment compatibility, the system is equivalent to the kinetic equation, and its hydrodynamic expansion identifies the wave and particle contributions to the Euler and Navier--Stokes fluxes.

The residual-based implicit iteration transfers macroscopic prediction through a discrete endpoint correction and synchronizes moments after each microscopic update. Its stationary states satisfy the intended discrete kinetic equation independently of the iteration controls. The smooth fixed-mesh analysis gives the formal fluid limit, while linear model analysis distinguishes the effects of the WPD scale, incomplete macroscopic solves and explicit collision remainders on convergence.

The BGK and Shakhov results support these mechanisms in one- and two-dimensional flows. Near the continuum limit, IWPD substantially reduces iteration counts and wall time relative to the conventional iterative scheme. Collision-model comparisons, parameter scans and mesh studies show how the gain depends on the collision remainder and the allocation of macroscopic work. These findings provide a basis for selecting iteration controls and improving the inner solvers.

\section*{Supplementary information}
Online Resource 1 contains additional profiles, velocity distributions, residual histories and parameter studies, together with the Apollo-section results. Supplementary figures, tables and sections are numbered with the prefix S.

\section*{Declarations}
\paragraph{Author contributions}
\textbf{Chang Liu:} Conceptualization, Methodology, Formal analysis, Software, Investigation, Visualization, Writing -- original draft, Funding acquisition. \textbf{Kun Xu:} Conceptualization, Methodology, Supervision, Writing -- review \& editing, Funding acquisition.

\paragraph{Competing interests}
The authors declare that they have no known competing financial interests or personal relationships that could have appeared to influence the work reported in this paper.

\paragraph{Funding}
{\sloppy Chang Liu is partially supported by the National Natural Science Foundation of China (12031001, 12102061), the Beijing Natural Science Foundation (Z230003), and the Presidential Foundation of the China Academy of Engineering Physics (YZJJZQ2022017). The authors are partially supported by the National Key R\&D Program of China (2022YFA1004500). Kun Xu is partially supported by the National Natural Science Foundation of China (12172316, 92371107) and the Hong Kong Research Grants Council (16301222, 16208324).\par}

\paragraph{Use of AI-assisted technologies}
The authors used Claude (Anthropic) and Codex (OpenAI) to assist with manuscript organization and language editing. The authors are responsible for reviewing the assisted revisions and for the content of the article.

\paragraph{Data availability}
The data supporting the results reported in this article are available from the corresponding author on request.

\paragraph{Code availability}
The code used in this study is available from the corresponding author on reasonable request.

\paragraph{Ethics approval and consent}
Not applicable.

\end{document}


\raggedbottom
\title{Supplementary material: An implicit multiscale solver for steady kinetic equations based on wave--particle decomposition}
\titlerunning{Supplementary material for the implicit WPD solver}
\author{Chang Liu \and Kun Xu}
\authorrunning{C. Liu and K. Xu}
\institute{Chang Liu (corresponding author) \at
Institute of Applied Physics and Computational Mathematics, Beijing, P.R. China\\
\email{liuchang@iapcm.ac.cn}
\and Kun Xu \at
Department of Mathematics and Department of Mechanical and Aerospace Engineering, Hong Kong University of Science and Technology, Clear Water Bay, Kowloon, Hong Kong\\
Shenzhen Research Institute, Hong Kong University of Science and Technology, Shenzhen, China\\
\email{makxu@ust.hk}}
\date{}
\maketitle

Online Resource 1 for \textit{Journal of Scientific Computing} supplements the main article. IWPD denotes the implicit wave--particle decomposition solver (WP in figures), CIS the conventional iterative scheme and NS the accompanying gas-kinetic Navier--Stokes scheme. Field comparisons use common contour levels. Wall values are evaluated in the cells adjacent to the surface. Except for the Apollo settings stated in Section~\ref{sec:apollo}, the discretizations, residual norms and comparison criteria are those of the main text.

\section{Normal shock}
\begin{figure}[H]
\centering\singlespacing\setlength{\figw}{1.0\linewidth}
\begin{subfigure}[t]{0.4739\figw}\centering
\raisebox{0.0000\figw}{\includegraphics[width=0.4739\figw]{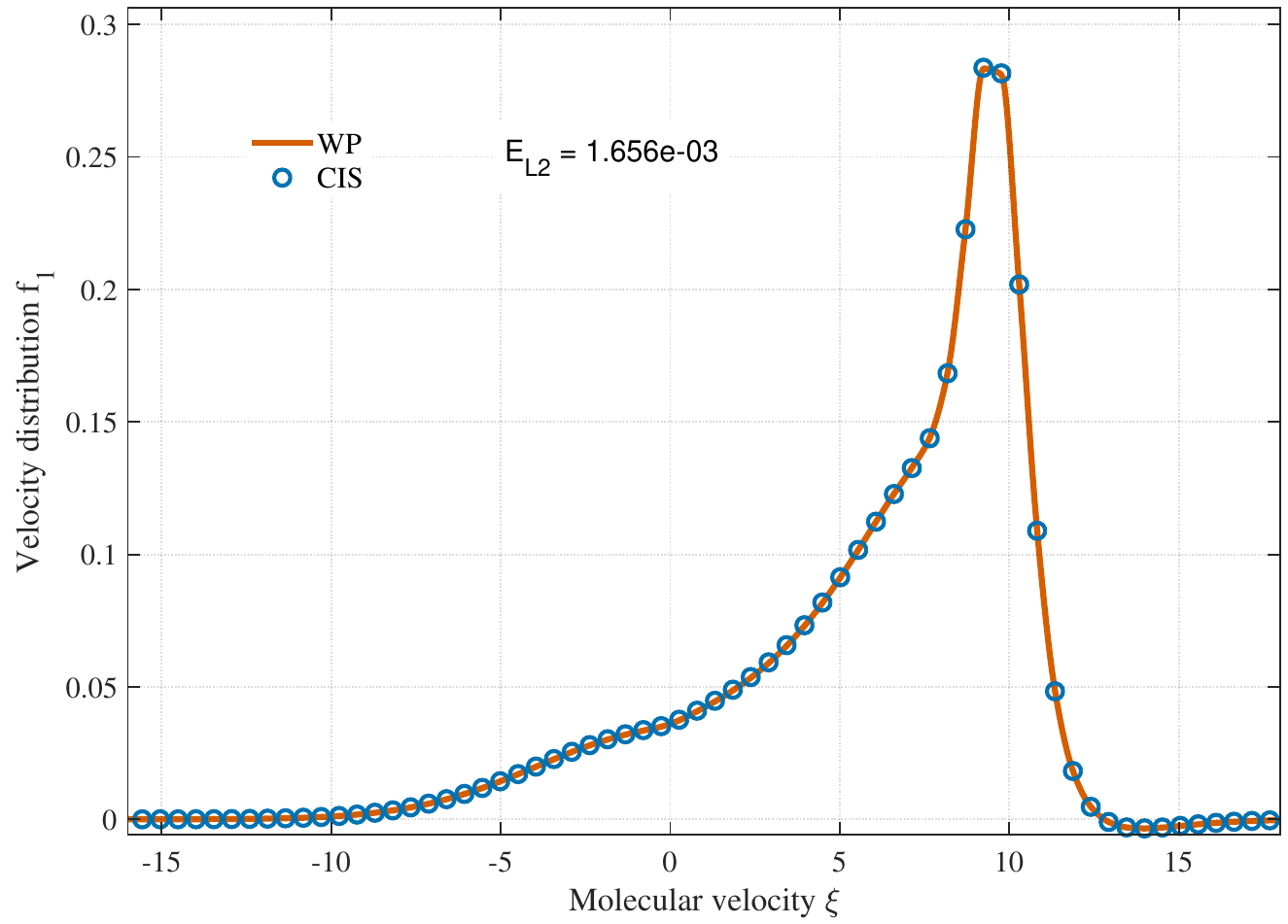}}
\caption{$x-x_{\rho\text{-mid}}=-3.52$.}
\end{subfigure}\hspace{0.0200\figw}%
\begin{subfigure}[t]{0.4739\figw}\centering
\raisebox{0.0000\figw}{\includegraphics[width=0.4739\figw]{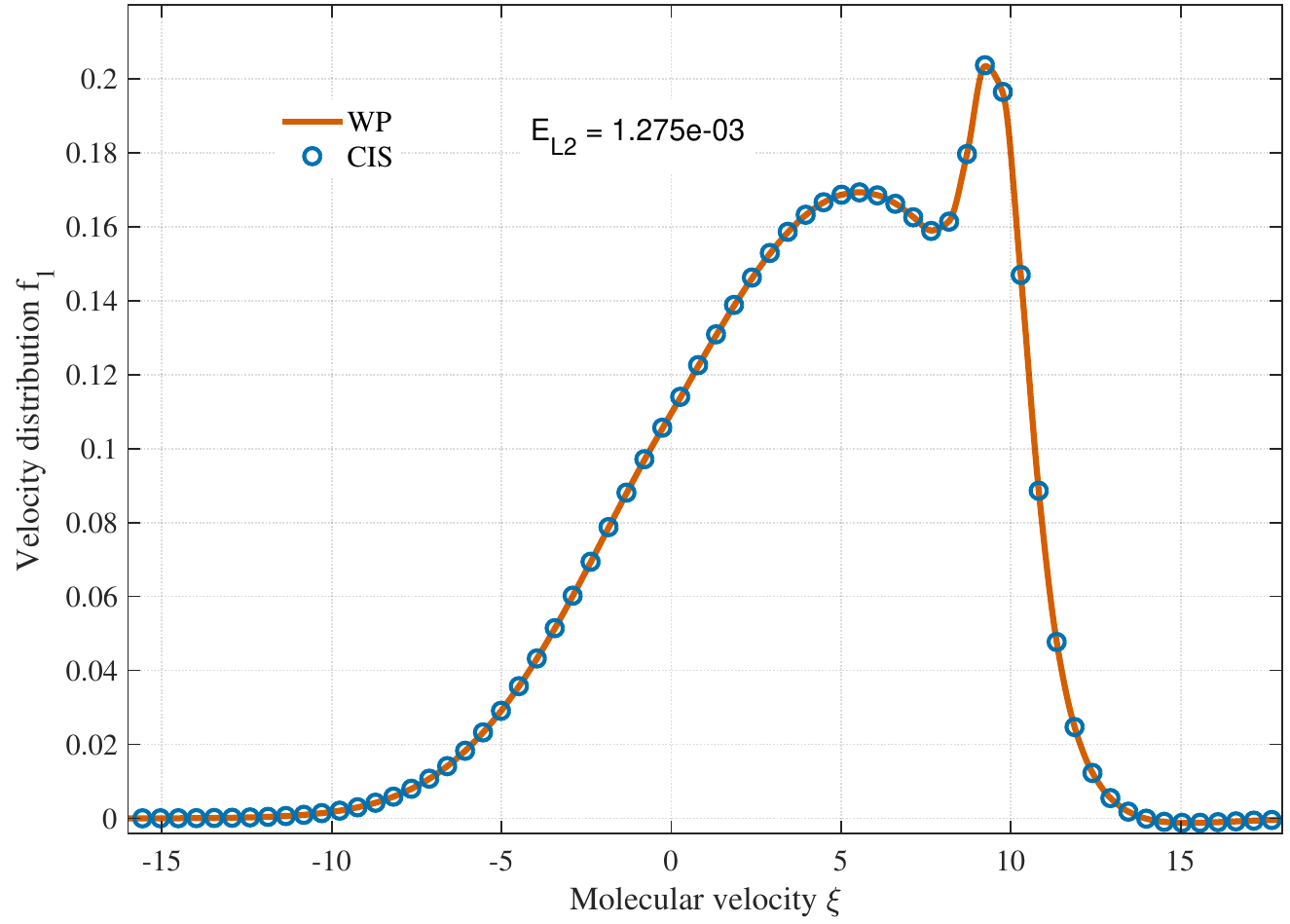}}
\caption{$x-x_{\rho\text{-mid}}=-0.80$.}
\end{subfigure}
\par\smallskip
\begin{subfigure}[t]{0.4739\figw}\centering
\raisebox{0.0000\figw}{\includegraphics[width=0.4739\figw]{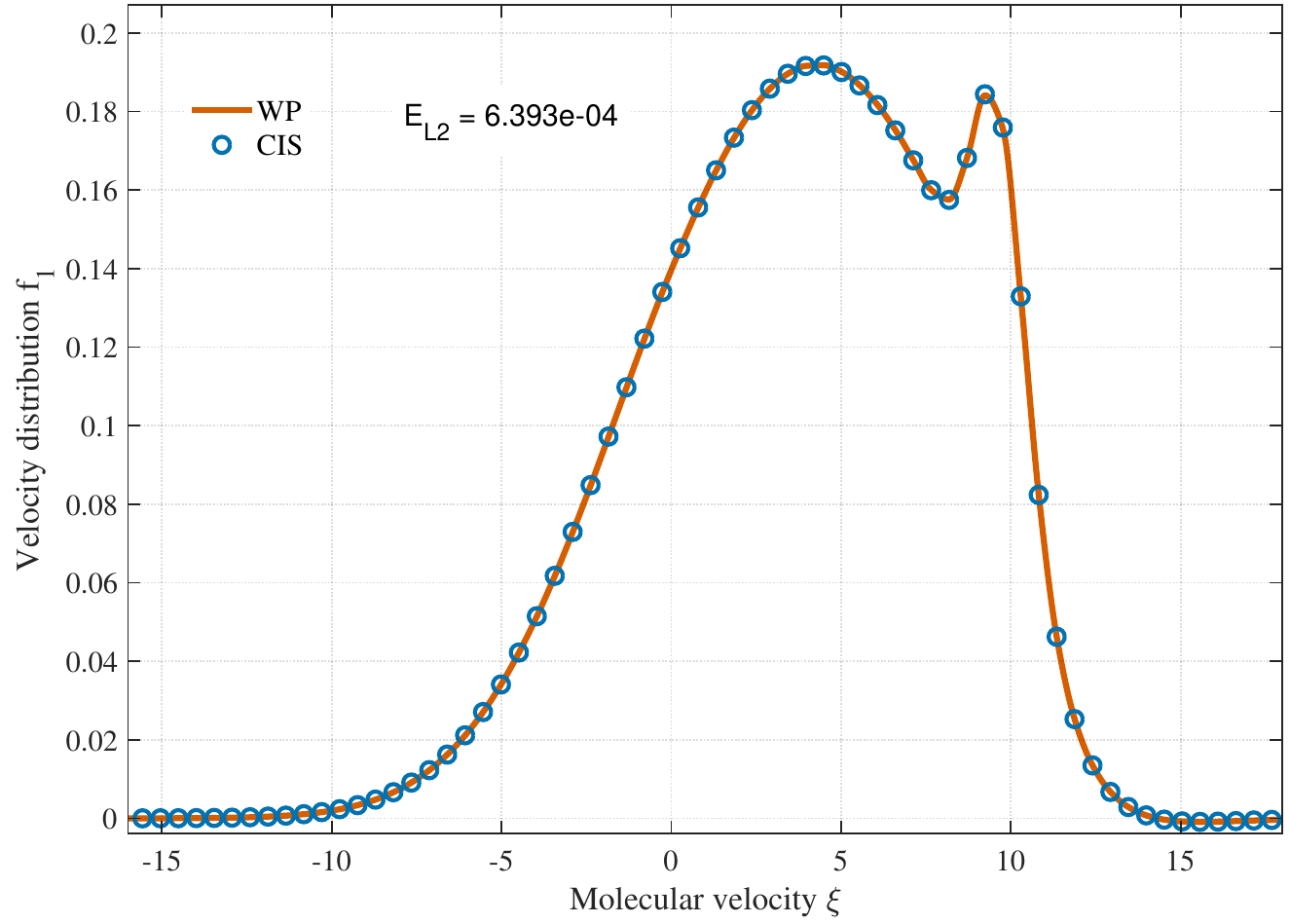}}
\caption{$x-x_{\rho\text{-mid}}=-0.12$.}
\end{subfigure}\hspace{0.0200\figw}%
\begin{subfigure}[t]{0.4739\figw}\centering
\raisebox{0.0000\figw}{\includegraphics[width=0.4739\figw]{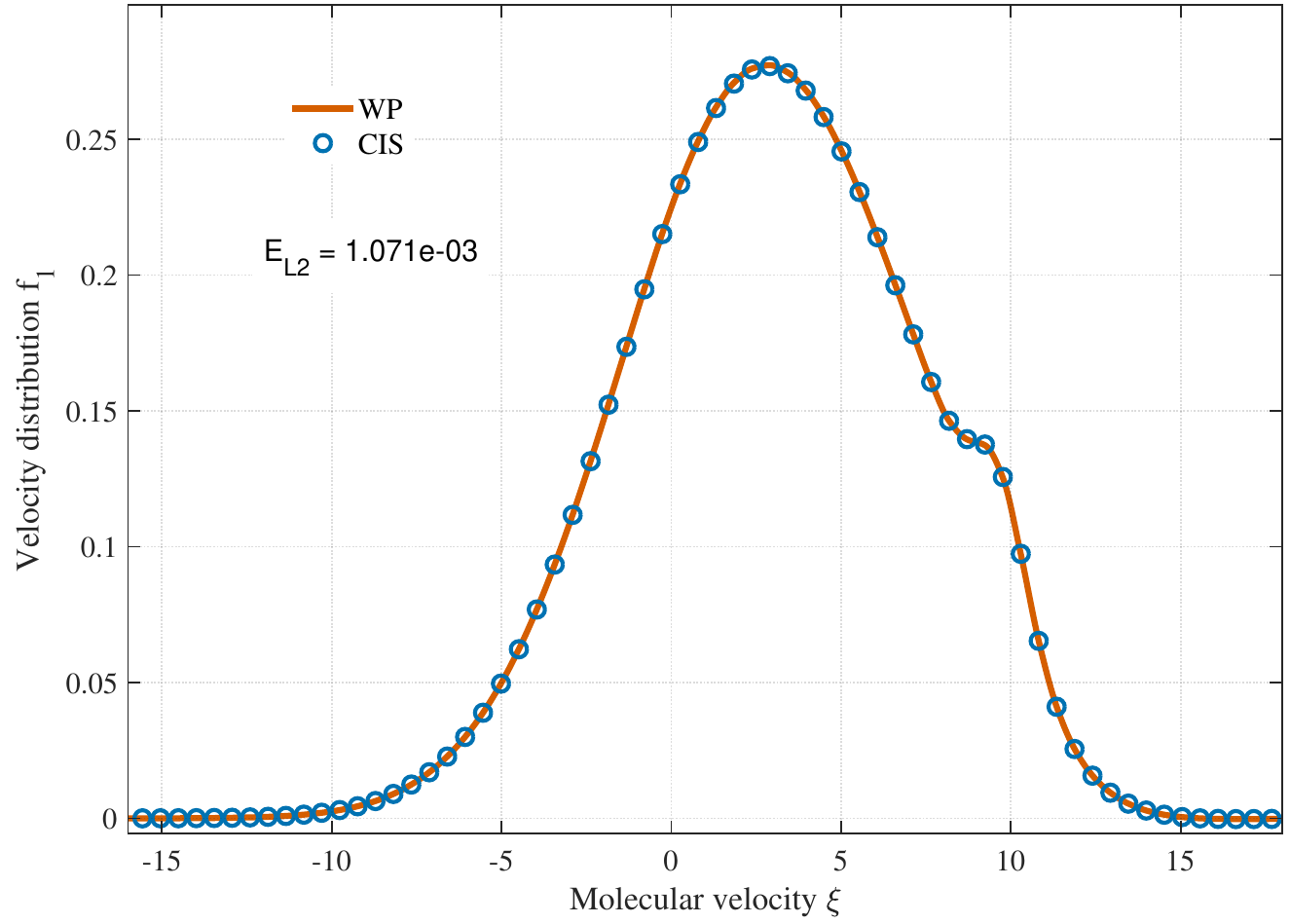}}
\caption{$x-x_{\rho\text{-mid}}=1.93$.}
\end{subfigure}
\caption{Normal shock at $\mathrm{Ma}=10$ and $\varepsilon=1$: reduced distributions $f_1(\xi)$ at four positions across the shock, ordered from upstream (a) to downstream (d), where $x_{\rho\text{-mid}}$ is the density midpoint. Lines denote IWPD and symbols CIS.}\label{fig:shock_dist}
\end{figure}

\clearpage
\section{Couette flow}
\begin{figure}[H]
\centering\singlespacing\setlength{\figw}{.65\linewidth}
\begin{subfigure}[t]{0.9494\figw}\centering
\makebox[\linewidth]{\raisebox{0.0039\figw}{\includegraphics[width=0.4543\figw]{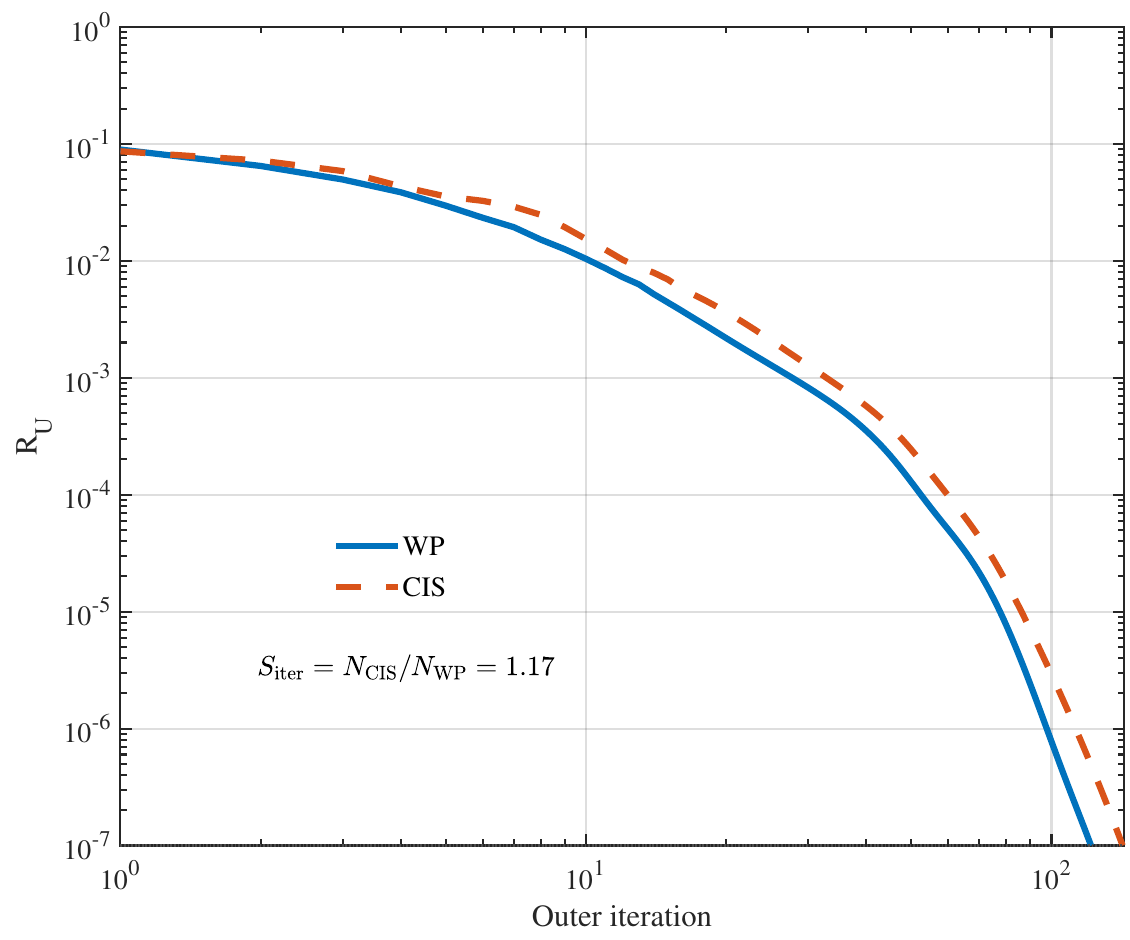}}\hspace{0.0200\figw}\raisebox{0.0000\figw}{\includegraphics[width=0.4752\figw]{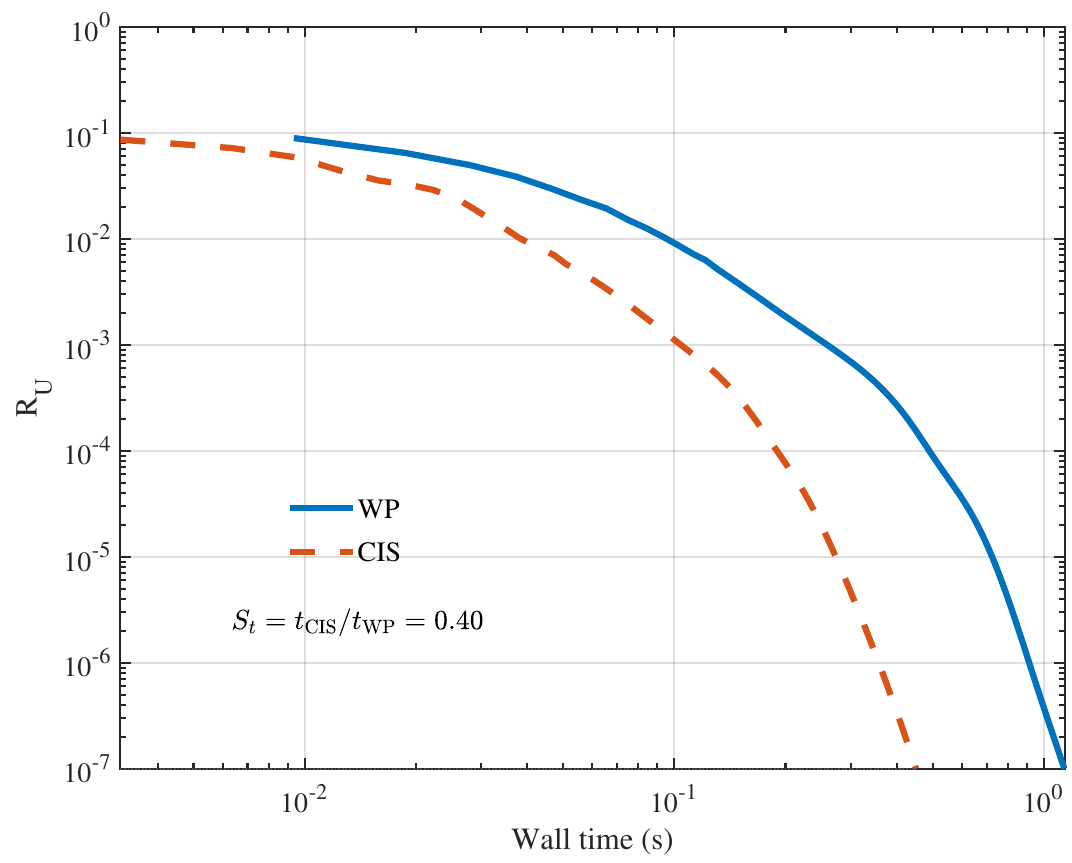}}}
\caption{$\varepsilon=1$.}
\end{subfigure}
\par\smallskip
\begin{subfigure}[t]{0.9494\figw}\centering
\makebox[\linewidth]{\raisebox{0.0034\figw}{\includegraphics[width=0.4553\figw]{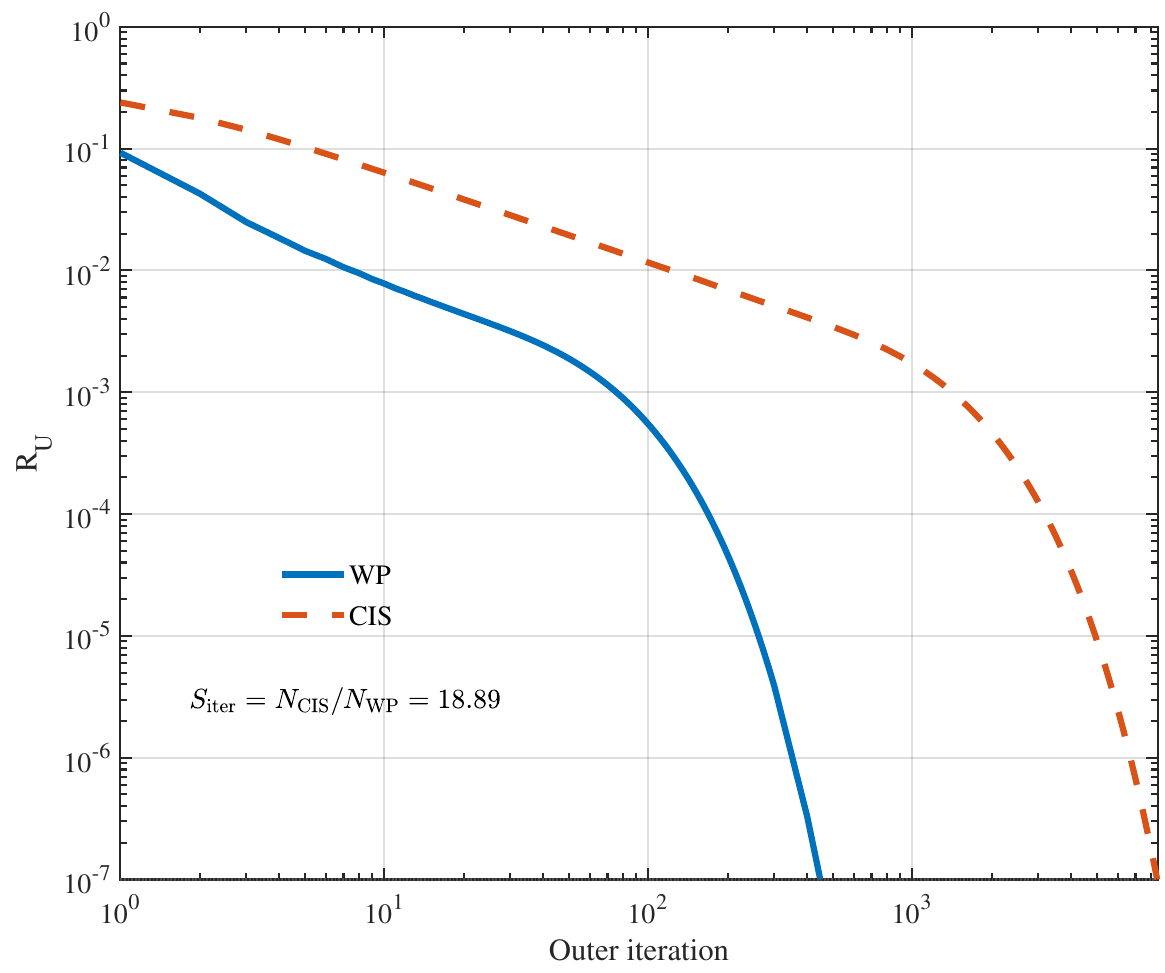}}\hspace{0.0200\figw}\raisebox{0.0000\figw}{\includegraphics[width=0.4741\figw]{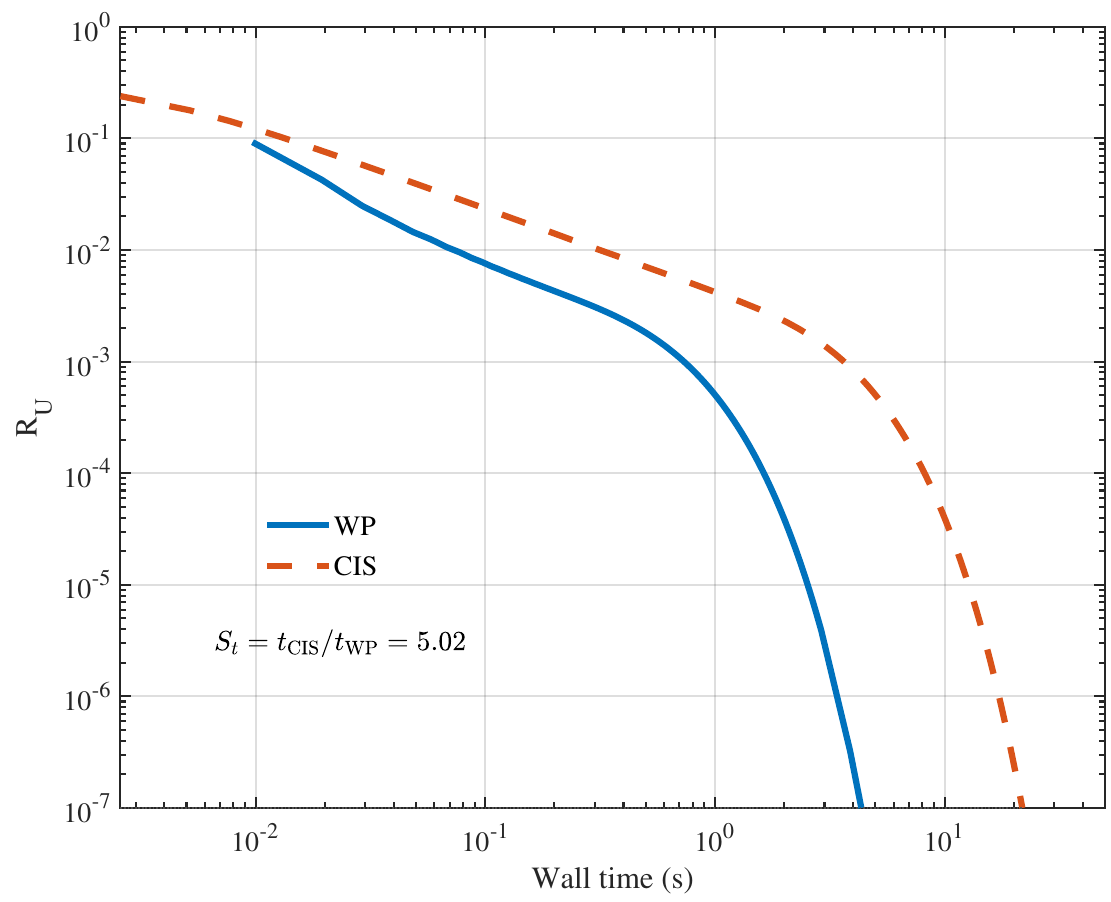}}}
\caption{$\varepsilon=10^{-2}$.}
\end{subfigure}
\par\smallskip
\begin{subfigure}[t]{0.9494\figw}\centering
\makebox[\linewidth]{\raisebox{0.0007\figw}{\includegraphics[width=0.4595\figw]{couette-residual-wp-si-couette-kn0d0001-2vs.pdf}}\hspace{0.0200\figw}\raisebox{0.0000\figw}{\includegraphics[width=0.4699\figw]{couette-residual-walltime-couette-kn0d0001-2vs.pdf}}}
\caption{$\varepsilon=10^{-4}$.}
\end{subfigure}
\caption{Couette flow: macroscopic residual norm $\|R_U\|$ versus outer iteration (left) and wall time (right) for IWPD (solid), CIS (dashed) and, at $\varepsilon=10^{-4}$, NS (dotted). The ratios $S_{\rm iter}$ and $S_t$ are evaluated at the common residual threshold.}\label{fig:couette_residual}
\end{figure}

\begin{figure}[H]
\centering\singlespacing\setlength{\figw}{.65\linewidth}
\begin{subfigure}[t]{0.4675\figw}\centering
\raisebox{0.0000\figw}{\includegraphics[width=0.4675\figw]{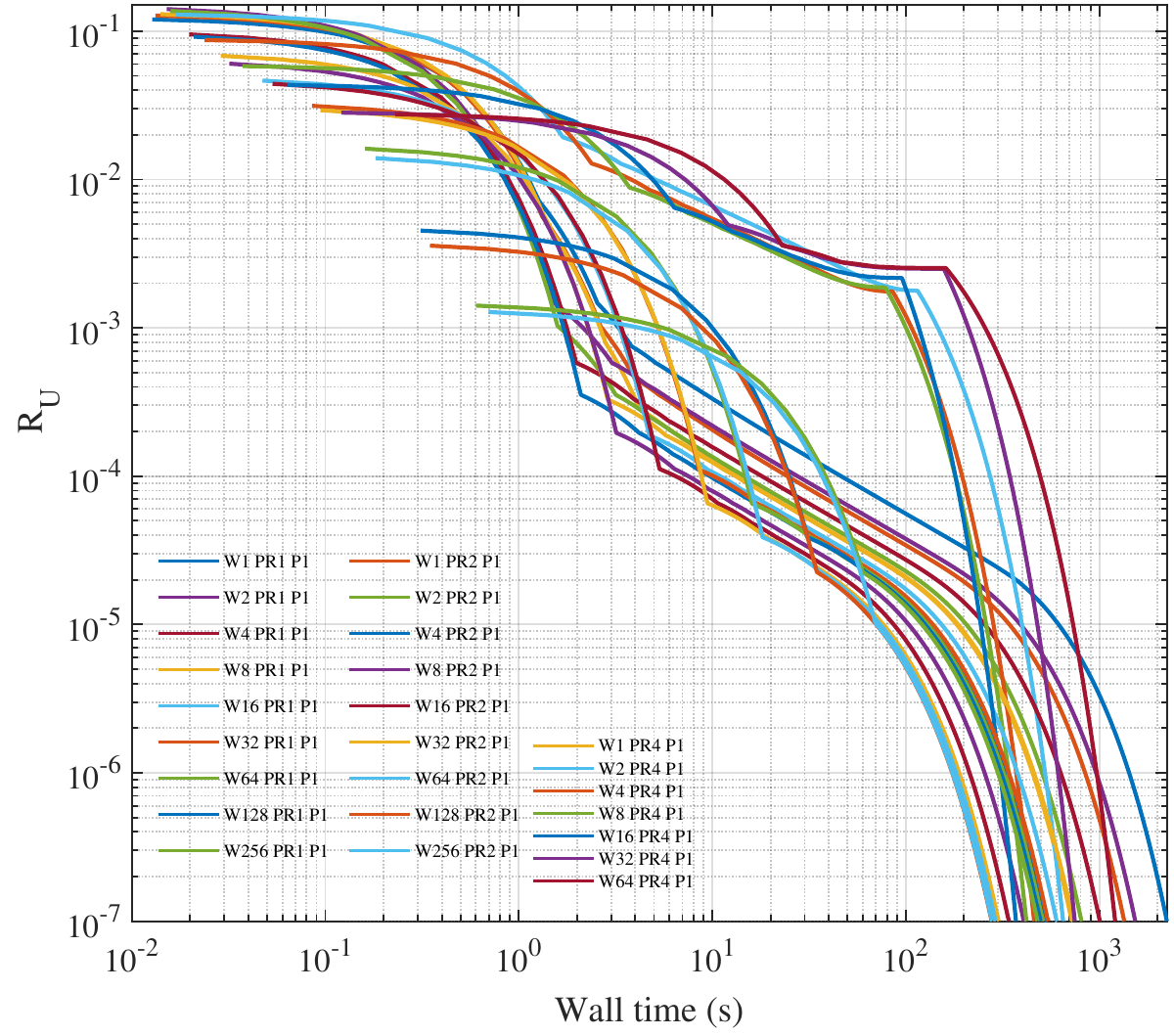}}
\caption{Residual versus wall time.}
\end{subfigure}\hspace{0.0200\figw}%
\begin{subfigure}[t]{0.5075\figw}\centering
\raisebox{0.0018\figw}{\includegraphics[width=0.5075\figw]{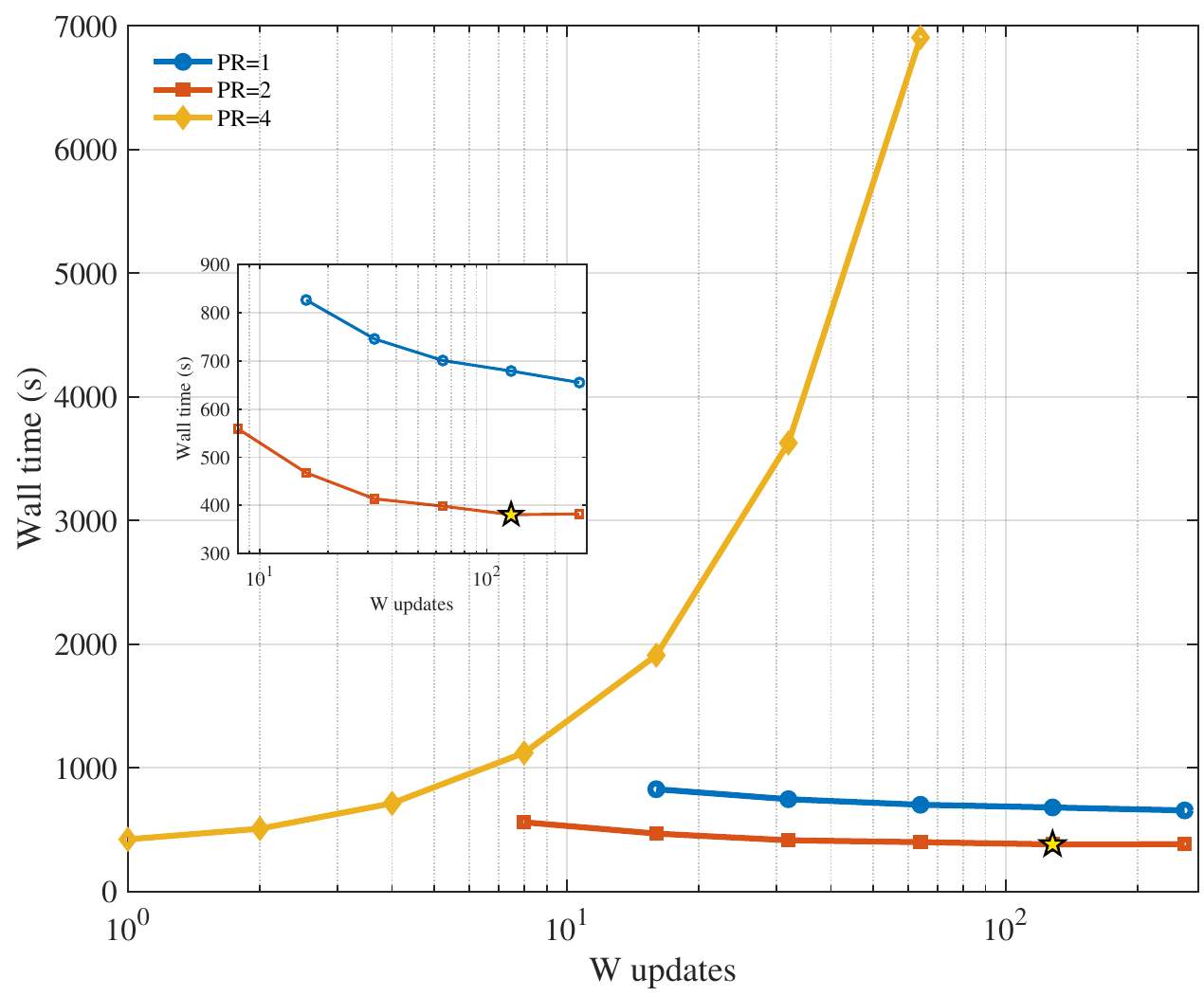}}
\caption{Wall time versus the number of W updates.}
\end{subfigure}
\caption{Couette flow at $\varepsilon=10^{-4}$: scan of the IWPD controls with P1, covering W1--W256 with PR1, PR2 and PR4. The inset in (b) enlarges the neighbourhood of the optimum, which is marked by a star.}\label{fig:couette_scan}
\end{figure}

\clearpage
\section{Lid-driven cavity}
\begin{figure}[H]
\centering\singlespacing\setlength{\figw}{.65\linewidth}
\begin{subfigure}[t]{0.4841\figw}\centering
\raisebox{0.0000\figw}{\includegraphics[width=0.4841\figw]{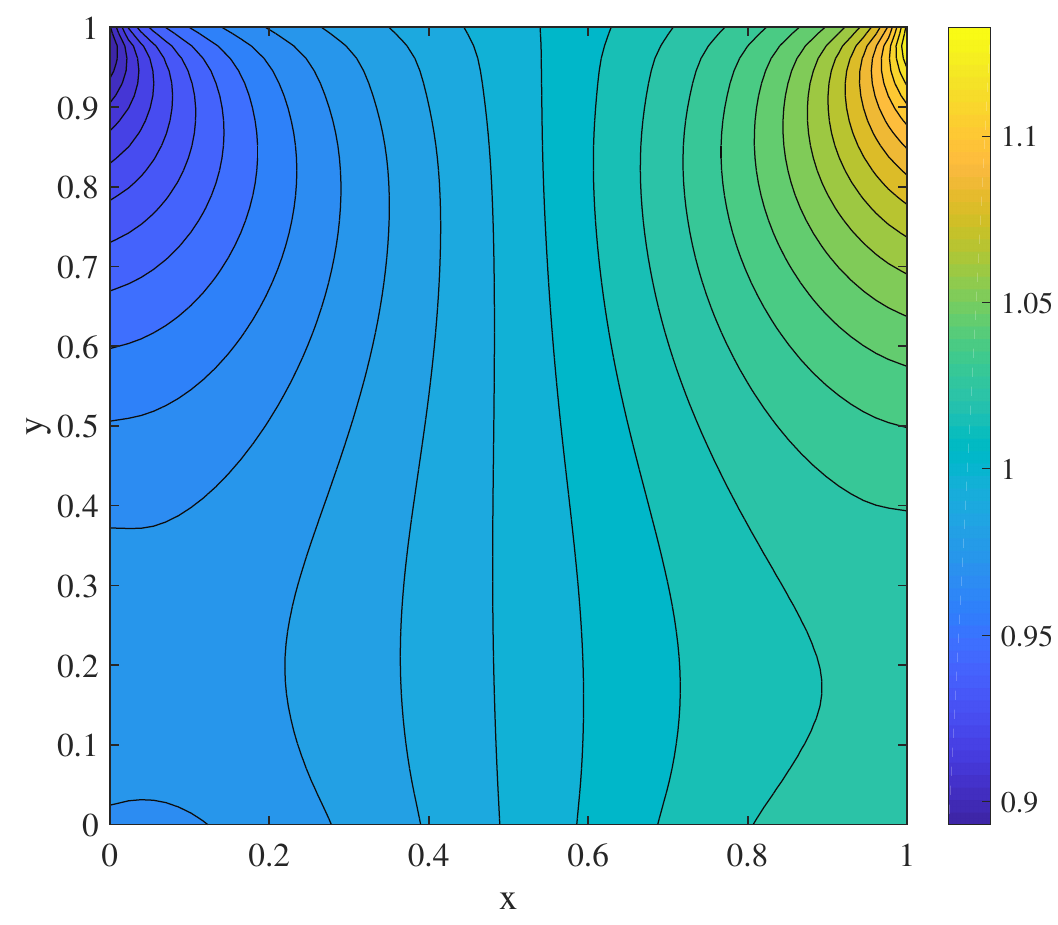}}
\caption{Density.}
\end{subfigure}\hspace{0.0200\figw}%
\begin{subfigure}[t]{0.4909\figw}\centering
\raisebox{0.0000\figw}{\includegraphics[width=0.4909\figw]{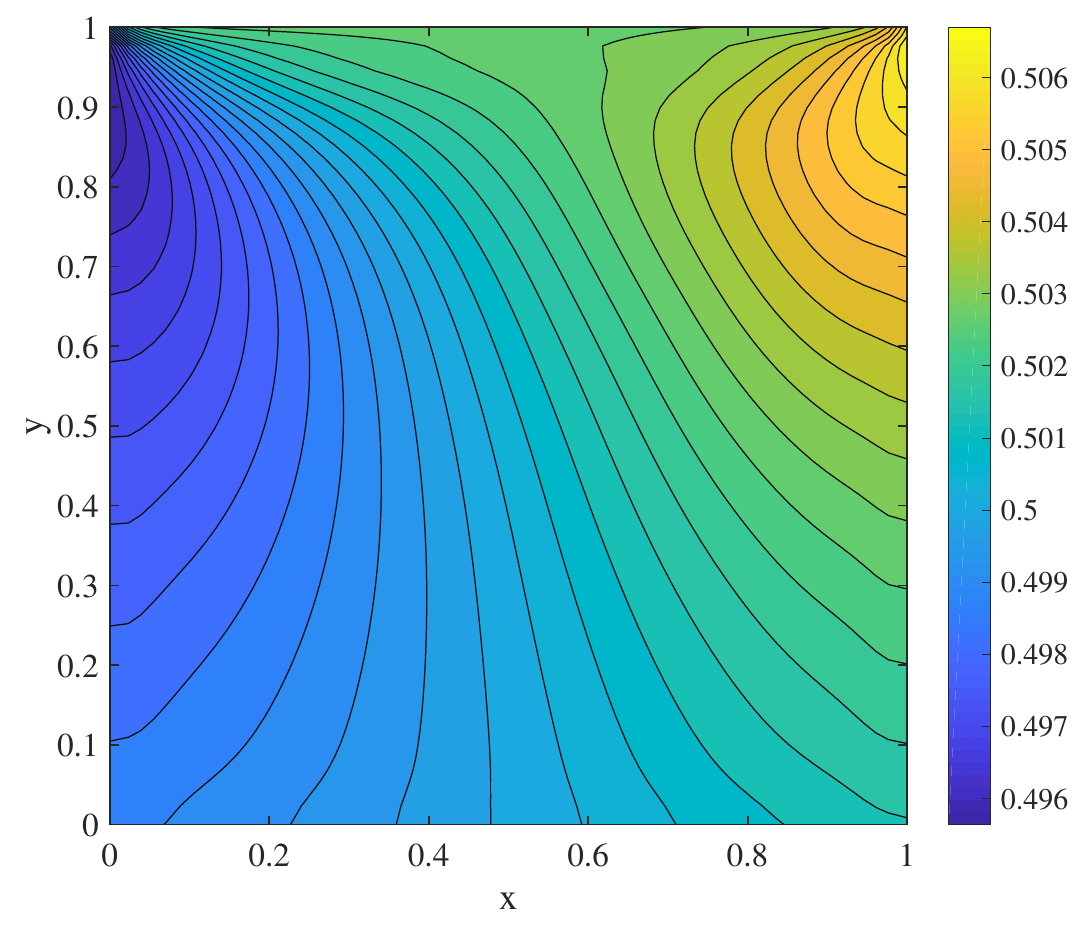}}
\caption{Temperature.}
\end{subfigure}
\par\smallskip
\begin{subfigure}[t]{0.4939\figw}\centering
\raisebox{0.0000\figw}{\includegraphics[width=0.4939\figw]{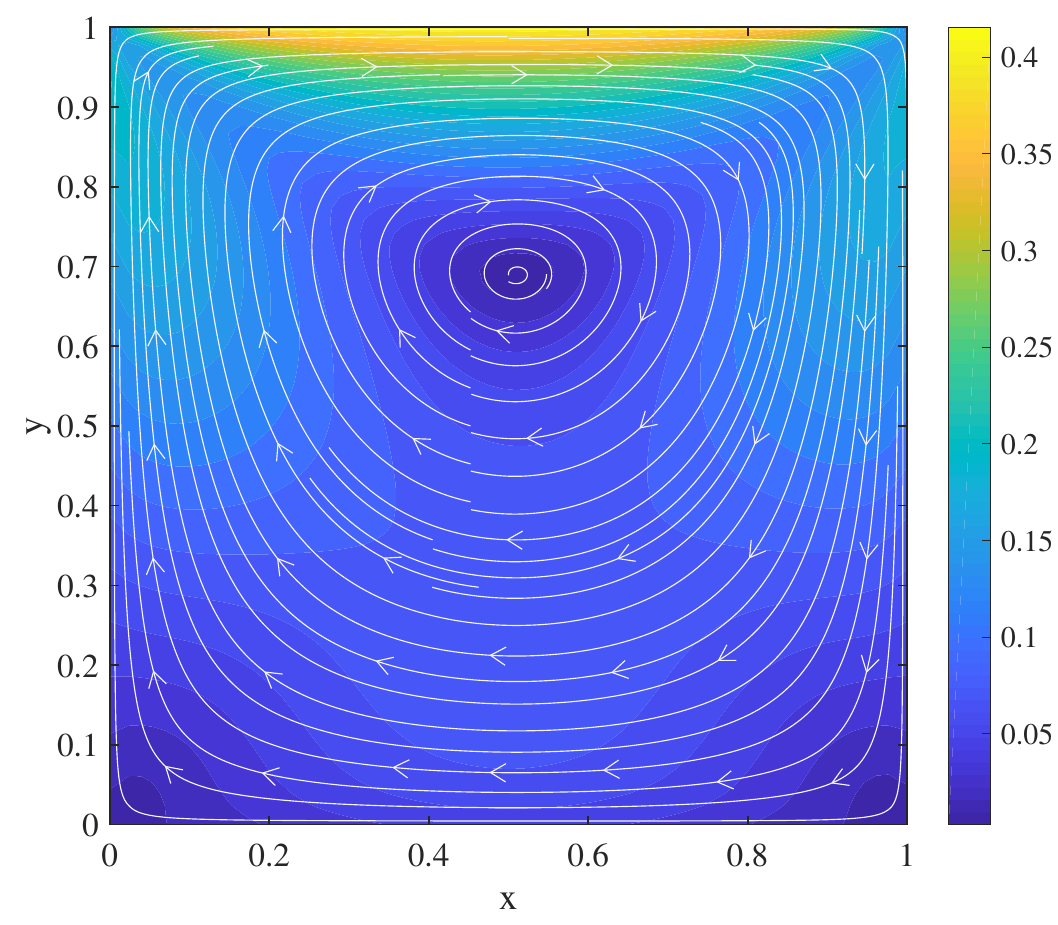}}
\caption{Speed and streamlines.}
\end{subfigure}\hspace{0.0200\figw}%
\begin{subfigure}[t]{0.4811\figw}\centering
\raisebox{0.0000\figw}{\includegraphics[width=0.4811\figw]{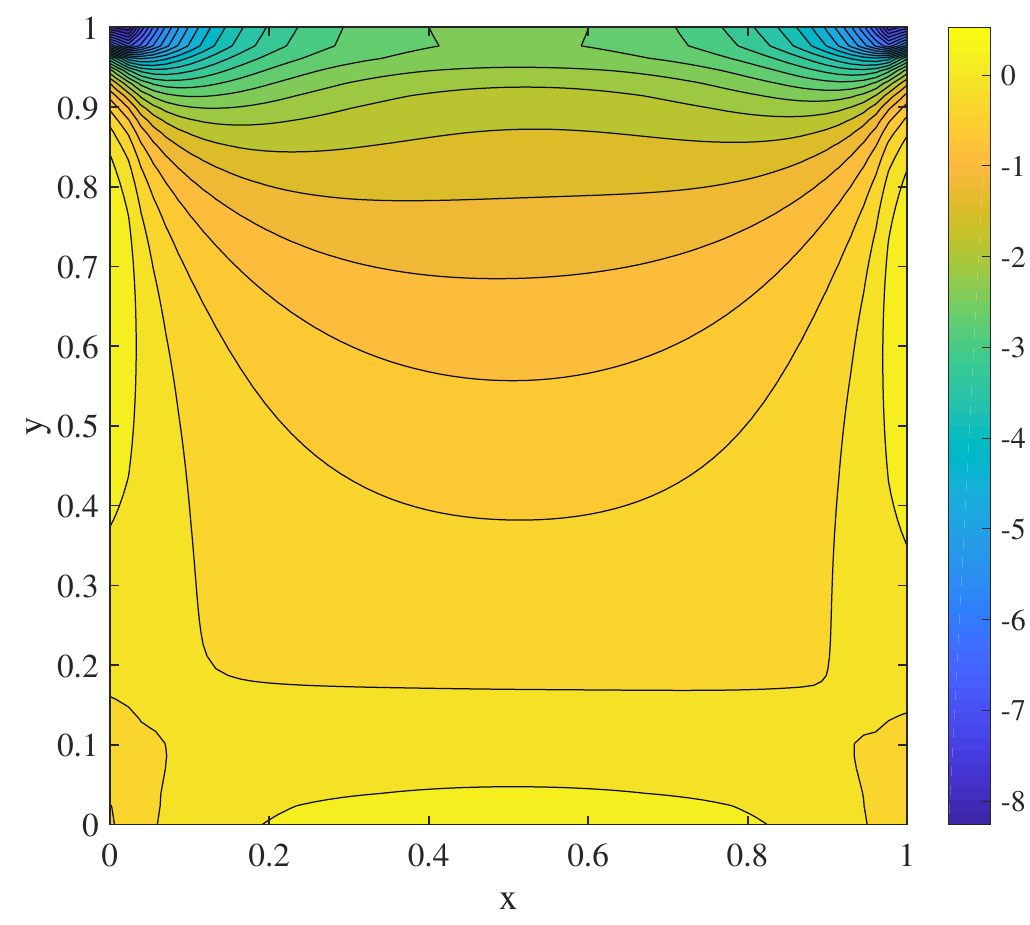}}
\caption{Vorticity.}
\end{subfigure}
\caption{Lid-driven cavity at $\varepsilon=1$: macroscopic fields of IWPD (filled contours) and CIS (lines, streamlines in (c)) on common levels.}\label{fig:cavity_kn1_fields}
\end{figure}

\begin{figure}[H]
\centering\singlespacing\setlength{\figw}{.65\linewidth}
\begin{subfigure}[t]{0.4841\figw}\centering
\raisebox{0.0000\figw}{\includegraphics[width=0.4841\figw]{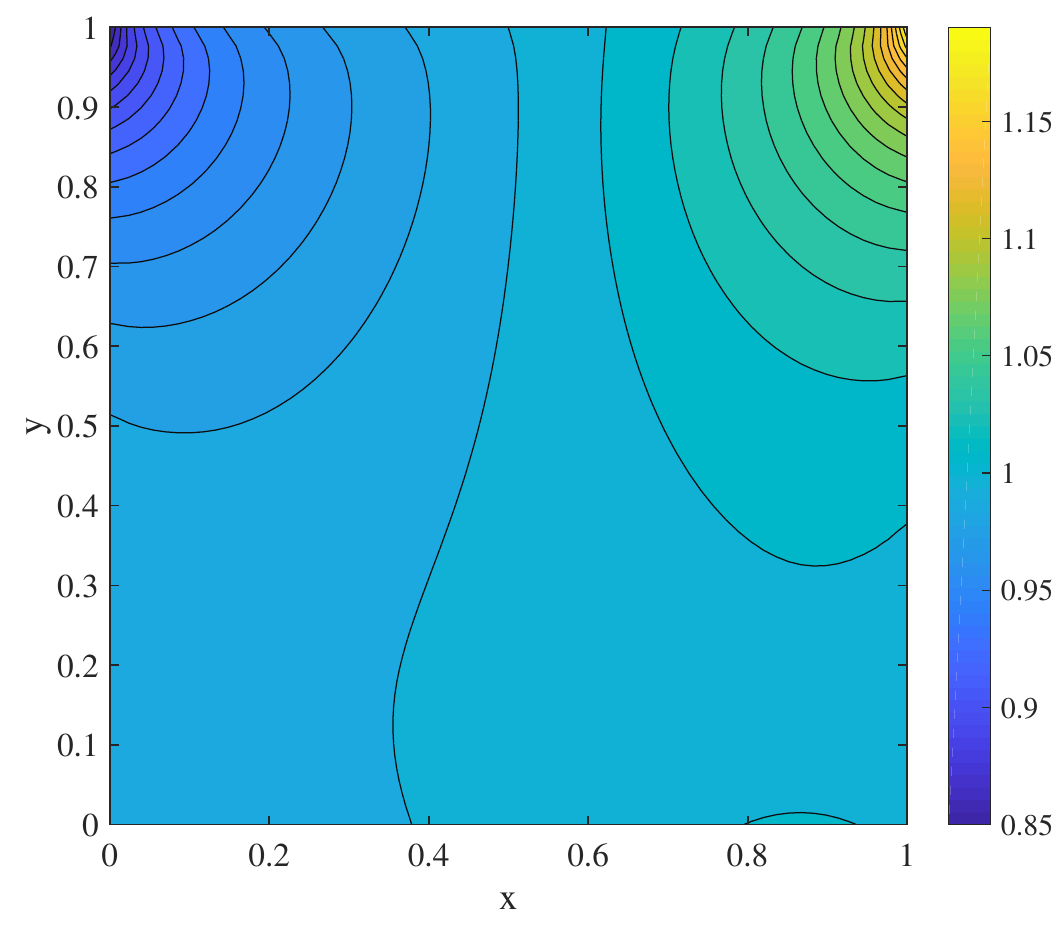}}
\caption{Density.}
\end{subfigure}\hspace{0.0200\figw}%
\begin{subfigure}[t]{0.4909\figw}\centering
\raisebox{0.0000\figw}{\includegraphics[width=0.4909\figw]{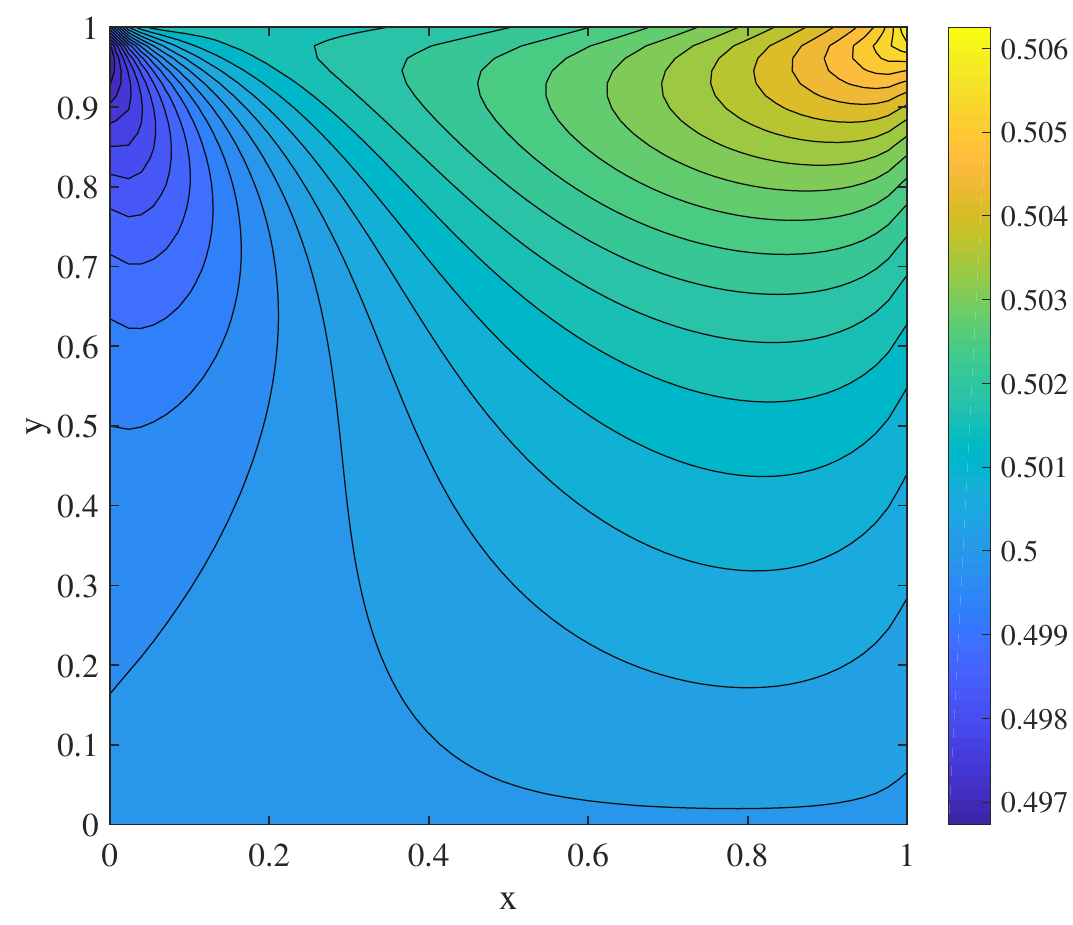}}
\caption{Temperature.}
\end{subfigure}
\par\smallskip
\begin{subfigure}[t]{0.4869\figw}\centering
\raisebox{0.0000\figw}{\includegraphics[width=0.4869\figw]{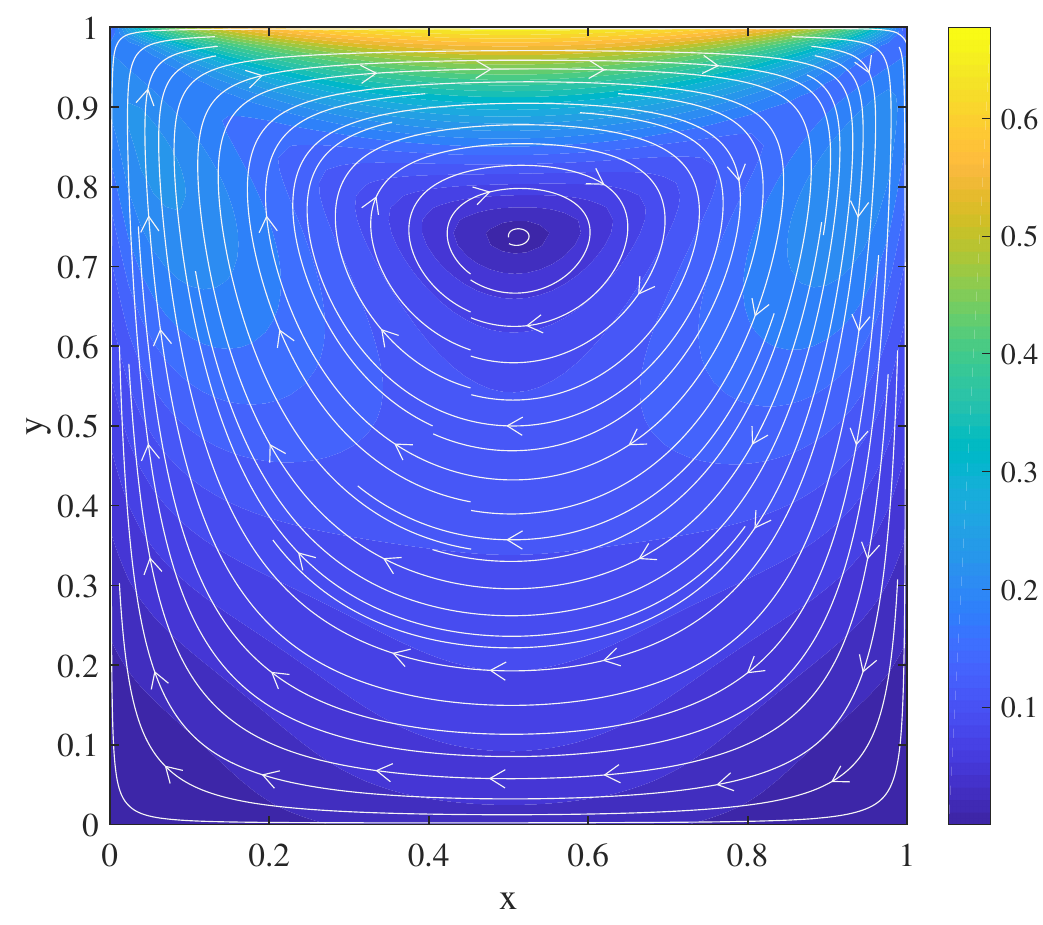}}
\caption{Speed and streamlines.}
\end{subfigure}\hspace{0.0200\figw}%
\begin{subfigure}[t]{0.4881\figw}\centering
\raisebox{0.0000\figw}{\includegraphics[width=0.4881\figw]{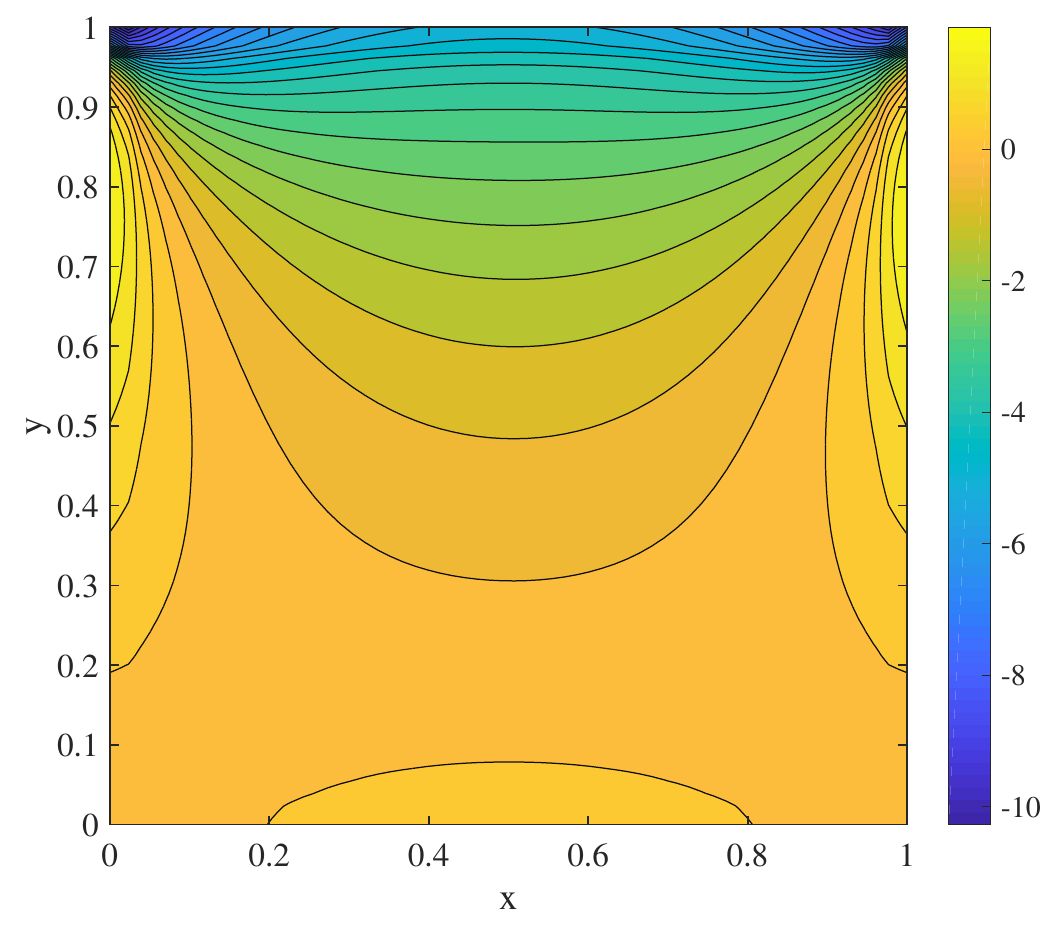}}
\caption{Vorticity.}
\end{subfigure}
\caption{Lid-driven cavity at $\varepsilon=0.075$: macroscopic fields of IWPD (filled contours) and CIS (lines, streamlines in (c)) on common levels.}\label{fig:cavity_kn0075_fields}
\end{figure}

\begin{figure}[H]
\centering\singlespacing\setlength{\figw}{.65\linewidth}
\begin{subfigure}[t]{0.4812\figw}\centering
\raisebox{0.0000\figw}{\includegraphics[width=0.4812\figw]{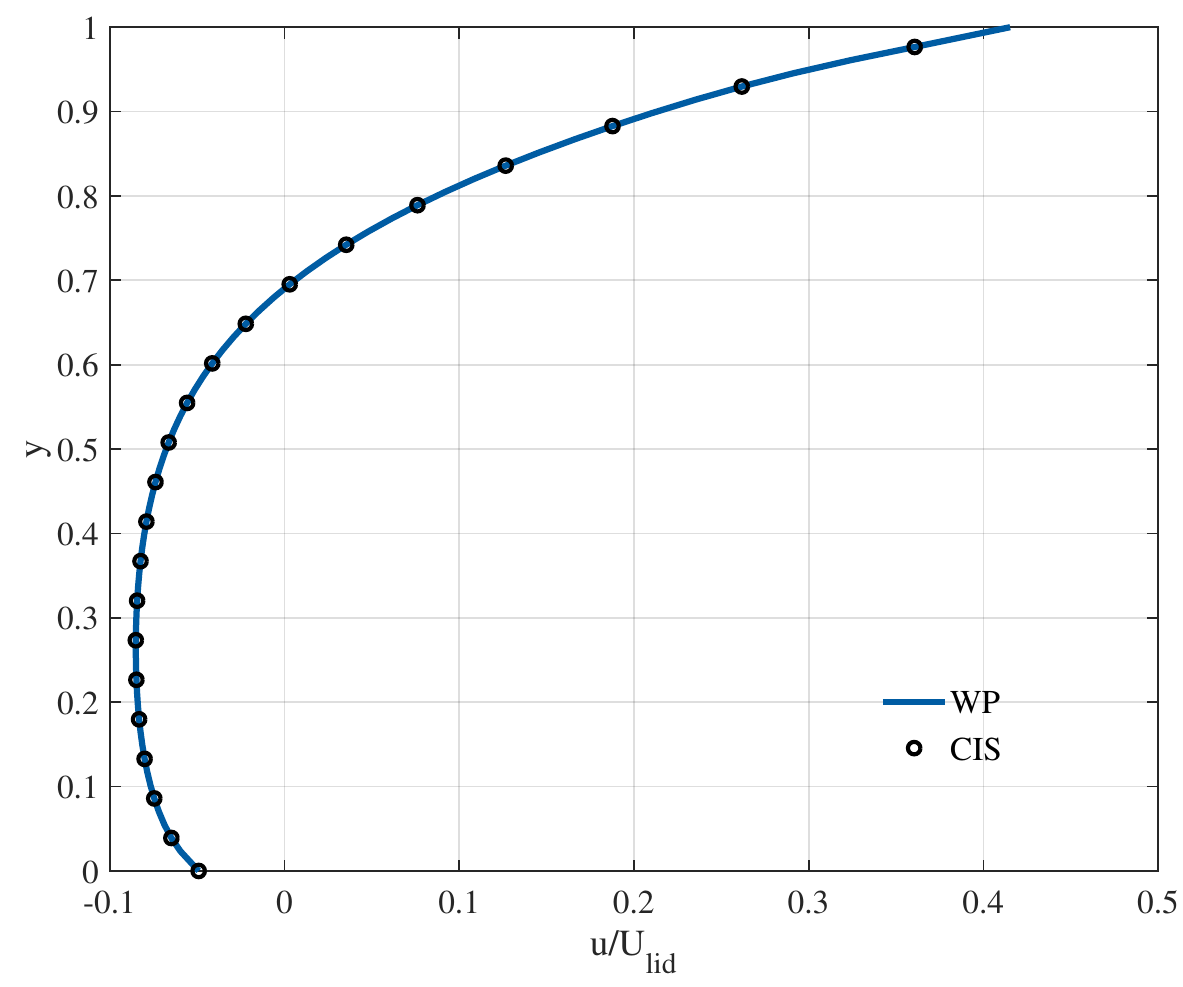}}
\caption{$u/U_{\rm lid}$ on the vertical centreline, $\varepsilon=1$.}
\end{subfigure}\hspace{0.0200\figw}%
\begin{subfigure}[t]{0.4938\figw}\centering
\raisebox{0.0063\figw}{\includegraphics[width=0.4938\figw]{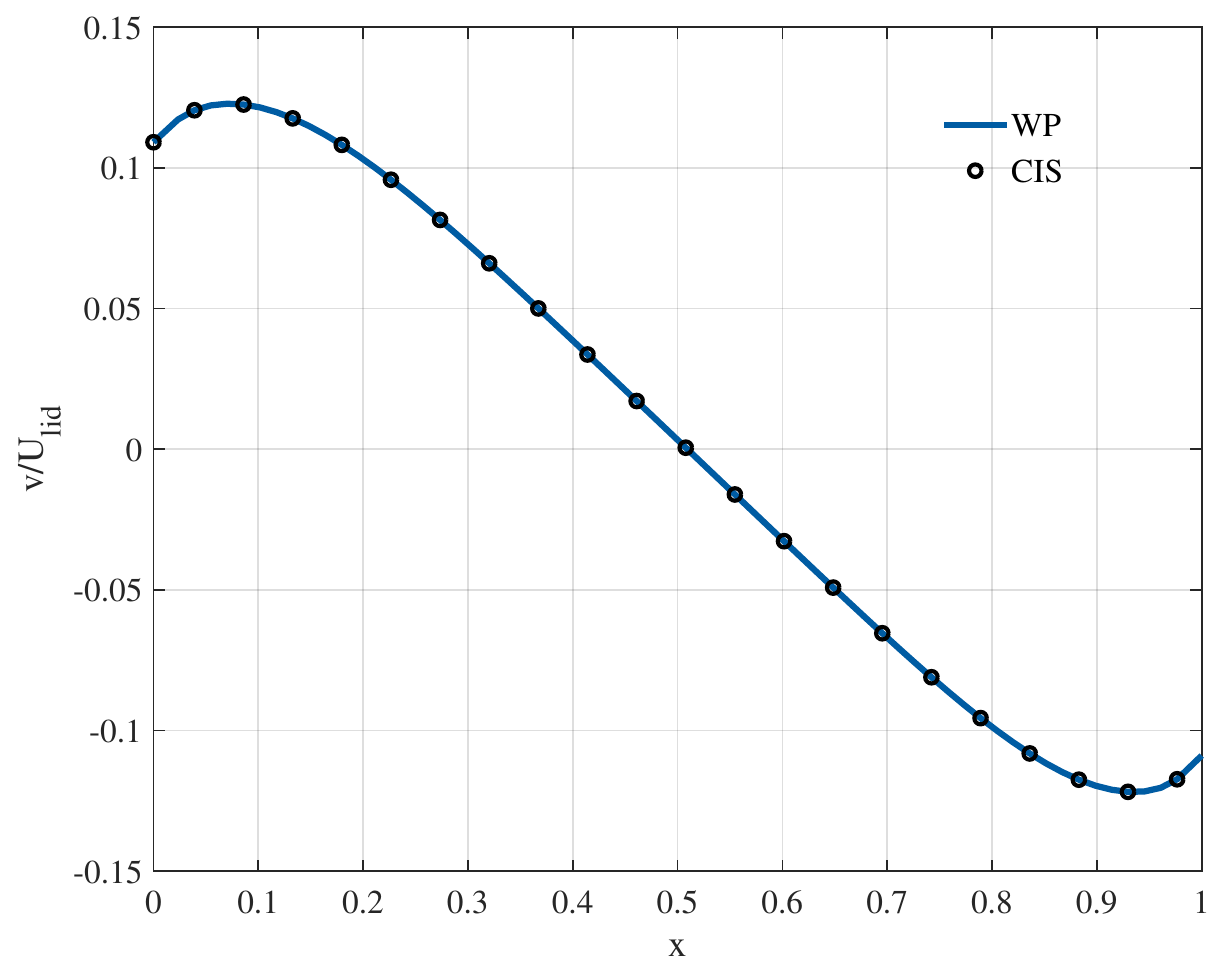}}
\caption{$v/U_{\rm lid}$ on the horizontal centreline, $\varepsilon=1$.}
\end{subfigure}
\par\smallskip
\begin{subfigure}[t]{0.4812\figw}\centering
\raisebox{0.0000\figw}{\includegraphics[width=0.4812\figw]{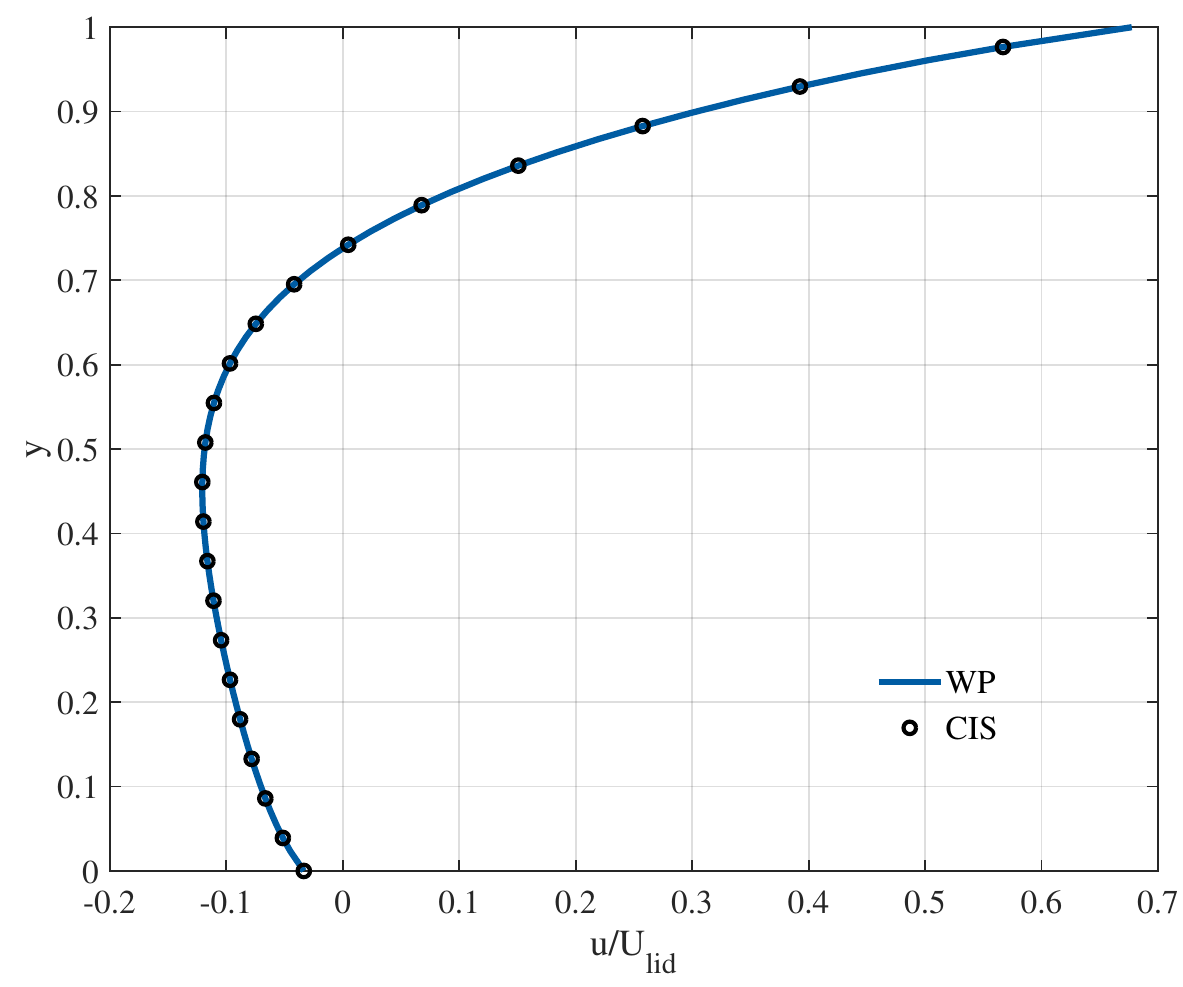}}
\caption{$u/U_{\rm lid}$ on the vertical centreline, $\varepsilon=0.075$.}
\end{subfigure}\hspace{0.0200\figw}%
\begin{subfigure}[t]{0.4938\figw}\centering
\raisebox{0.0063\figw}{\includegraphics[width=0.4938\figw]{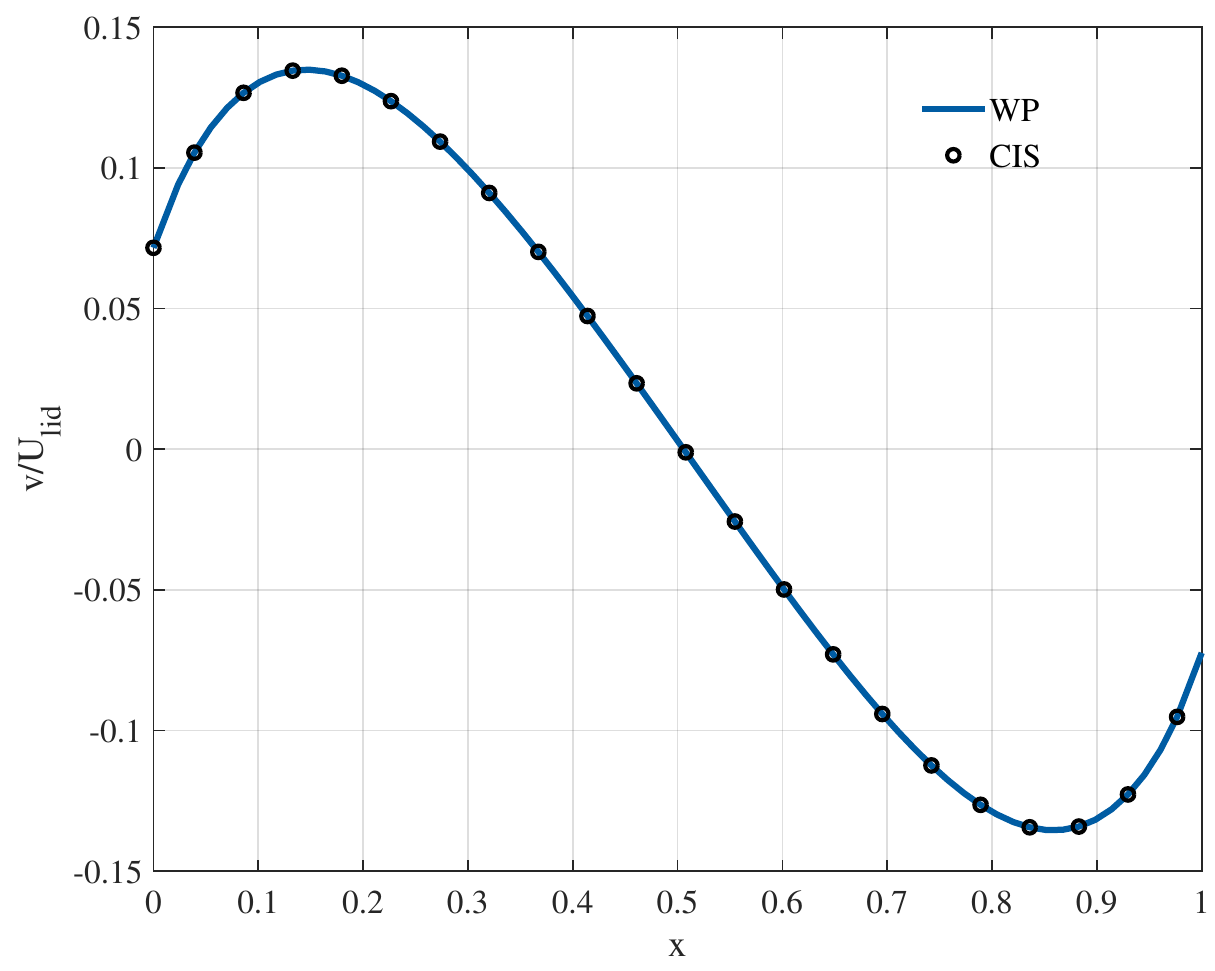}}
\caption{$v/U_{\rm lid}$ on the horizontal centreline, $\varepsilon=0.075$.}
\end{subfigure}
\caption{Lid-driven cavity at $\varepsilon=1$ and 0.075: centreline velocity profiles. Lines denote IWPD and symbols CIS.}\label{fig:cavity_rarefied_profiles}
\end{figure}

\begin{figure}[H]
\centering\singlespacing\setlength{\figw}{.65\linewidth}
\begin{subfigure}[t]{0.4869\figw}\centering
\raisebox{0.0000\figw}{\includegraphics[width=0.4869\figw]{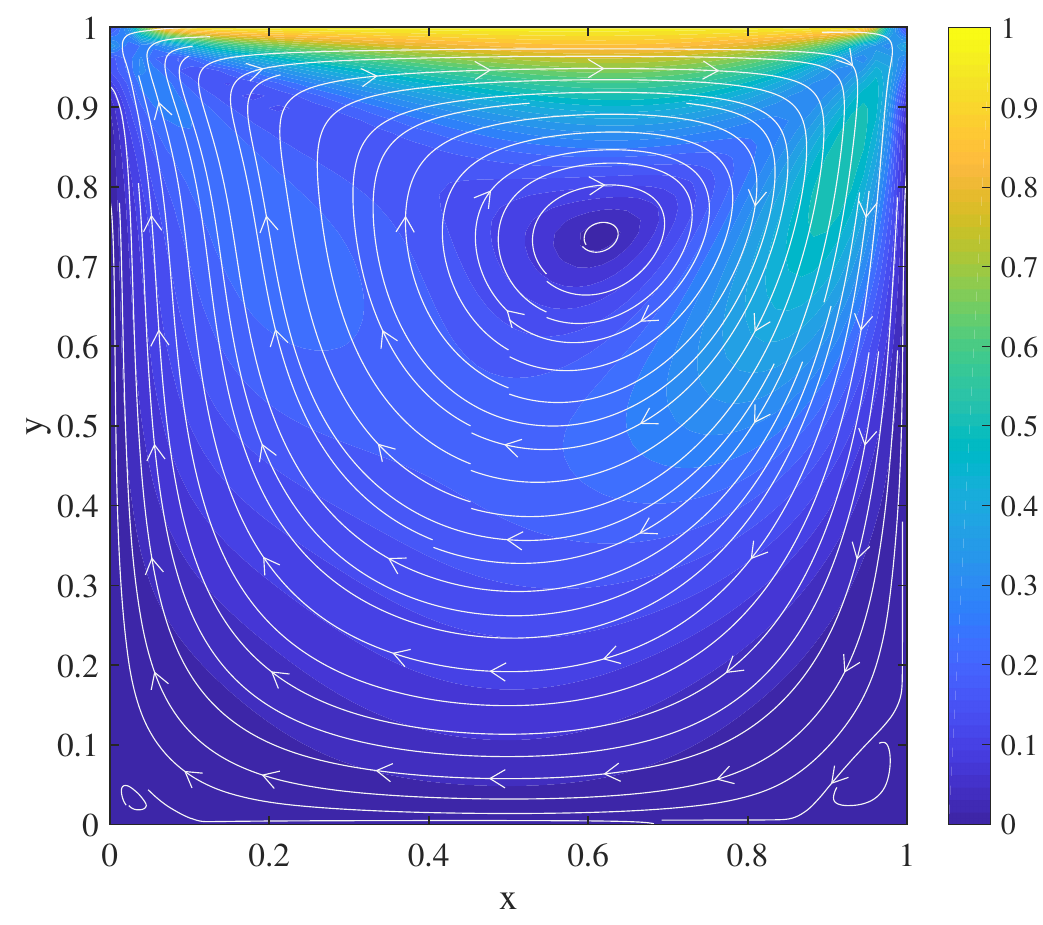}}
\caption{Speed and streamlines.}
\end{subfigure}\hspace{0.0200\figw}%
\begin{subfigure}[t]{0.4881\figw}\centering
\raisebox{0.0000\figw}{\includegraphics[width=0.4881\figw]{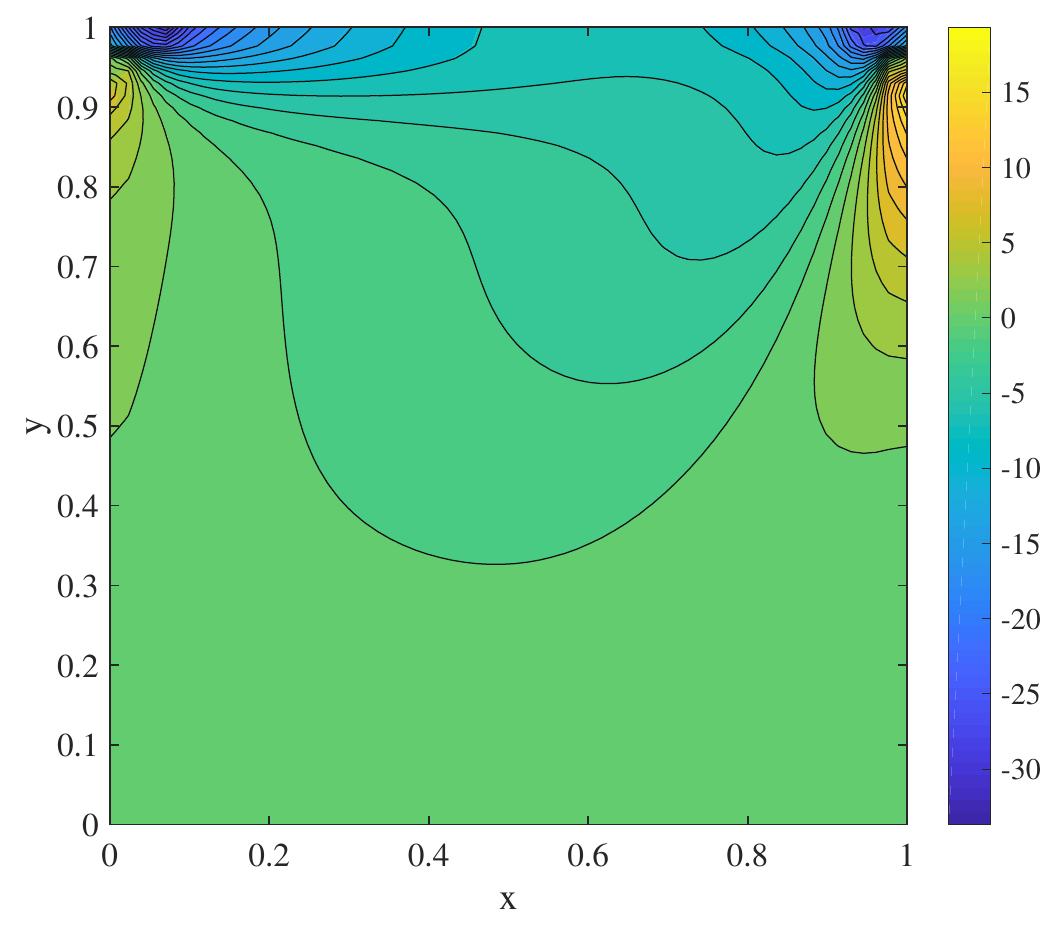}}
\caption{Vorticity.}
\end{subfigure}
\par\smallskip
\begin{subfigure}[t]{0.4774\figw}\centering
\raisebox{0.0000\figw}{\includegraphics[width=0.4774\figw]{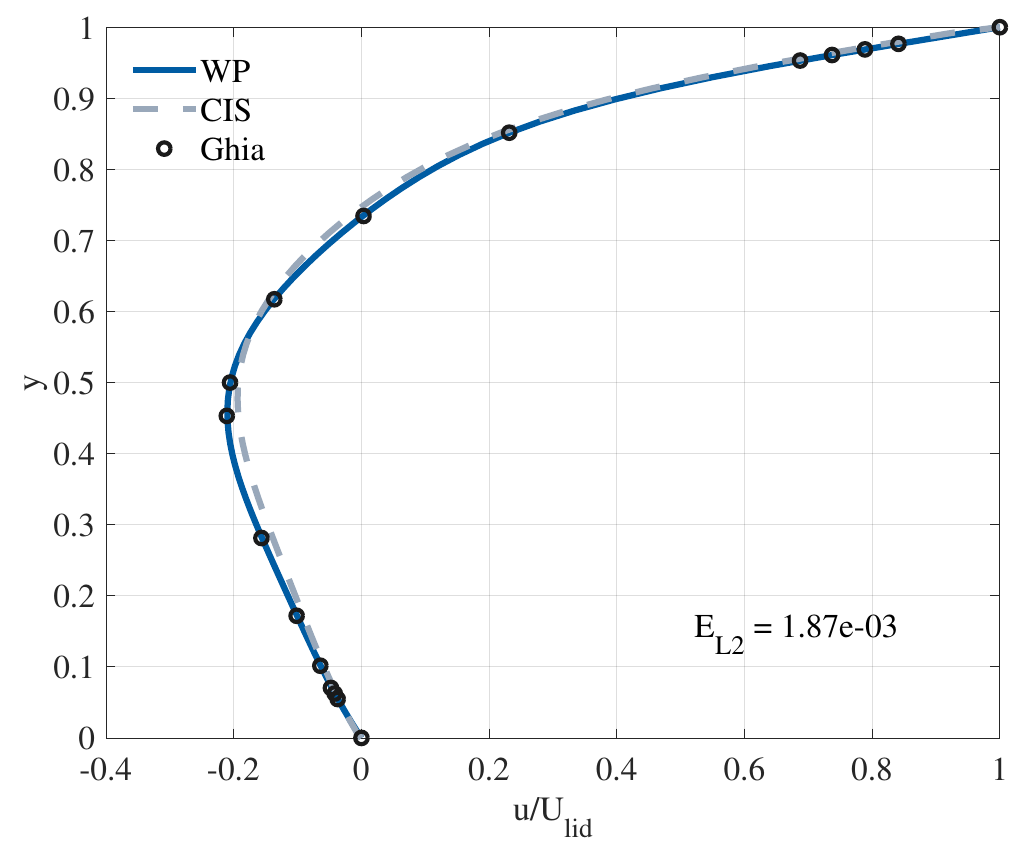}}
\caption{$u/U_{\rm lid}$ on the vertical centreline.}
\end{subfigure}\hspace{0.0200\figw}%
\begin{subfigure}[t]{0.4976\figw}\centering
\raisebox{0.0069\figw}{\includegraphics[width=0.4976\figw]{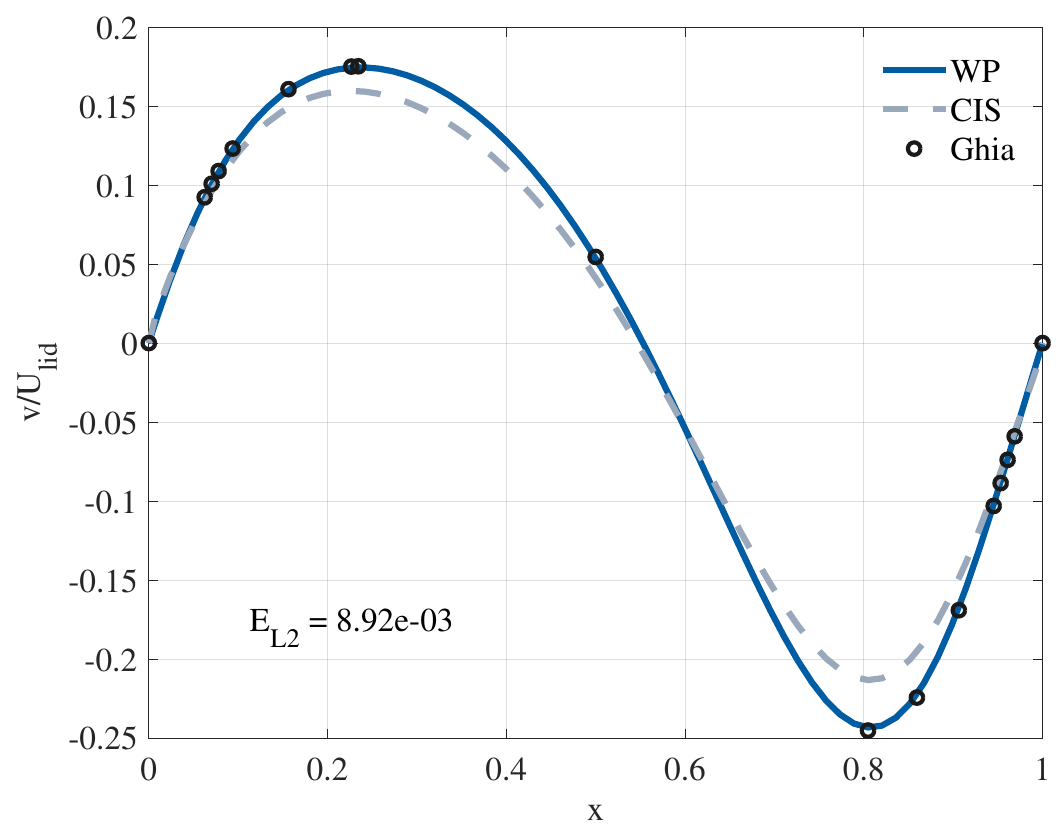}}
\caption{$v/U_{\rm lid}$ on the horizontal centreline.}
\end{subfigure}
\caption{Lid-driven cavity at $\mathrm{Re}=100$: (a,b) fields of IWPD (filled contours) and NS (lines, streamlines in (a)); (c,d) centreline profiles of IWPD (solid) and CIS (dashed) compared with the benchmark data of Ghia et al.~\cite{ghia1982} (symbols); the panels give the relative difference $E_{L_2}$ between IWPD and the benchmark data.}\label{fig:cavity_re100}
\end{figure}

\begin{figure}[H]
\centering\singlespacing\setlength{\figw}{.90\linewidth}
\begin{subfigure}[t]{0.6915\figw}\centering
\makebox[\linewidth]{\raisebox{0.0000\figw}{\includegraphics[width=0.3417\figw]{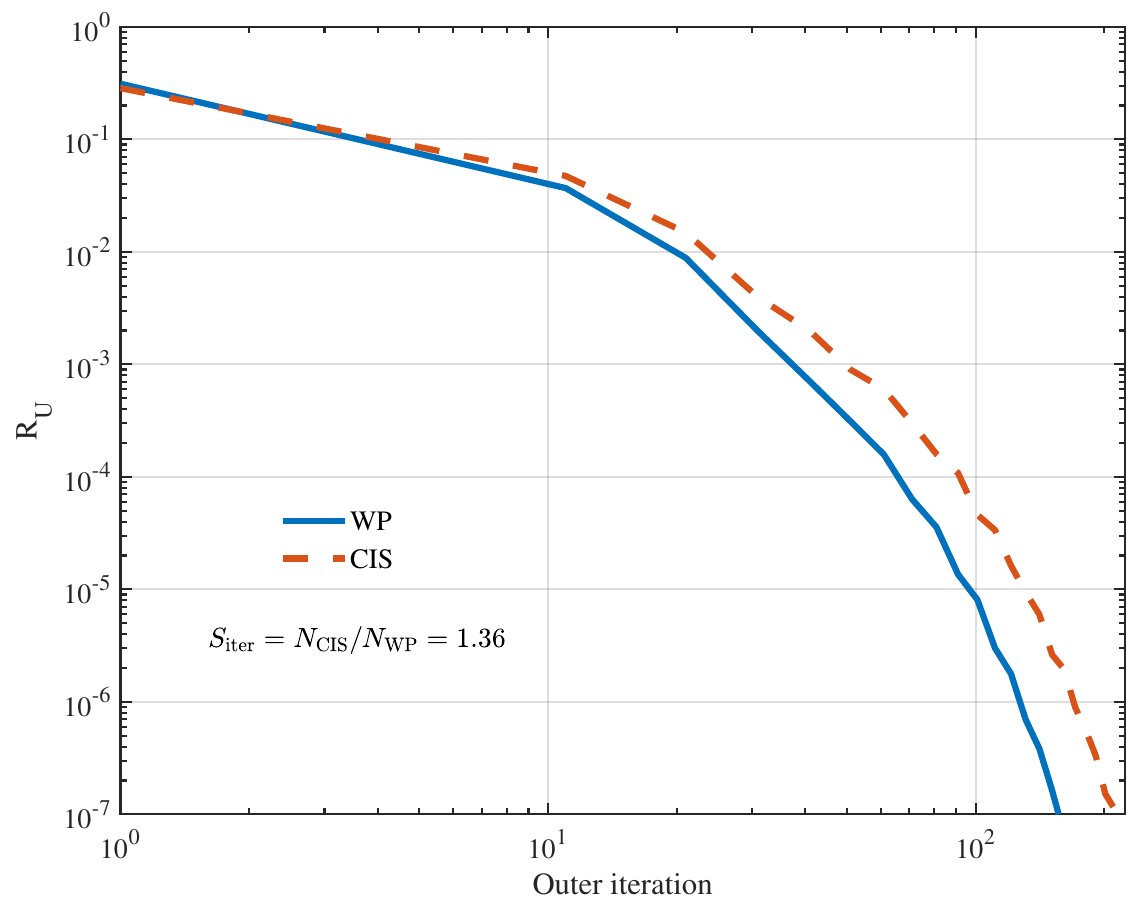}}\hspace{0.0200\figw}\raisebox{0.0001\figw}{\includegraphics[width=0.3298\figw]{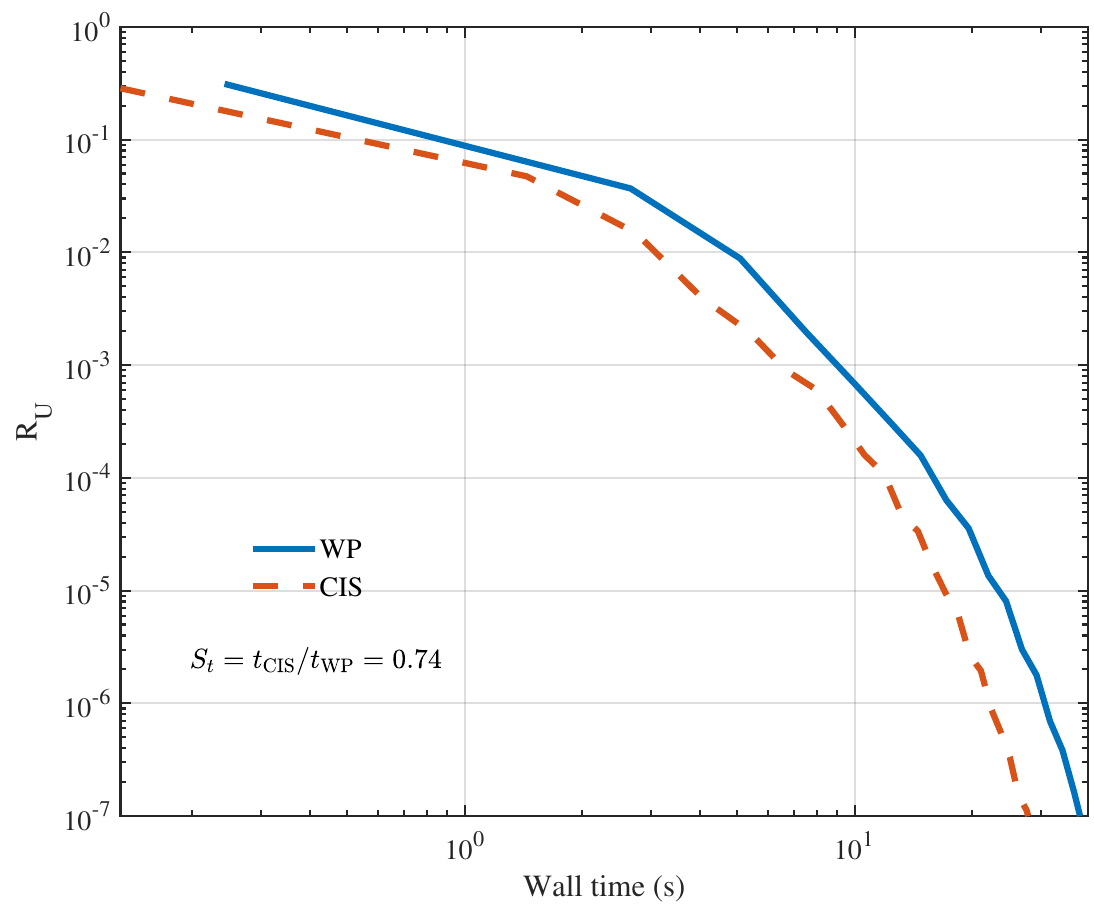}}}
\caption{$\varepsilon=1$.}
\end{subfigure}
\par\smallskip
\begin{subfigure}[t]{0.6915\figw}\centering
\makebox[\linewidth]{\raisebox{0.0020\figw}{\includegraphics[width=0.3340\figw]{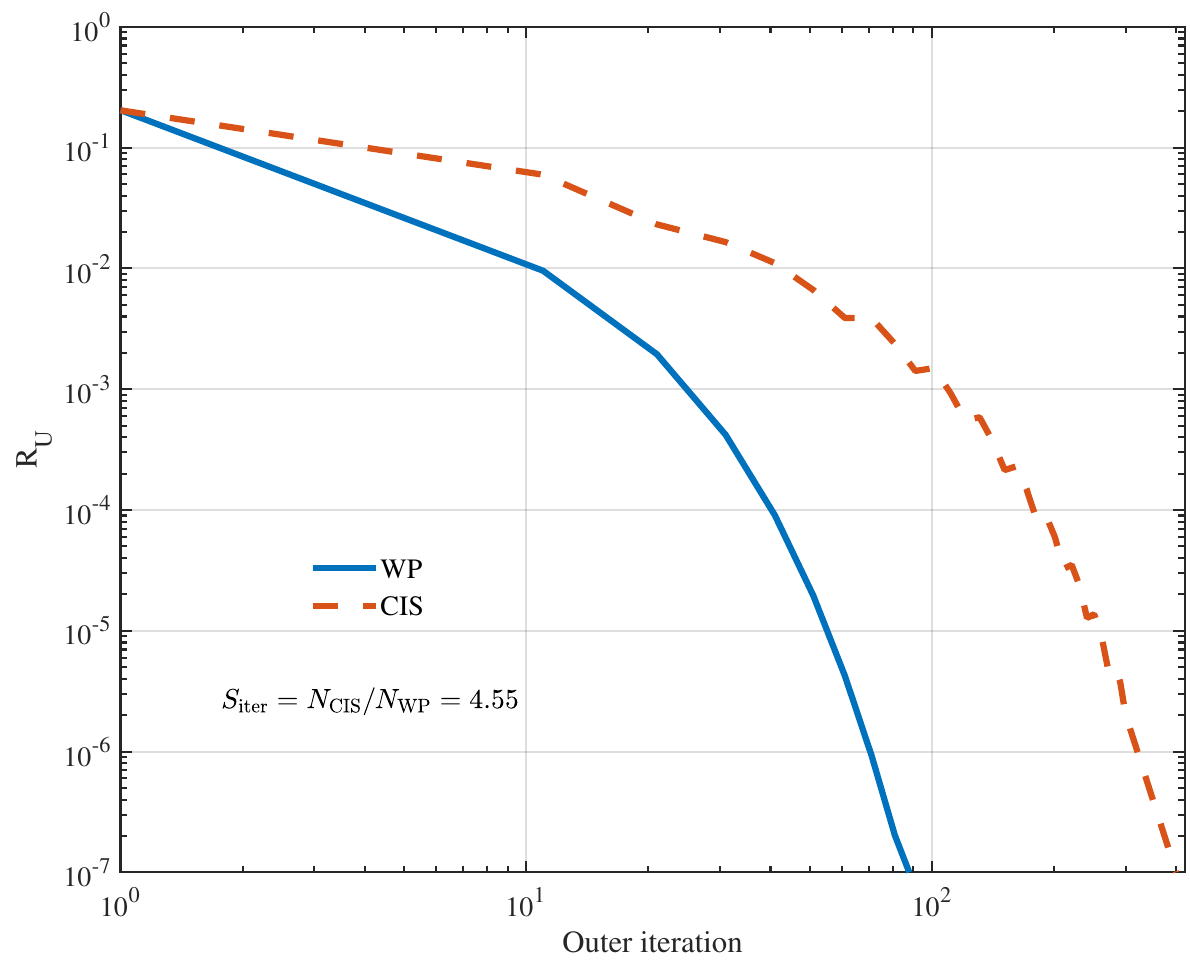}}\hspace{0.0200\figw}\raisebox{0.0000\figw}{\includegraphics[width=0.3375\figw]{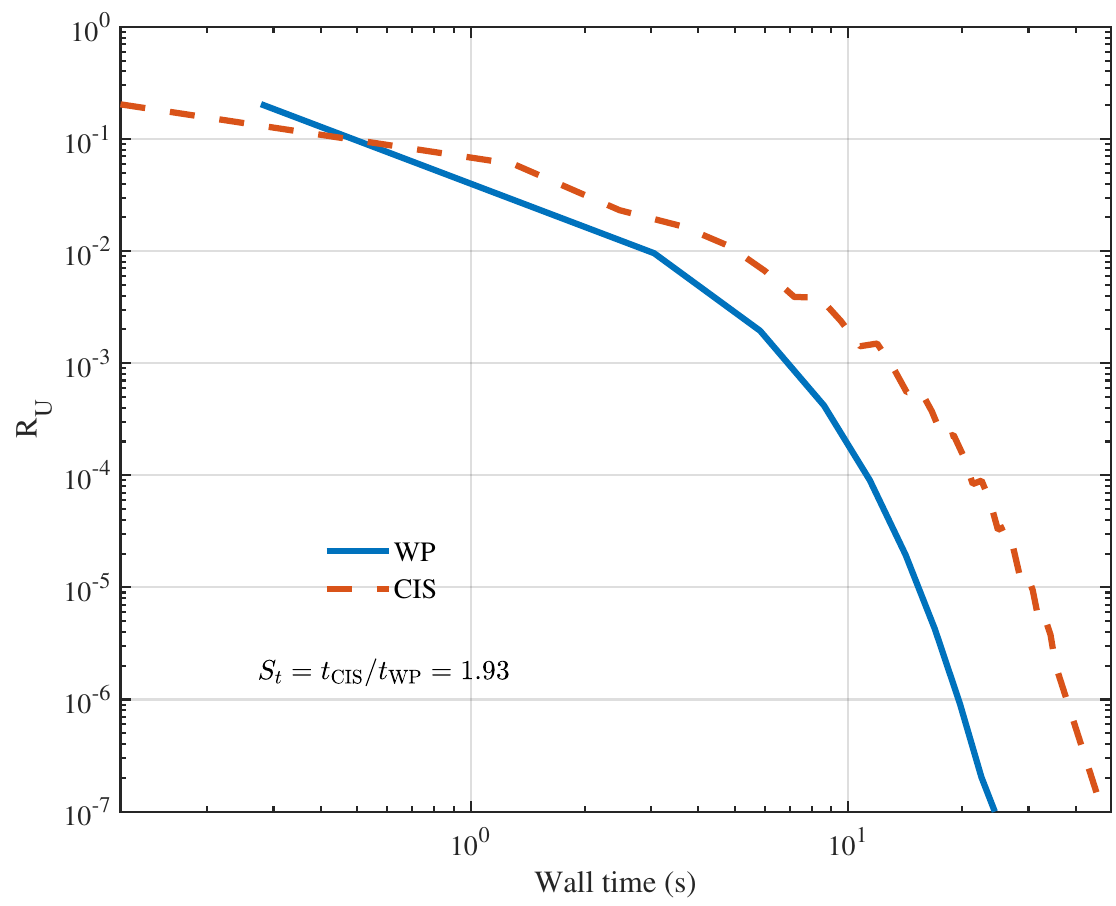}}}
\caption{$\varepsilon=0.075$.}
\end{subfigure}
\par\smallskip
\begin{subfigure}[t]{0.6915\figw}\centering
\makebox[\linewidth]{\raisebox{0.0000\figw}{\includegraphics[width=0.3418\figw]{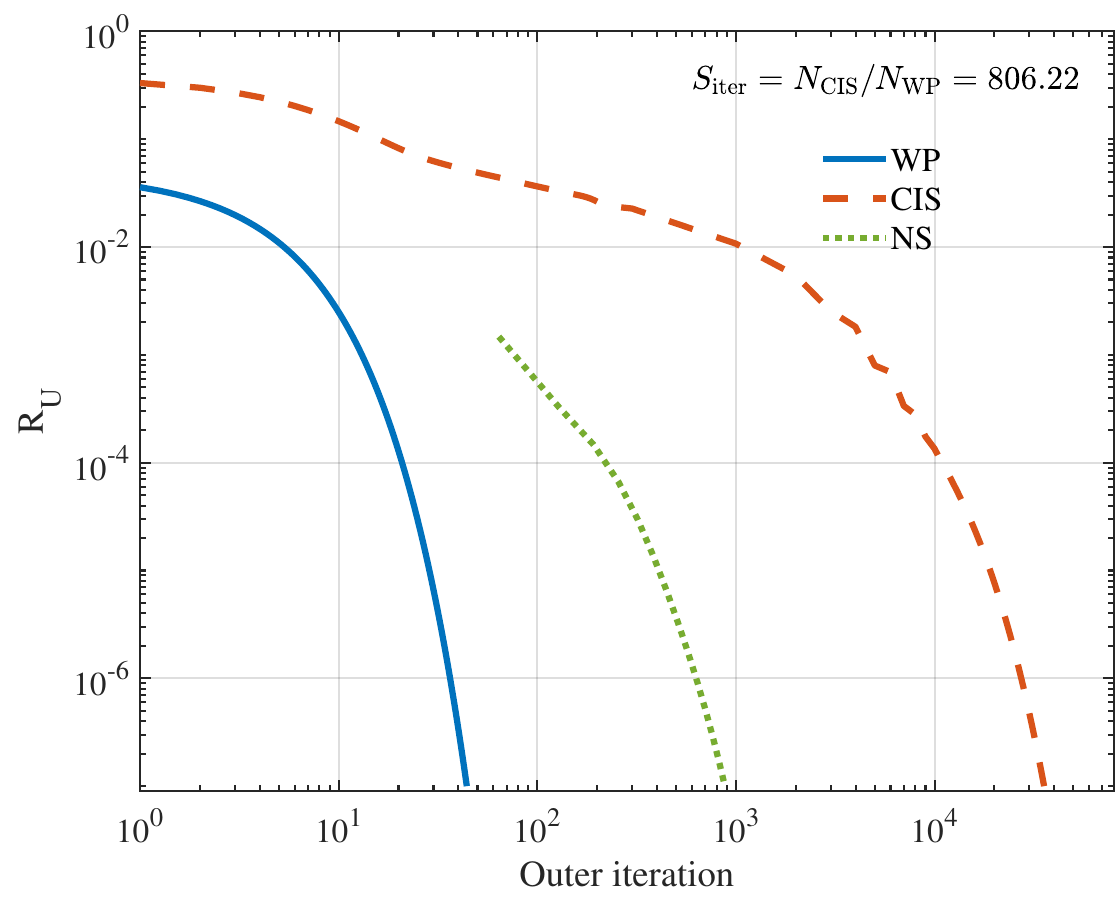}}\hspace{0.0200\figw}\raisebox{0.0023\figw}{\includegraphics[width=0.3297\figw]{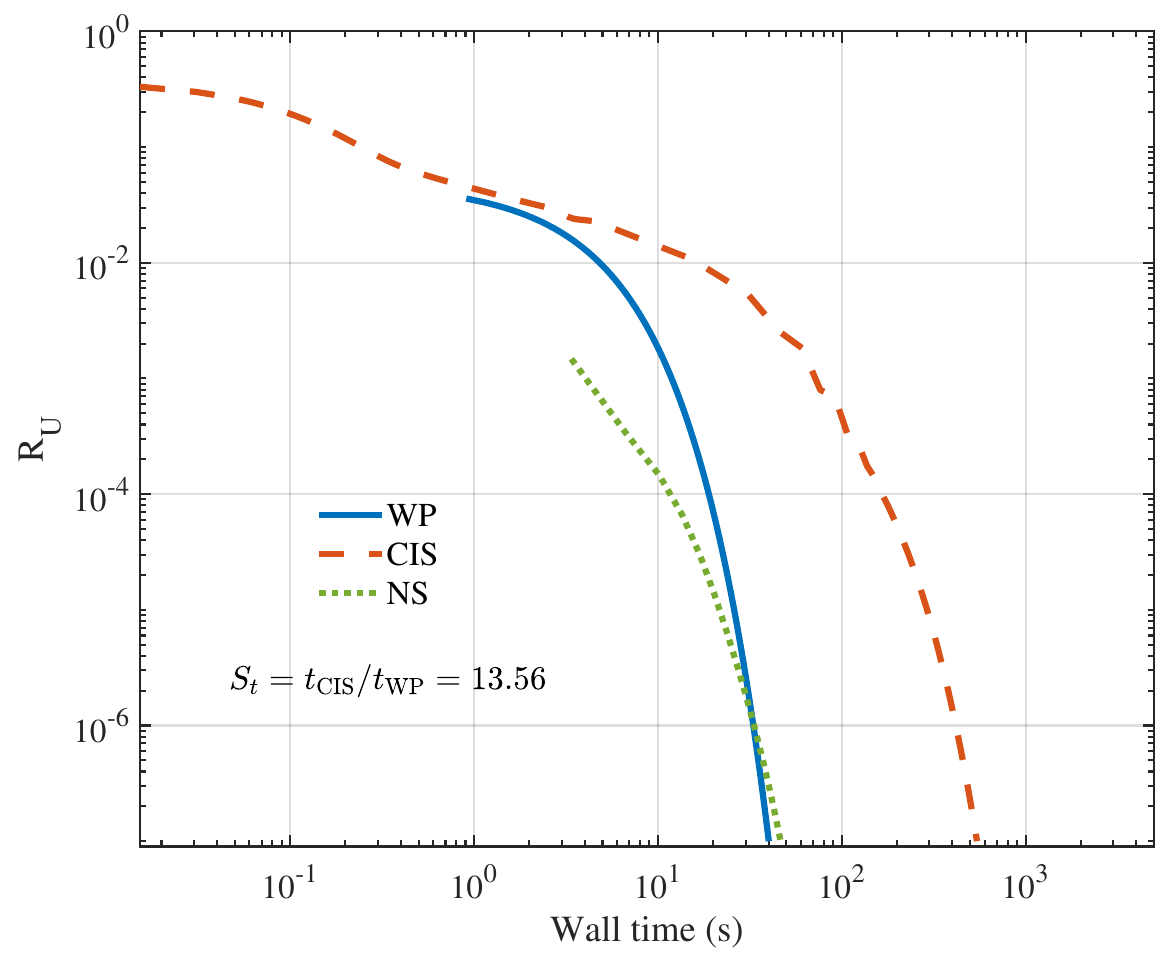}}}
\caption{$\mathrm{Re}=100$.}
\end{subfigure}
\par\smallskip
\begin{subfigure}[t]{0.6915\figw}\centering
\makebox[\linewidth]{\raisebox{0.0010\figw}{\includegraphics[width=0.3324\figw]{cavity-re1000_2vs_residual_iteration.pdf}}\hspace{0.0200\figw}\raisebox{0.0000\figw}{\includegraphics[width=0.3391\figw]{cavity-re1000_2vs_residual_walltime.pdf}}}
\caption{$\mathrm{Re}=1000$.}
\end{subfigure}
\caption{Lid-driven cavity: macroscopic residual versus outer iteration (left) and wall time (right) at $\varepsilon=1$~(a), $\varepsilon=0.075$~(b), $\mathrm{Re}=100$~(c) and $\mathrm{Re}=1000$~(d), for IWPD (solid), CIS (dashed) and, in the continuum cases, NS (dotted).}\label{fig:cavity_residual}
\end{figure}

\begin{figure}[H]
\centering\singlespacing\setlength{\figw}{1.0\linewidth}
\begin{subfigure}[t]{0.4920\figw}\centering
\raisebox{0.0000\figw}{\includegraphics[width=0.4920\figw]{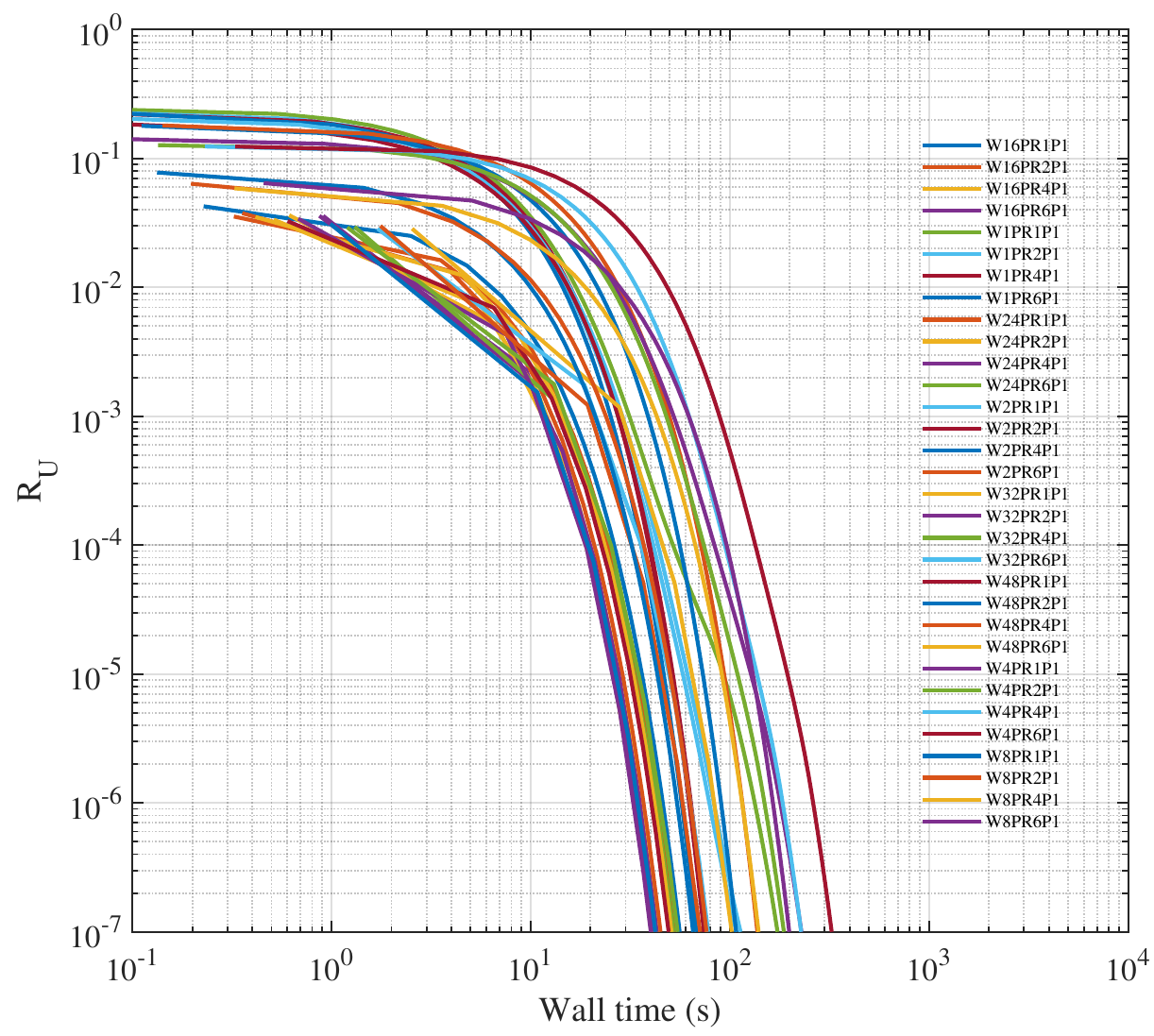}}
\caption{Residual histories, $\mathrm{Re}=100$.}
\end{subfigure}\hspace{0.0200\figw}%
\begin{subfigure}[t]{0.4830\figw}\centering
\raisebox{0.0000\figw}{\includegraphics[width=0.4830\figw]{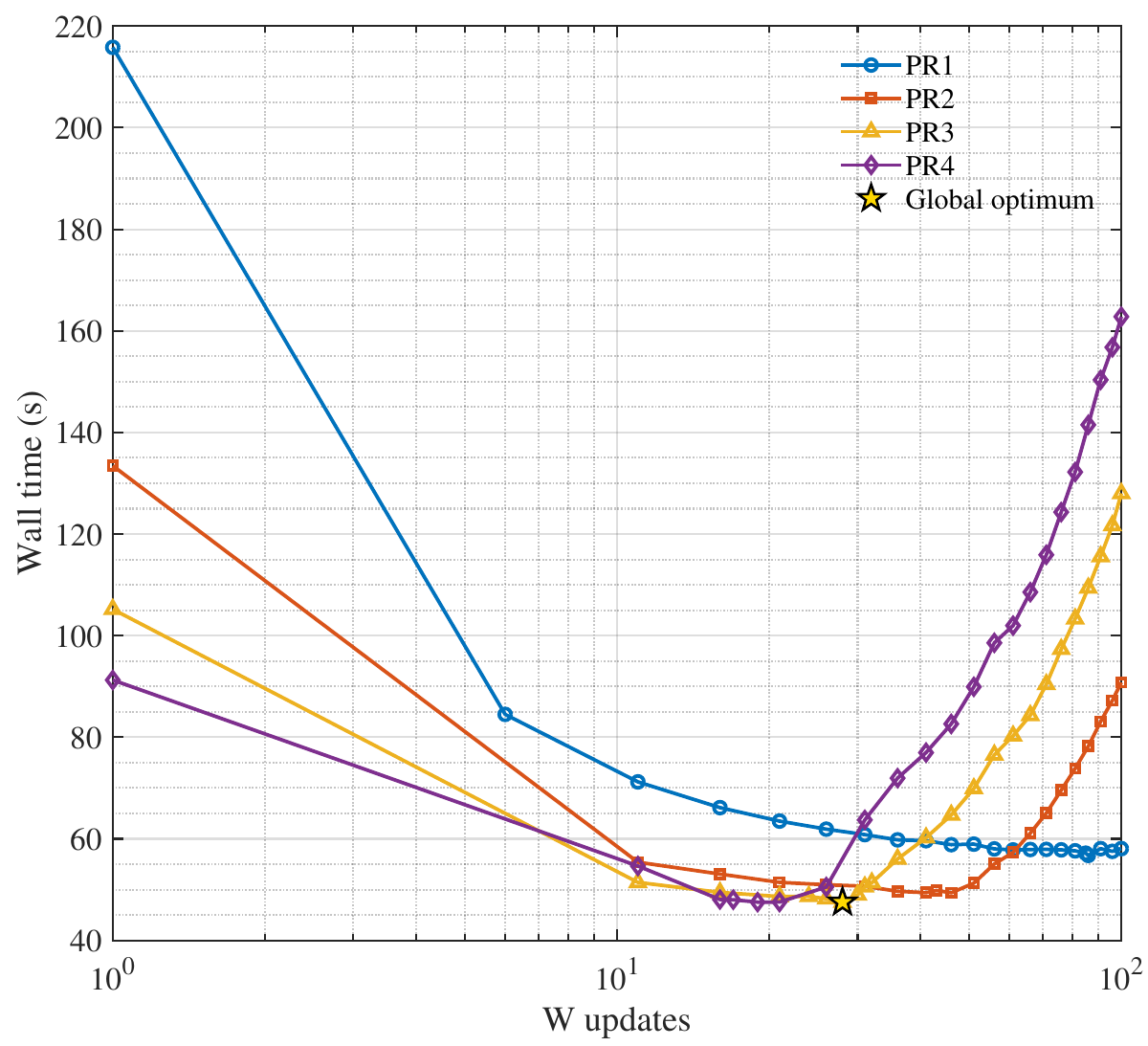}}
\caption{Wall time versus the number of W updates, $\mathrm{Re}=100$.}
\end{subfigure}
\par\smallskip
\begin{subfigure}[t]{0.4841\figw}\centering
\raisebox{0.0000\figw}{\includegraphics[width=0.4841\figw]{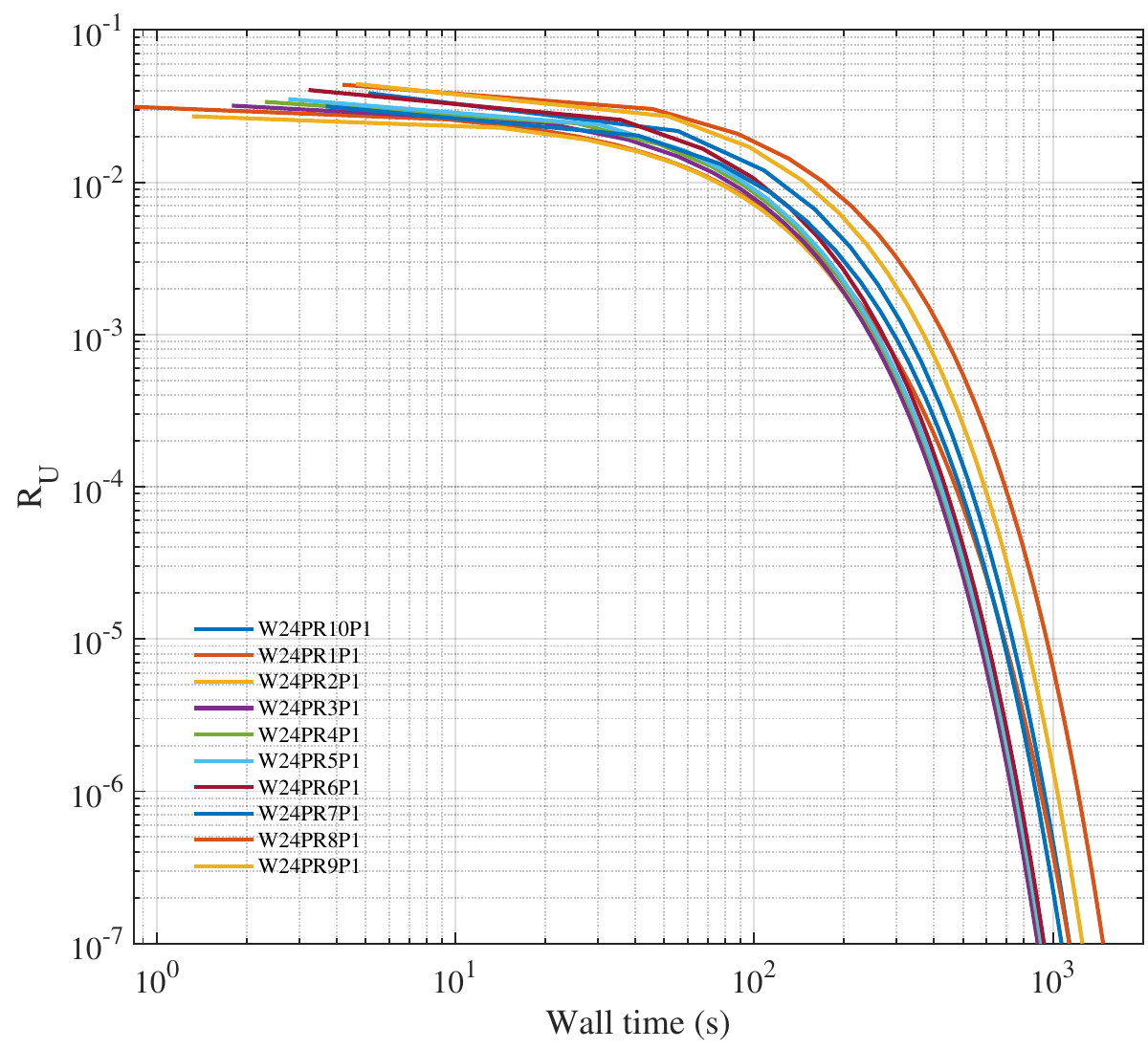}}
\caption{Residual histories, $\mathrm{Re}=1000$.}
\end{subfigure}\hspace{0.0200\figw}%
\begin{subfigure}[t]{0.4909\figw}\centering
\raisebox{0.0000\figw}{\includegraphics[width=0.4909\figw]{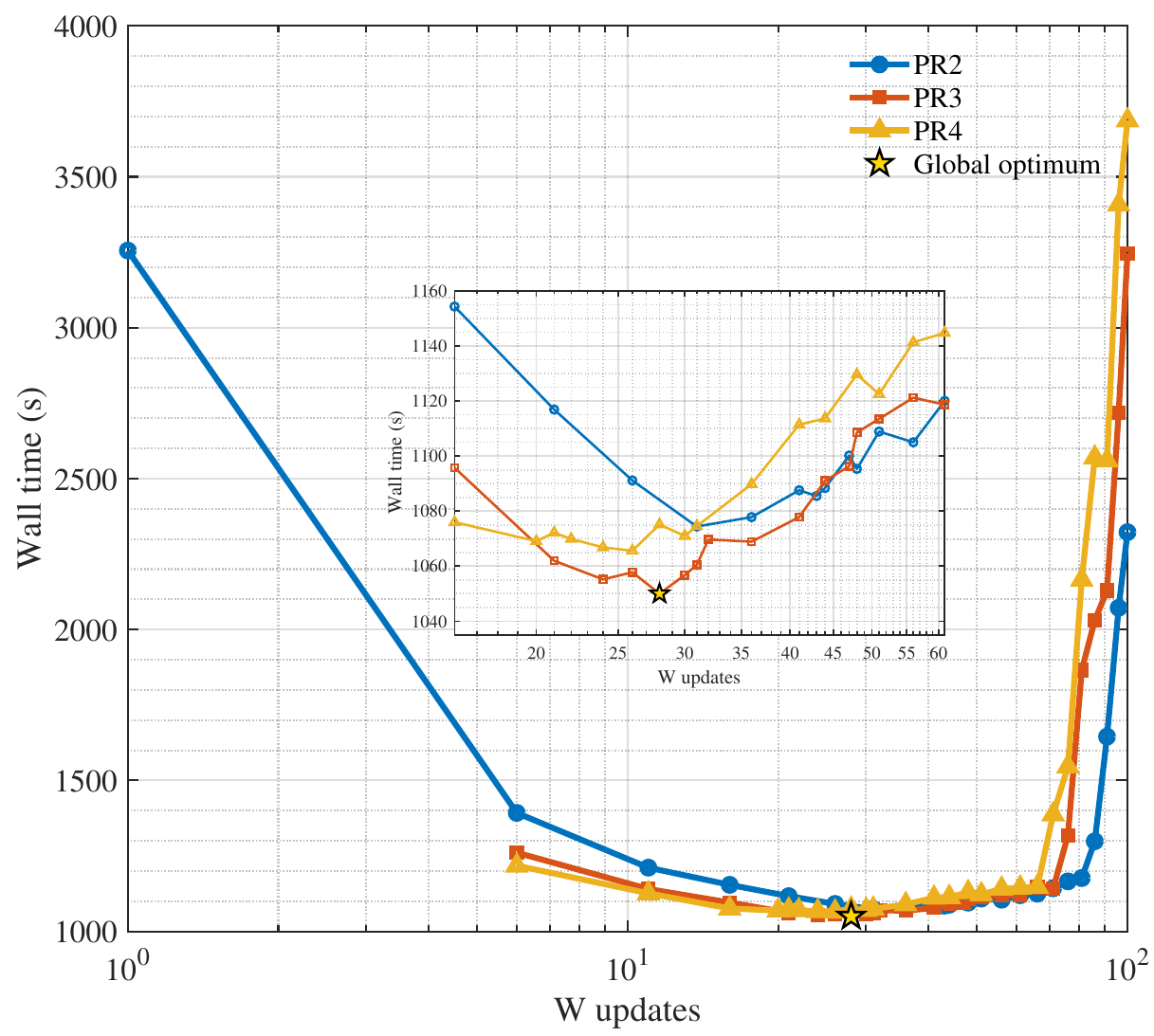}}
\caption{Wall time versus the number of W updates, $\mathrm{Re}=1000$.}
\end{subfigure}
\caption{Lid-driven cavity: scans of the IWPD controls with P1. The left panels show the residual histories of the scanned configurations versus wall time, and the right panels the wall time to $\|R_U\|\leq10^{-8}$ versus the number of W updates for different numbers of PR pairs, with the optimum marked by a star.}\label{fig:cavity_scan}
\end{figure}

\clearpage
\section{Circular cylinder}
\begin{figure}[H]
\centering\singlespacing\setlength{\figw}{.63\linewidth}
\begin{subfigure}[t]{0.4892\figw}\centering
\raisebox{0.0000\figw}{\includegraphics[width=0.4892\figw]{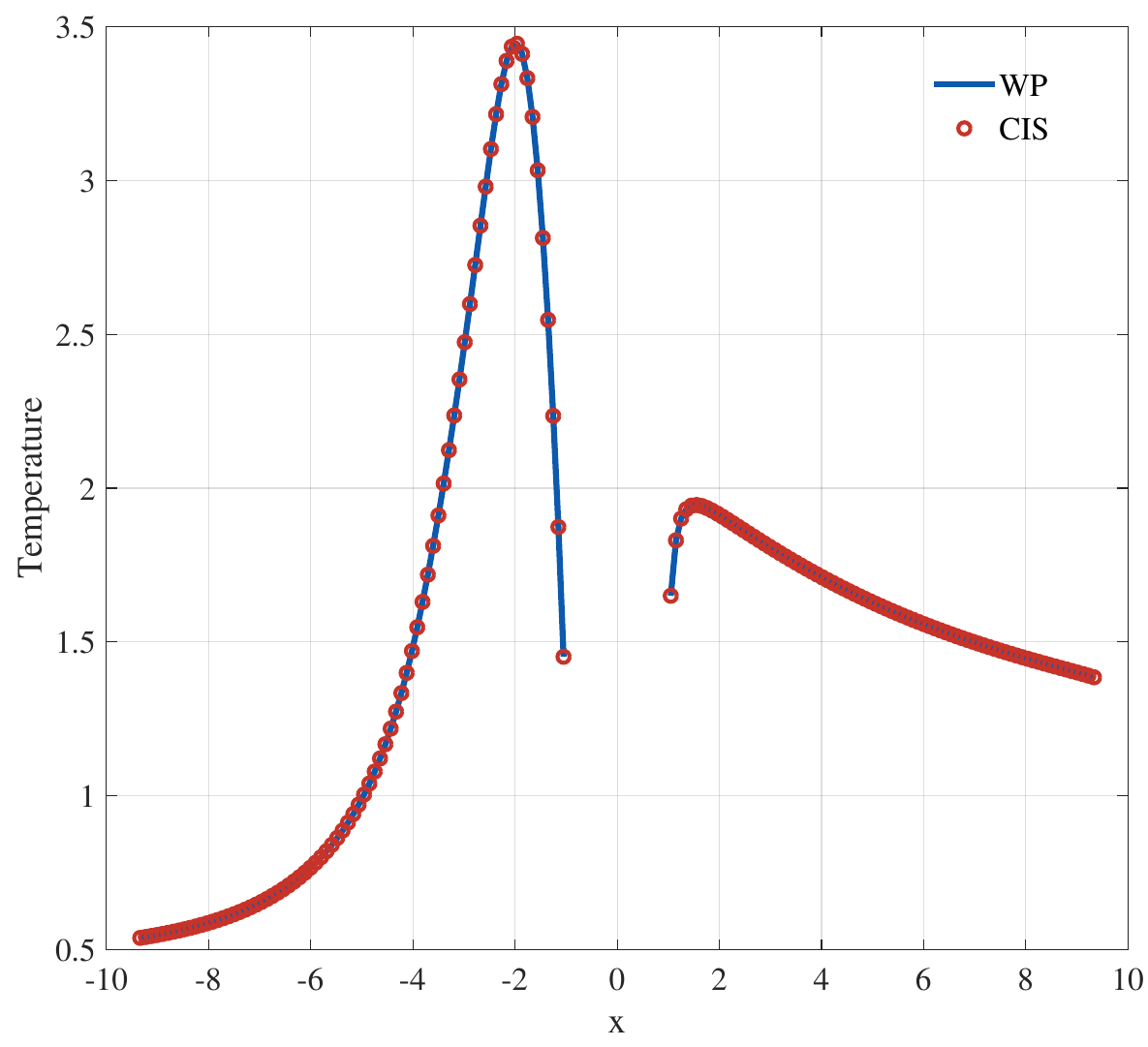}}
\caption{Stagnation-line temperature.}
\end{subfigure}\hspace{0.0200\figw}%
\begin{subfigure}[t]{0.4858\figw}\centering
\raisebox{0.0000\figw}{\includegraphics[width=0.4858\figw]{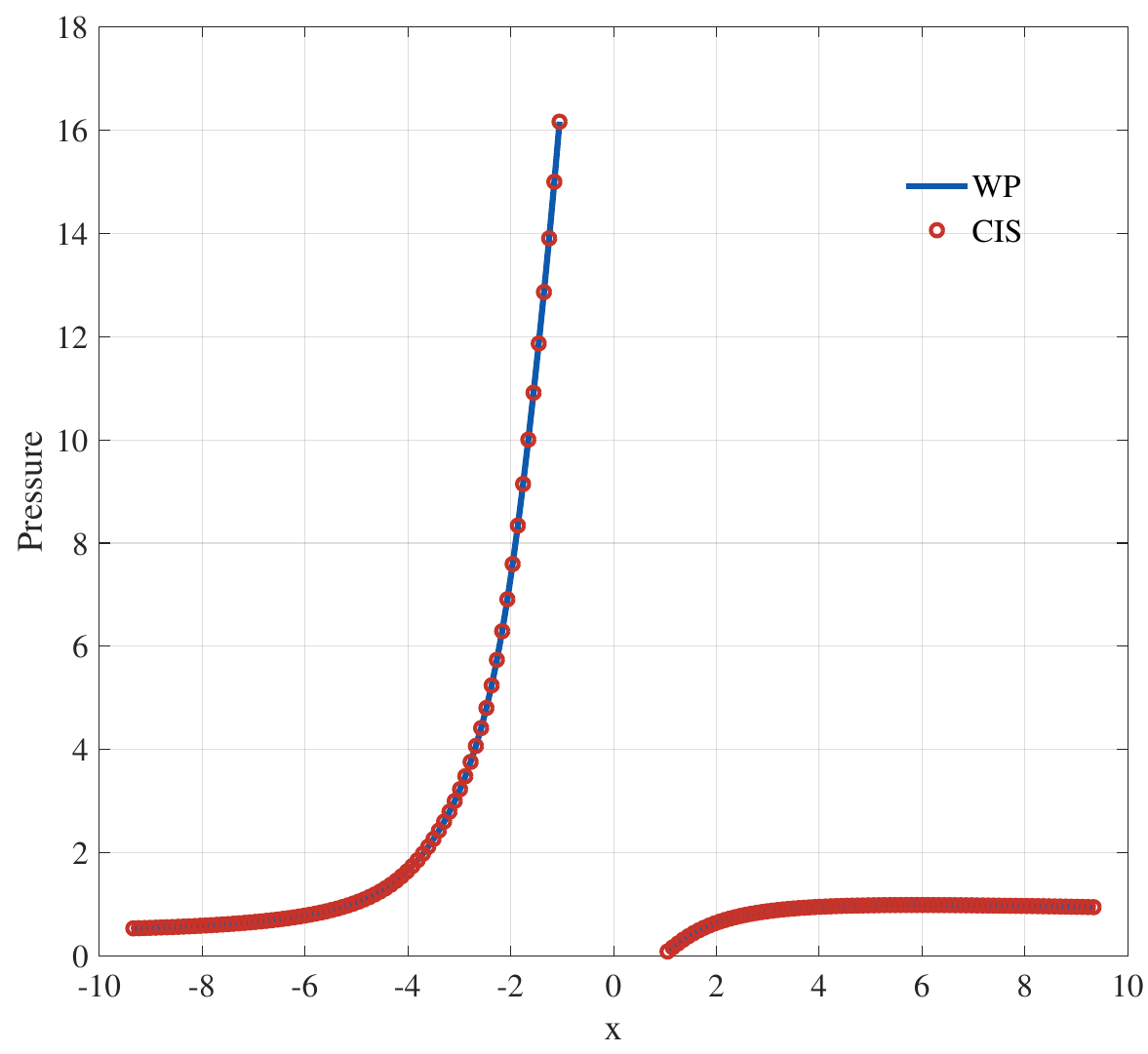}}
\caption{Stagnation-line pressure.}
\end{subfigure}
\par\smallskip
\begin{subfigure}[t]{0.4927\figw}\centering
\raisebox{0.0000\figw}{\includegraphics[width=0.4927\figw]{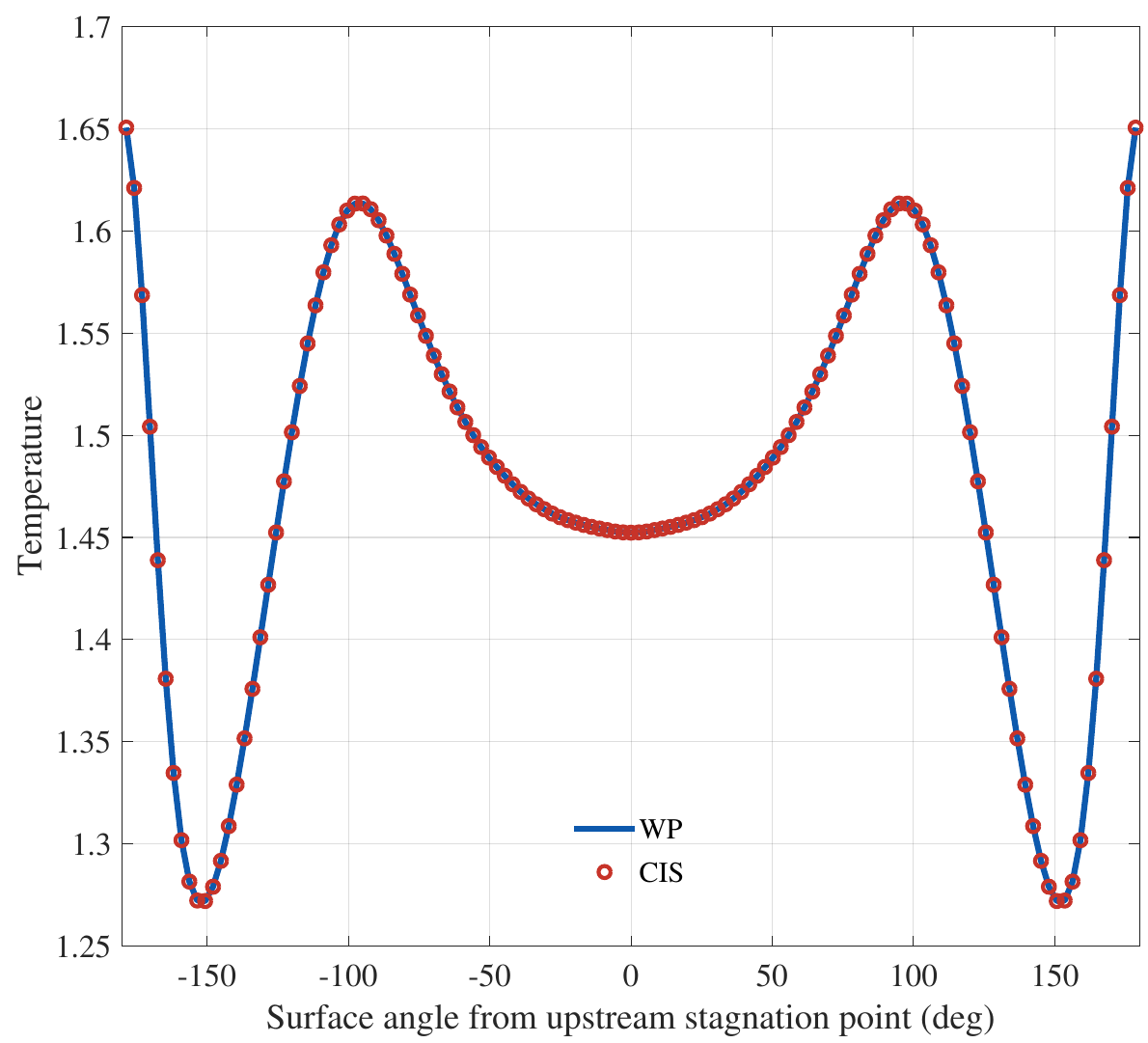}}
\caption{Wall temperature.}
\end{subfigure}\hspace{0.0200\figw}%
\begin{subfigure}[t]{0.4823\figw}\centering
\raisebox{0.0000\figw}{\includegraphics[width=0.4823\figw]{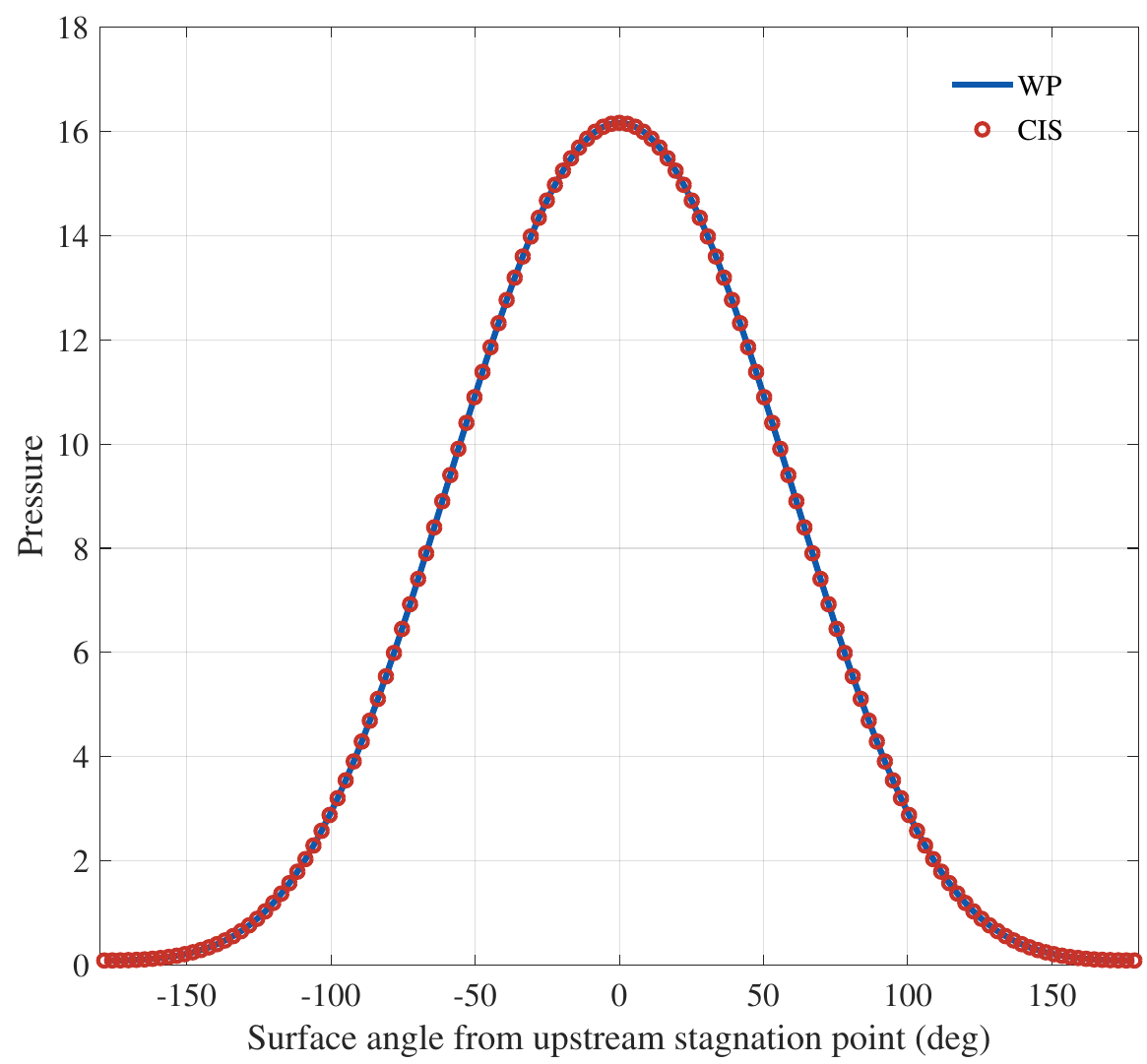}}
\caption{Wall pressure.}
\end{subfigure}
\caption{Circular cylinder at $\mathrm{Ma}=5$ and $\varepsilon=1$: temperature and pressure along the stagnation line (a,b) and on the wall versus the angle from the upstream stagnation point (c,d). Lines denote IWPD and symbols CIS. Here and in the following figures the wall values are those of the cells adjacent to the wall.}\label{fig:cyl_kn1_profiles}
\end{figure}

\begin{figure}[H]
\centering\singlespacing\setlength{\figw}{.63\linewidth}
\begin{subfigure}[t]{0.4875\figw}\centering
\raisebox{0.0000\figw}{\includegraphics[width=0.4875\figw]{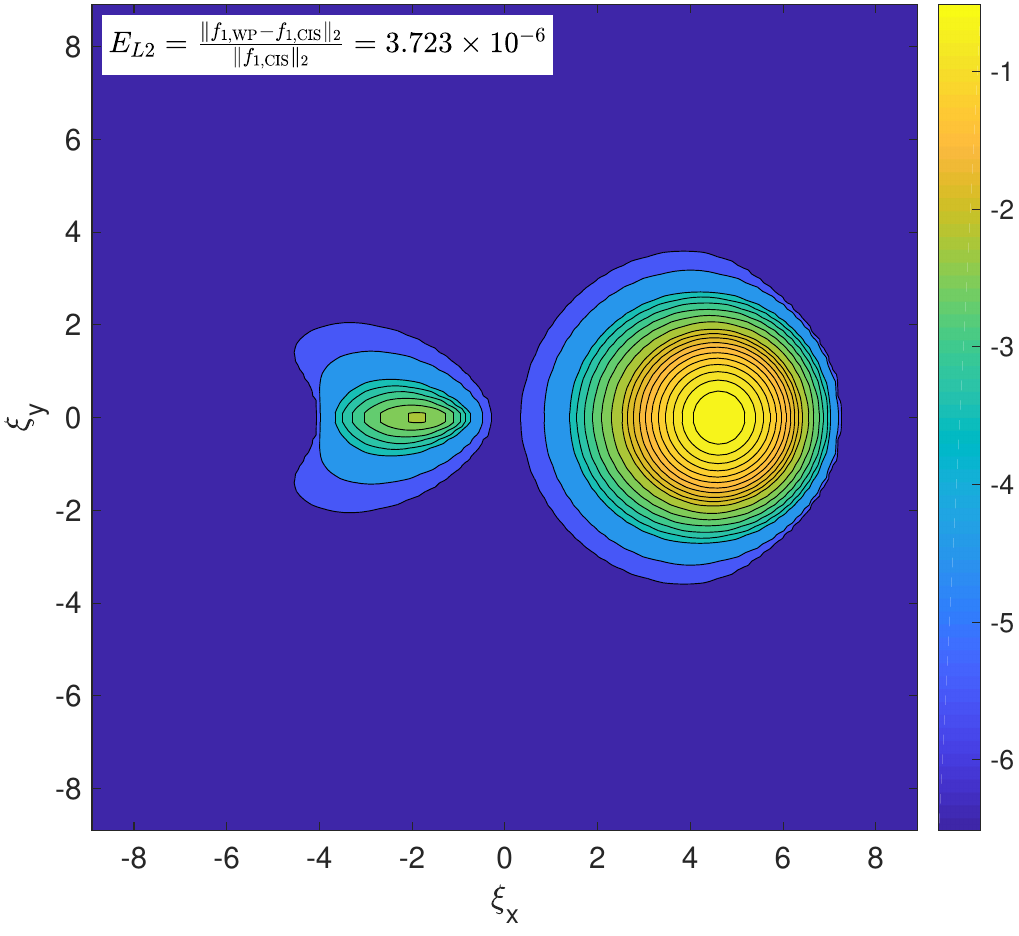}}
\caption{$x=-7.29$.}
\end{subfigure}\hspace{0.0200\figw}%
\begin{subfigure}[t]{0.4875\figw}\centering
\raisebox{0.0000\figw}{\includegraphics[width=0.4875\figw]{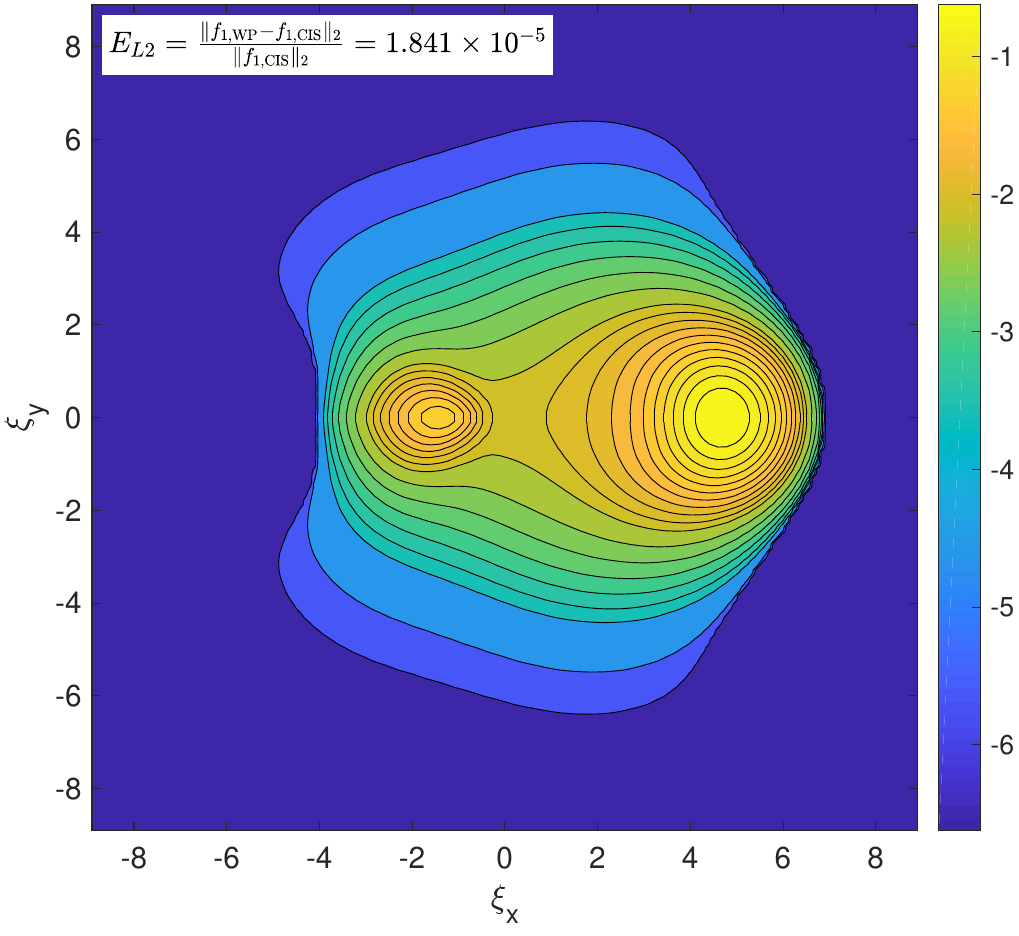}}
\caption{$x=-3.09$.}
\end{subfigure}
\par\smallskip
\begin{subfigure}[t]{0.4875\figw}\centering
\raisebox{0.0000\figw}{\includegraphics[width=0.4875\figw]{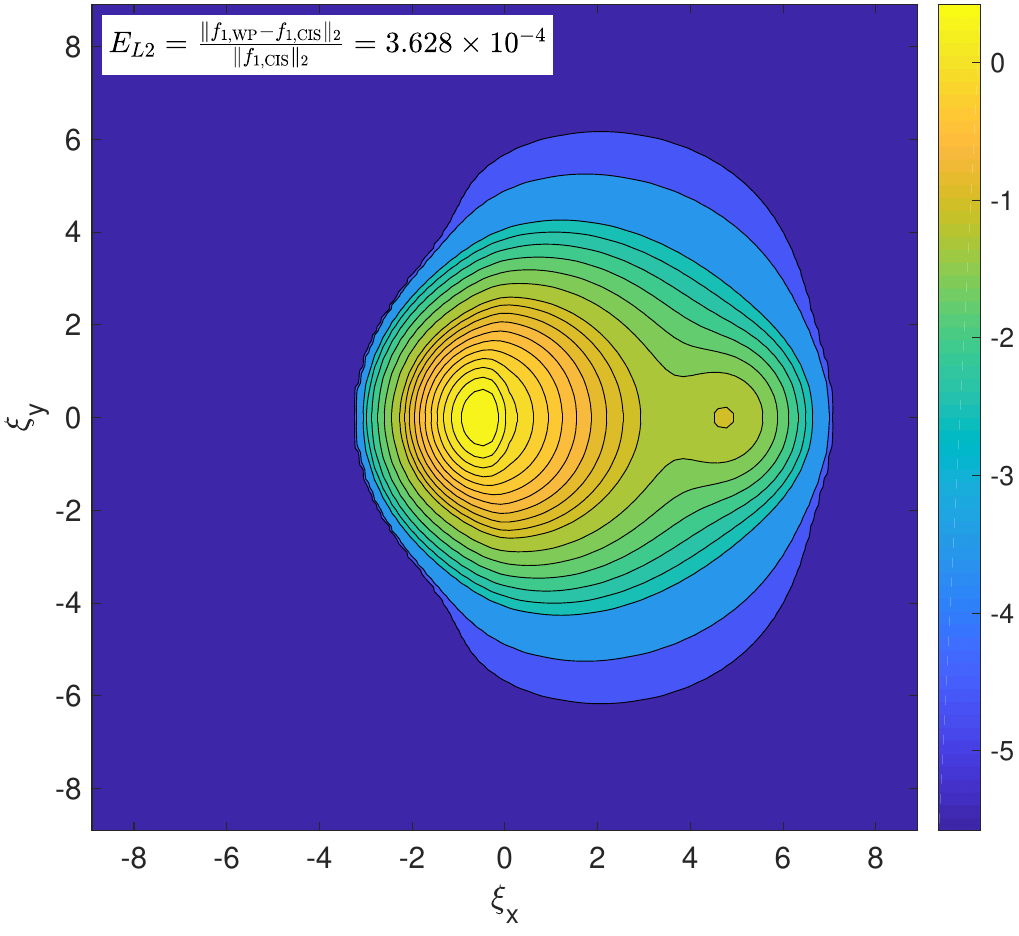}}
\caption{$x=-1.05$.}
\end{subfigure}\hspace{0.0200\figw}%
\begin{subfigure}[t]{0.4875\figw}\centering
\raisebox{0.0000\figw}{\includegraphics[width=0.4875\figw]{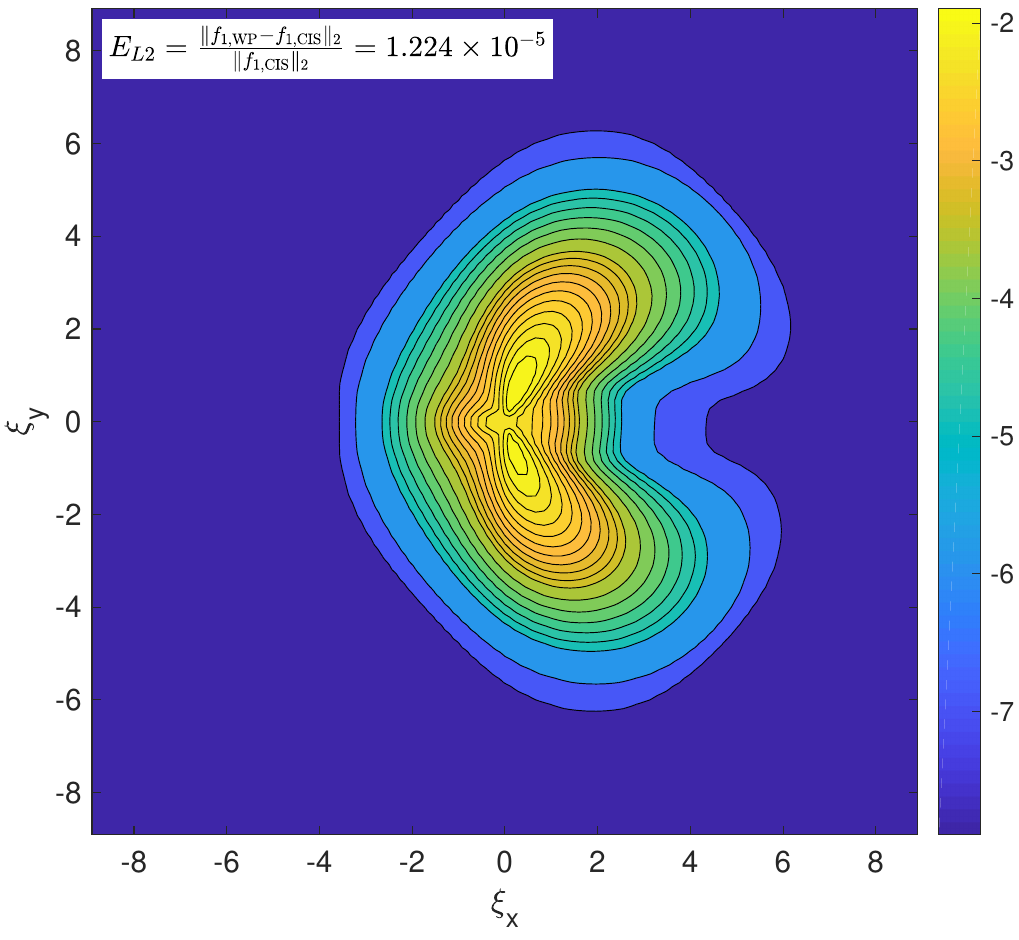}}
\caption{$x=1.05$.}
\end{subfigure}
\caption{Circular cylinder at $\mathrm{Ma}=5$ and $\varepsilon=1$: reduced distributions $f_1(\xi_x,\xi_y)$ on a logarithmic scale at three points ahead of the body (a--c) and one point behind it (d). IWPD is shown by filled contours and CIS by lines.}\label{fig:cyl_kn1_dist_log}
\end{figure}

\begin{figure}[H]
\centering\singlespacing\setlength{\figw}{.65\linewidth}
\begin{subfigure}[t]{0.4856\figw}\centering
\raisebox{0.0000\figw}{\includegraphics[width=0.4856\figw]{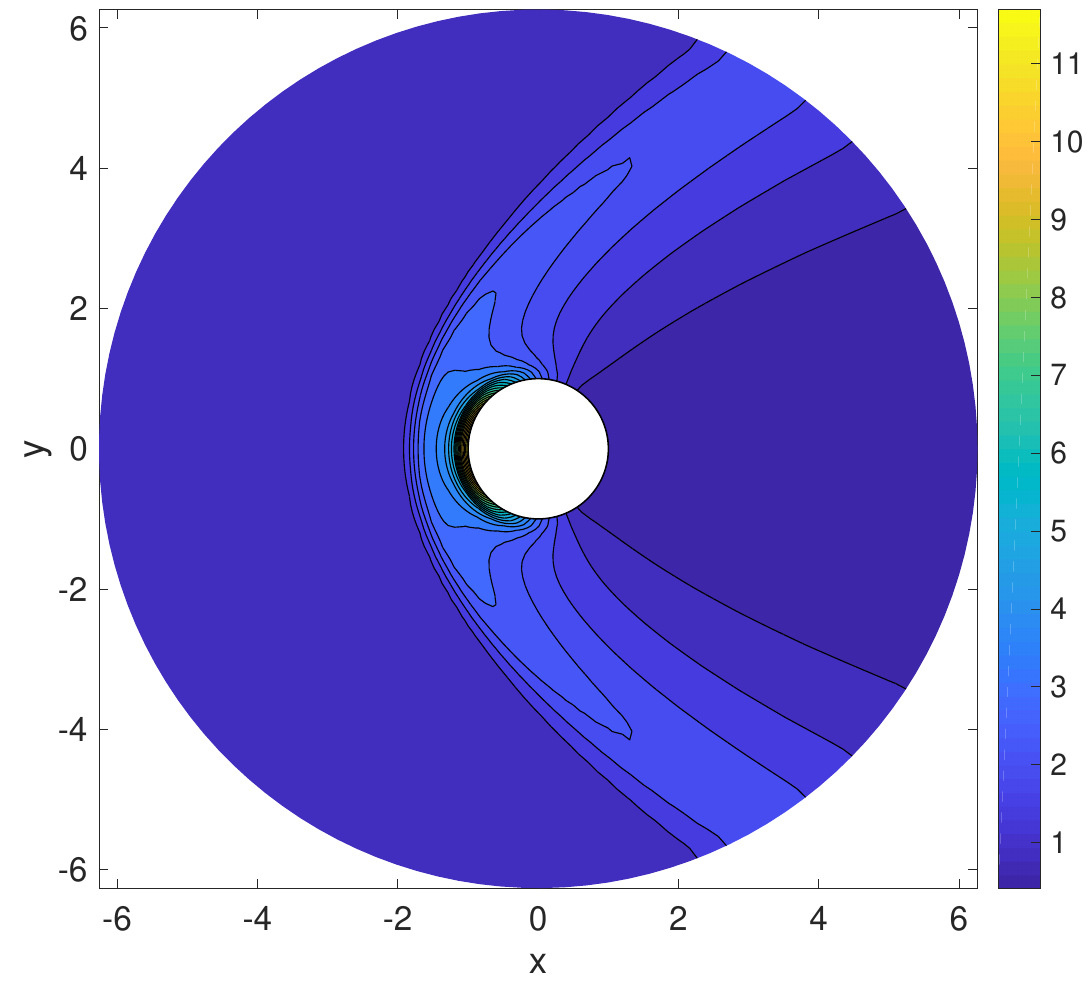}}
\caption{Density.}
\end{subfigure}\hspace{0.0200\figw}%
\begin{subfigure}[t]{0.4894\figw}\centering
\raisebox{0.0000\figw}{\includegraphics[width=0.4894\figw]{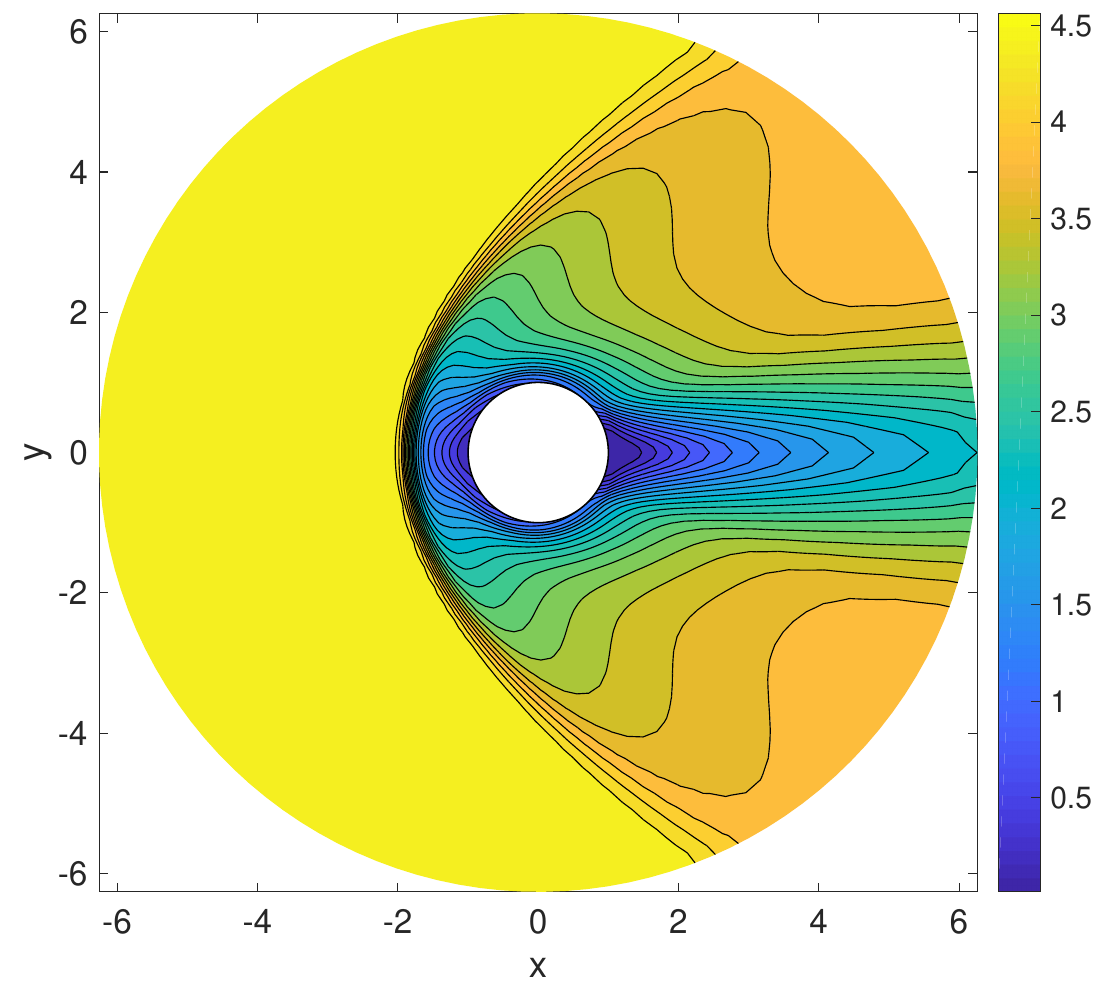}}
\caption{Speed.}
\end{subfigure}
\par\smallskip
\begin{subfigure}[t]{0.4894\figw}\centering
\raisebox{0.0000\figw}{\includegraphics[width=0.4894\figw]{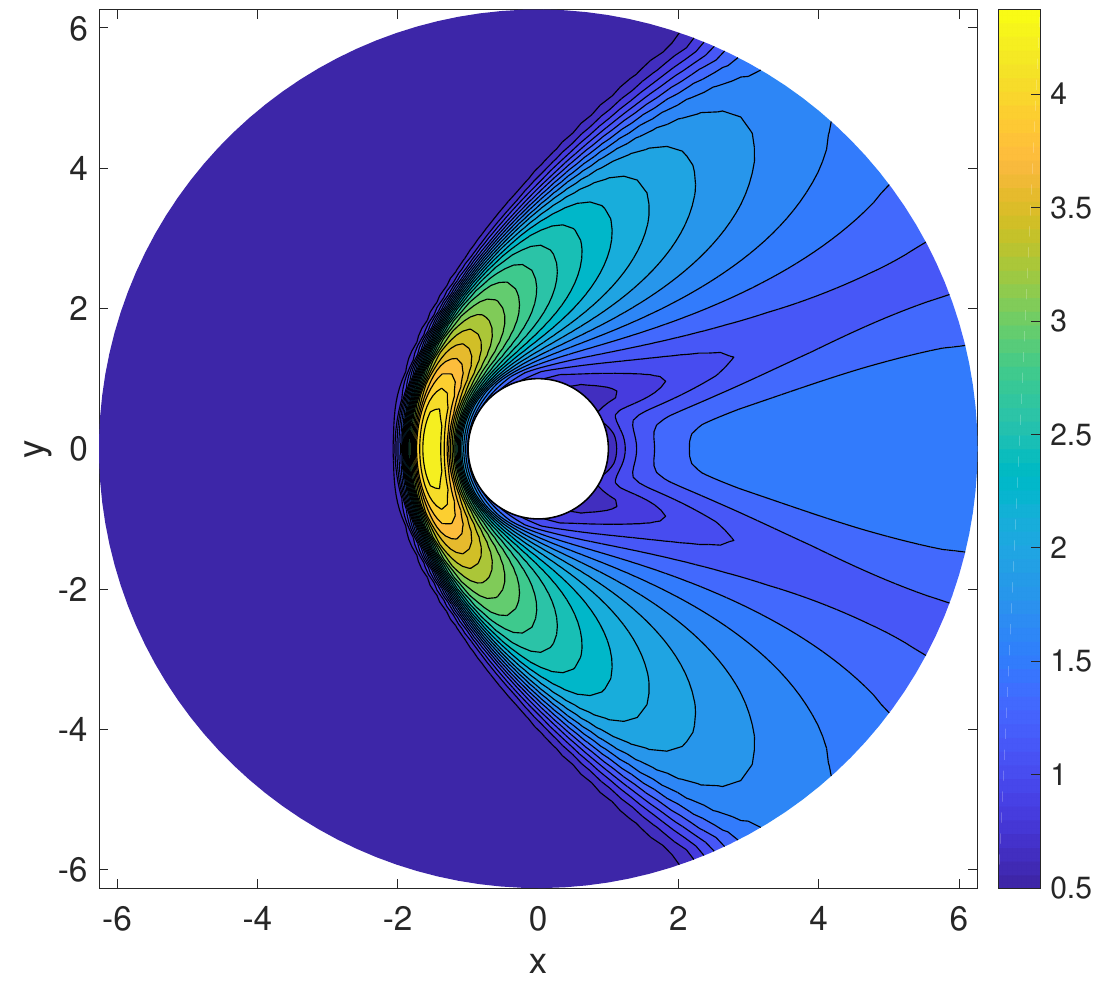}}
\caption{Temperature.}
\end{subfigure}\hspace{0.0200\figw}%
\begin{subfigure}[t]{0.4856\figw}\centering
\raisebox{0.0000\figw}{\includegraphics[width=0.4856\figw]{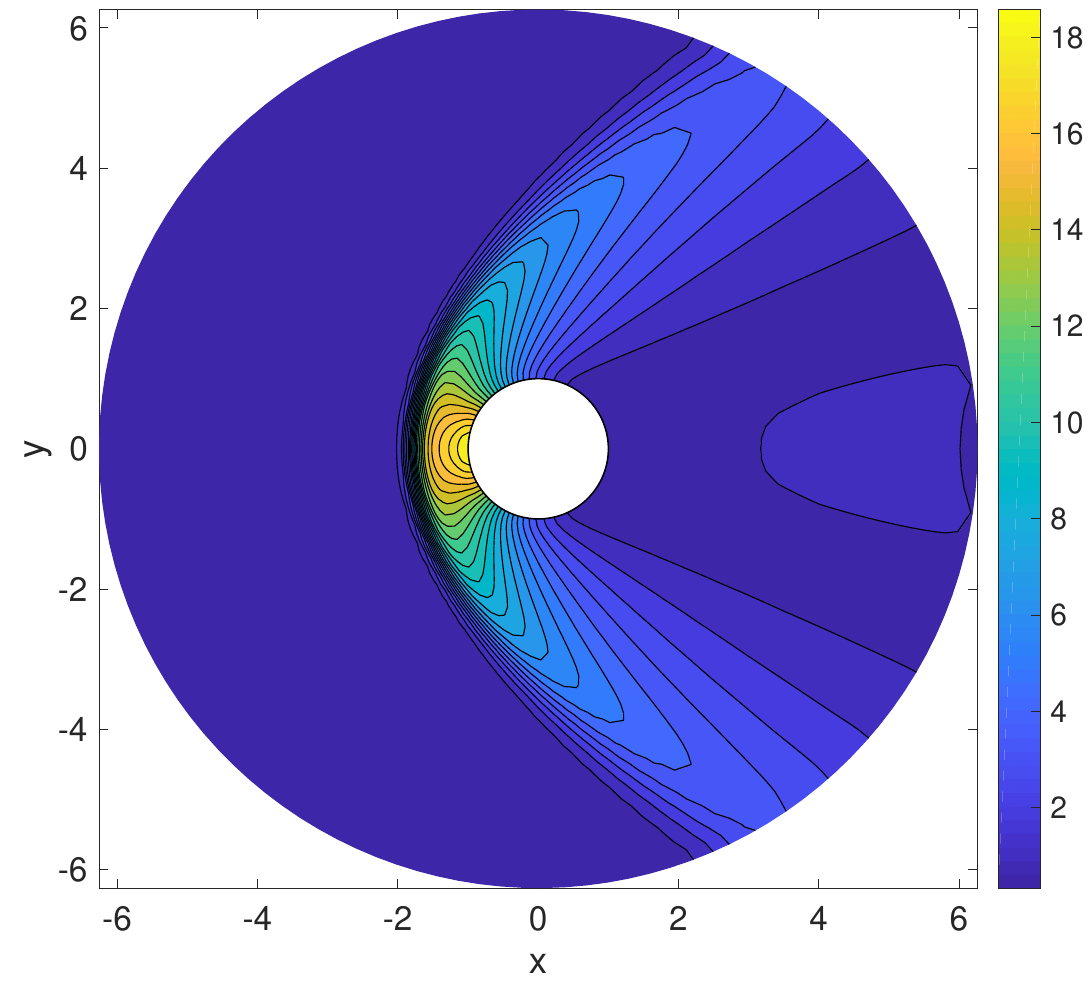}}
\caption{Pressure.}
\end{subfigure}
\caption{Circular cylinder at $\mathrm{Ma}=5$ and $\varepsilon=10^{-2}$: macroscopic fields of IWPD (filled contours) and CIS (lines) on common levels.}\label{fig:cyl_kn1em2_fields}
\end{figure}

\begin{figure}[H]
\centering\singlespacing\setlength{\figw}{.65\linewidth}
\begin{subfigure}[t]{0.4892\figw}\centering
\raisebox{0.0000\figw}{\includegraphics[width=0.4892\figw]{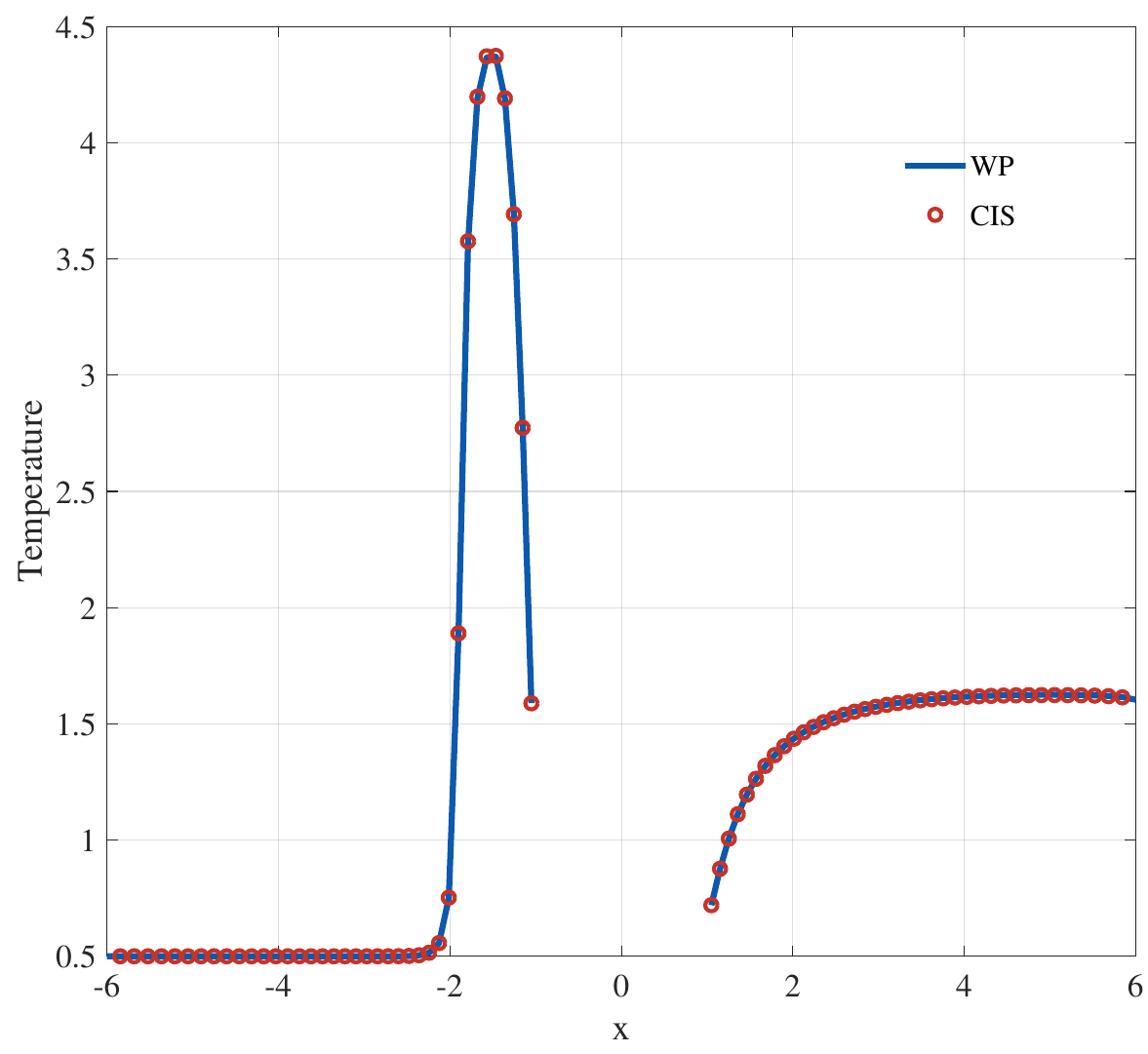}}
\caption{Stagnation-line temperature.}
\end{subfigure}\hspace{0.0200\figw}%
\begin{subfigure}[t]{0.4858\figw}\centering
\raisebox{0.0000\figw}{\includegraphics[width=0.4858\figw]{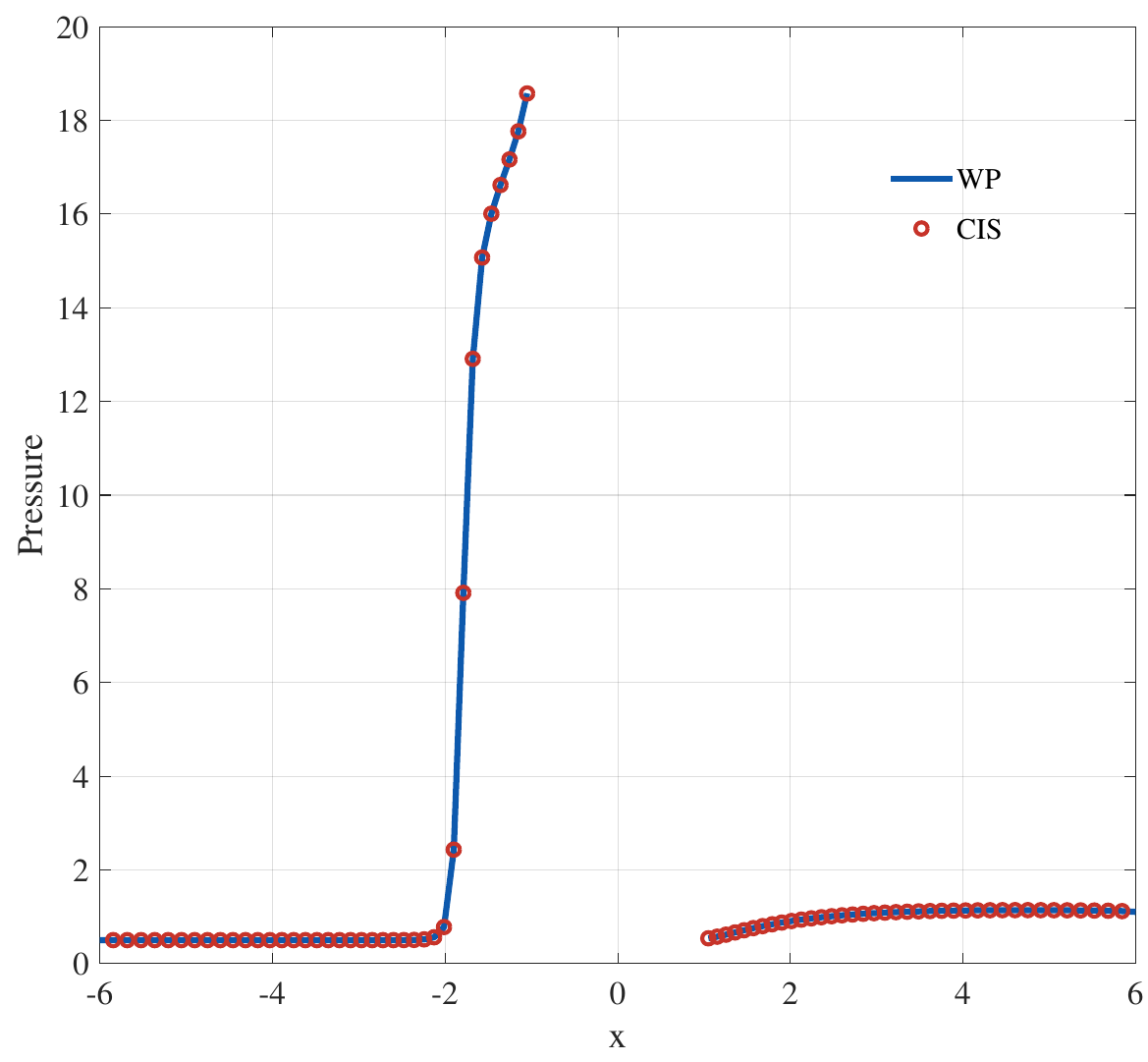}}
\caption{Stagnation-line pressure.}
\end{subfigure}
\par\smallskip
\begin{subfigure}[t]{0.4892\figw}\centering
\raisebox{0.0000\figw}{\includegraphics[width=0.4892\figw]{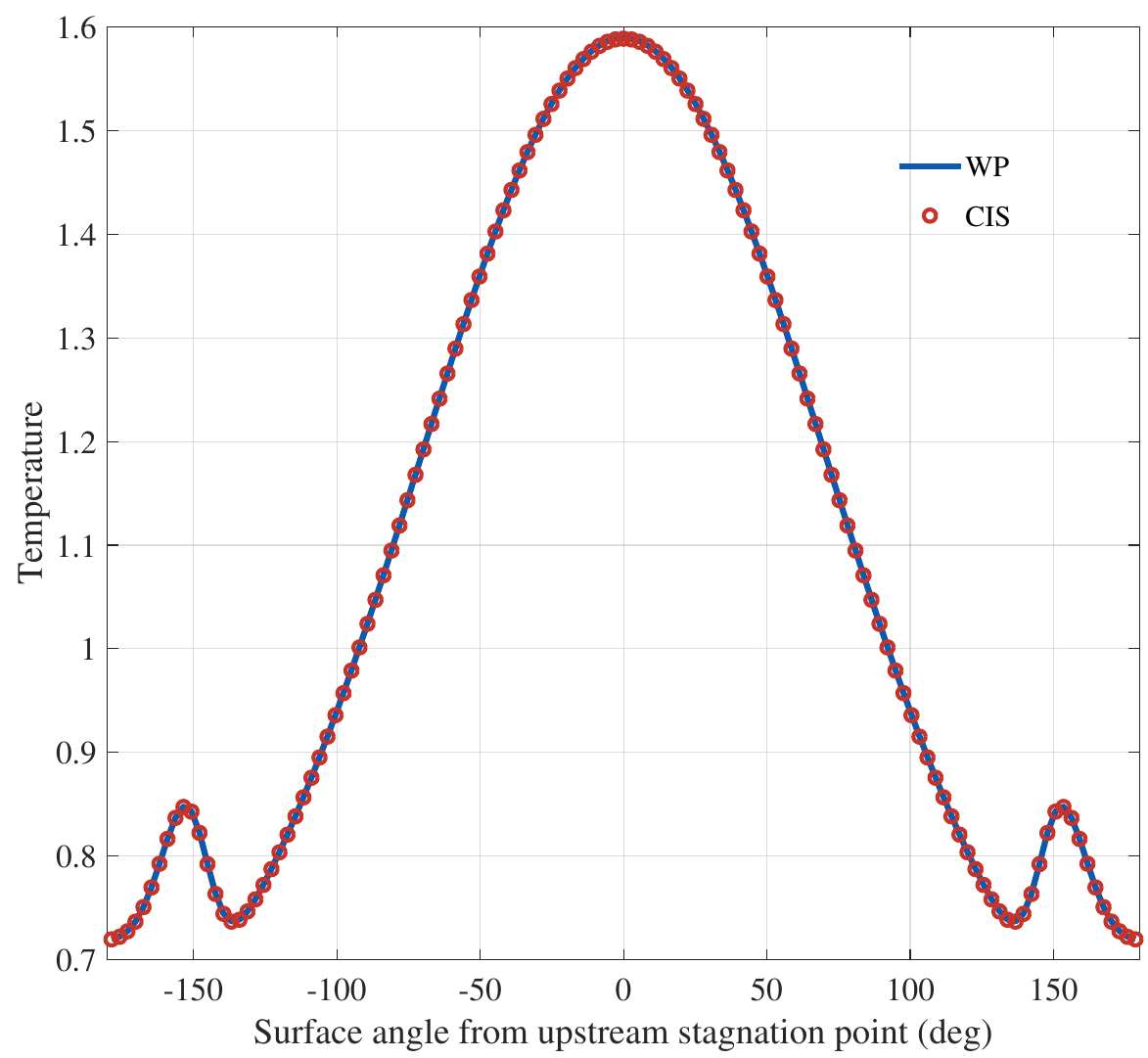}}
\caption{Wall temperature.}
\end{subfigure}\hspace{0.0200\figw}%
\begin{subfigure}[t]{0.4858\figw}\centering
\raisebox{0.0000\figw}{\includegraphics[width=0.4858\figw]{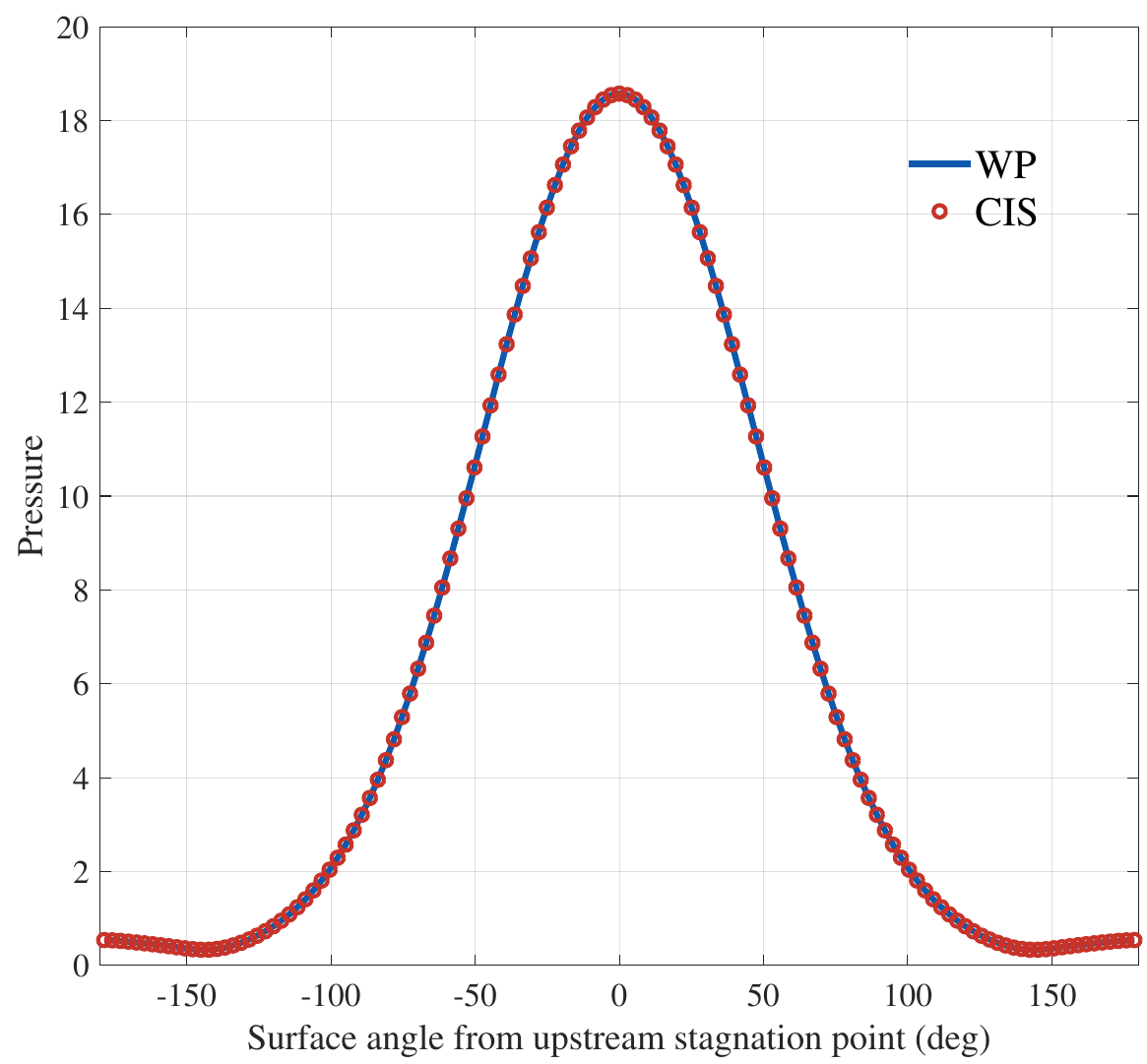}}
\caption{Wall pressure.}
\end{subfigure}
\caption{Circular cylinder at $\mathrm{Ma}=5$ and $\varepsilon=10^{-2}$: temperature and pressure along the stagnation line (a,b) and on the wall (c,d). Lines denote IWPD and symbols CIS.}\label{fig:cyl_kn1em2_profiles}
\end{figure}

\begin{figure}[H]
\centering\singlespacing\setlength{\figw}{.65\linewidth}
\begin{subfigure}[t]{0.4856\figw}\centering
\raisebox{0.0000\figw}{\includegraphics[width=0.4856\figw]{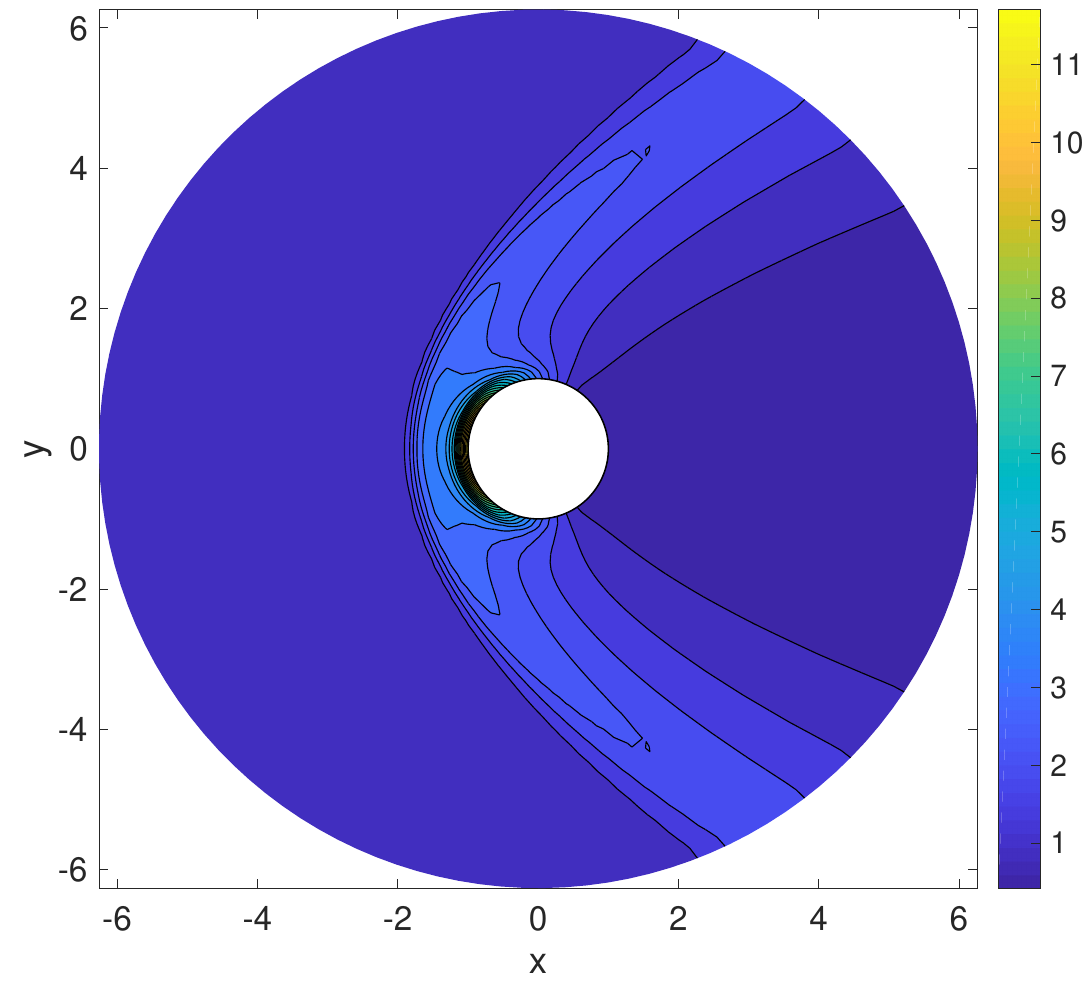}}
\caption{Density.}
\end{subfigure}\hspace{0.0200\figw}%
\begin{subfigure}[t]{0.4894\figw}\centering
\raisebox{0.0000\figw}{\includegraphics[width=0.4894\figw]{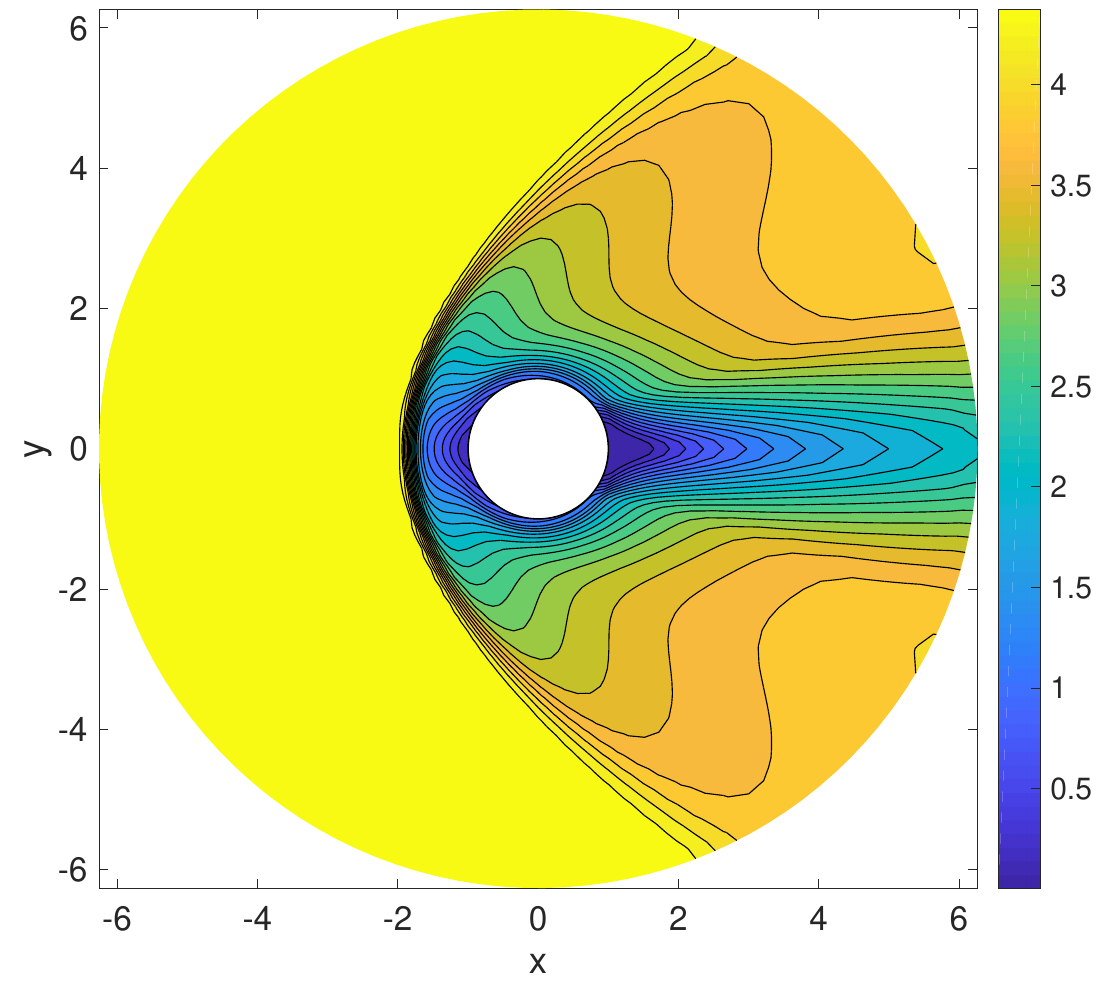}}
\caption{Speed.}
\end{subfigure}
\par\smallskip
\begin{subfigure}[t]{0.4894\figw}\centering
\raisebox{0.0000\figw}{\includegraphics[width=0.4894\figw]{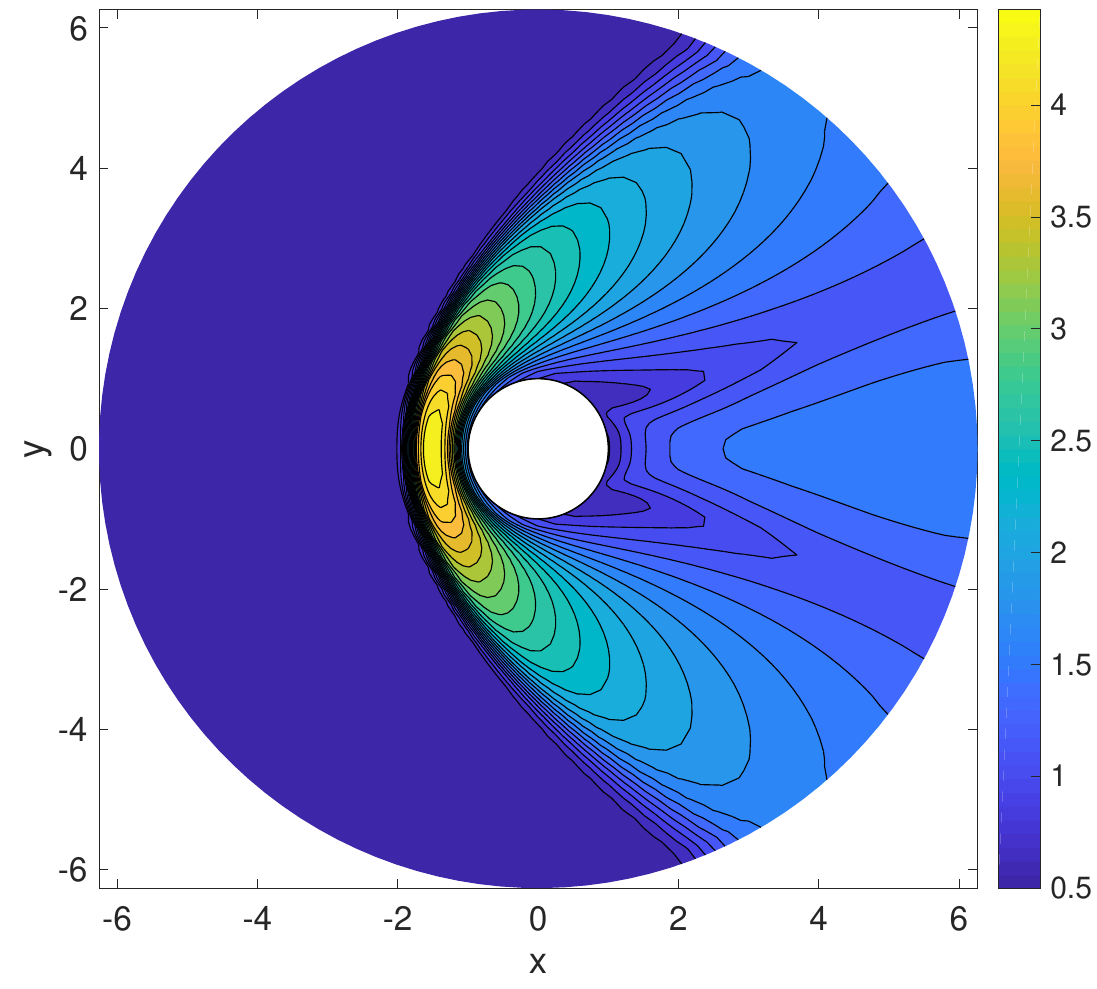}}
\caption{Temperature.}
\end{subfigure}\hspace{0.0200\figw}%
\begin{subfigure}[t]{0.4856\figw}\centering
\raisebox{0.0000\figw}{\includegraphics[width=0.4856\figw]{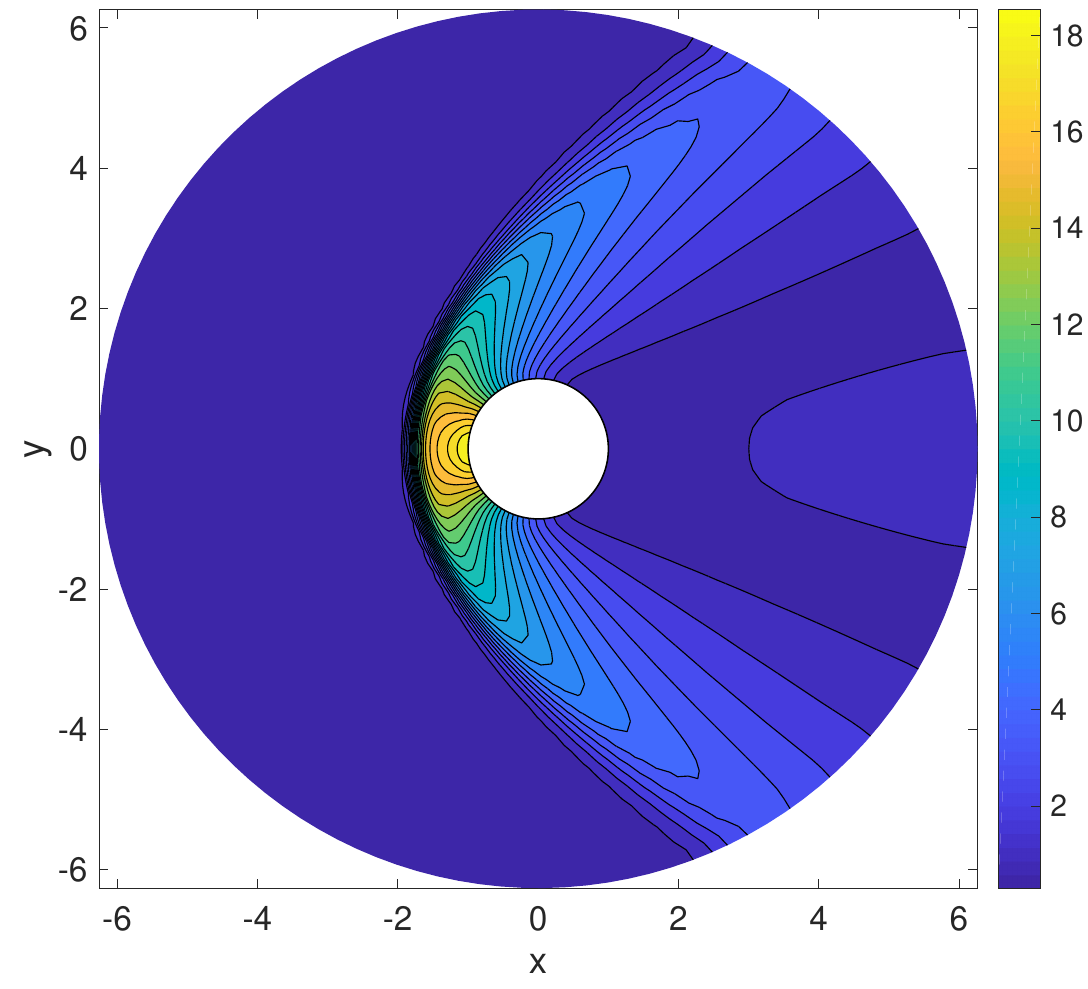}}
\caption{Pressure.}
\end{subfigure}
\caption{Circular cylinder at $\mathrm{Ma}=5$ and $\varepsilon=10^{-4}$: macroscopic fields of IWPD (filled contours) and NS (lines) on common levels.}\label{fig:cyl_kn1em4_fields}
\end{figure}

\begin{figure}[H]
\centering\singlespacing\setlength{\figw}{.65\linewidth}
\begin{subfigure}[t]{0.4892\figw}\centering
\raisebox{0.0000\figw}{\includegraphics[width=0.4892\figw]{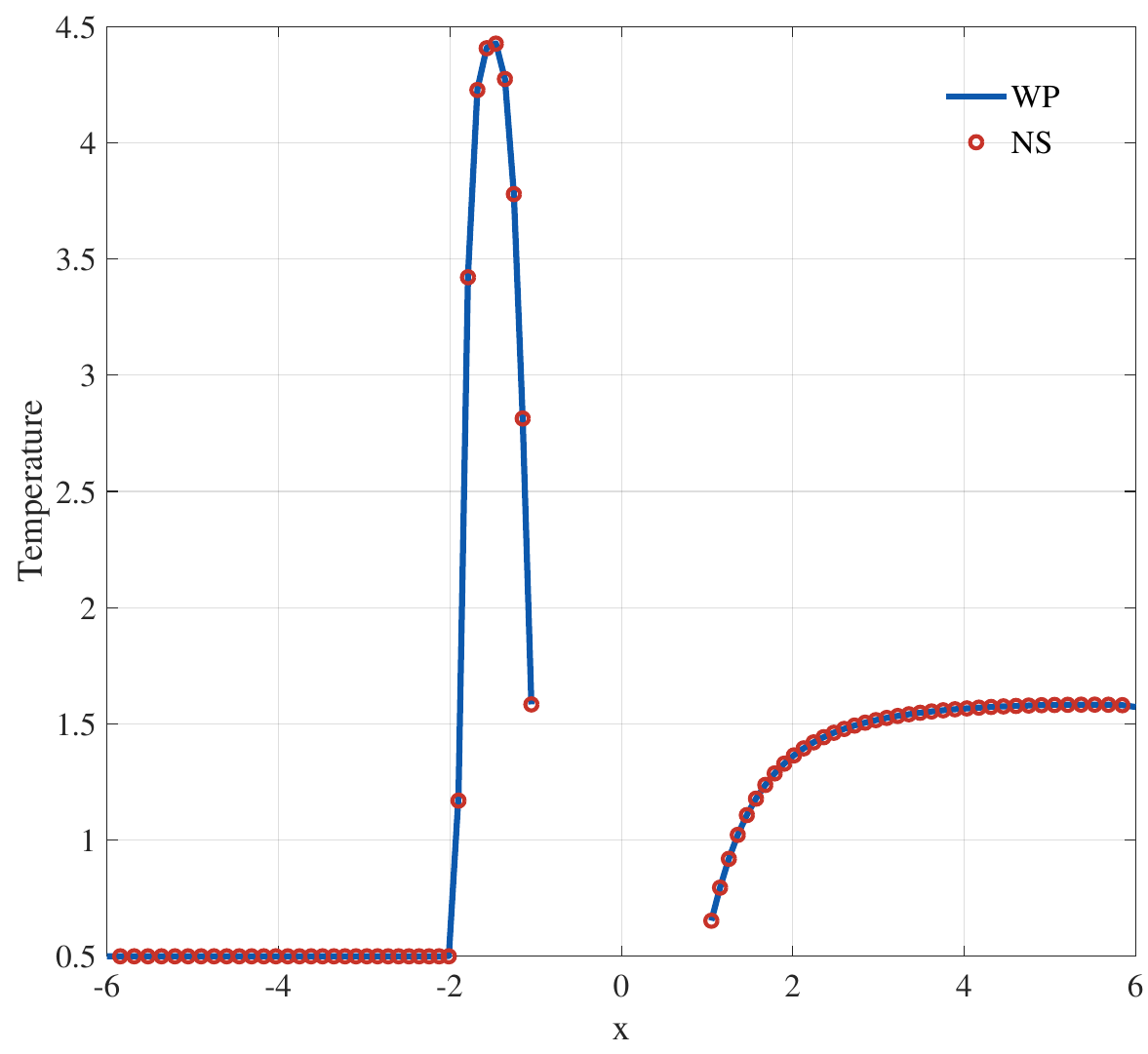}}
\caption{Stagnation-line temperature.}
\end{subfigure}\hspace{0.0200\figw}%
\begin{subfigure}[t]{0.4858\figw}\centering
\raisebox{0.0000\figw}{\includegraphics[width=0.4858\figw]{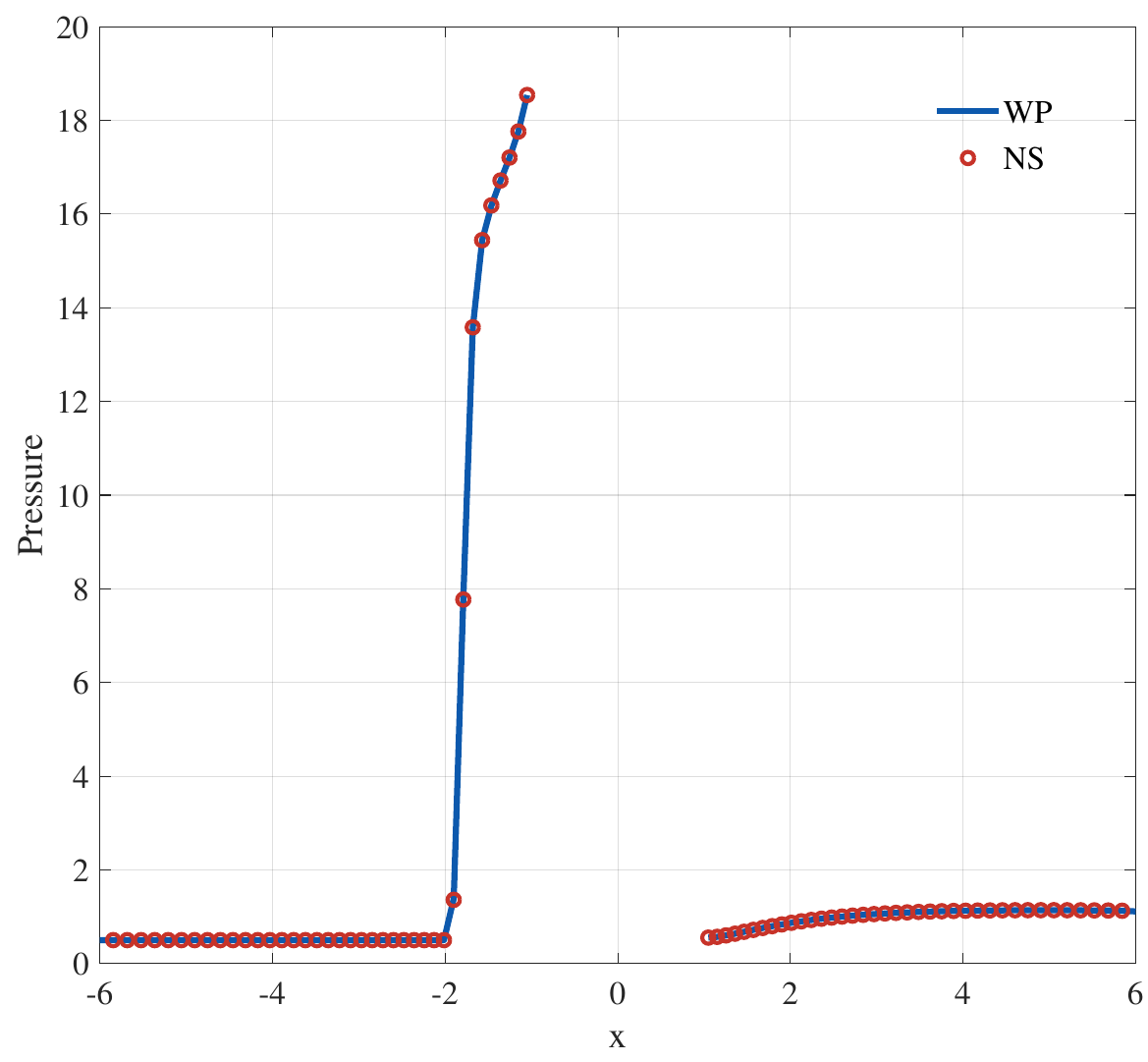}}
\caption{Stagnation-line pressure.}
\end{subfigure}
\par\smallskip
\begin{subfigure}[t]{0.4892\figw}\centering
\raisebox{0.0000\figw}{\includegraphics[width=0.4892\figw]{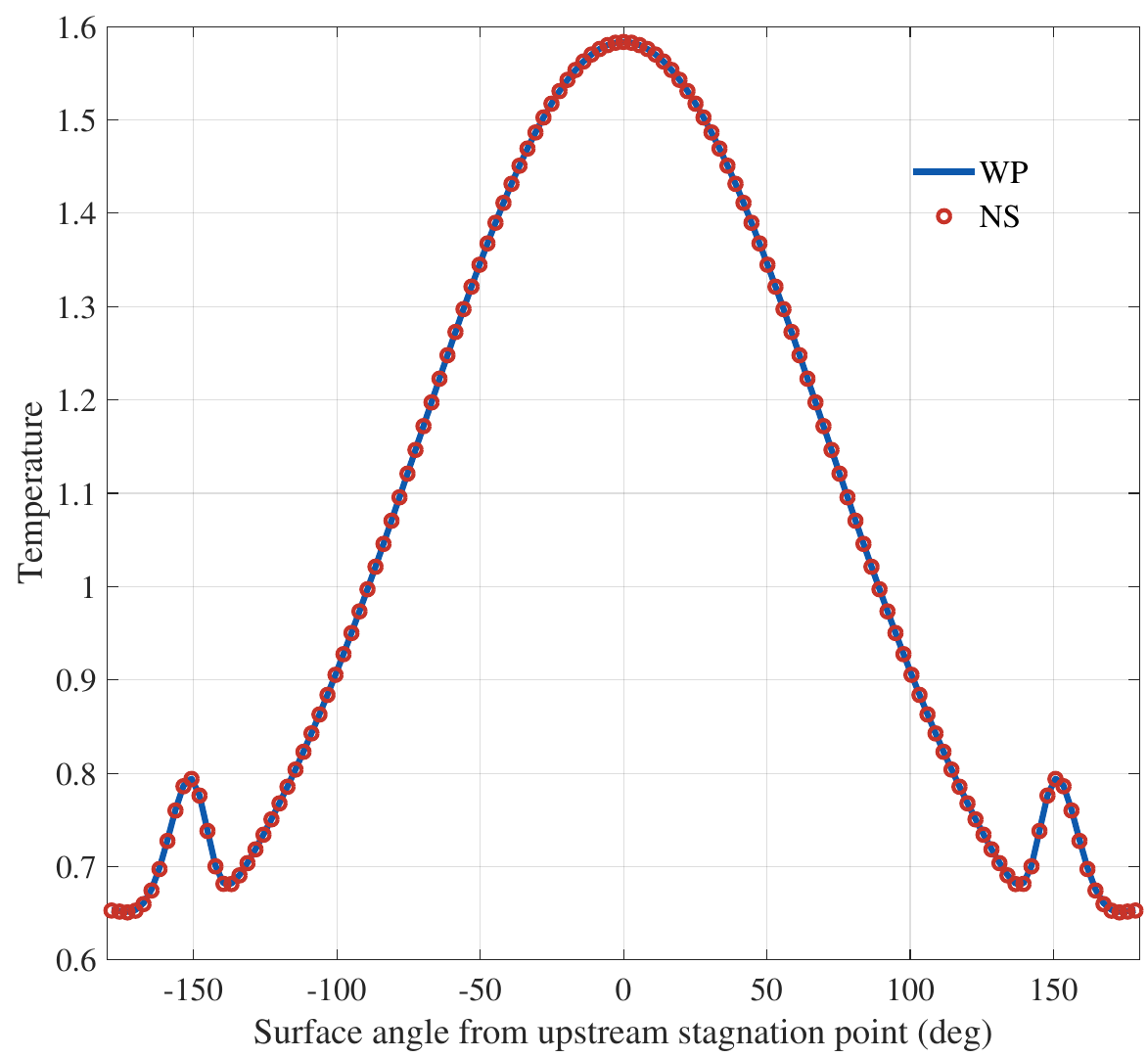}}
\caption{Wall temperature.}
\end{subfigure}\hspace{0.0200\figw}%
\begin{subfigure}[t]{0.4858\figw}\centering
\raisebox{0.0000\figw}{\includegraphics[width=0.4858\figw]{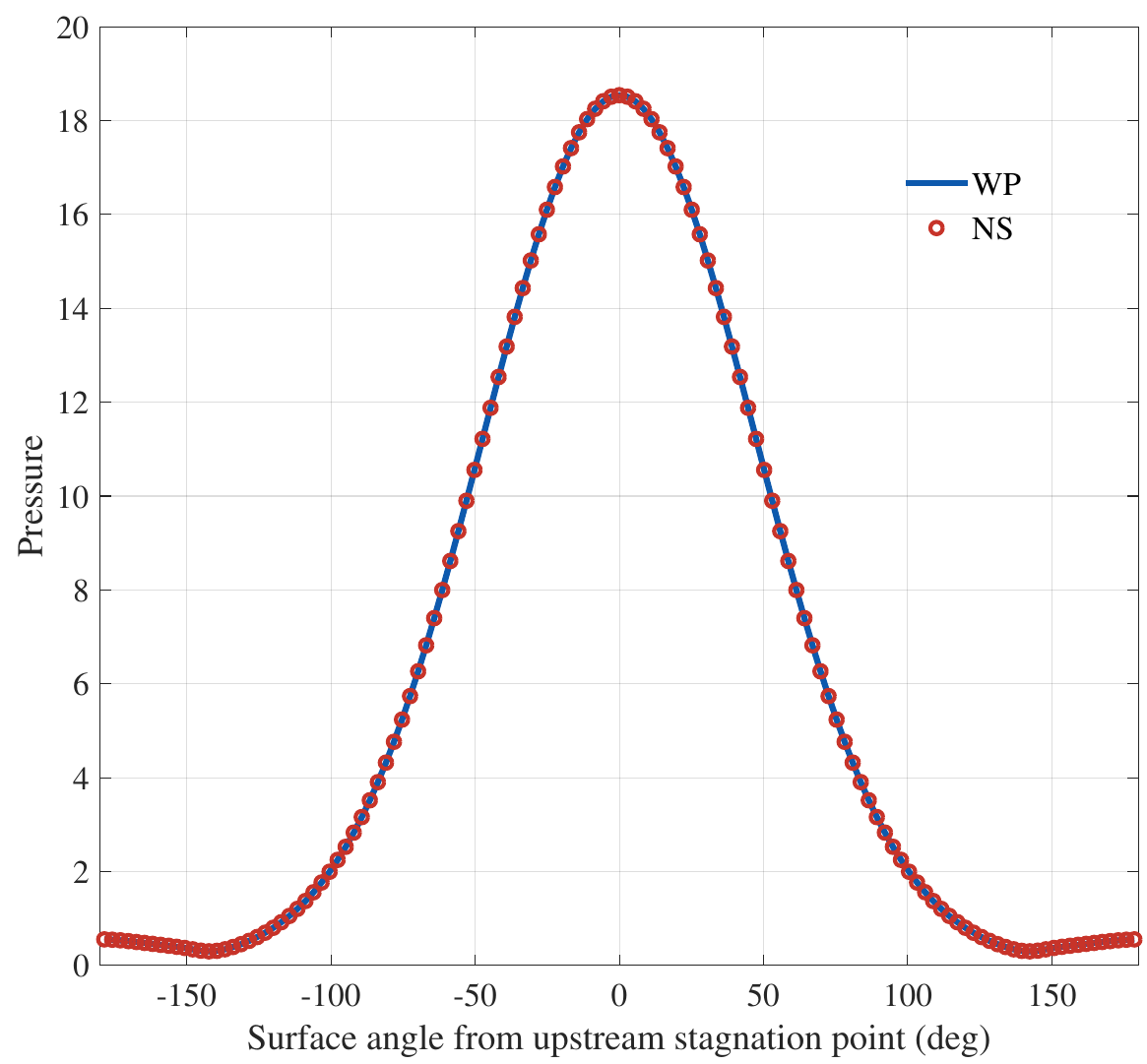}}
\caption{Wall pressure.}
\end{subfigure}
\caption{Circular cylinder at $\mathrm{Ma}=5$ and $\varepsilon=10^{-4}$: temperature and pressure along the stagnation line (a,b) and on the wall (c,d). Lines denote IWPD and symbols NS.}\label{fig:cyl_kn1em4_profiles}
\end{figure}

\begin{figure}[H]
\centering\singlespacing\setlength{\figw}{.90\linewidth}
\begin{subfigure}[t]{0.9950\figw}\centering
\makebox[\linewidth]{\raisebox{0.0000\figw}{\includegraphics[width=0.5032\figw]{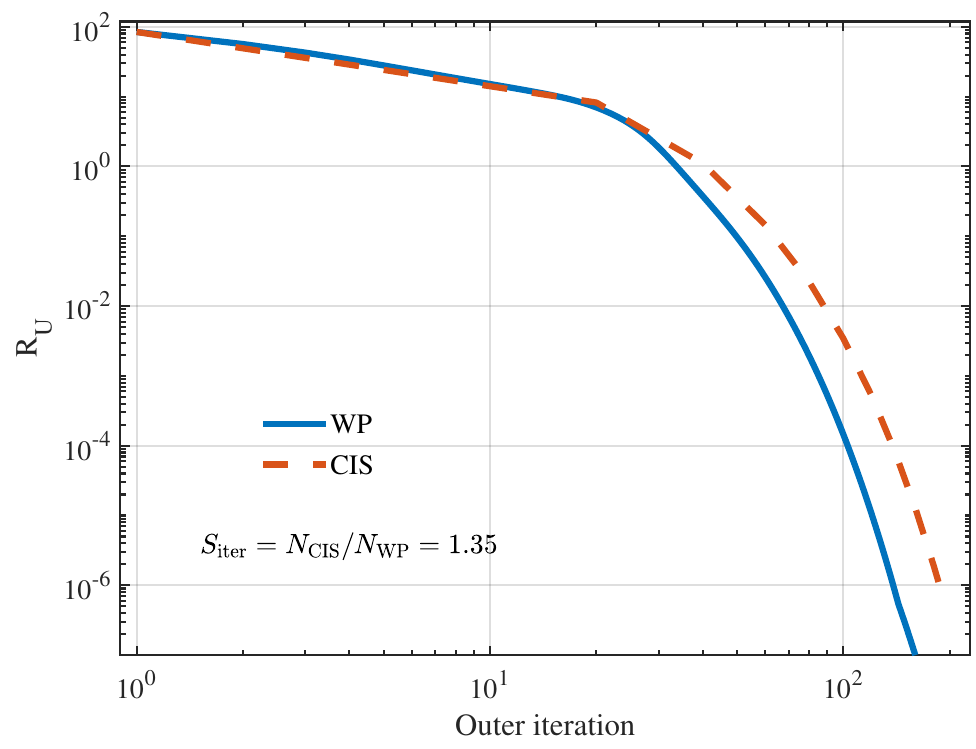}}\hspace{0.0200\figw}\raisebox{0.0045\figw}{\includegraphics[width=0.4718\figw]{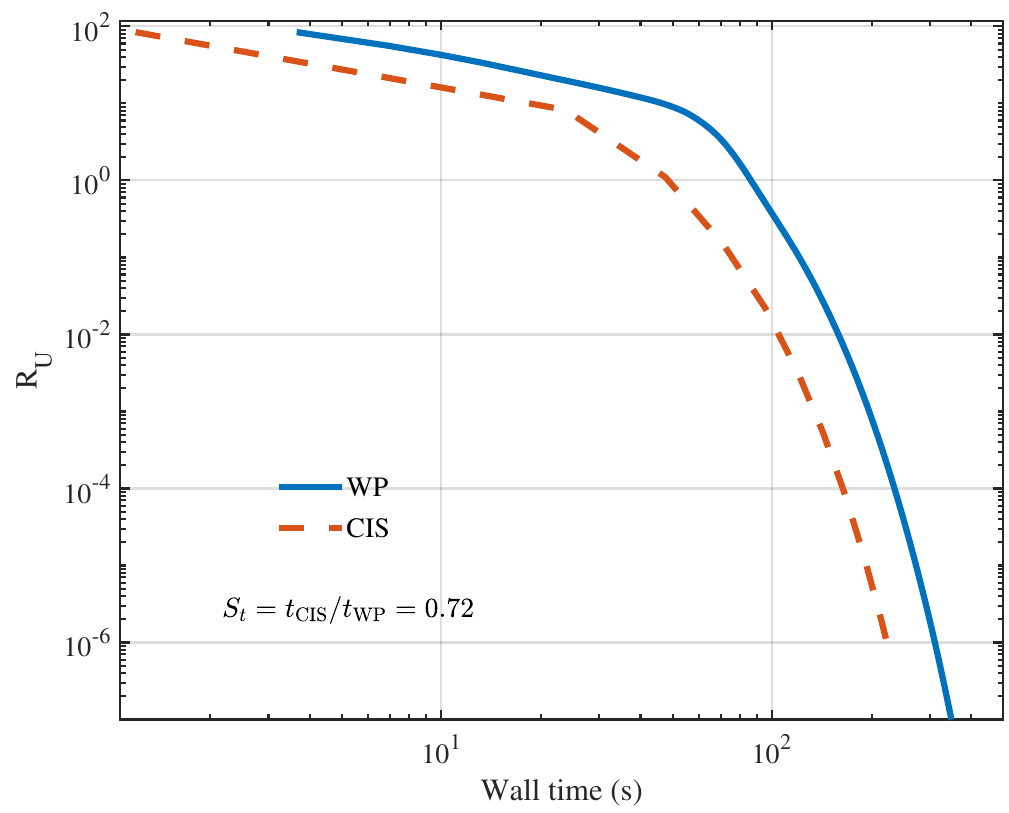}}}
\caption{$\varepsilon=1$.}
\end{subfigure}
\par\smallskip
\begin{subfigure}[t]{0.9950\figw}\centering
\makebox[\linewidth]{\raisebox{0.0000\figw}{\includegraphics[width=0.4818\figw]{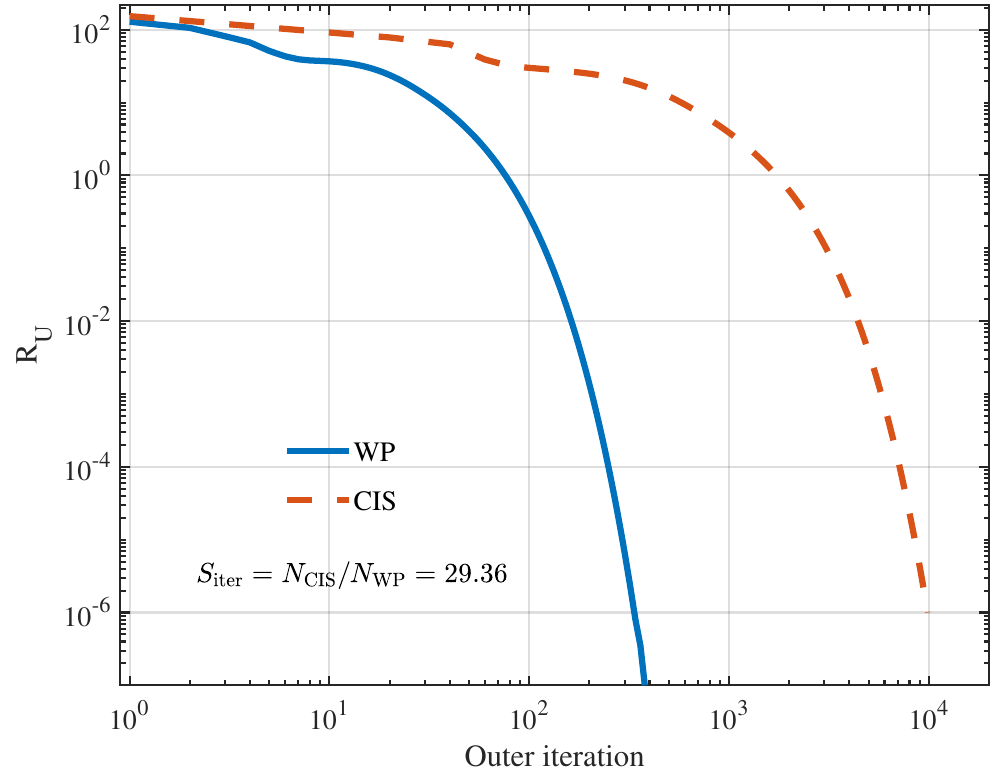}}\hspace{0.0200\figw}\raisebox{0.0002\figw}{\includegraphics[width=0.4932\figw]{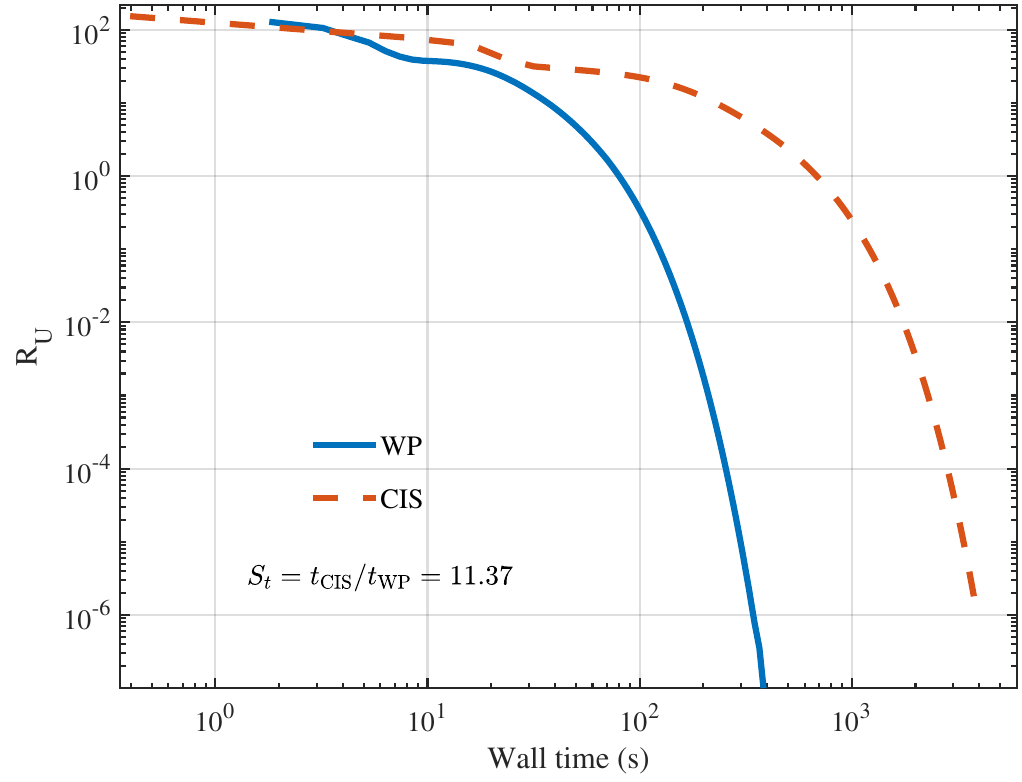}}}
\caption{$\varepsilon=10^{-2}$.}
\end{subfigure}
\par\smallskip
\begin{subfigure}[t]{0.9950\figw}\centering
\makebox[\linewidth]{\raisebox{0.0056\figw}{\includegraphics[width=0.4800\figw]{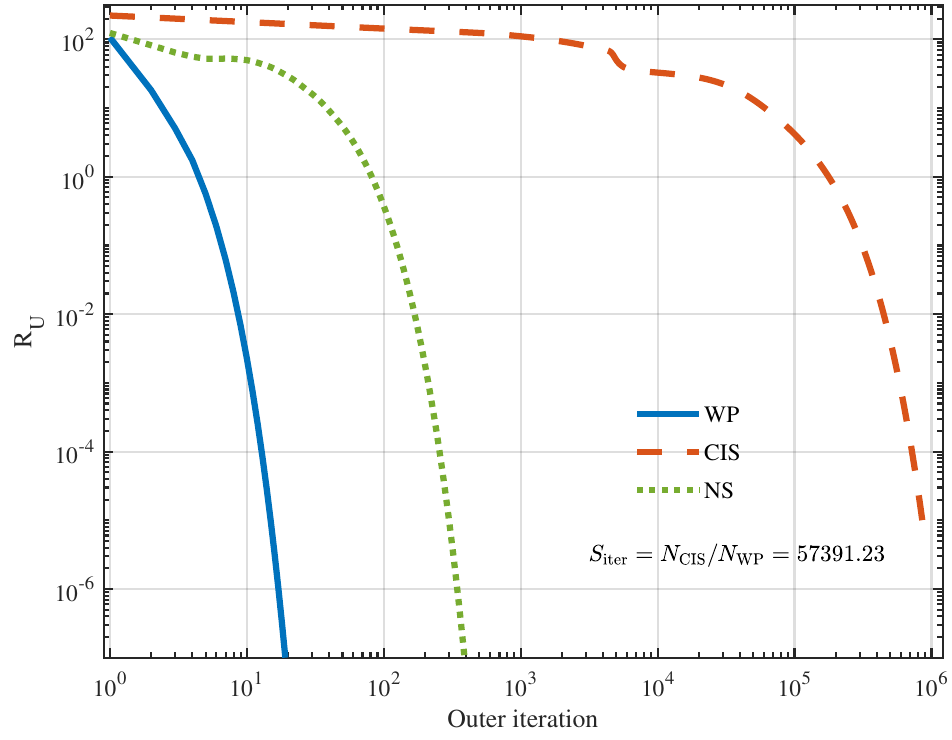}}\hspace{0.0200\figw}\raisebox{0.0000\figw}{\includegraphics[width=0.4950\figw]{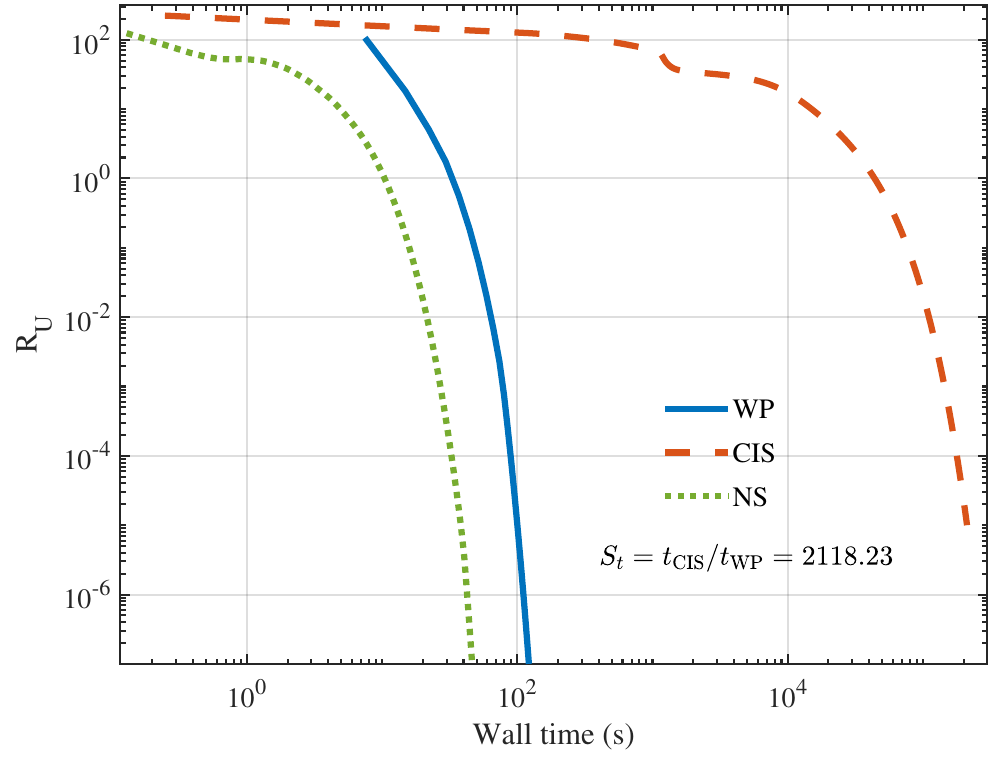}}}
\caption{$\varepsilon=10^{-4}$.}
\end{subfigure}
\caption{Circular cylinder at $\mathrm{Ma}=5$: residual versus iteration (left) and wall time (right), with IWPD solid, CIS dashed and NS dotted. Grey line: $\|R_U\|=10^{-6}$. Ratios use $10^{-6}$, except $10^{-5}$ at $\varepsilon=10^{-4}$.}\label{fig:cyl_residual}
\end{figure}

\clearpage
\section{X38 section}
\begin{figure}[H]
\centering\singlespacing\setlength{\figw}{.80\linewidth}
\begin{subfigure}[t]{0.4875\figw}\centering
\raisebox{0.0000\figw}{\includegraphics[width=0.4875\figw]{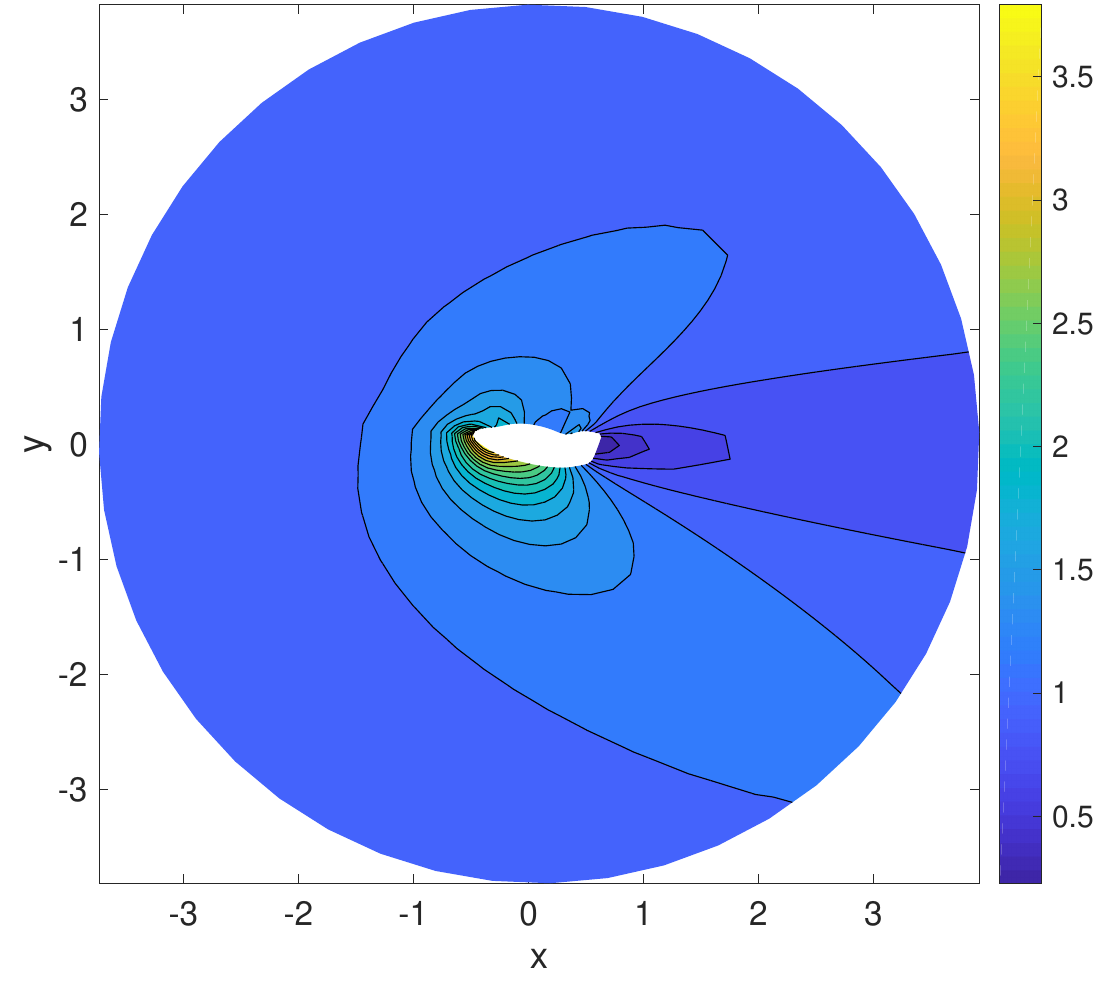}}
\caption{Density.}
\end{subfigure}\hspace{0.0200\figw}%
\begin{subfigure}[t]{0.4875\figw}\centering
\raisebox{0.0000\figw}{\includegraphics[width=0.4875\figw]{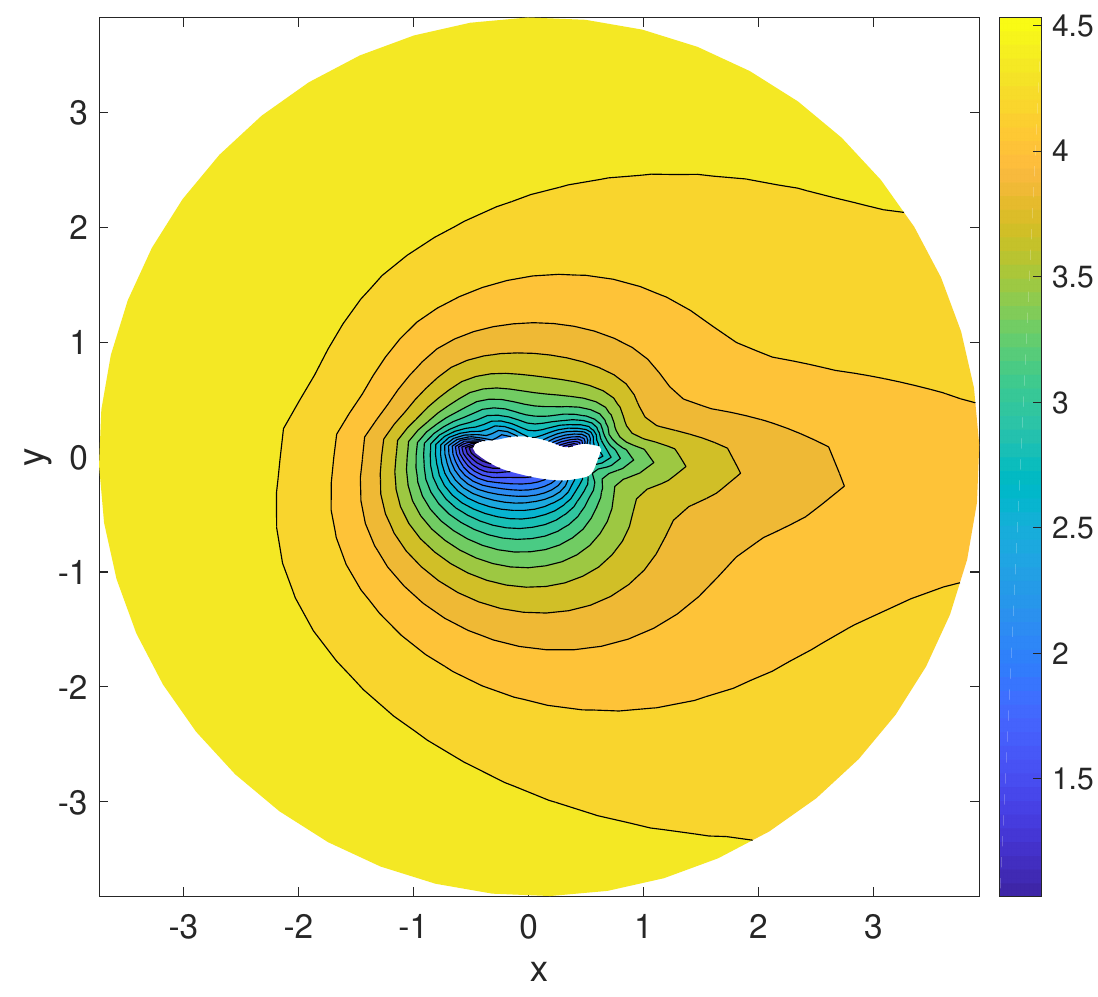}}
\caption{Speed.}
\end{subfigure}
\par\smallskip
\begin{subfigure}[t]{0.4931\figw}\centering
\raisebox{0.0000\figw}{\includegraphics[width=0.4931\figw]{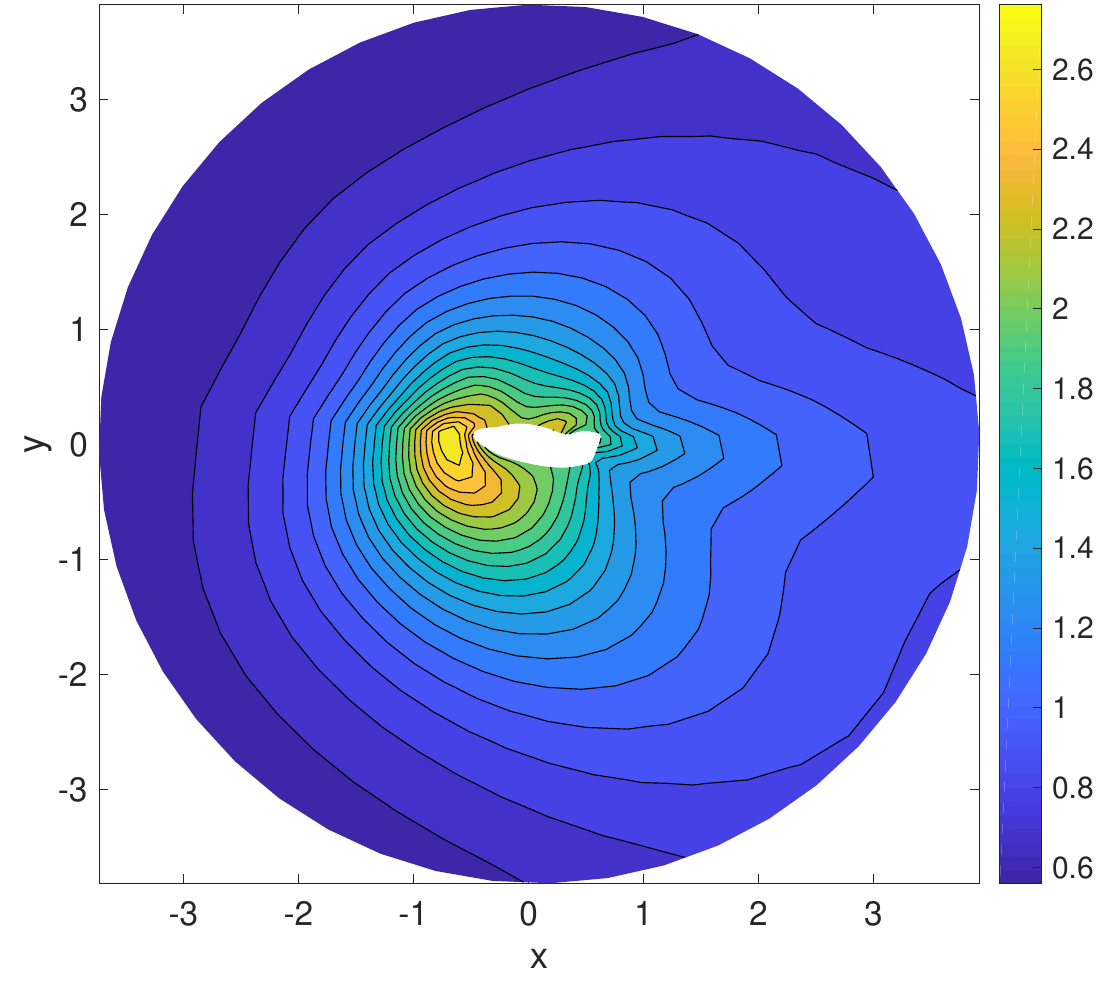}}
\caption{Temperature.}
\end{subfigure}\hspace{0.0200\figw}%
\begin{subfigure}[t]{0.4819\figw}\centering
\raisebox{0.0000\figw}{\includegraphics[width=0.4819\figw]{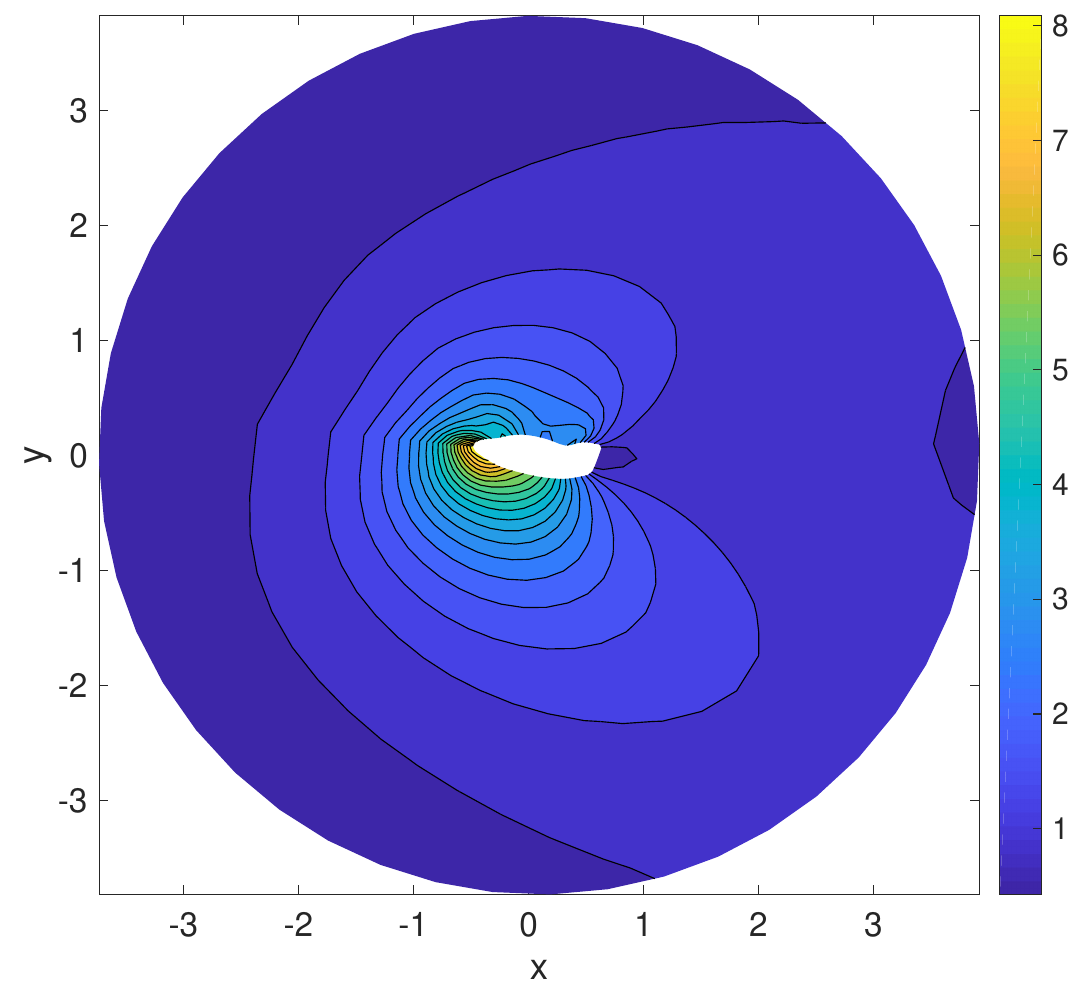}}
\caption{Pressure.}
\end{subfigure}
\caption{X38 section at $\mathrm{Ma}=5$ and $\varepsilon=1$: macroscopic fields of IWPD (filled contours) and CIS (lines) on common levels.}\label{fig:x38_kn1_fields}
\end{figure}

\begin{figure}[H]
\centering\singlespacing\setlength{\figw}{.85\linewidth}
\begin{subfigure}[t]{0.4930\figw}\centering
\raisebox{0.0000\figw}{\includegraphics[width=0.4930\figw]{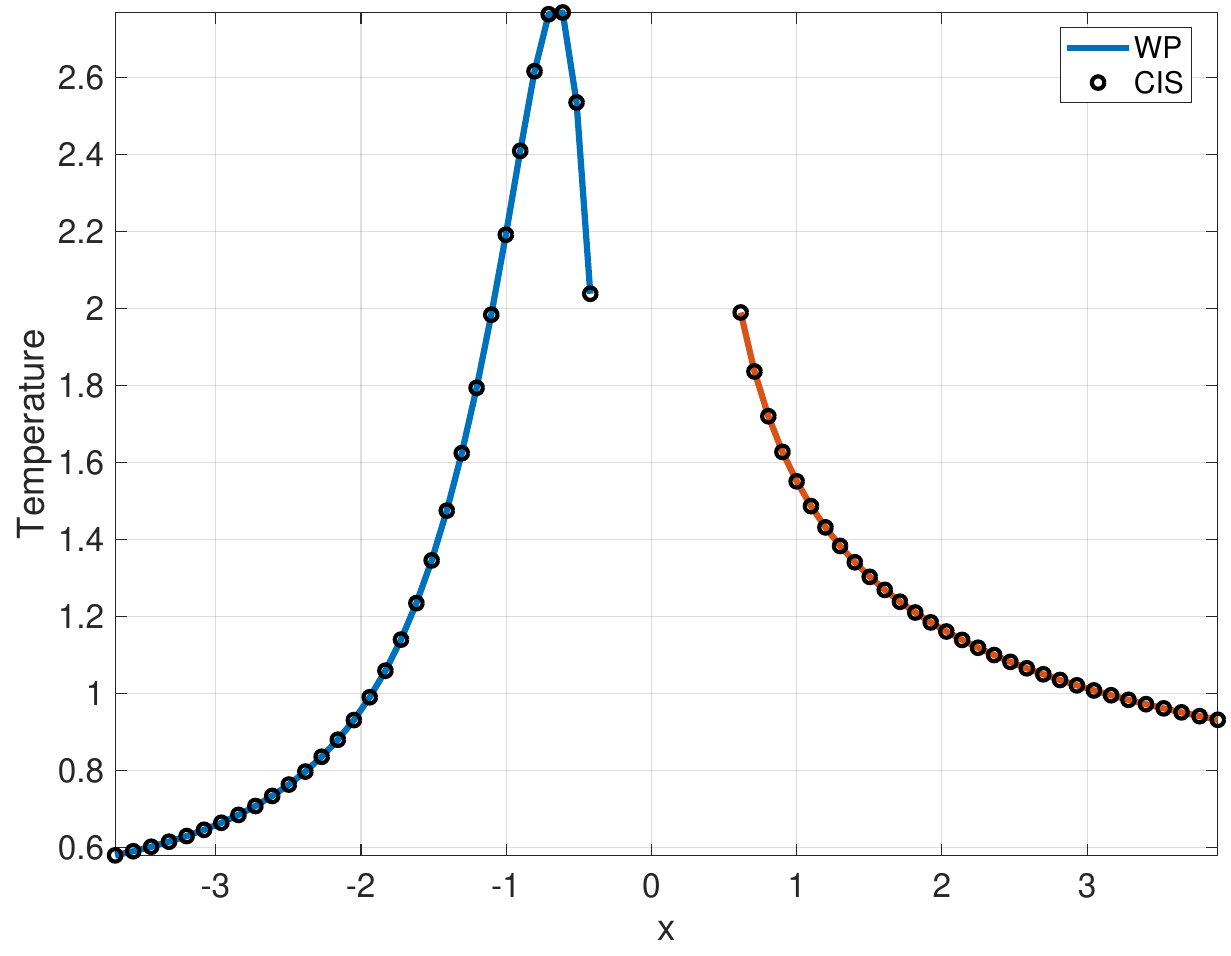}}
\caption{Temperature.}
\end{subfigure}\hspace{0.0200\figw}%
\begin{subfigure}[t]{0.4820\figw}\centering
\raisebox{0.0000\figw}{\includegraphics[width=0.4820\figw]{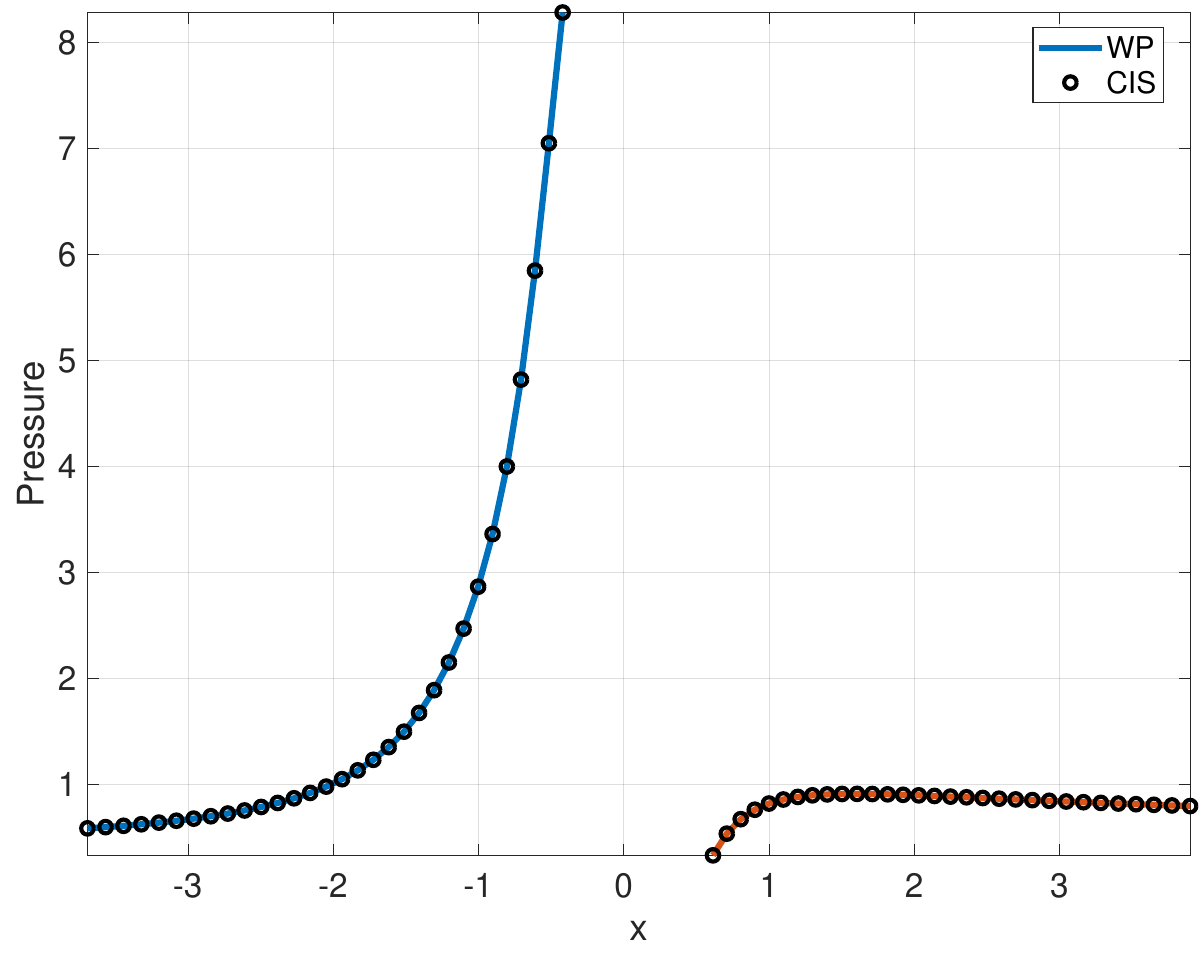}}
\caption{Pressure.}
\end{subfigure}
\caption{X38 section at $\varepsilon=1$: temperature and pressure along the stagnation line. Lines denote IWPD and symbols CIS.}\label{fig:x38_kn1_profiles}
\end{figure}

\begin{figure}[H]
\centering\singlespacing\setlength{\figw}{.65\linewidth}
\begin{subfigure}[t]{0.4803\figw}\centering
\raisebox{0.0000\figw}{\includegraphics[width=0.4803\figw]{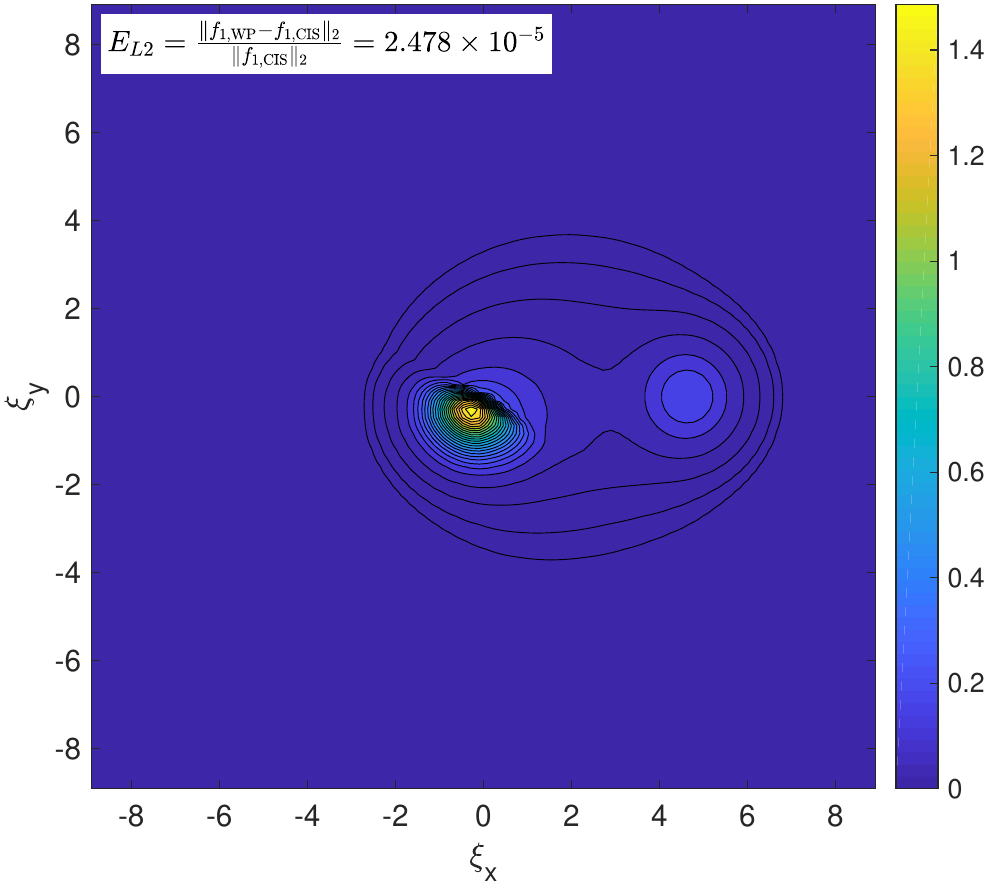}}
\caption{Sample 5, $(x,y)=(-0.38,-0.03)$.}
\end{subfigure}\hspace{0.0200\figw}%
\begin{subfigure}[t]{0.4947\figw}\centering
\raisebox{0.0000\figw}{\includegraphics[width=0.4947\figw]{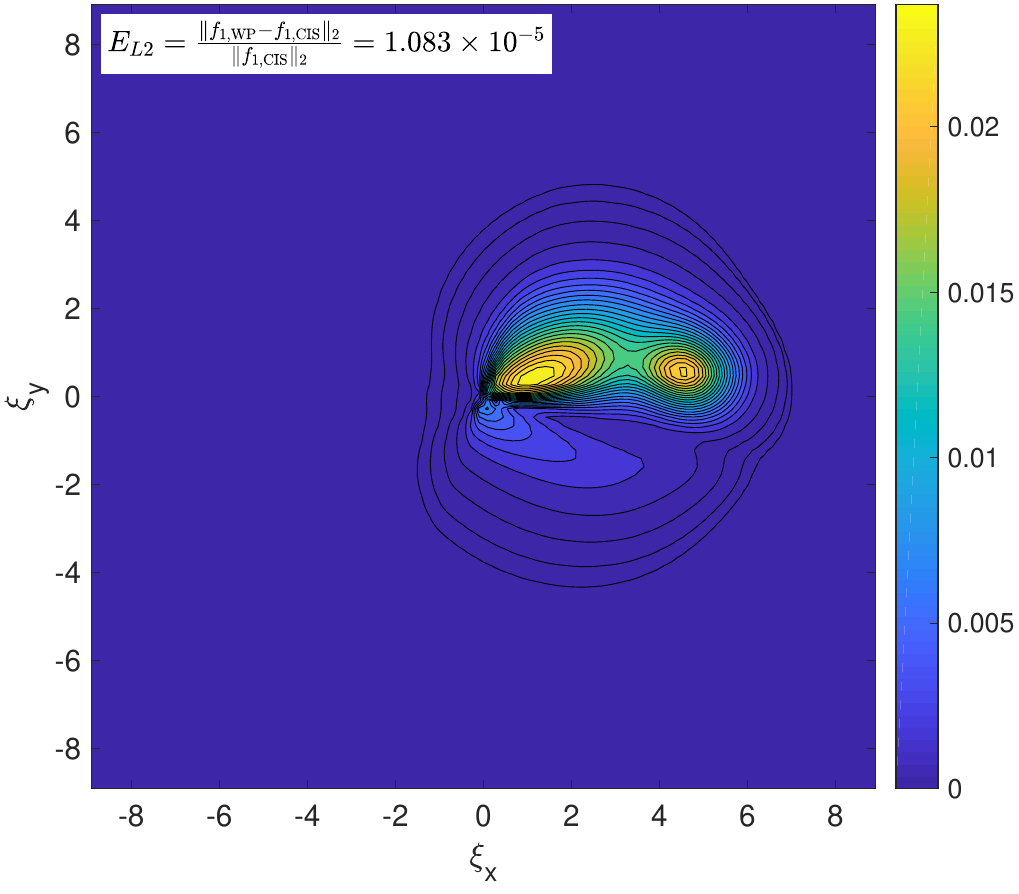}}
\caption{Sample 6, $(x,y)=(0.60,-0.04)$.}
\end{subfigure}
\par\smallskip
\begin{subfigure}[t]{0.4875\figw}\centering
\raisebox{0.0000\figw}{\includegraphics[width=0.4875\figw]{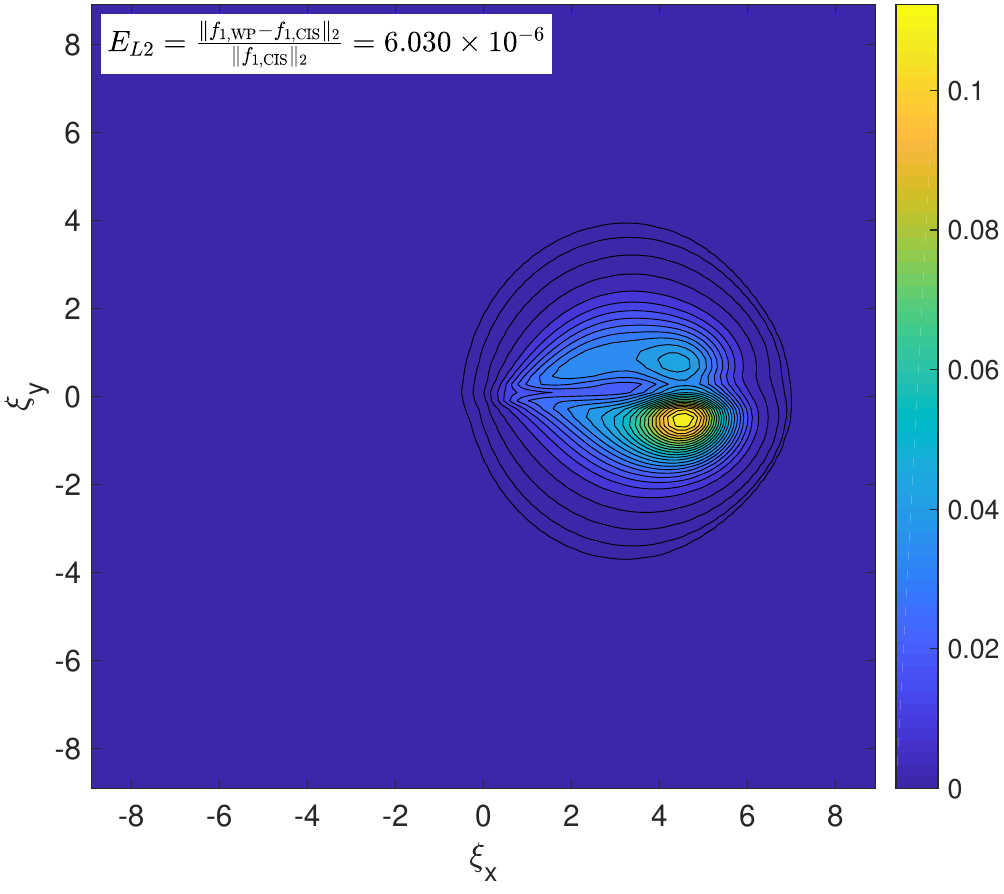}}
\caption{Sample 7, $(x,y)=(1.32,0.07)$.}
\end{subfigure}\hspace{0.0200\figw}%
\begin{subfigure}[t]{0.4875\figw}\centering
\raisebox{0.0000\figw}{\includegraphics[width=0.4875\figw]{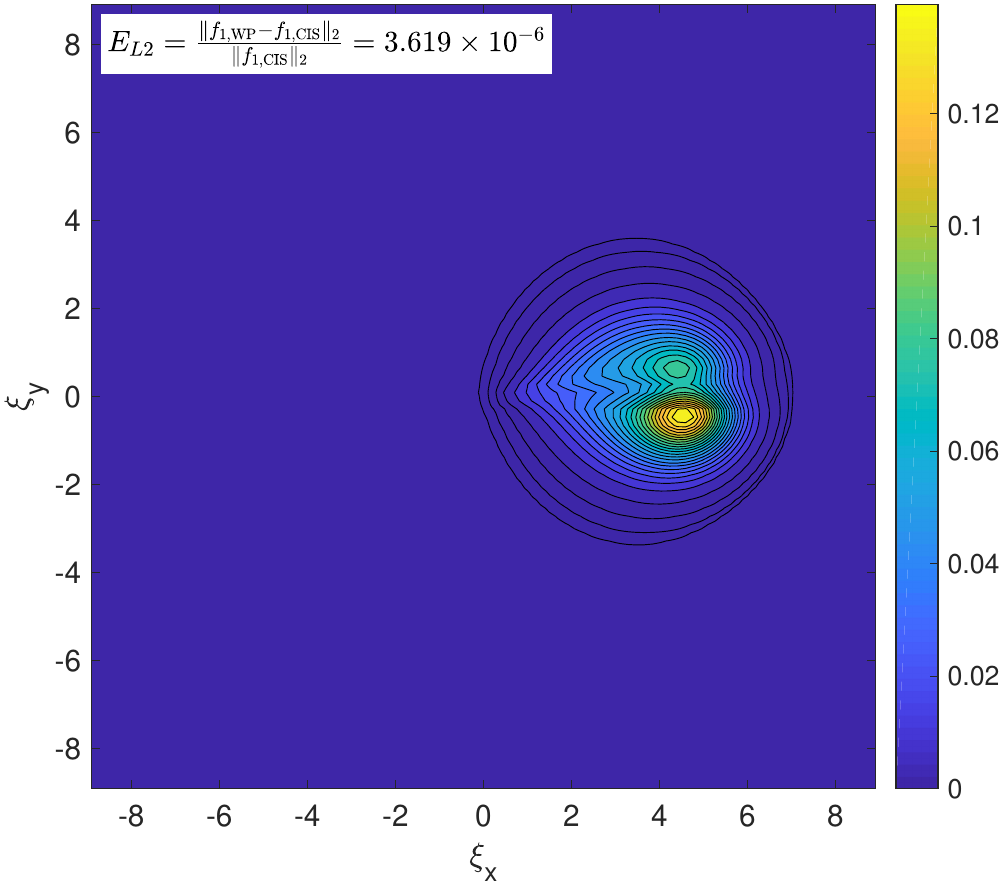}}
\caption{Sample 8, $(x,y)=(2.27,0.08)$.}
\end{subfigure}
\caption{X38 section at $\varepsilon=1$: reduced distributions $f_1(\xi_x,\xi_y)$ on a linear scale at four points along the centreline. IWPD is shown by filled contours and CIS by lines.}\label{fig:x38_kn1_dist_lin}
\end{figure}

\begin{figure}[H]
\centering\singlespacing\setlength{\figw}{.65\linewidth}
\begin{subfigure}[t]{0.4875\figw}\centering
\raisebox{0.0000\figw}{\includegraphics[width=0.4875\figw]{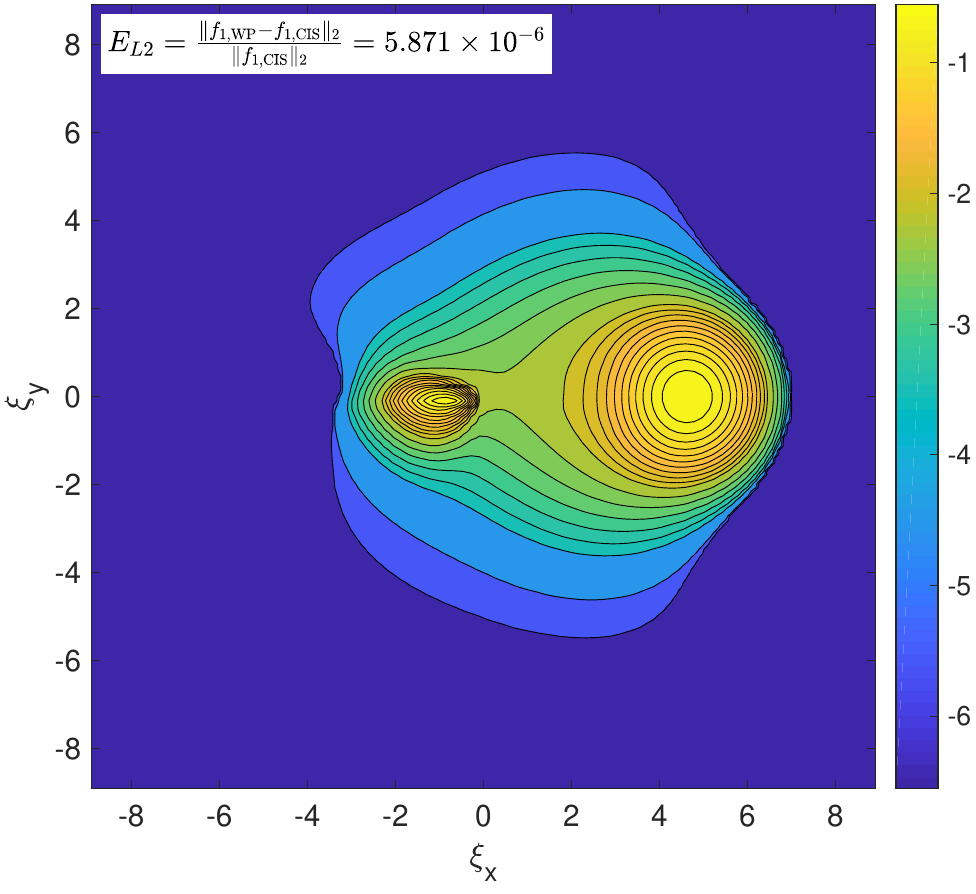}}
\caption{Sample 4, $(x,y)=(-1.14,0.01)$.}
\end{subfigure}\hspace{0.0200\figw}%
\begin{subfigure}[t]{0.4875\figw}\centering
\raisebox{0.0000\figw}{\includegraphics[width=0.4875\figw]{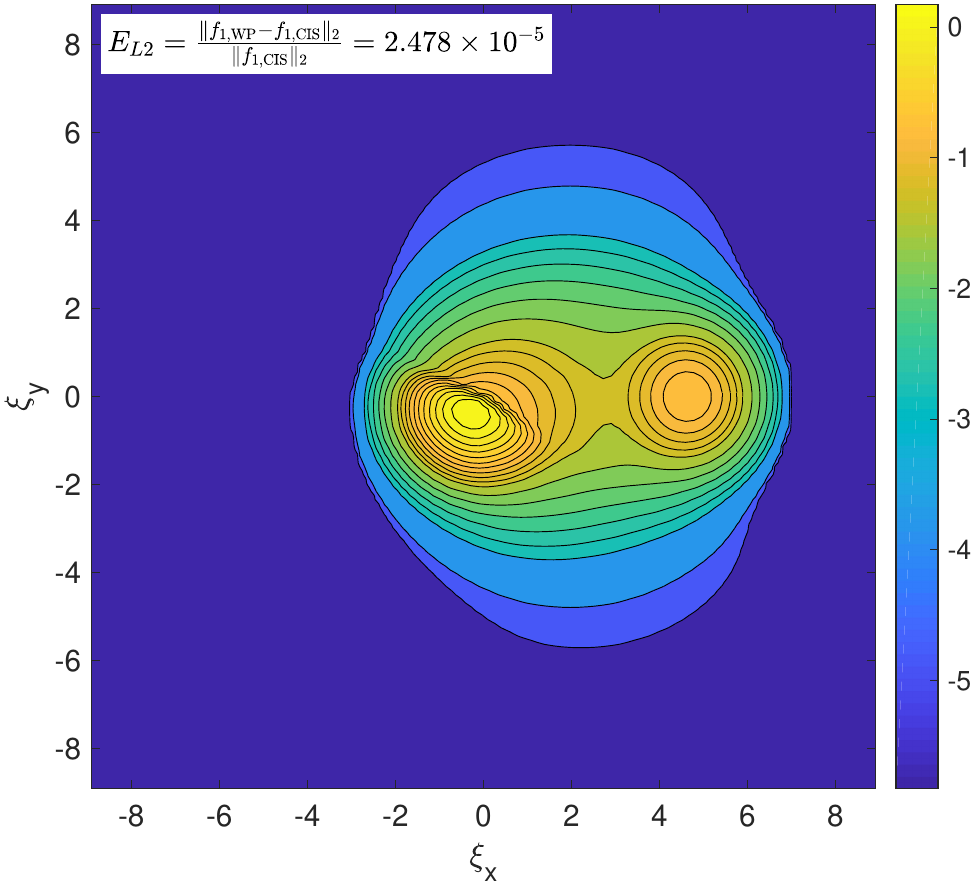}}
\caption{Sample 5, $(x,y)=(-0.38,-0.03)$.}
\end{subfigure}
\par\smallskip
\begin{subfigure}[t]{0.4875\figw}\centering
\raisebox{0.0000\figw}{\includegraphics[width=0.4875\figw]{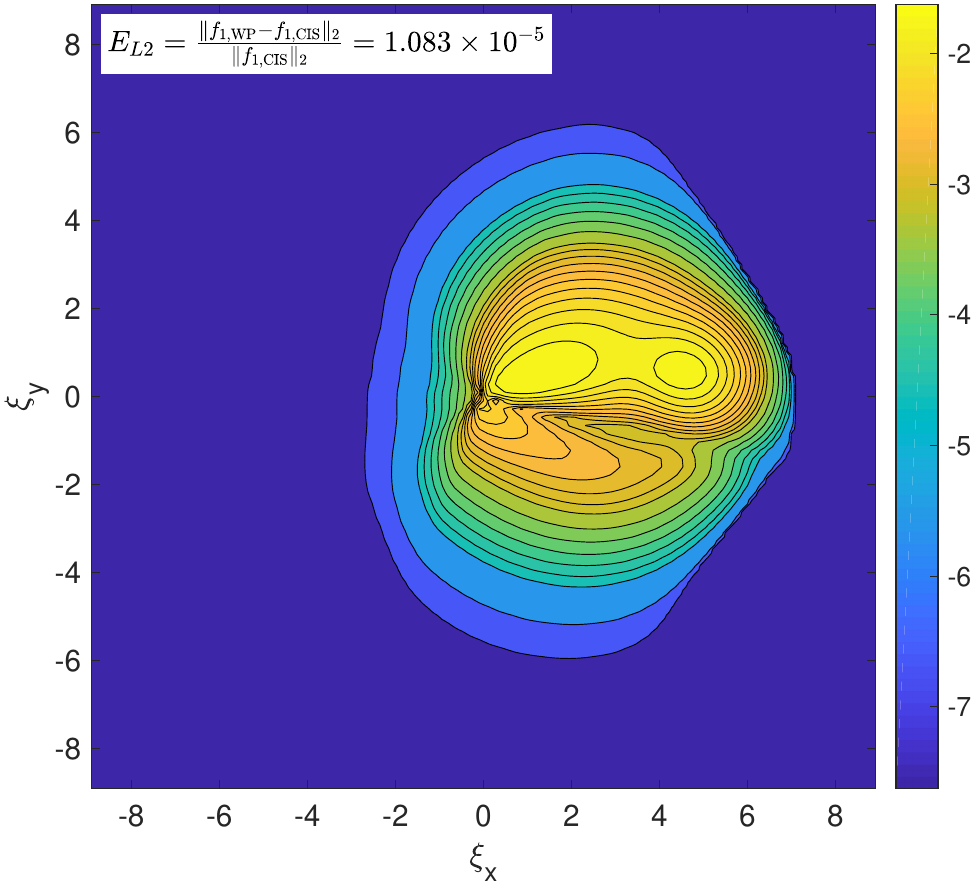}}
\caption{Sample 6, $(x,y)=(0.60,-0.04)$.}
\end{subfigure}\hspace{0.0200\figw}%
\begin{subfigure}[t]{0.4875\figw}\centering
\raisebox{0.0000\figw}{\includegraphics[width=0.4875\figw]{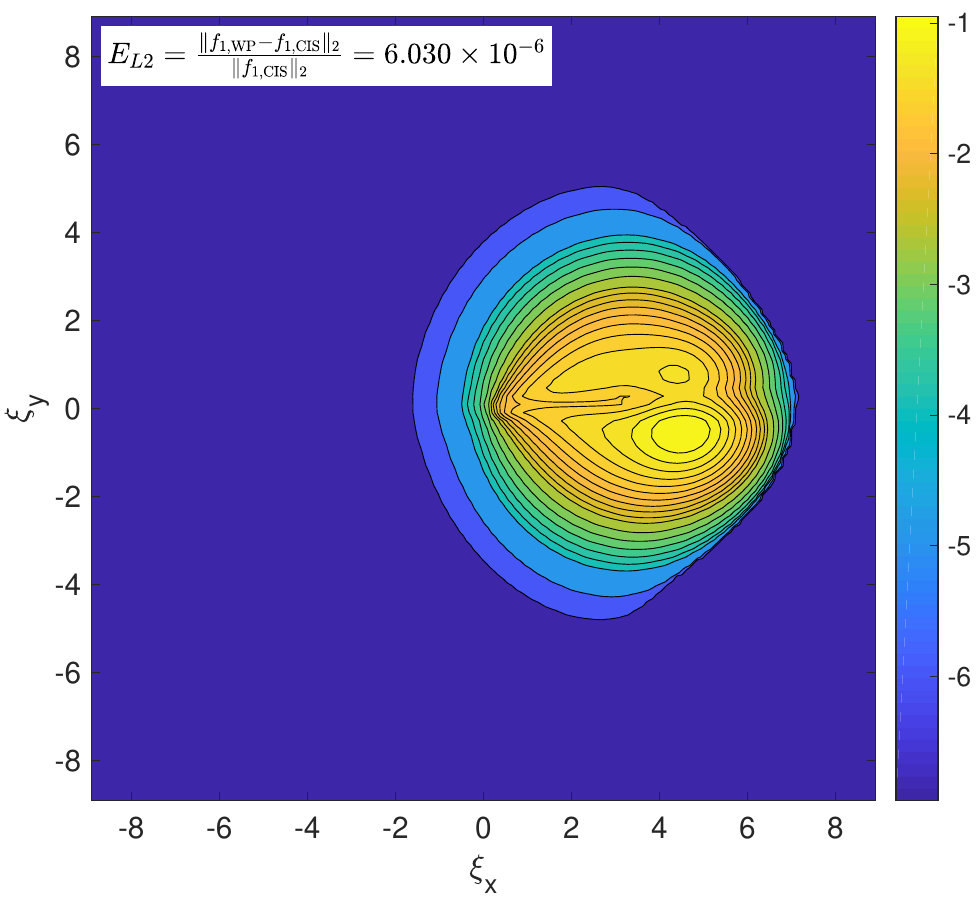}}
\caption{Sample 7, $(x,y)=(1.32,0.07)$.}
\end{subfigure}
\caption{X38 section at $\varepsilon=1$: reduced distributions $f_1(\xi_x,\xi_y)$ on a logarithmic scale at four points along the centreline. IWPD is shown by filled contours and CIS by lines.}\label{fig:x38_kn1_dist_log}
\end{figure}

\begin{figure}[H]
\centering\singlespacing\setlength{\figw}{.65\linewidth}
\begin{subfigure}[t]{0.4875\figw}\centering
\raisebox{0.0000\figw}{\includegraphics[width=0.4875\figw]{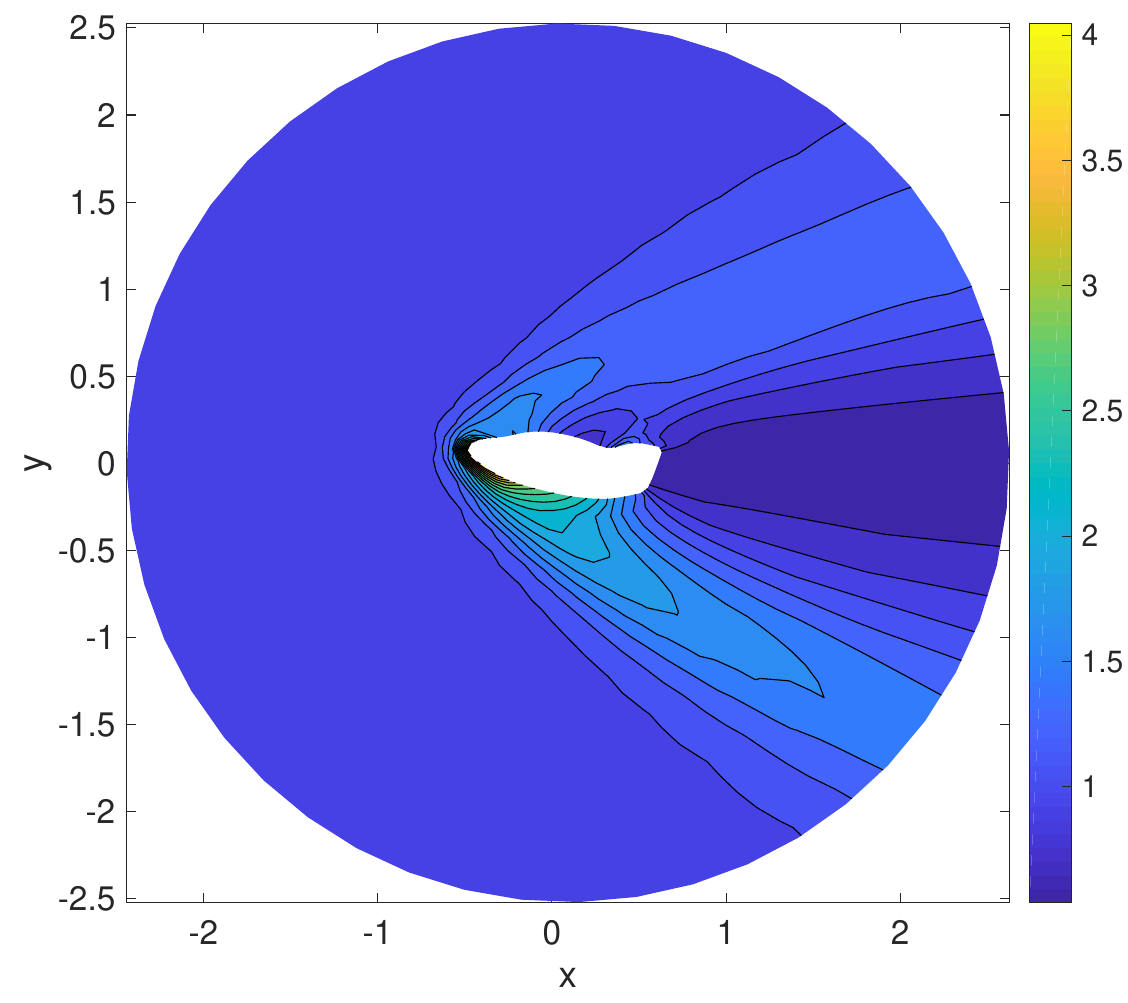}}
\caption{Density.}
\end{subfigure}\hspace{0.0200\figw}%
\begin{subfigure}[t]{0.4875\figw}\centering
\raisebox{0.0000\figw}{\includegraphics[width=0.4875\figw]{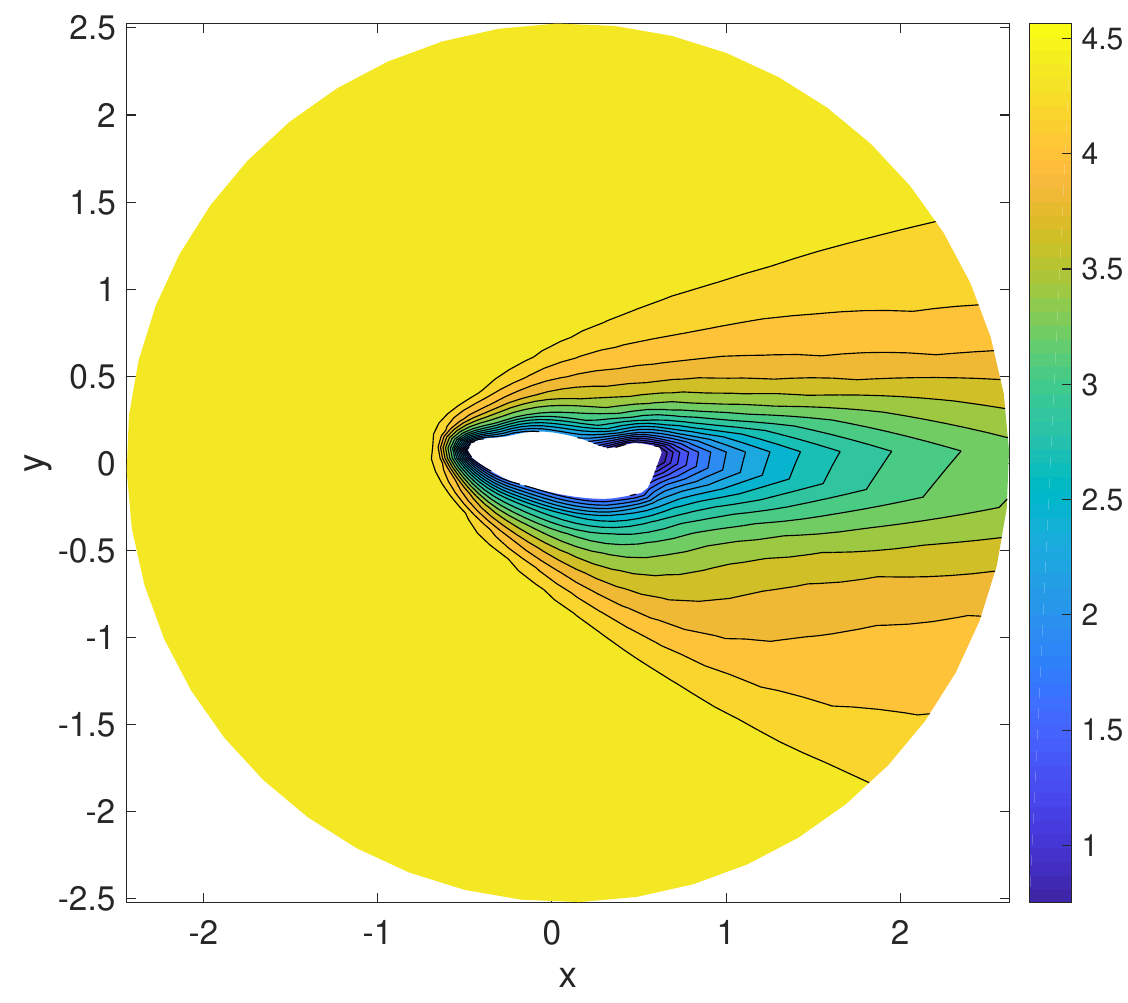}}
\caption{Speed.}
\end{subfigure}
\par\smallskip
\begin{subfigure}[t]{0.4930\figw}\centering
\raisebox{0.0000\figw}{\includegraphics[width=0.4930\figw]{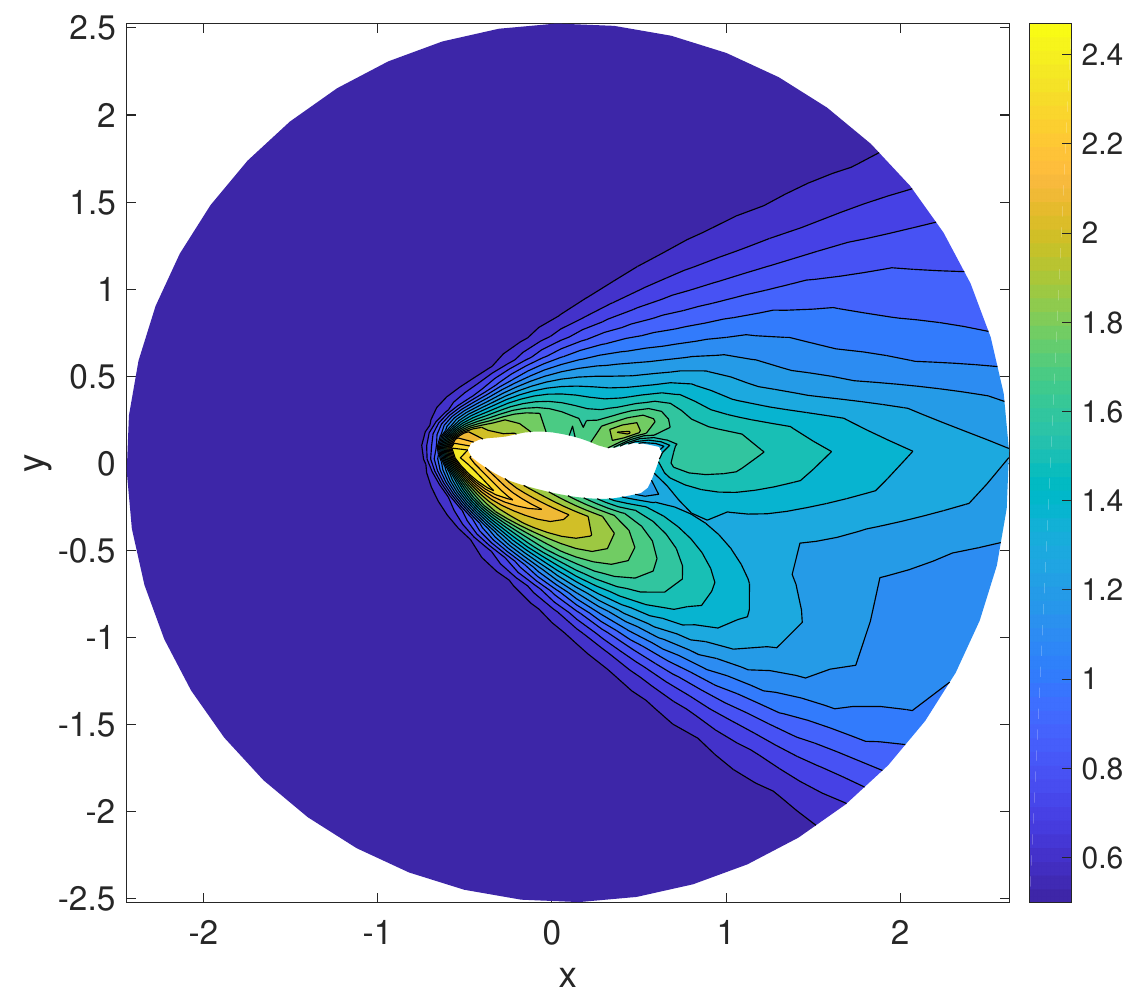}}
\caption{Temperature.}
\end{subfigure}\hspace{0.0200\figw}%
\begin{subfigure}[t]{0.4820\figw}\centering
\raisebox{0.0000\figw}{\includegraphics[width=0.4820\figw]{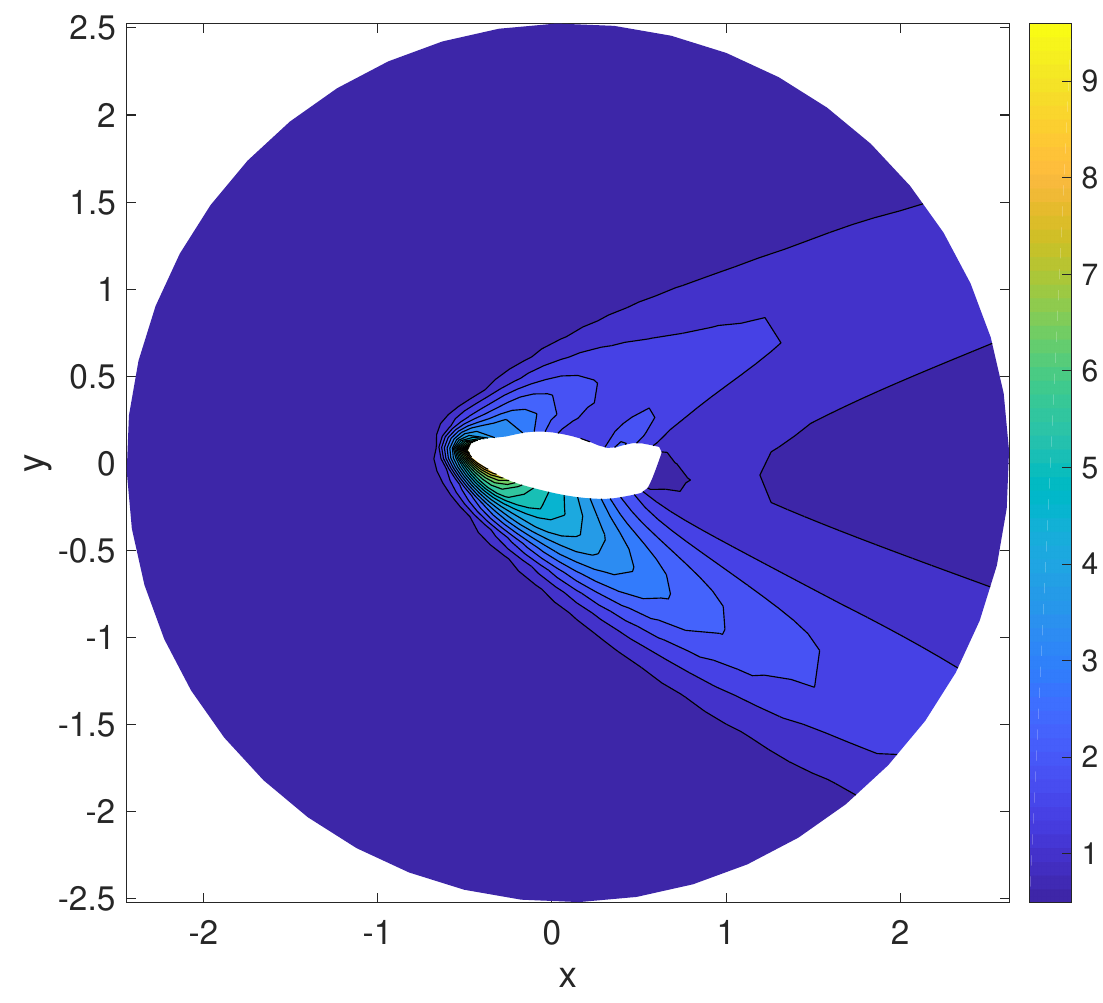}}
\caption{Pressure.}
\end{subfigure}
\caption{X38 section at $\mathrm{Ma}=5$ and $\varepsilon=10^{-2}$: macroscopic fields of IWPD (filled contours) and CIS (lines) on common levels.}\label{fig:x38_kn1em2_fields}
\end{figure}

\begin{figure}[H]
\centering\singlespacing\setlength{\figw}{.65\linewidth}
\begin{subfigure}[t]{0.4868\figw}\centering
\raisebox{0.0109\figw}{\includegraphics[width=0.4868\figw]{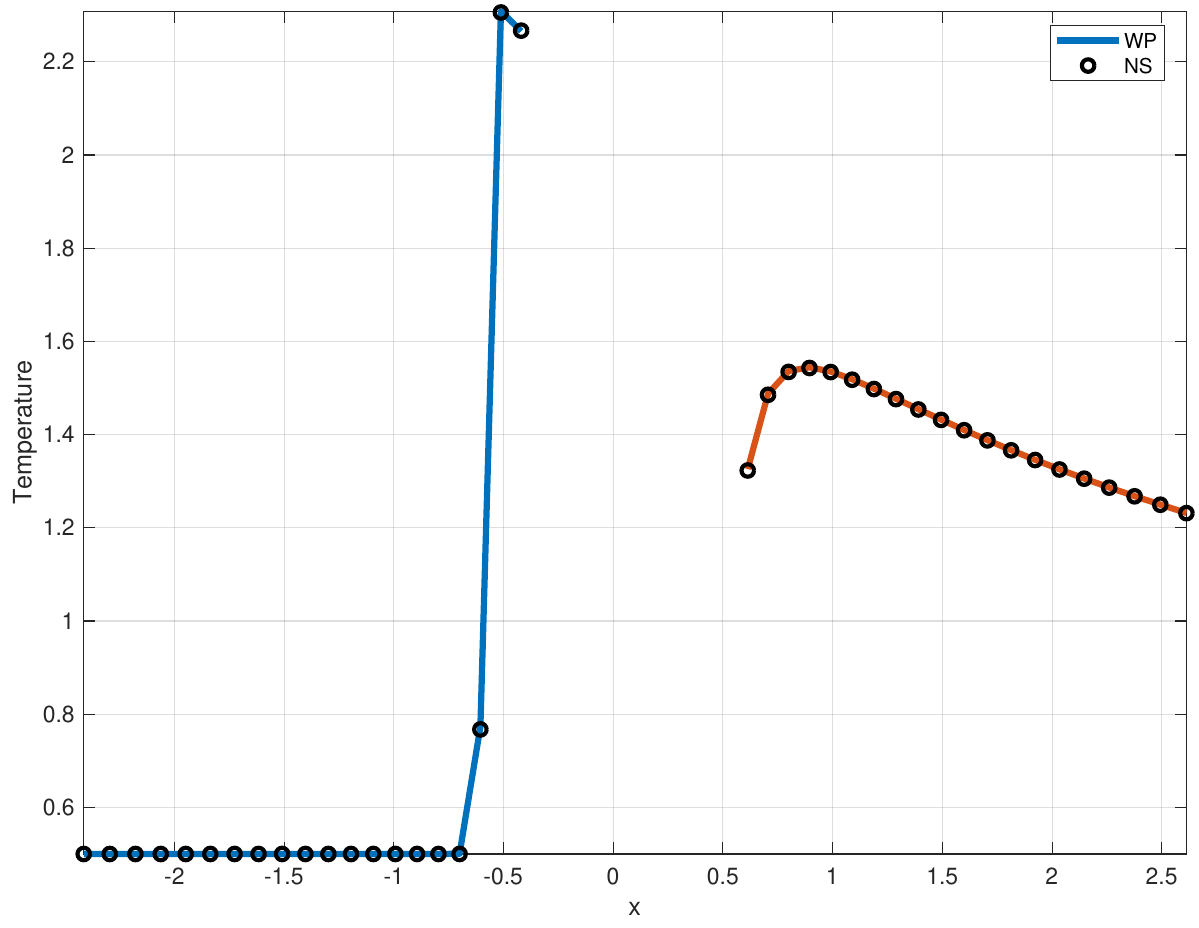}}
\caption{Stagnation-line temperature.}
\end{subfigure}\hspace{0.0200\figw}%
\begin{subfigure}[t]{0.4882\figw}\centering
\raisebox{0.0000\figw}{\includegraphics[width=0.4882\figw]{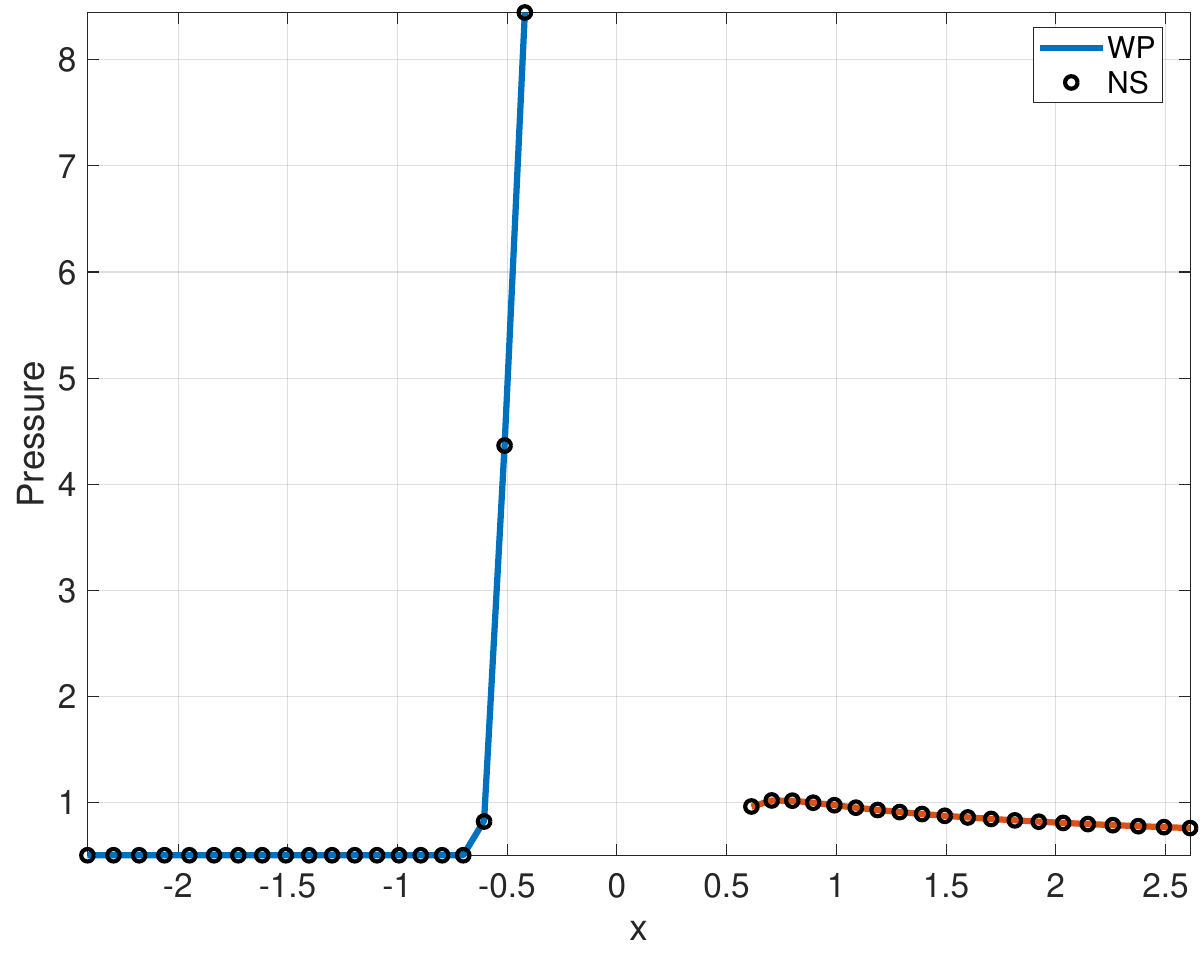}}
\caption{Stagnation-line pressure.}
\end{subfigure}
\par\smallskip
\begin{subfigure}[t]{0.4930\figw}\centering
\raisebox{0.0000\figw}{\includegraphics[width=0.4930\figw]{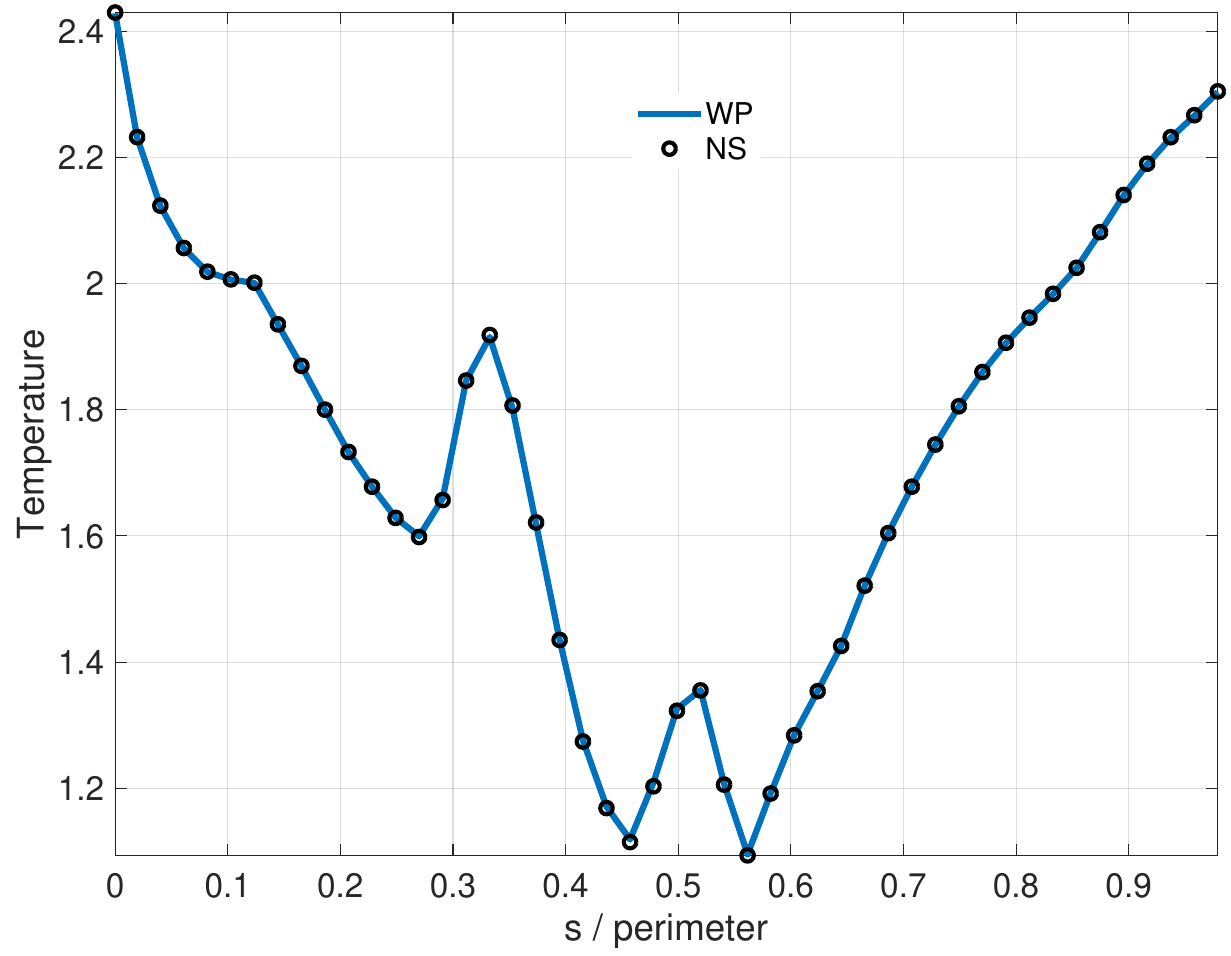}}
\caption{Wall temperature.}
\end{subfigure}\hspace{0.0200\figw}%
\begin{subfigure}[t]{0.4820\figw}\centering
\raisebox{0.0000\figw}{\includegraphics[width=0.4820\figw]{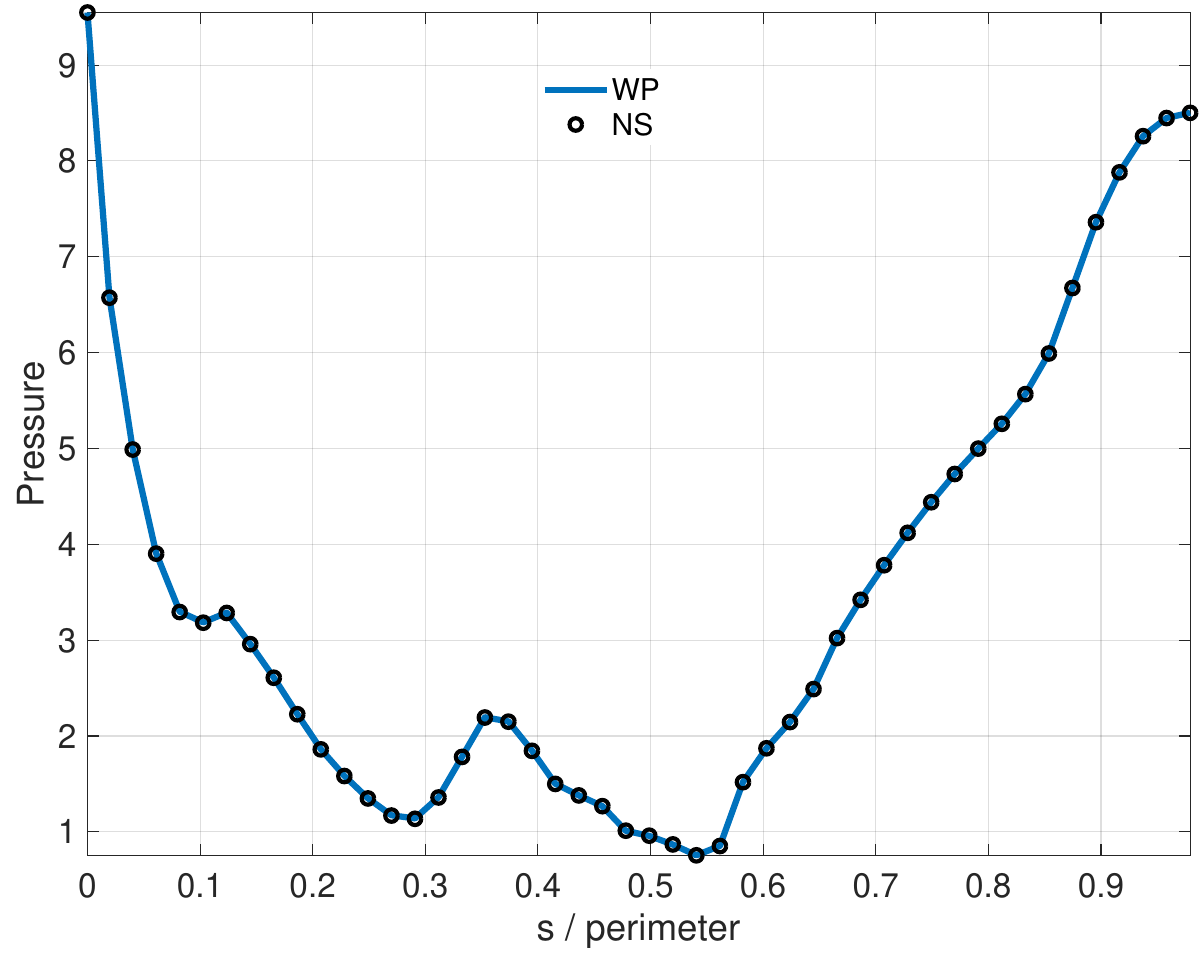}}
\caption{Wall pressure.}
\end{subfigure}
\caption{X38 section at $\varepsilon=10^{-4}$: temperature and pressure along the stagnation line (a,b) and on the wall versus the normalized arc length (c,d). Lines denote IWPD and symbols NS.}\label{fig:x38_kn1em4_profiles}
\end{figure}

\clearpage
\section{Additional verification: Apollo section}\label{sec:apollo}
The Apollo section is normalized by its maximum diameter and placed at an $18.5^\circ$ angle of attack. The body-fitted O-grids have outer radius $3.9$ at $\varepsilon=1$ and $2.6$ at $10^{-2}$ and $10^{-4}$. Walls are stationary and diffuse; the incoming far-field distribution is the free-stream Maxwellian. The model, viscosity law and residual definitions are those of the main text.

\begin{table}[H]\centering\small
\caption{Apollo discretizations and controls. The CFL column gives explicit/implicit values; the velocity domain is $[-9,9]^2$.}\label{tab:apollo_settings}
\resizebox{\linewidth}{!}{\begin{tabular}{llllll}\toprule
Case & Regime & Physical mesh & Velocity grid & IWPD control & CFL\\\midrule
Apollo   & $\varepsilon=1$, $\Ma=5$ & $64\times38$ O-grid & $96^2$ on $[-9,9]^2$ & W1/PR2/P1 & 0.2/1000\\
Apollo   & $\varepsilon=10^{-2}$, $10^{-4}$ & $64\times24$ O-grid & $96^2$, $64^2$ on $[-9,9]^2$ & W1, W120/PR2/P1 & 0.2/1000\\
\bottomrule\end{tabular}}
\end{table}

\begin{table}[H]\centering\small
\caption{Apollo field differences from CIS at $\varepsilon=1$ and $10^{-2}$ and from NS at $10^{-4}$, using the main-text definition of $E_{L_2}$.}\label{tab:apollo_accuracy}
\resizebox{\linewidth}{!}{\begin{tabular}{llrrrr}\toprule
Case & $\varepsilon$ & Density & Speed & Temperature & Pressure\\\midrule
Apollo   & $1$       & $3.18\times10^{-5}$ & $4.43\times10^{-6}$ & $8.63\times10^{-6}$ & $1.09\times10^{-5}$\\
         & $10^{-2}$ & $4.13\times10^{-4}$ & $1.72\times10^{-4}$ & $7.32\times10^{-4}$ & $3.21\times10^{-4}$\\
         & $10^{-4}$ & $4.40\times10^{-4}$ & $4.02\times10^{-4}$ & $1.21\times10^{-3}$ & $8.47\times10^{-4}$\\
\bottomrule\end{tabular}}
\end{table}

\begin{table}[H]\centering\small
\caption{Apollo efficiency relative to CIS at the common macroscopic residual $10^{-6}$. The NS time ratio uses tolerance $10^{-7}$.}\label{tab:apollo_efficiency}
\resizebox{\linewidth}{!}{\begin{tabular}{llcrrr}\toprule
Case & Regime & IWPD control & $S_{\rm iter}$ & $S_t$ & $t_{\rm WP}/t_{\rm NS}$\\\midrule
Apollo   & $\varepsilon=1$       & W1/PR2/P1   & 1.36     & 0.72    & --\\
Apollo   & $\varepsilon=10^{-2}$ & W1/PR2/P1   & 26.56    & 10.25   & --\\
Apollo   & $\varepsilon=10^{-4}$ & W120/PR2/P1 & 36531.84 & 1417.95 & 3.99\\
\bottomrule\end{tabular}}
\end{table}

\begin{table}[H]\centering\small
\caption{Apollo BGK--Shakhov comparison at $\varepsilon=10^{-4}$ with W120/PR2/P1 and tolerance $10^{-7}$. The Shakhov factor is fitted over $[10^{-6},10^{-3}]$; for BGK it is the largest successive residual ratio below $10^{-1}$.}\label{tab:apollo_bgk}
\resizebox{\linewidth}{!}{\begin{tabular}{llrrrr}\toprule
Case & Model & Outer iterations & Factor & $t_{\rm WP}$ (s) & $t_{\rm WP}/t_{\rm NS}$\\\midrule
Apollo   & Shakhov & 19 & 0.331 & 35.47 & 3.99\\
         & BGK     & 8  & 0.051 & 15.8  & 1.78\\
\bottomrule\end{tabular}}
\end{table}

The field differences are listed in Table~\ref{tab:apollo_accuracy}. The profiles show the asymmetric shock layer associated with the angle of attack. The rarefied distributions contain incident and wall-reflected populations, with two lobes in the wake. Across the sampled X38 and Apollo distributions, the IWPD--CIS differences at $\varepsilon=1$ range from $2.89\times10^{-6}$ to $3.23\times10^{-5}$.

At $\varepsilon=10^{-4}$, IWPD reaches the macroscopic residual $10^{-6}$ in $19$ outer iterations, compared with $573\,864$ for CIS; NS requires $251$ steps. The IWPD and NS times are $35.47$ and $8.89$~s. Tables~\ref{tab:apollo_efficiency} and~\ref{tab:apollo_bgk} summarize the efficiency and collision-model comparisons. The corresponding BGK run reaches tolerance in $8$ outer iterations, supporting the main-text interpretation of the explicit collision remainder.

\begin{figure}[H]
\centering\singlespacing\setlength{\figw}{.65\linewidth}
\begin{subfigure}[t]{0.4814\figw}\centering
\raisebox{0.0000\figw}{\includegraphics[width=0.4814\figw]{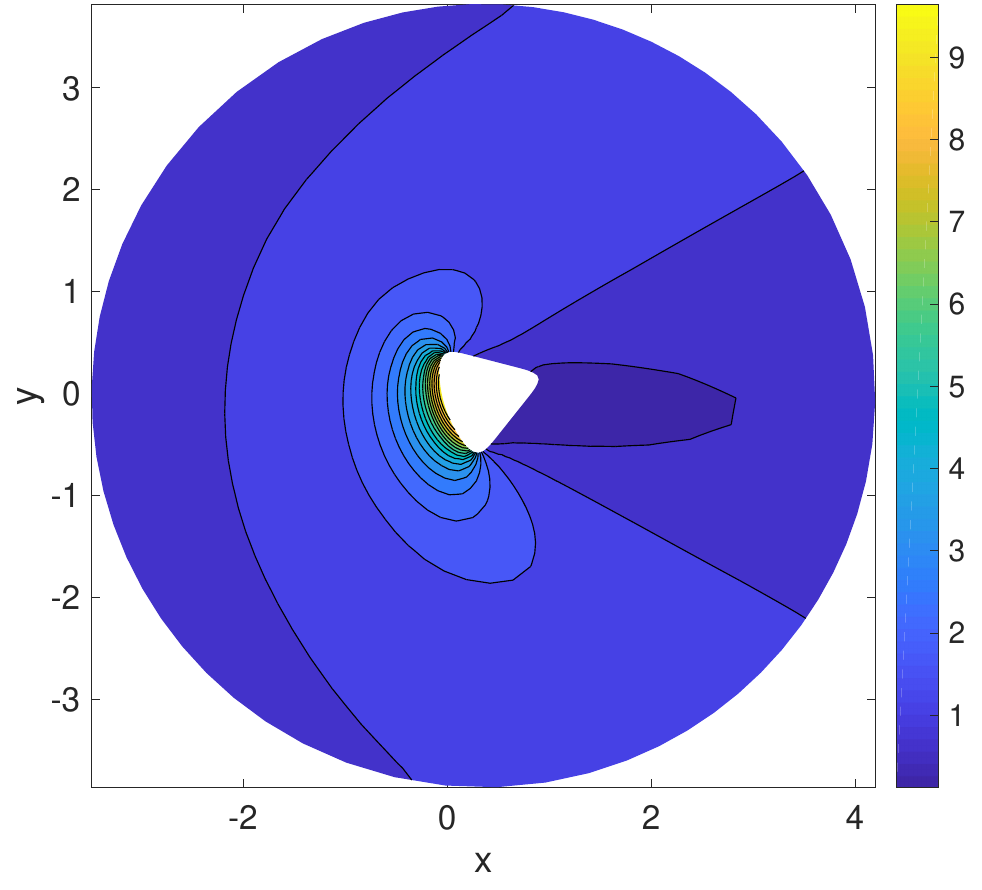}}
\caption{Density.}
\end{subfigure}\hspace{0.0200\figw}%
\begin{subfigure}[t]{0.4936\figw}\centering
\raisebox{0.0000\figw}{\includegraphics[width=0.4936\figw]{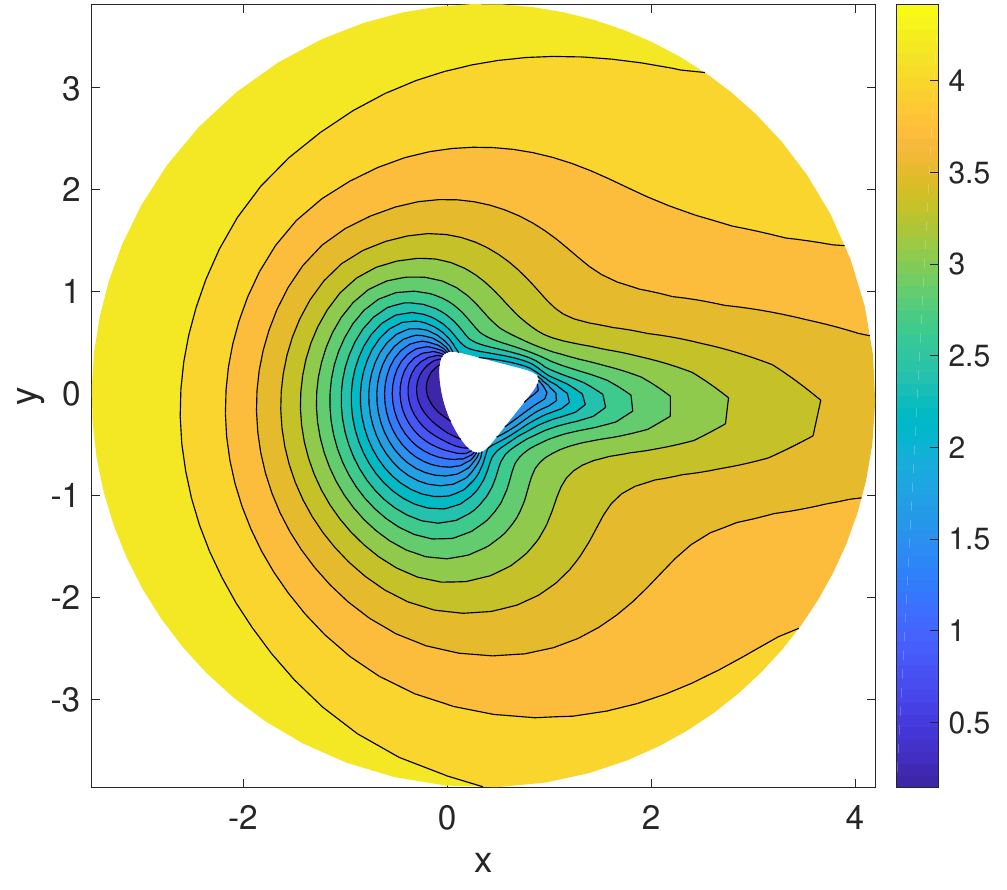}}
\caption{Speed.}
\end{subfigure}
\par\smallskip
\begin{subfigure}[t]{0.4895\figw}\centering
\raisebox{0.0000\figw}{\includegraphics[width=0.4895\figw]{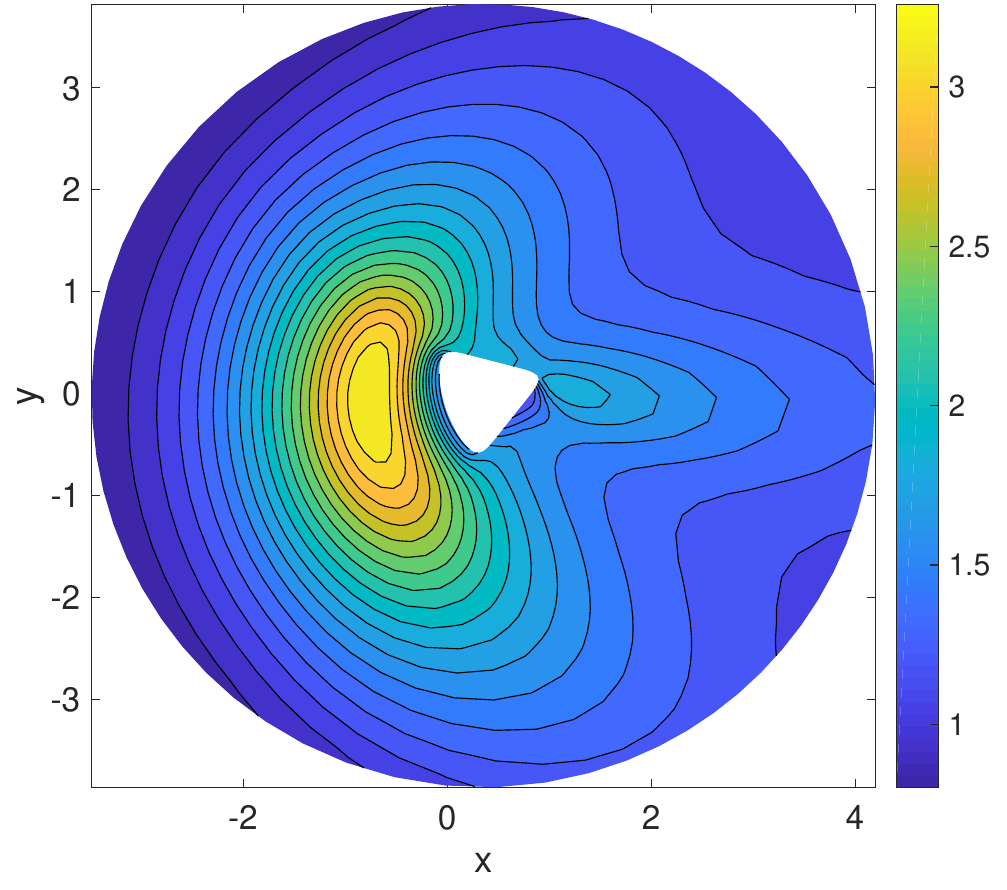}}
\caption{Temperature.}
\end{subfigure}\hspace{0.0200\figw}%
\begin{subfigure}[t]{0.4855\figw}\centering
\raisebox{0.0000\figw}{\includegraphics[width=0.4855\figw]{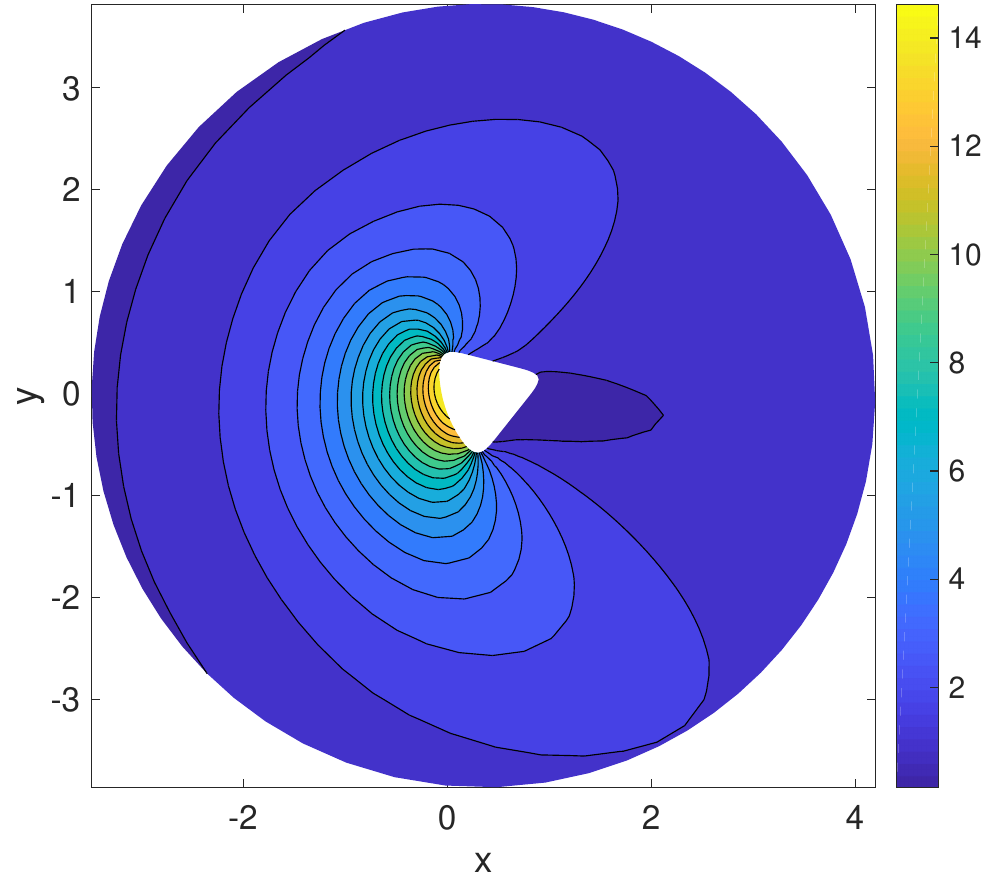}}
\caption{Pressure.}
\end{subfigure}
\caption{Apollo section at $\mathrm{Ma}=5$, $18.5^\circ$ angle of attack and $\varepsilon=1$: macroscopic fields of IWPD (filled contours) and CIS (lines) on common levels.}\label{fig:apollo_kn1_fields}
\end{figure}

\begin{figure}[H]
\centering\singlespacing\setlength{\figw}{.65\linewidth}
\begin{subfigure}[t]{0.4895\figw}\centering
\raisebox{0.0000\figw}{\includegraphics[width=0.4895\figw]{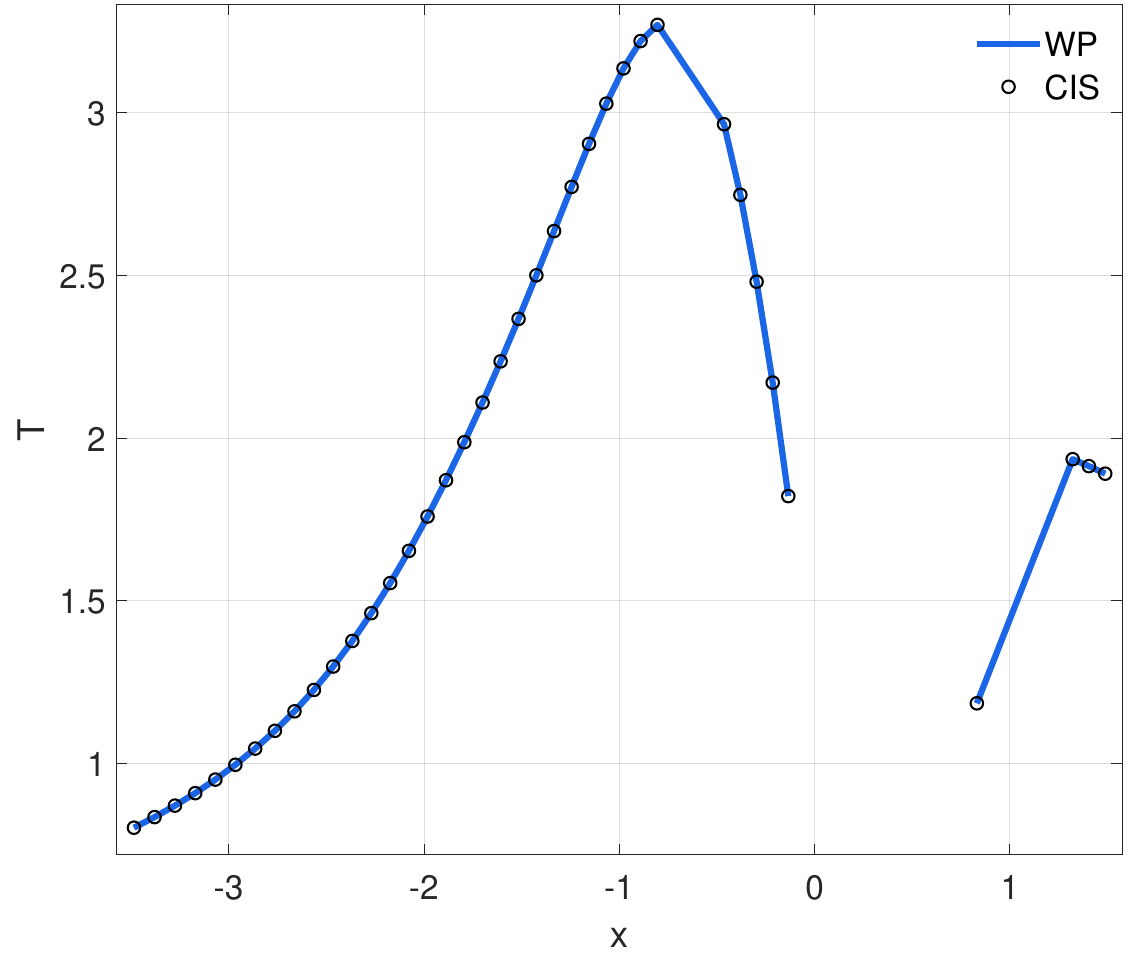}}
\caption{Stagnation-line temperature.}
\end{subfigure}\hspace{0.0200\figw}%
\begin{subfigure}[t]{0.4855\figw}\centering
\raisebox{0.0000\figw}{\includegraphics[width=0.4855\figw]{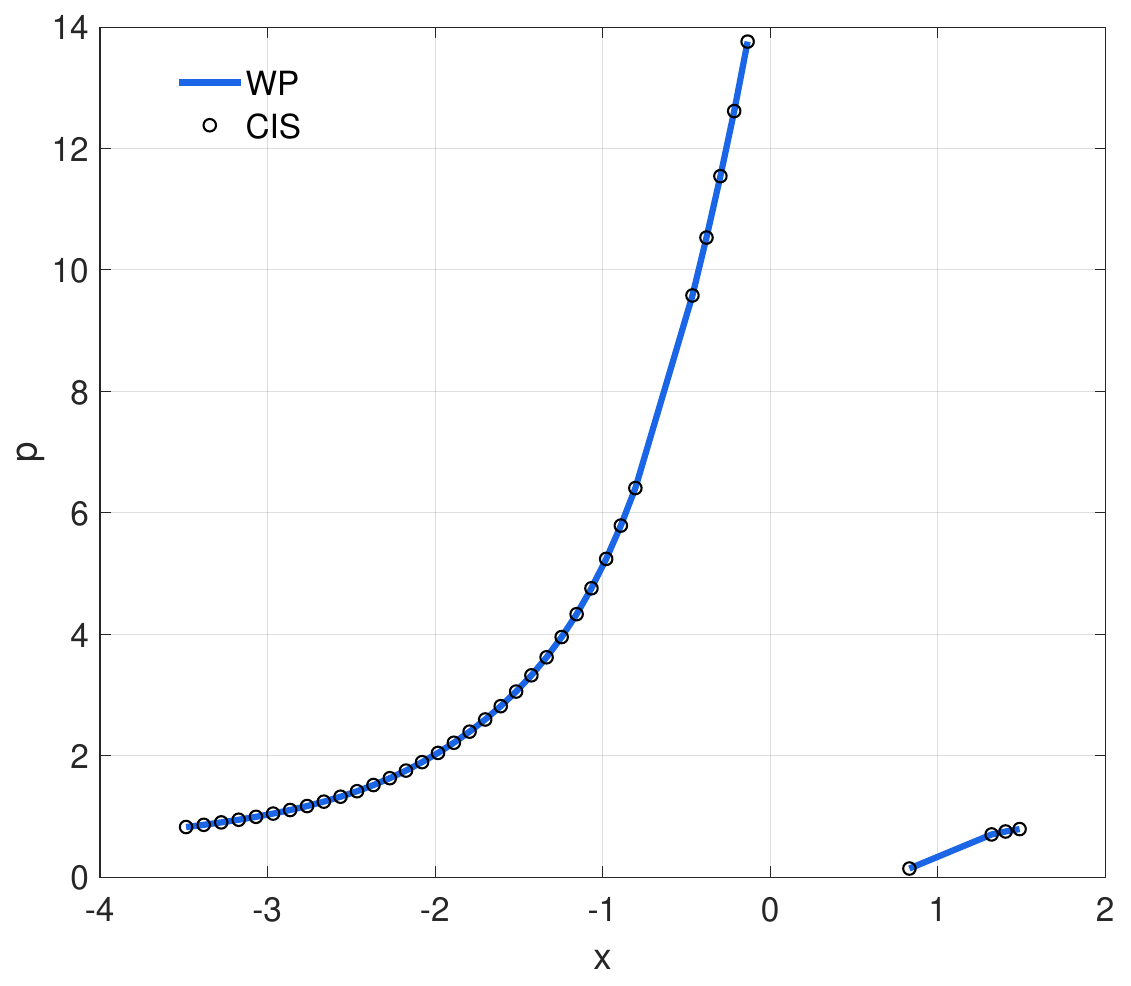}}
\caption{Stagnation-line pressure.}
\end{subfigure}
\par\smallskip
\begin{subfigure}[t]{0.4911\figw}\centering
\raisebox{0.0000\figw}{\includegraphics[width=0.4911\figw]{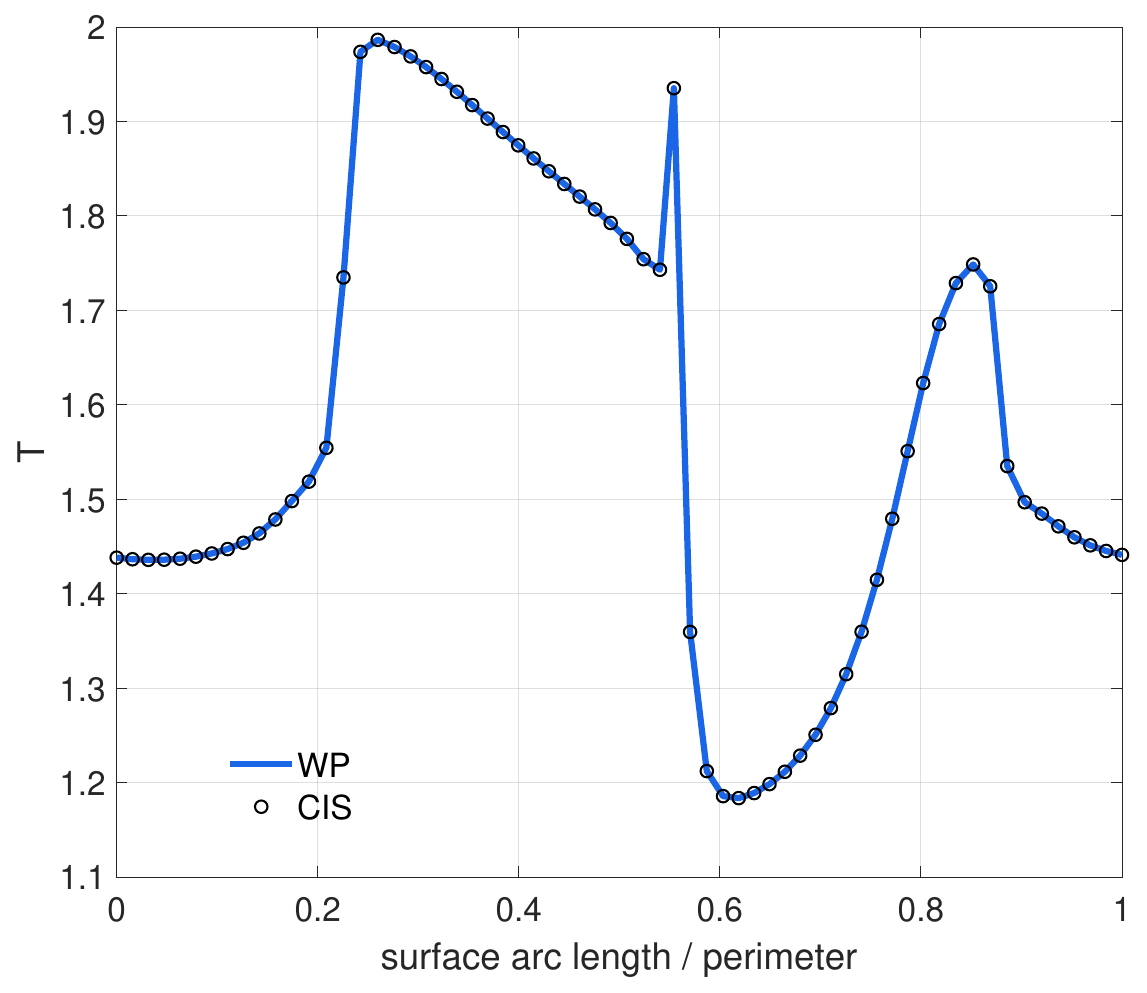}}
\caption{Wall temperature.}
\end{subfigure}\hspace{0.0200\figw}%
\begin{subfigure}[t]{0.4839\figw}\centering
\raisebox{0.0000\figw}{\includegraphics[width=0.4839\figw]{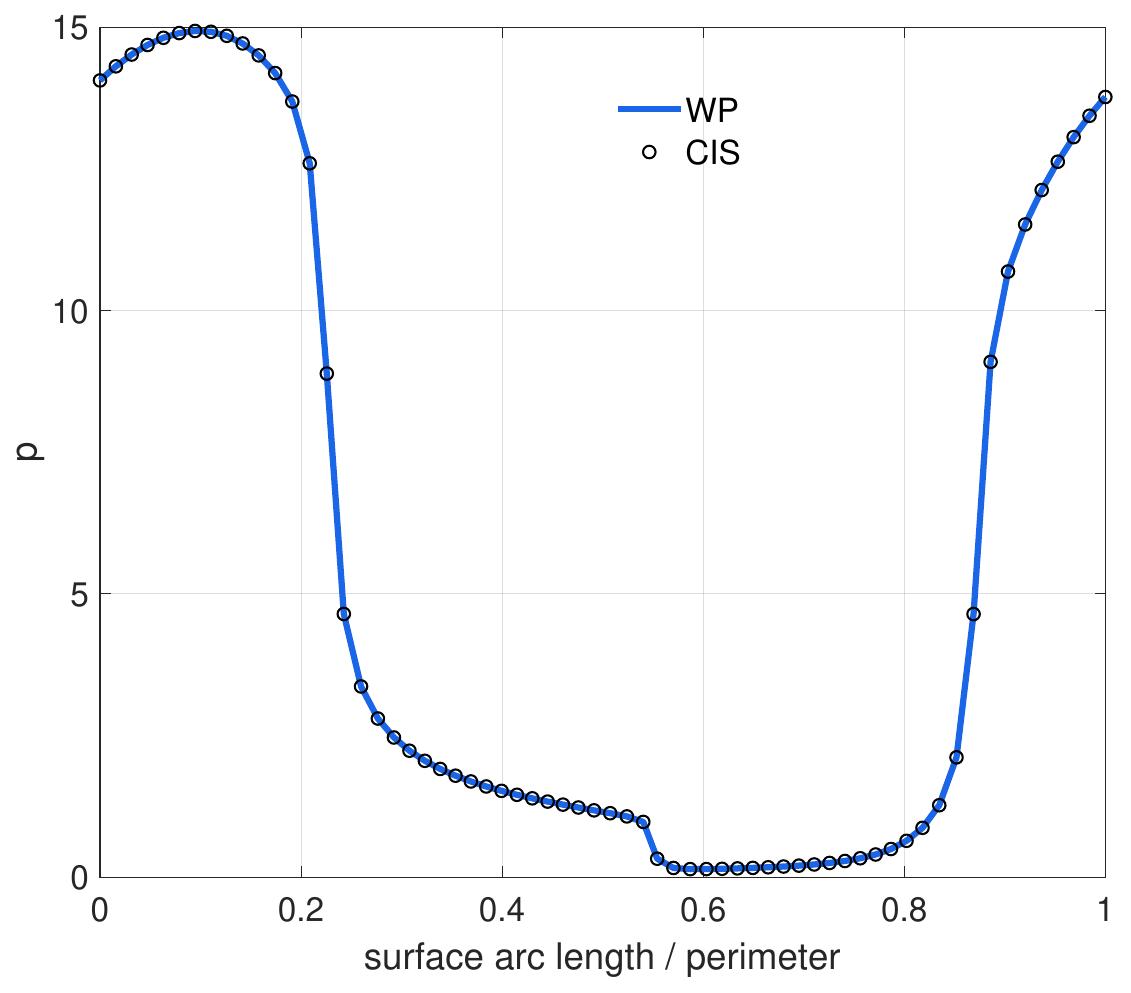}}
\caption{Wall pressure.}
\end{subfigure}
\caption{Apollo section at $\varepsilon=1$: temperature and pressure along the stagnation line (a,b) and on the wall versus the normalized arc length (c,d). Lines denote IWPD and symbols CIS.}\label{fig:apollo_kn1_profiles}
\end{figure}

\begin{figure}[H]
\centering\singlespacing\setlength{\figw}{.65\linewidth}
\begin{subfigure}[t]{0.4875\figw}\centering
\raisebox{0.0000\figw}{\includegraphics[width=0.4875\figw]{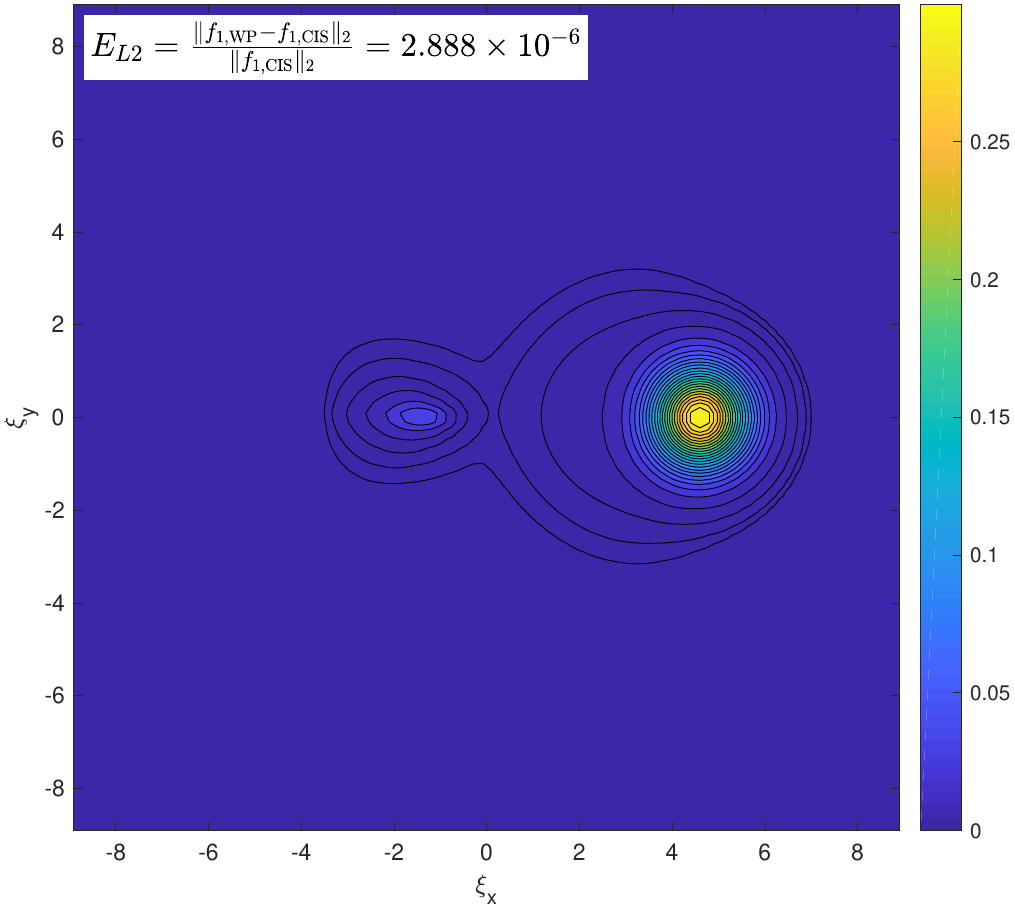}}
\caption{Sample 2, $(x,y)=(-2.66,0.05)$.}
\end{subfigure}\hspace{0.0200\figw}%
\begin{subfigure}[t]{0.4875\figw}\centering
\raisebox{0.0000\figw}{\includegraphics[width=0.4875\figw]{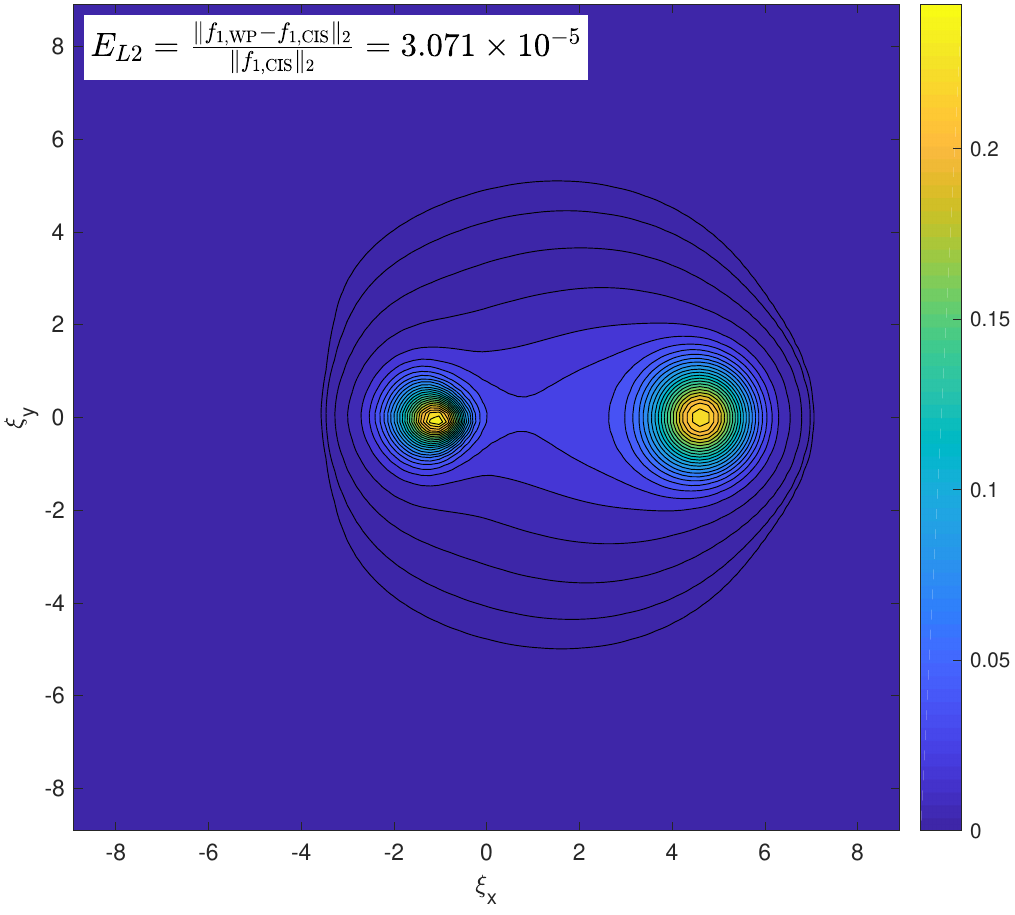}}
\caption{Sample 4, $(x,y)=(-0.89,0.01)$.}
\end{subfigure}
\par\smallskip
\begin{subfigure}[t]{0.4841\figw}\centering
\raisebox{0.0013\figw}{\includegraphics[width=0.4841\figw]{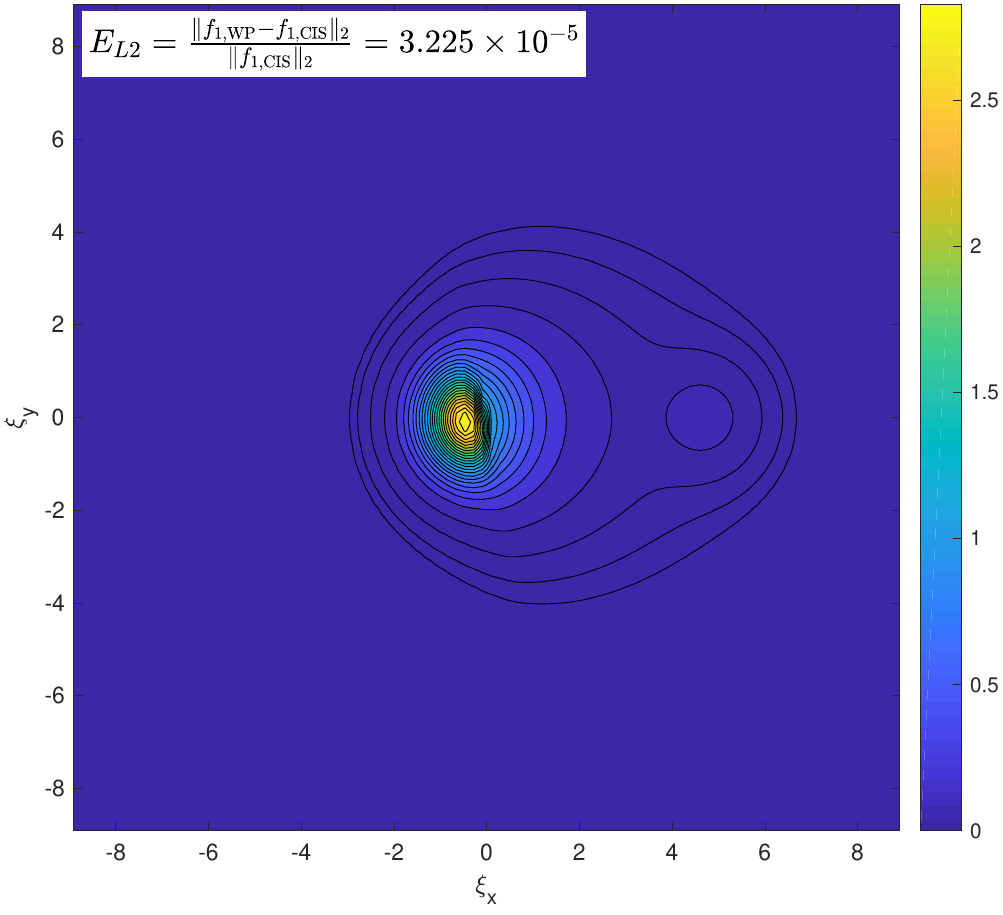}}
\caption{Sample 5, $(x,y)=(-0.06,0.01)$.}
\end{subfigure}\hspace{0.0200\figw}%
\begin{subfigure}[t]{0.4909\figw}\centering
\raisebox{0.0000\figw}{\includegraphics[width=0.4909\figw]{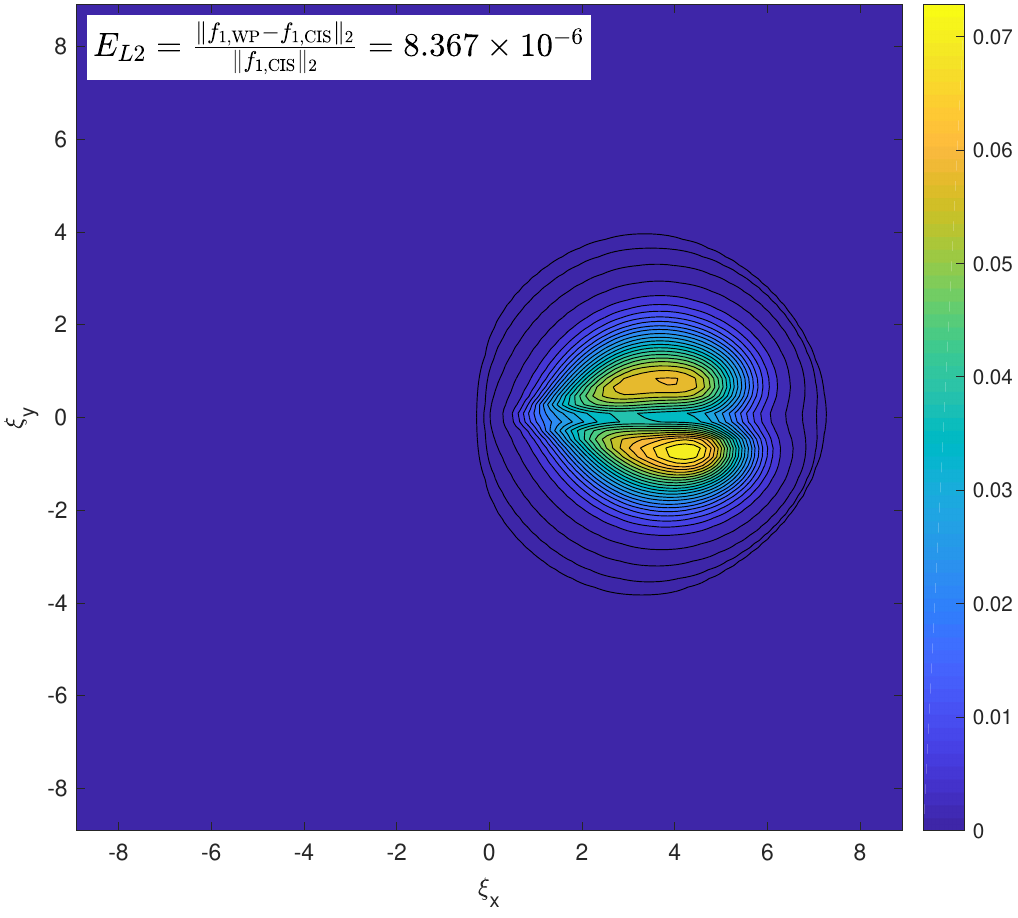}}
\caption{Sample 10, $(x,y)=(4.20,-0.08)$.}
\end{subfigure}
\caption{Apollo section at $\varepsilon=1$: reduced distributions $f_1(\xi_x,\xi_y)$ at three points ahead of the body (a--c) and one point in the wake (d). IWPD is shown by filled contours and CIS by lines.}\label{fig:apollo_kn1_dist}
\end{figure}

\begin{figure}[H]
\centering\singlespacing\setlength{\figw}{.65\linewidth}
\begin{subfigure}[t]{0.4855\figw}\centering
\raisebox{0.0000\figw}{\includegraphics[width=0.4855\figw]{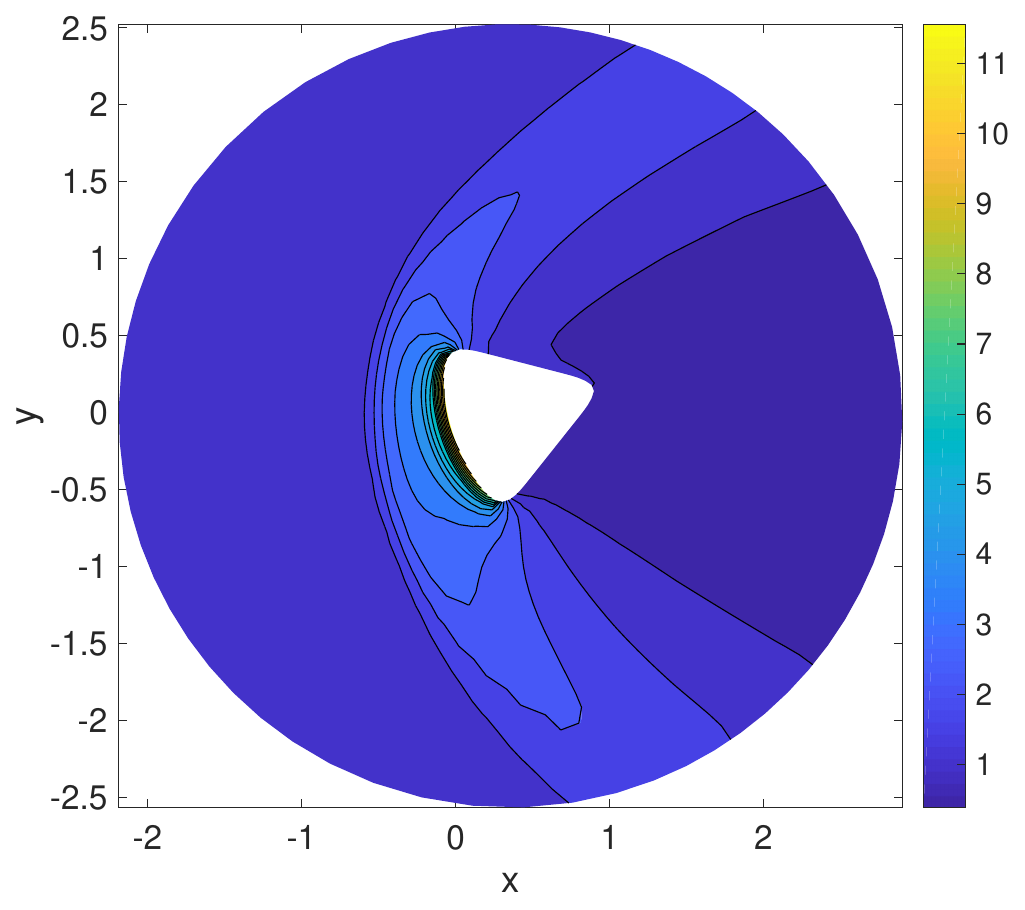}}
\caption{Density.}
\end{subfigure}\hspace{0.0200\figw}%
\begin{subfigure}[t]{0.4895\figw}\centering
\raisebox{0.0000\figw}{\includegraphics[width=0.4895\figw]{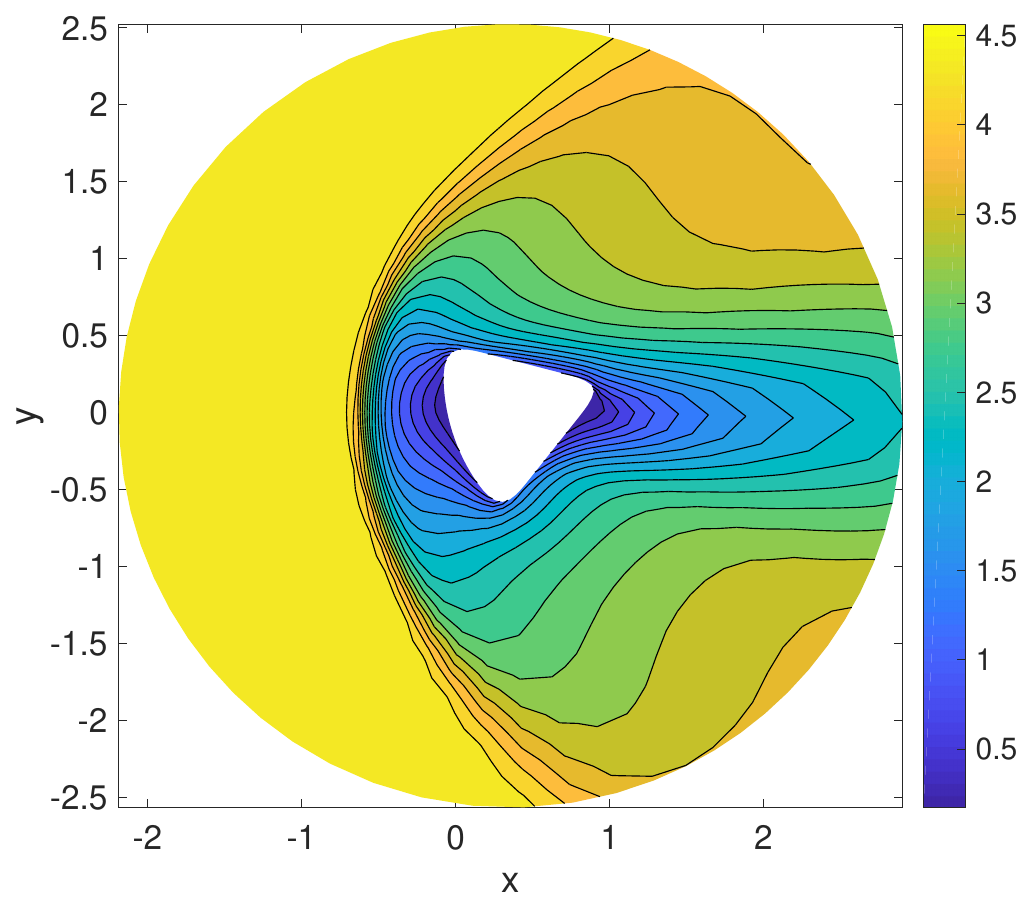}}
\caption{Speed.}
\end{subfigure}
\par\smallskip
\begin{subfigure}[t]{0.4792\figw}\centering
\raisebox{0.0116\figw}{\includegraphics[width=0.4792\figw]{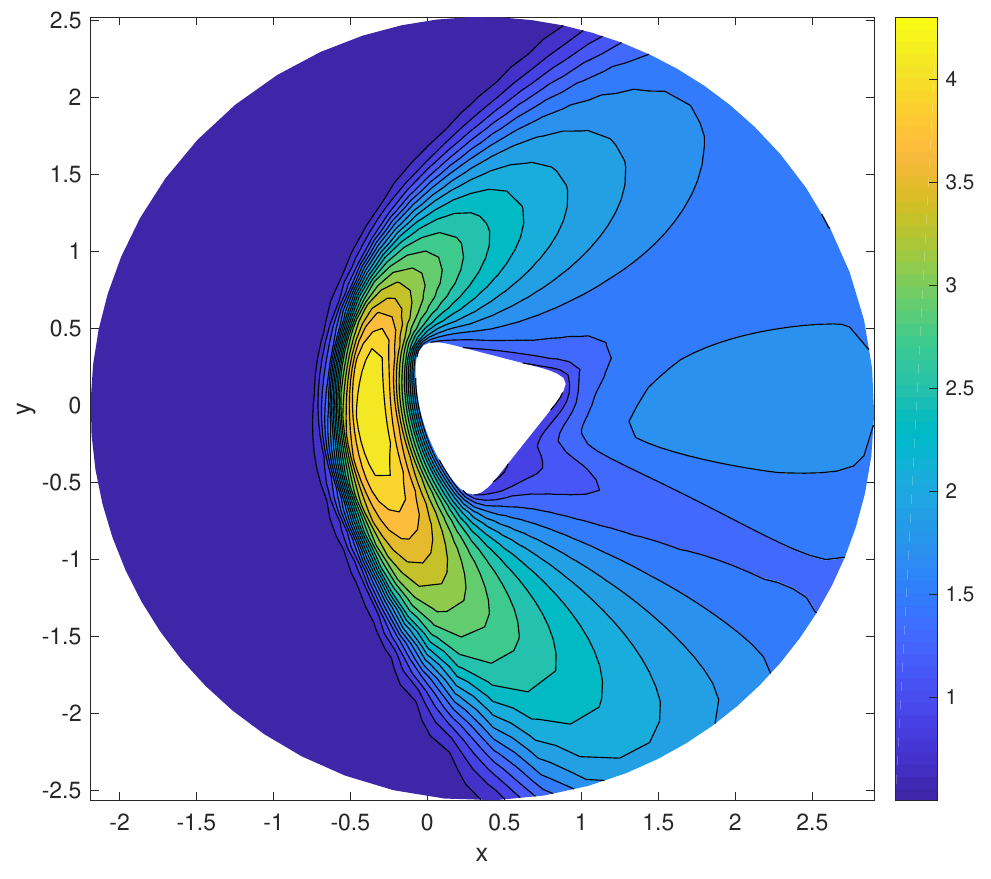}}
\caption{Temperature.}
\end{subfigure}\hspace{0.0200\figw}%
\begin{subfigure}[t]{0.4958\figw}\centering
\raisebox{0.0000\figw}{\includegraphics[width=0.4958\figw]{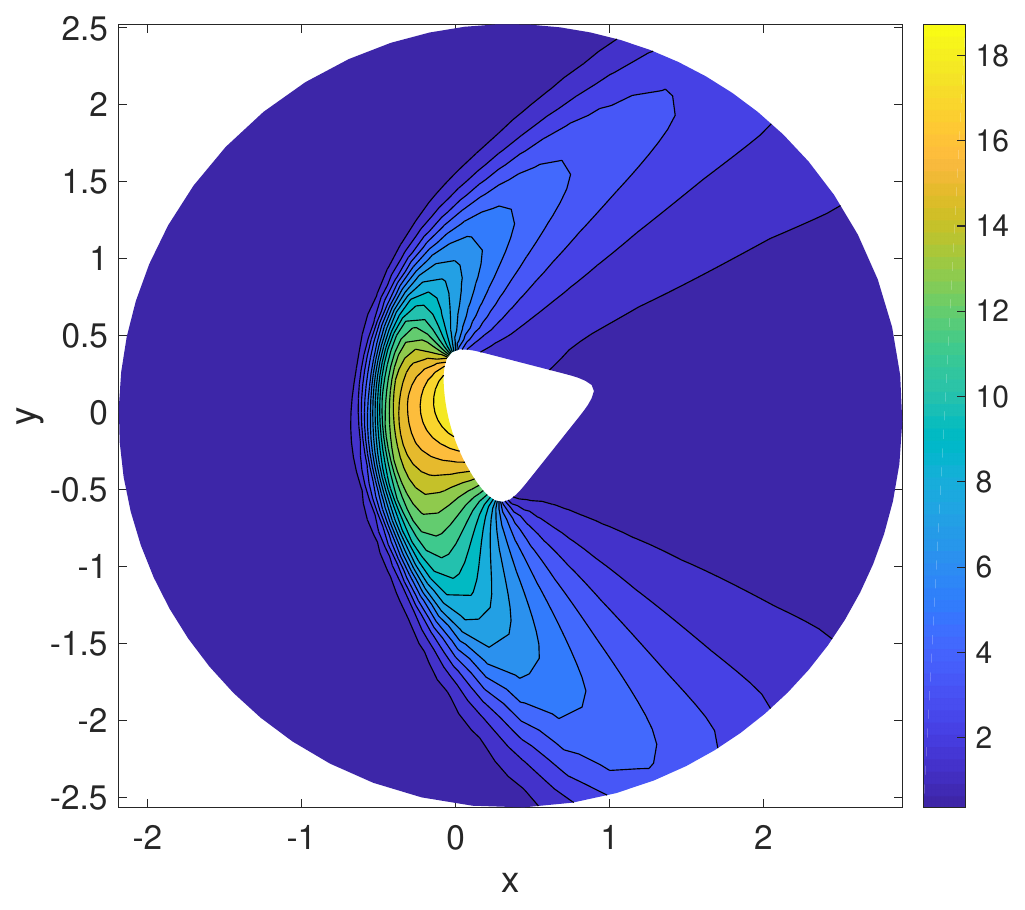}}
\caption{Pressure.}
\end{subfigure}
\caption{Apollo section at $\mathrm{Ma}=5$ and $\varepsilon=10^{-2}$: macroscopic fields of IWPD (filled contours) and CIS (lines) on common levels.}\label{fig:apollo_kn1em2_fields}
\end{figure}

\begin{figure}[H]
\centering\singlespacing\setlength{\figw}{.65\linewidth}
\begin{subfigure}[t]{0.4910\figw}\centering
\raisebox{0.0000\figw}{\includegraphics[width=0.4910\figw]{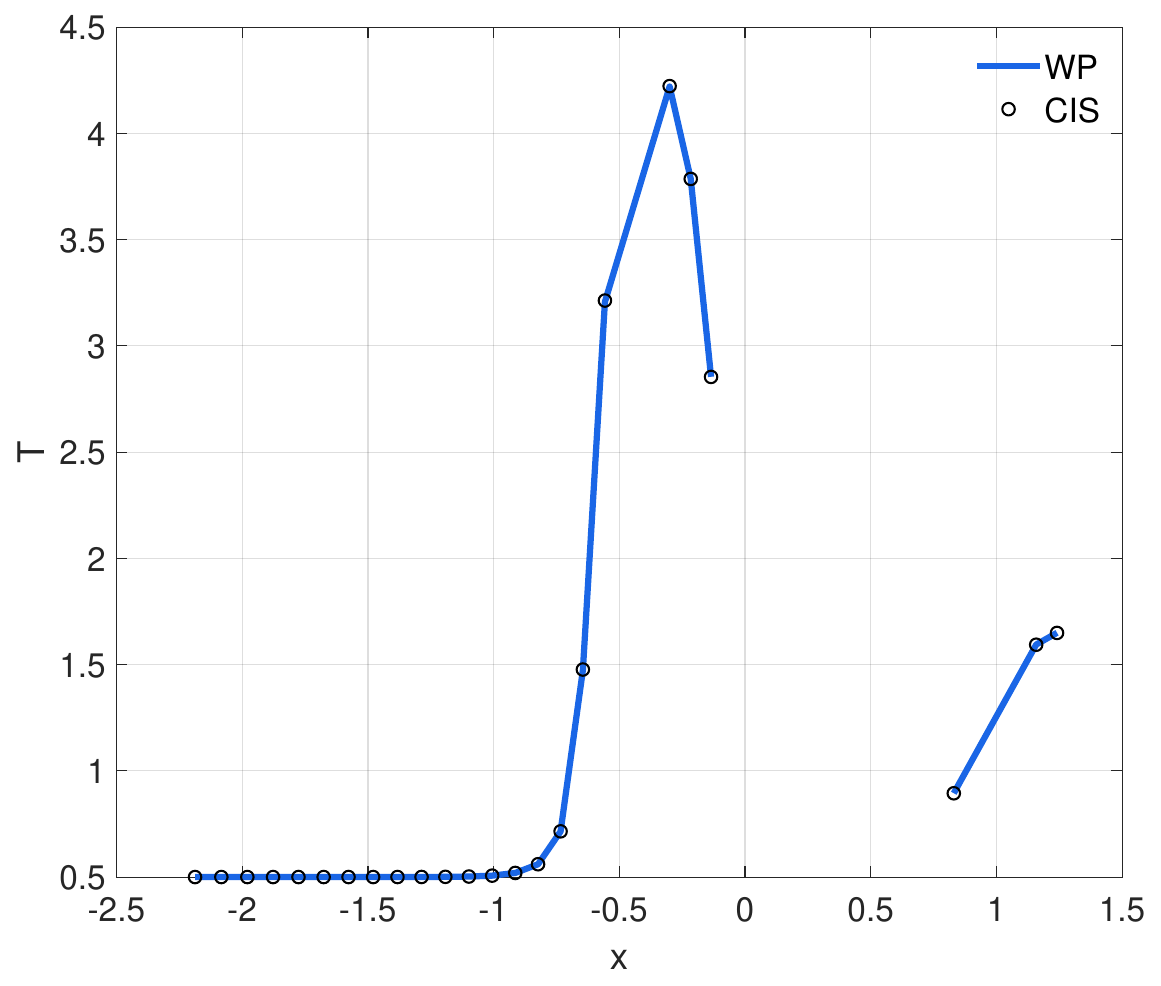}}
\caption{Stagnation-line temperature.}
\end{subfigure}\hspace{0.0200\figw}%
\begin{subfigure}[t]{0.4840\figw}\centering
\raisebox{0.0000\figw}{\includegraphics[width=0.4840\figw]{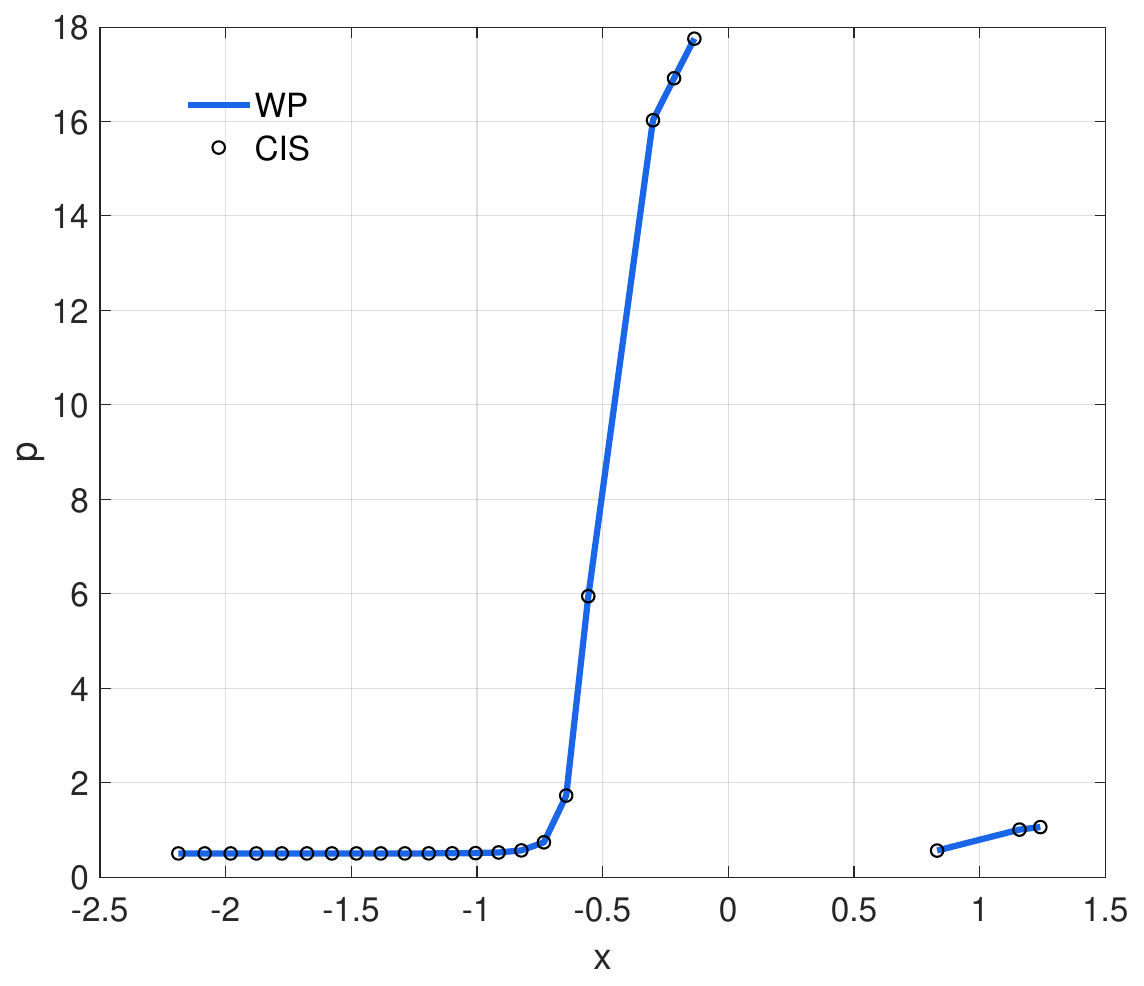}}
\caption{Stagnation-line pressure.}
\end{subfigure}
\par\smallskip
\begin{subfigure}[t]{0.4911\figw}\centering
\raisebox{0.0000\figw}{\includegraphics[width=0.4911\figw]{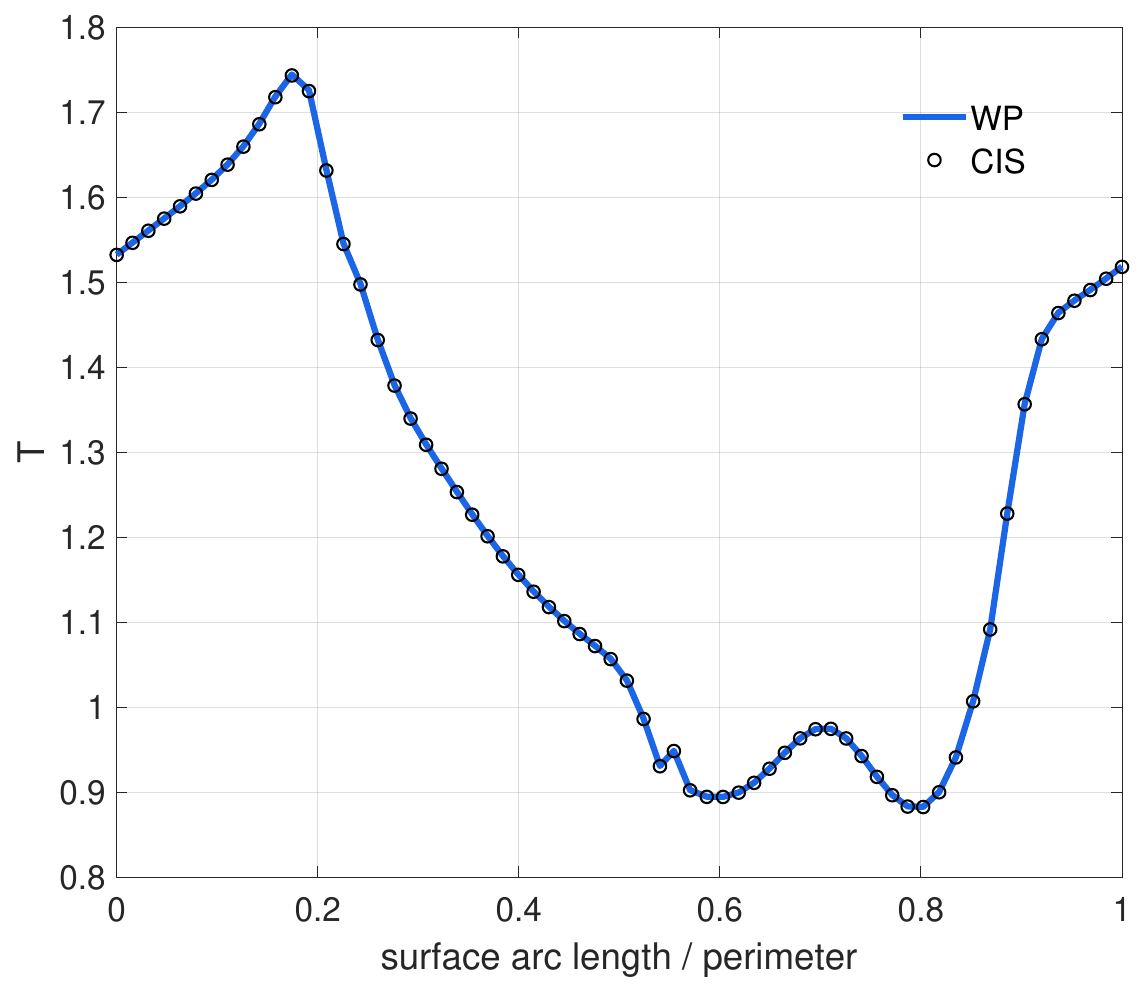}}
\caption{Wall temperature.}
\end{subfigure}\hspace{0.0200\figw}%
\begin{subfigure}[t]{0.4839\figw}\centering
\raisebox{0.0000\figw}{\includegraphics[width=0.4839\figw]{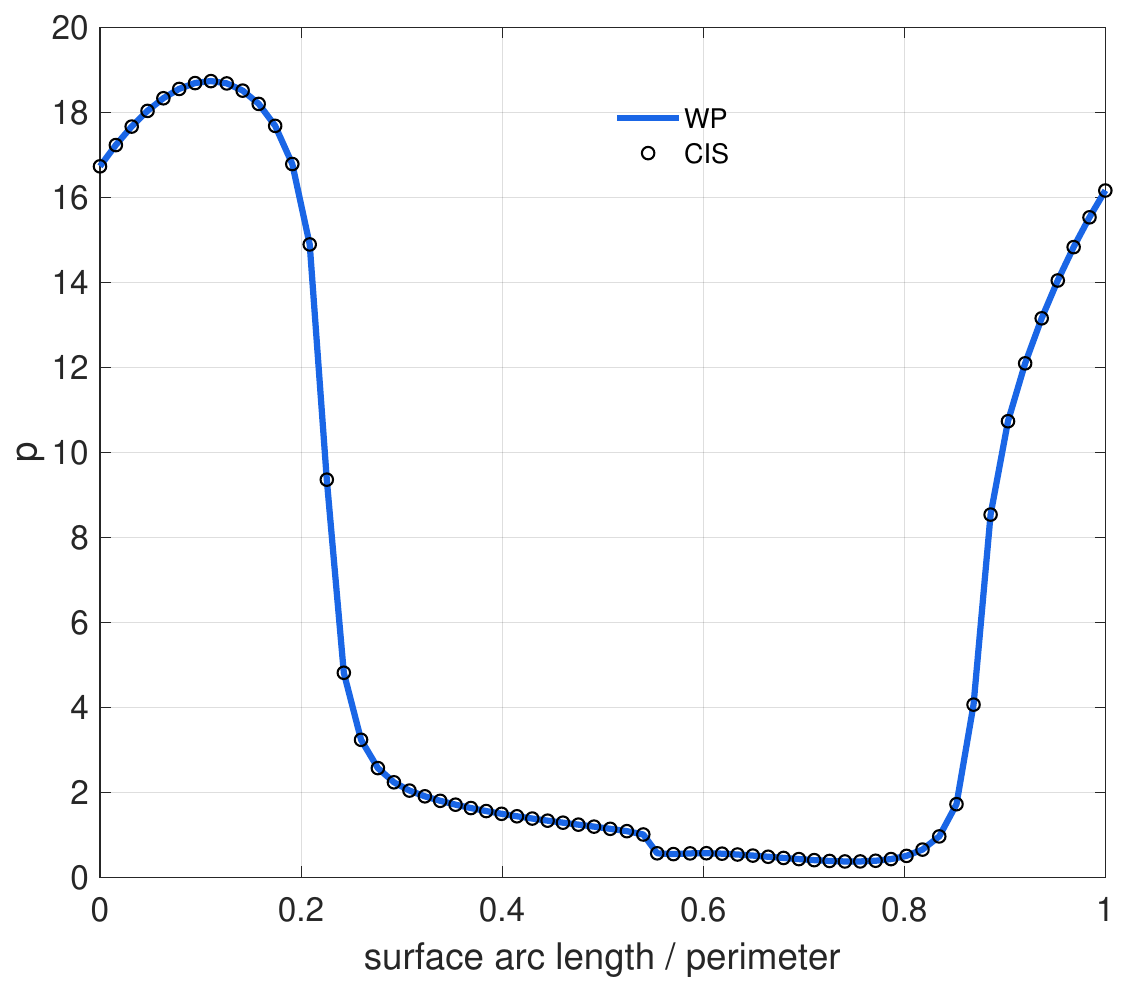}}
\caption{Wall pressure.}
\end{subfigure}
\caption{Apollo section at $\varepsilon=10^{-2}$: temperature and pressure along the stagnation line (a,b) and on the wall (c,d). Lines denote IWPD and symbols CIS.}\label{fig:apollo_kn1em2_profiles}
\end{figure}

\begin{figure}[H]
\centering\singlespacing\setlength{\figw}{.65\linewidth}
\begin{subfigure}[t]{0.4910\figw}\centering
\raisebox{0.0000\figw}{\includegraphics[width=0.4910\figw]{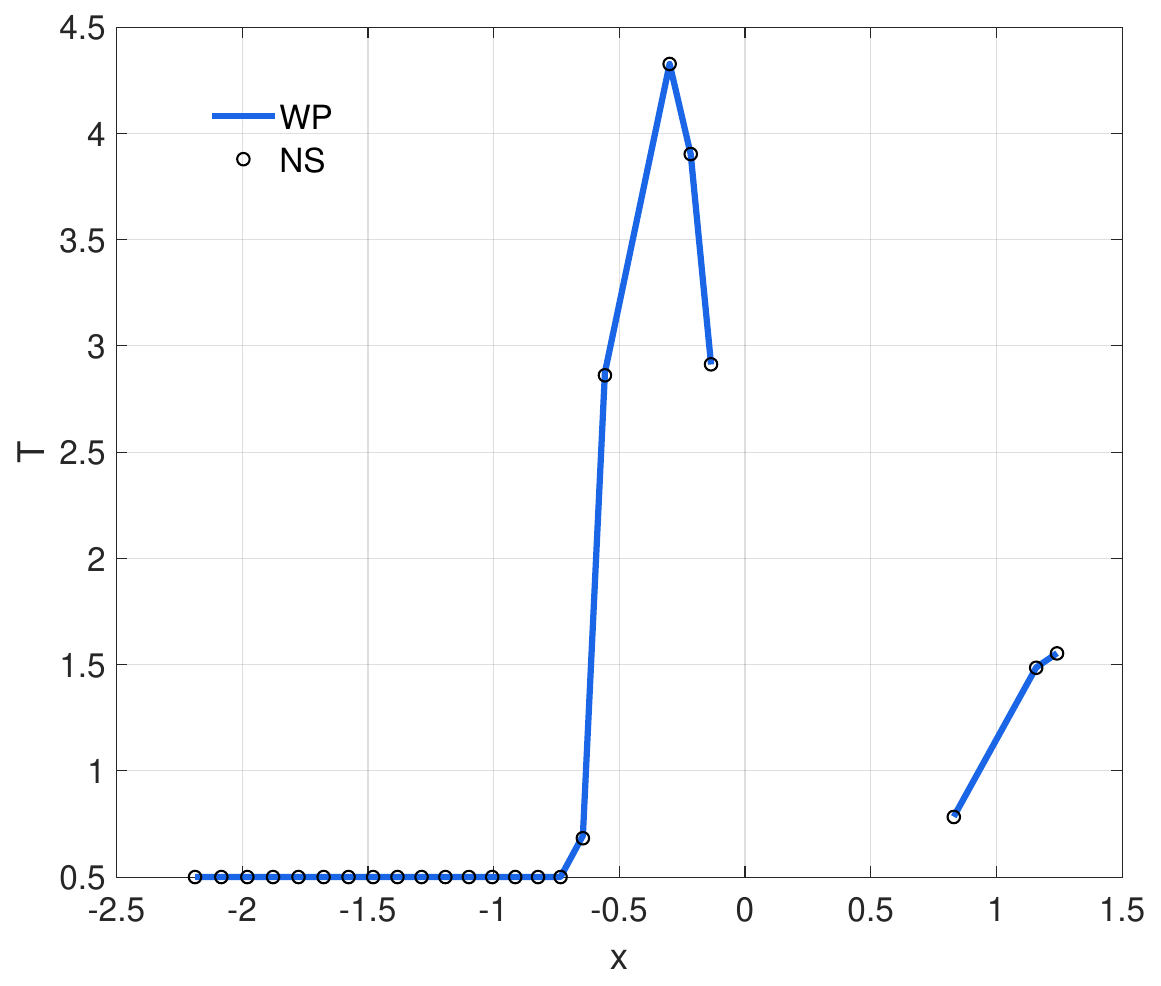}}
\caption{Stagnation-line temperature.}
\end{subfigure}\hspace{0.0200\figw}%
\begin{subfigure}[t]{0.4840\figw}\centering
\raisebox{0.0000\figw}{\includegraphics[width=0.4840\figw]{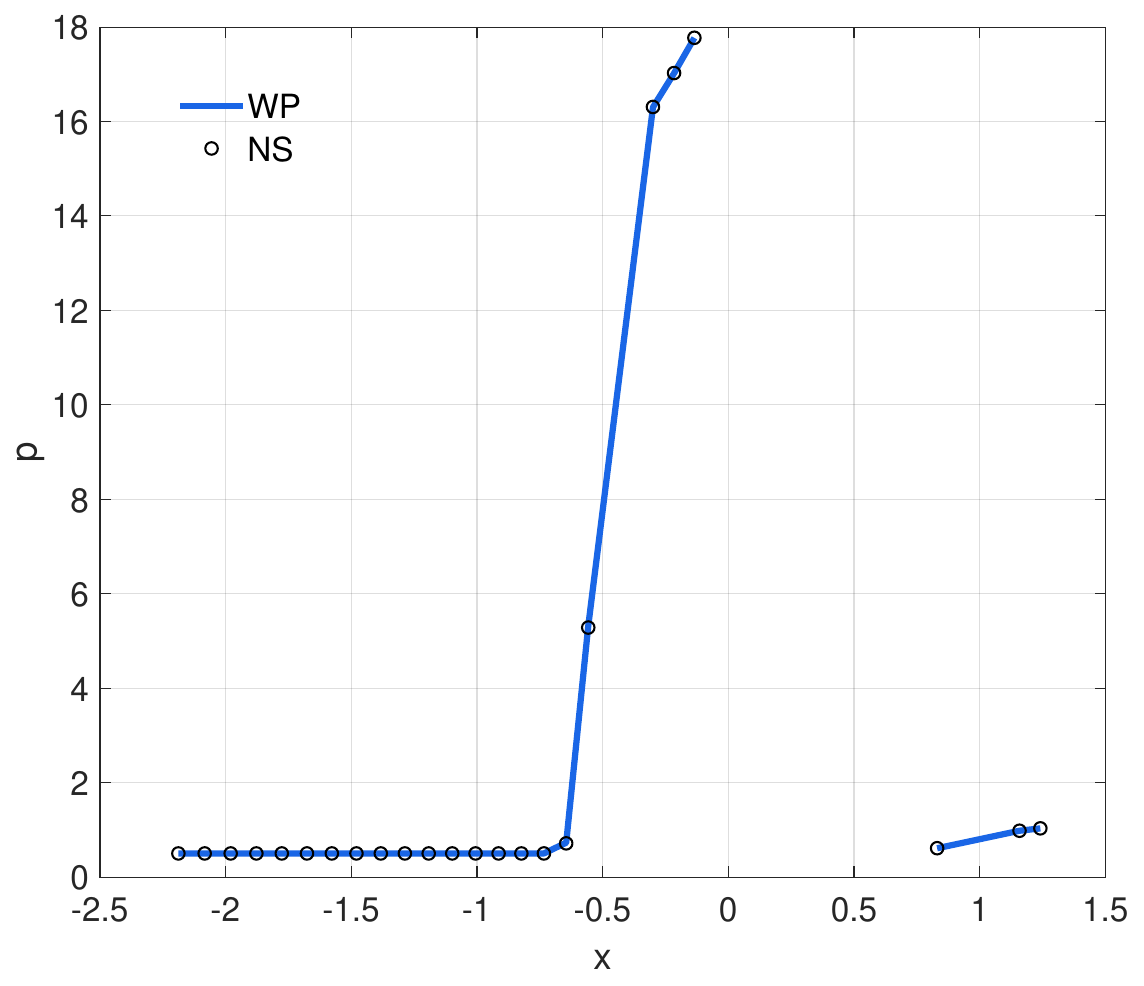}}
\caption{Stagnation-line pressure.}
\end{subfigure}
\par\smallskip
\begin{subfigure}[t]{0.4911\figw}\centering
\raisebox{0.0000\figw}{\includegraphics[width=0.4911\figw]{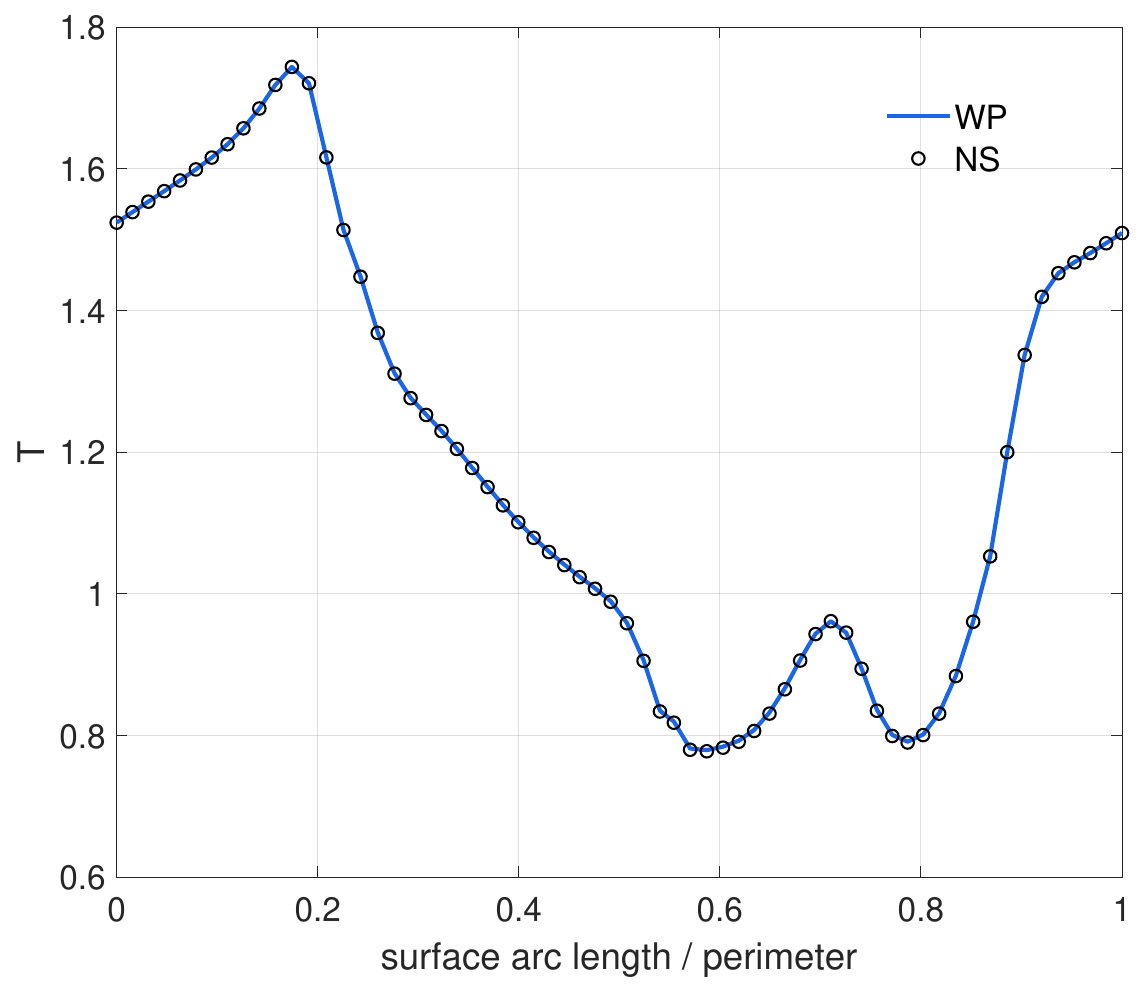}}
\caption{Wall temperature.}
\end{subfigure}\hspace{0.0200\figw}%
\begin{subfigure}[t]{0.4839\figw}\centering
\raisebox{0.0000\figw}{\includegraphics[width=0.4839\figw]{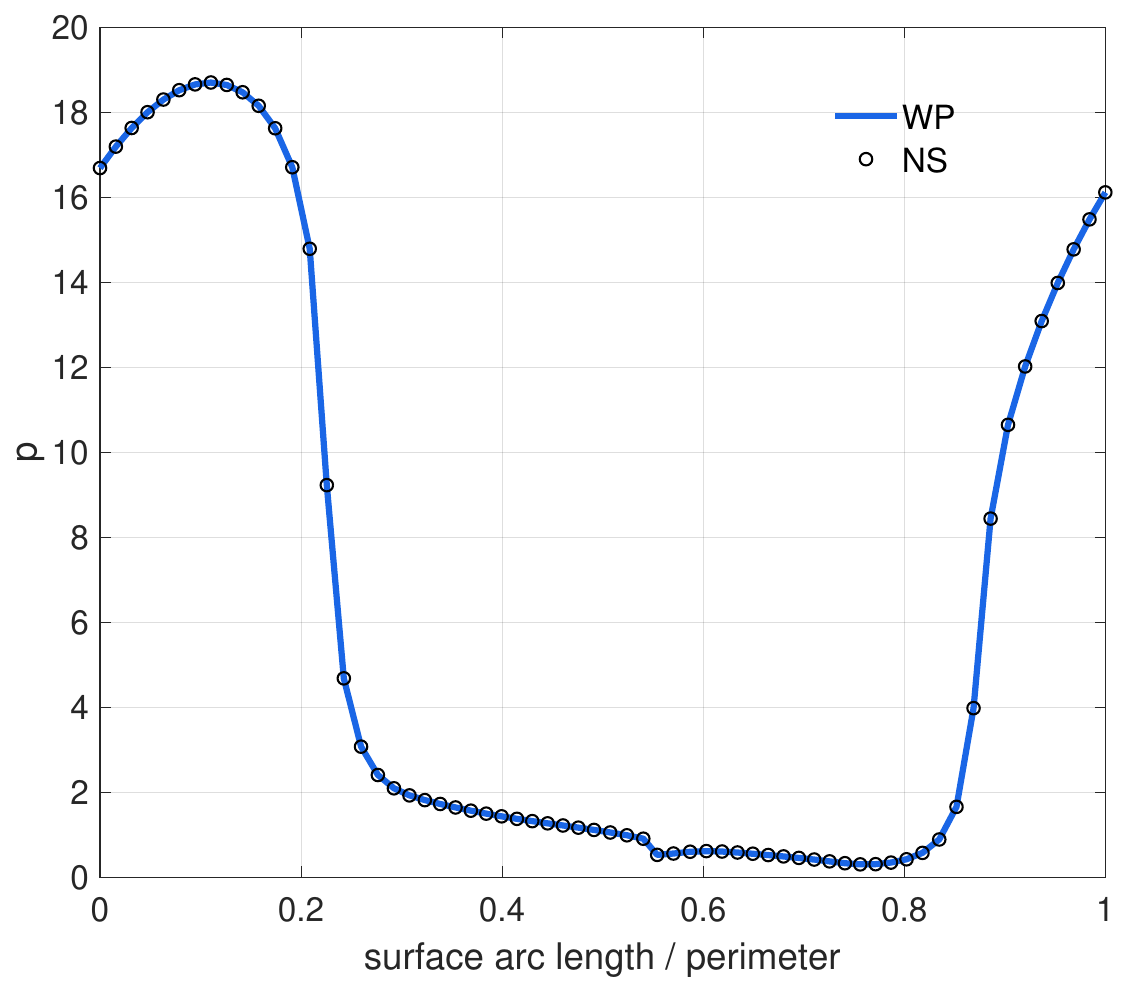}}
\caption{Wall pressure.}
\end{subfigure}
\caption{Apollo section at $\varepsilon=10^{-4}$: temperature and pressure along the stagnation line (a,b) and on the wall (c,d). Lines denote IWPD and symbols NS.}\label{fig:apollo_kn1em4_profiles}
\end{figure}

\begin{figure}[H]
\centering\singlespacing\setlength{\figw}{1.0\linewidth}
\begin{subfigure}[t]{0.4855\figw}\centering
\raisebox{0.0000\figw}{\includegraphics[width=0.4855\figw]{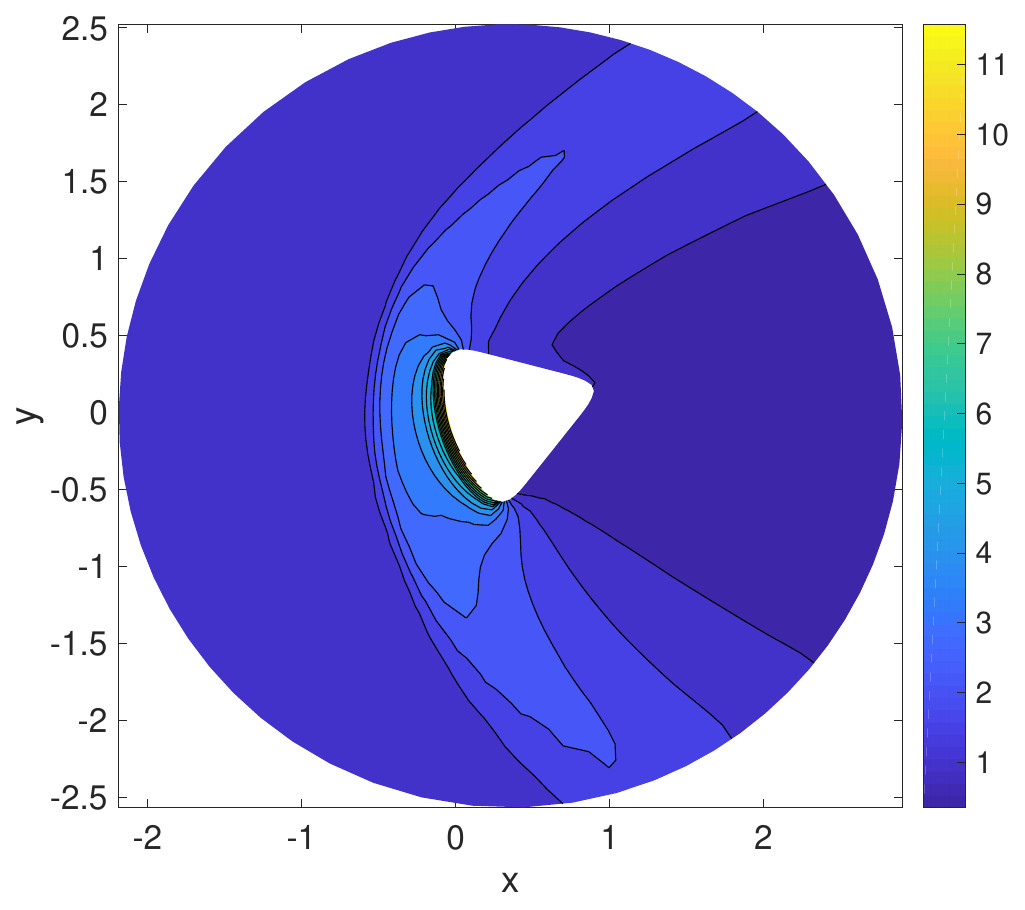}}
\caption{Density.}
\end{subfigure}\hspace{0.0200\figw}%
\begin{subfigure}[t]{0.4895\figw}\centering
\raisebox{0.0000\figw}{\includegraphics[width=0.4895\figw]{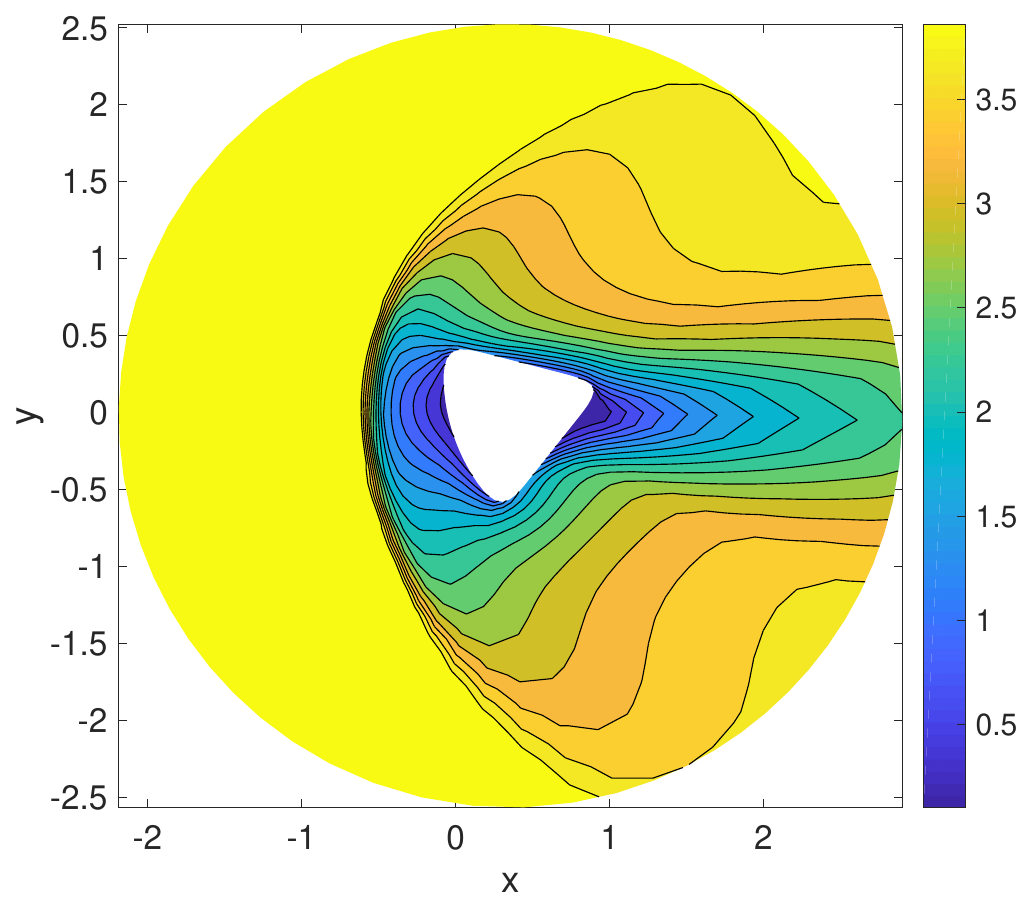}}
\caption{Speed.}
\end{subfigure}
\par\smallskip
\begin{subfigure}[t]{0.4792\figw}\centering
\raisebox{0.0116\figw}{\includegraphics[width=0.4792\figw]{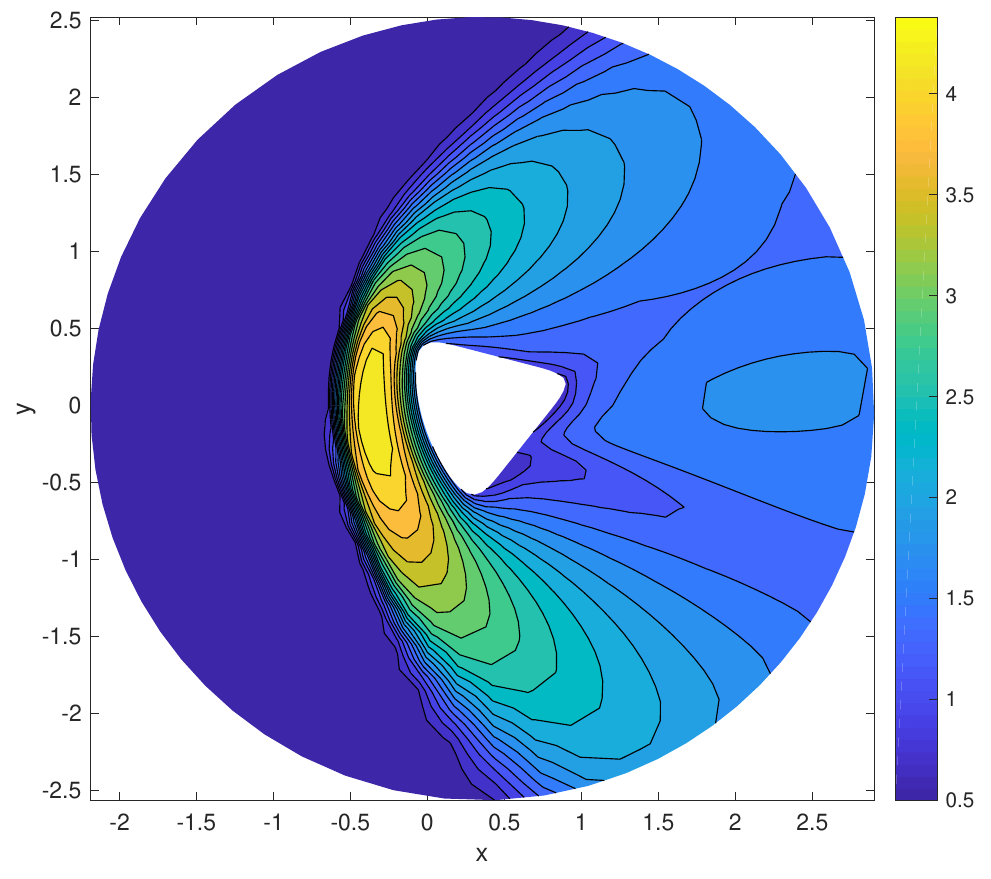}}
\caption{Temperature.}
\end{subfigure}\hspace{0.0200\figw}%
\begin{subfigure}[t]{0.4958\figw}\centering
\raisebox{0.0000\figw}{\includegraphics[width=0.4958\figw]{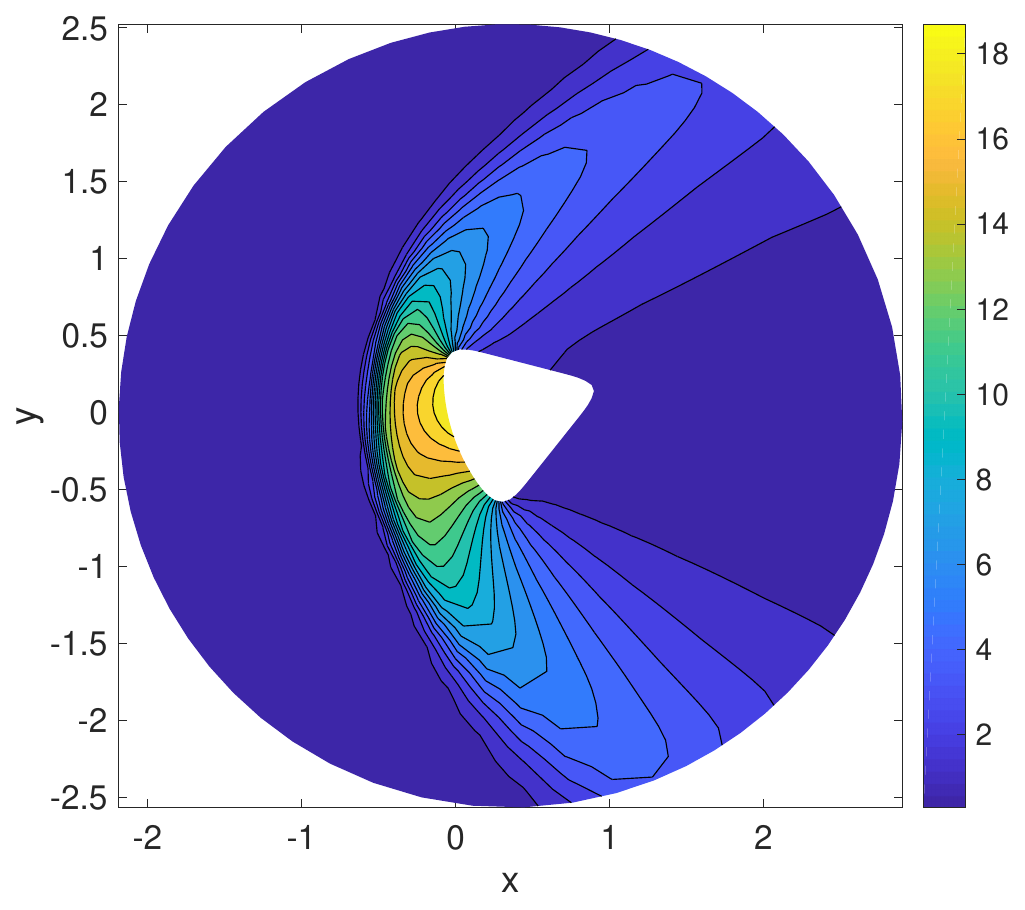}}
\caption{Pressure.}
\end{subfigure}
\caption{Apollo section at $\mathrm{Ma}=5$ and $\varepsilon=10^{-4}$: macroscopic fields of IWPD (filled contours) and NS (lines) on common levels.}\label{fig:apollo_kn1em4_fields}
\end{figure}
\begin{figure}[H]
\centering\singlespacing\setlength{\figw}{.90\linewidth}
\begin{subfigure}[t]{0.9950\figw}\centering
\makebox[\linewidth]{\raisebox{0.0000\figw}{\includegraphics[width=0.4964\figw]{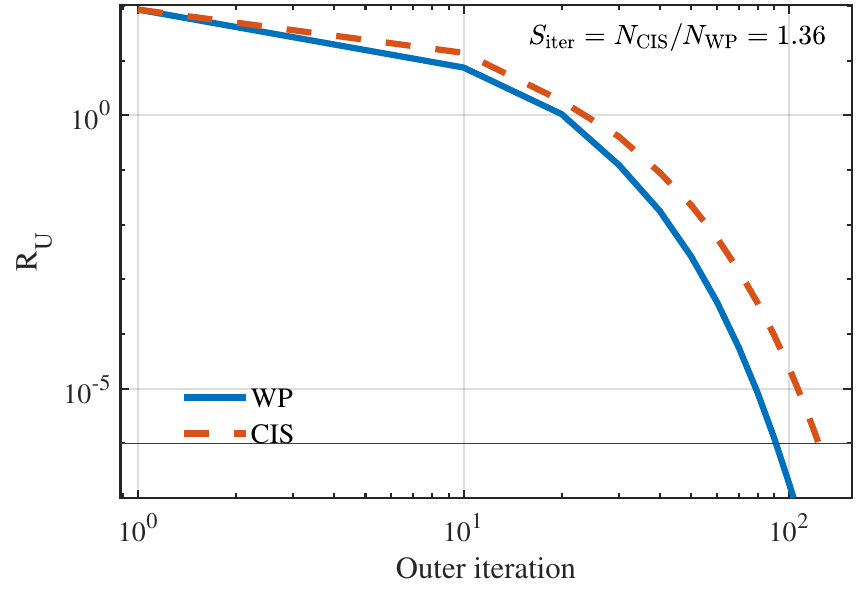}}\hspace{0.0200\figw}\raisebox{0.0031\figw}{\includegraphics[width=0.4786\figw]{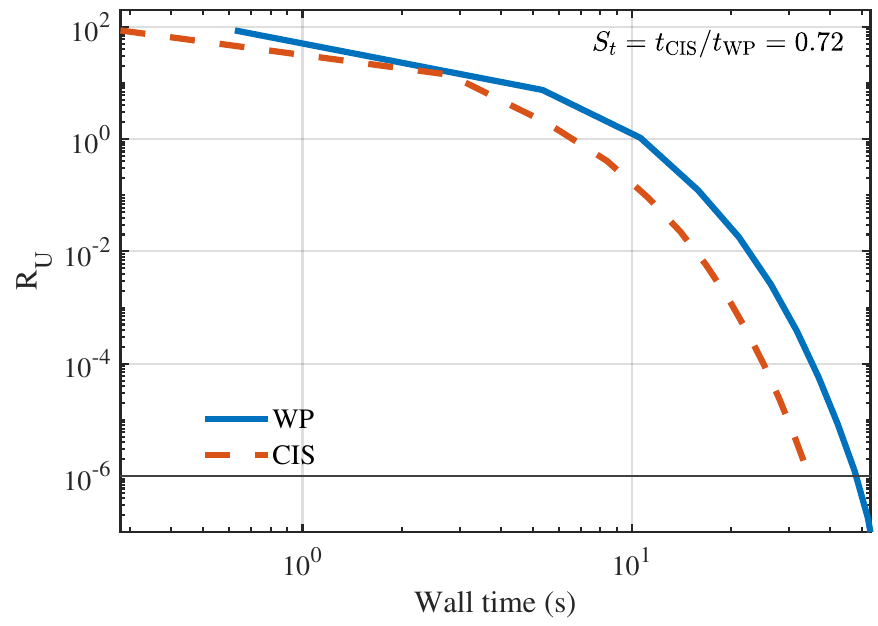}}}
\caption{$\varepsilon=1$.}
\end{subfigure}
\par\smallskip
\begin{subfigure}[t]{0.9950\figw}\centering
\makebox[\linewidth]{\raisebox{0.0000\figw}{\includegraphics[width=0.4944\figw]{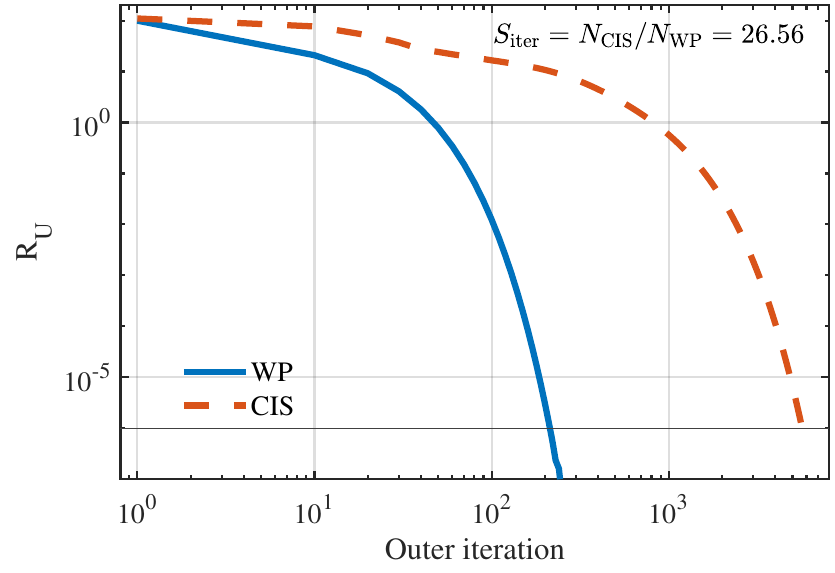}}\hspace{0.0200\figw}\raisebox{0.0029\figw}{\includegraphics[width=0.4806\figw]{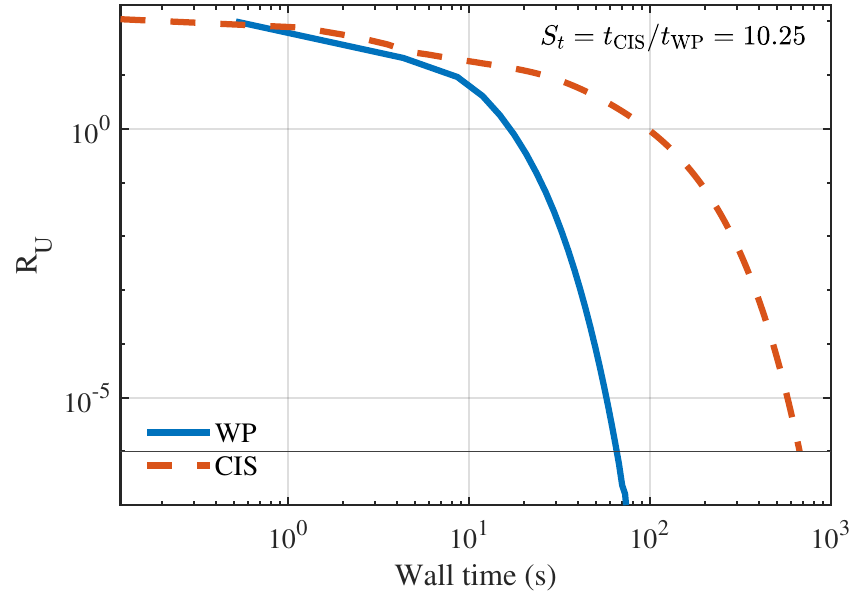}}}
\caption{$\varepsilon=10^{-2}$.}
\end{subfigure}
\par\smallskip
\begin{subfigure}[t]{0.9950\figw}\centering
\makebox[\linewidth]{\raisebox{0.0000\figw}{\includegraphics[width=0.4942\figw]{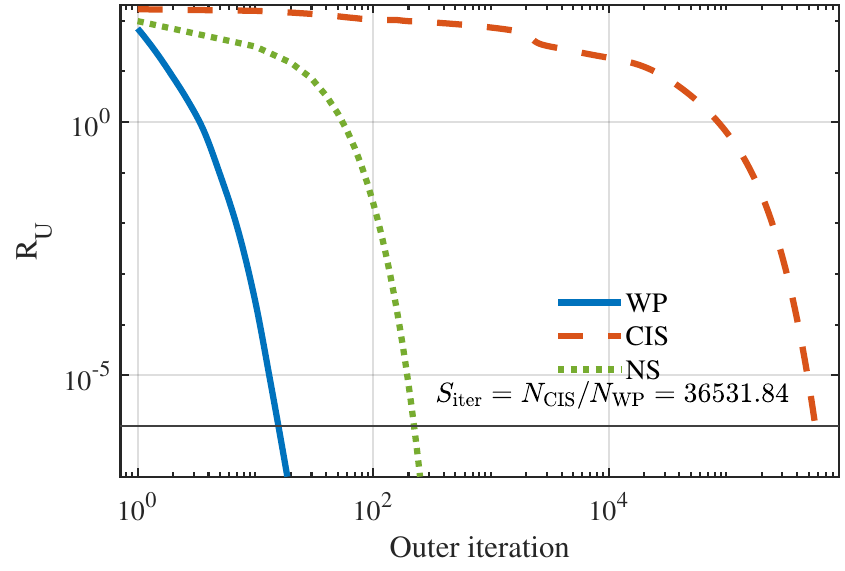}}\hspace{0.0200\figw}\raisebox{0.0071\figw}{\includegraphics[width=0.4808\figw]{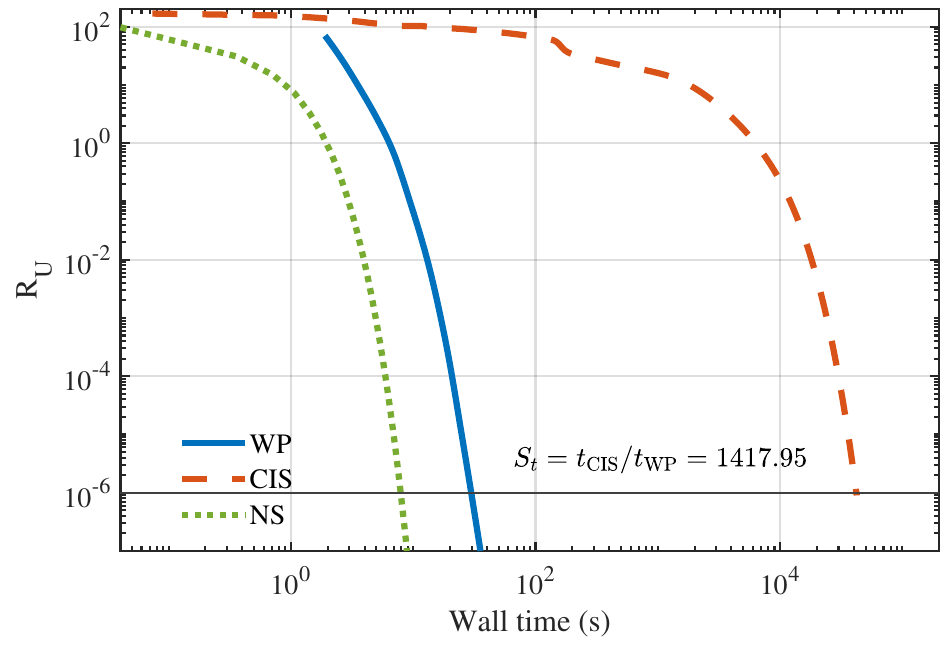}}}
\caption{$\varepsilon=10^{-4}$.}
\end{subfigure}
\caption{Apollo section at $\mathrm{Ma}=5$: macroscopic residual versus iteration (left) and wall time (right) for IWPD (solid), CIS (dashed) and, at $\varepsilon=10^{-4}$, NS (dotted). The grey line marks $\|R_U\|=10^{-6}$, at which $S_{\rm iter}$ and $S_t$ are evaluated.}\label{fig:apollo_residual}
\end{figure}